\documentclass[twocolumn]{autart}    

\usepackage{graphicx}          
  \usepackage{amsmath,amssymb,amsfonts}
\usepackage{graphicx}
\usepackage{textcomp}
\usepackage[compatibility=false]{caption}
\usepackage{subcaption}

\usepackage{algorithm}

\usepackage{algpseudocode}
\usepackage{tikz}
\usepackage{dcolumn}                      
\usepackage[normalem]{ulem}
\usepackage{bbm}
\usepackage{bm}
\usepackage{chngcntr}
\counterwithout{figure}{section}  
\usepackage{nicematrix}

\usepackage{comment}
\usetikzlibrary{positioning}
\definecolor{subsectioncolor}{rgb}{0.067,0.627,0.859}
\newtheorem{proposition}{Proposition}
\newtheorem{lemma}{Lemma}

\newtheorem{theorem}{Theorem}
\newcommand{\myqed}{\hfill \mbox{$\square$}}
\definecolor{mygreen}{HTML}{00F78D}
\newcommand{\highlight}[1]{\colorbox{mygreen!30}{\makebox[0.4em][c]{$#1$}}}

\newtheorem{proof}{Proof}

\newtheorem{definition}{Definition}
\newtheorem{remark}{Remark}
\newtheorem{eg}{Example}
\usepackage[numbers,sort&compress]{natbib}

\begin{document}

\begin{frontmatter}
\title{Structural sign herdability of linear time-invariant systems: theory and design for arbitrary network structures}

\thanks[footnoteinfo]{This paper was not presented at any IFAC 
meeting. \\ Corresponding author: Twinkle~Tripathy}

\author{Pradeep~M}\ead{pradeepm22@iitk.ac.in},    
\author{Twinkle~Tripathy}\ead{ttripathy@iitk.ac.in}  

\address{Department of Electrical Engineering, Indian Institute of Technology Kanpur, Uttar Pradesh, India 208016.}  

\begin{abstract} 
The objective of this paper is to investigate graph-theoretic conditions for structural herdability of an LTI system. In particular, we are interested in the structural sign $(\mathcal{SS})$ herdability of a system wherein the underlying digraph representing it is signed. Structural herdability finds applications in various domains like power networks, biological networks, opinion dynamics, multi-robot shepherding, \textit{etc}. We begin the analysis by introducing a layered graph representation $\mathcal{G}_s$ of the signed digraph $\mathcal{G}$; such a representation allows us to capture the signed distances between the nodes with ease. We construct a subgraph of $\mathcal{G}_s$ that characterizes paths of identical signs between layers and uniform path lengths, referred to as a layer-wise unisigned graph $\mathcal{LUG}(\mathcal{G}_s)$. A special subgraph of an $\mathcal{LUG}(\mathcal{G}_s)$, denoted as an $\mathcal{LUG^{H}}(\mathcal{G}_s)$, is key to achieving $\mathcal{SS}$ herdability. This is because we prove that an LTI system is $\mathcal{SS}$ herdable if and only if there exists an $\mathcal{LUG^{H}}(\mathcal{G}_s)$ which covers all the nodes of the given digraph. To the best of our knowledge, such a graphical test is one of the first methods which allows us to check $\mathcal{SS}$ herdability for arbitrary digraph topologies. Interestingly, the analysis also reveals that a system can be $\mathcal{SS}$ herdable even in the presence of (signed and layer) dilation in the associated digraph (note that such a behaviour has been shown to be impossible in directed trees). Additionally, we also extend these results to digraphs with multiple leader and driver nodes. In order to illustrate all the results, we present numerous examples throughout the paper.
\end{abstract}
\end{frontmatter}
\section{Introduction} 
\label{Sec:1}
 In a networked system, controllability deals with the ability to independently control all nodes in the network \cite{liu2011controllability} and drive them to any desired state. If the system cannot reach a point in the state space using a suitable control input in a finite time, then the system is not completely controllable. How do we address the case when the system lacks complete controllability, but the desired states lie within the reachable space of the nodes? In most practical systems, reaching a certain subset or region in the state space is sufficient rather than analyzing full controllability. For example, in heating systems, the desired values are always positive. Similarly, a negative level in a tank is illogical. The property of a system to reach specifically the positive orthant of $\mathbb{R}^n$ is termed herdability. It has been extensively studied across various domains, including finance and market analysis \cite{devenow1996rational},\cite{welch2000herding},\cite{wermers1999mutual}, cognitive science \cite{raafat2009herding}, social networks \cite{baddeley2013herding}, \cite{hubert2019consensus}. In a networked system, herdability refers to the ability of a particular set of nodes, termed leader nodes, to receive an external input signal and drive the states of the rest of the follower nodes in the same direction or towards a particular subset, where both cooperative and antagonistic interactions coexist among the followers.

One of the earliest works on herdability \cite{ruf2018herdable} proved that input connectivity is a necessary condition for complete herdability. Subsequent studies have investigated herdability based on graph network topologies and structures. In \cite{meng2020leader}, leader selection strategies for herdability in structurally balanced signed graphs were studied using graph walks, with extensions to weakly balanced signed graphs. The study in \cite{de2023herdability} further extended this framework to analyze the herdability of clustering balanced directed graphs with $k$ clusters. Moreover, several algebraic conditions for the herdability of different graph topologies, particularly tree graphs, have been investigated under specific single-leader assumptions. Sufficient conditions were established in \cite{she2019characterizing} to guarantee the herdability of followers in signed networks through 1-walk and 2-walk from the leader. Similarly, in \cite{she2020characterizing}, the controllable subspace of such dynamics was analyzed using a generalized equitable partition. The study revealed that the herdability of the original graph can be inferred from its quotient graph. Recently \cite{shen2025herdability} studied the herdability of switching signed networks, with necessary and sufficient conditions established using union graphs, graph partitions, and switching quotient graphs; a leader selection method is also proposed for structurally balanced cases. In \cite{shen2025target}, it was shown that target herdability is fully characterized by the length and sign of walks from leader nodes to target nodes. Building upon this foundation, later works investigated target herdability in structurally balanced networks and established corresponding analytical conditions.

Although structural controllability has been rigorously studied in the literature, investigations on structural herdability are still in their initial stages.
 In \cite{ruf2019herdability}, several sufficient conditions for the herdability of a system were derived, and a concept analogous to strong structural controllability and sign controllability was investigated within the framework of herdability, termed sign herdability. Sign herdability examines network herdability by leveraging the sign patterns of the associated graphs. Some studies have been conducted on undirected graphs, and research on herdability in digraphs, especially when the digraph is not acyclic, and directed acyclic graphs (non-tree topologies) is limited.  
In this study, we develop a novel graph-theoretic framework for analyzing and ensuring herdability in a structural sense in a directed graph.

Our main contributions are as follows.
\begin{itemize}
 
    \item \textit{Formulation of structural sign herdability}:
    While previous studies \cite{ruf2018herdable, ruf2019herdability,de2023herdability} address herdability based on sign pattern of the graph, we propose a unified structural framework that jointly characterizes the effects of both sign patterns and edge weight on the herdability of directed graphs. We formally define the notion of structural sign $(\mathcal{SS})$ herdability and show that it can be analyzed for digraphs with arbitrary topologies. The proposed framework enables the systematic study of $(\mathcal{SS})$ herdability in large-scale networks with unknown or uncertain edge weights, offering a clearer understanding of how structure influences herdability.
    
      \item \textit{Effects of signed dilation}: In \cite{ruf2018herdable, ruf2019herdability}, it is shown that a graph is not completely herdable if it contains signed dilation. In this work, We introduce a signed dilation set to systematically explore the impact of signed dilation on $(\mathcal{SS})$ herdability and to determine the criteria necessary for $(\mathcal{SS})$ herdability in digraphs with such structural features.
      we also study the effect of layer dilation on $(\mathcal{SS})$ herdability of arbitrary digraphs, a notion that has been appeared in literature \cite{she2019characterizing,de2023herdability}, particularly for tree graphs.

    \item \textit{ Significance of $\mathcal{LUG^{H}}(\mathcal{G}_s)$ in $\mathcal{SS}$ herdability}: We propose a layerwise unisigned graph structure $(\mathcal{LUG})$ and sign-matching of nodes with respect to edges closely related to the study in \cite{liu2011controllability}. Together, they aid in the formulation of $\mathcal{LUG^{H}}(\mathcal{G}_s)$, a version of $(\mathcal{LUG})$ using which we derive the necessary and sufficient conditions for $(\mathcal{SS})$ herdability. Through rigorous analysis we test and confirm our conclusions. 
    
   \item \textit{$\mathcal{SS}$ herdability analysis in Multi-Leader Networks}: This work is further extended to digraphs with multiple leaders, where each leader receives the input $u(t)$. Subsequently, we investigate an extended scenario in which multiple driver nodes are each associated with distinct control inputs.
\end{itemize}
    
The remainder of this paper is organized as follows. In Section \ref{Sec:2}, we describe the notation and review the notions of structural controllability and herdability. Section \ref{Sec:3} presents the problem. In Section \ref{Sec:4}, we define the structural sign herdability and sign-structured controllability matrix. In this section, we also propose a signed layered graph. In section \ref{Sec:5} we introduce the signed dilation set and provide the necessary and sufficient conditions for $\mathcal{SS}$ herdability of arbitrary digraphs, and discuss the implications. In Section \ref{Sec:6}, we extend these results to the herdability of digraphs with multiple leaders. Finally, Section \ref{Sec:7} concludes the paper and discusses future work.

\section{Notation and Background} 
\label{Sec:2}
\subsection{Notation and Matrix Theory}\label{Notation}
The set of real numbers is denoted as $\mathbb{R}$ and the set of integers is denoted by $\mathbb{Z}$. $\mathbb{R}^+$ and $\mathbb{Z}^+$ represent nonnegative real numbers and nonnegative integer number sets, respectively.
  For a vector $\bm{k}$, $[k]_i$ denotes $i^{th}$ entry of the vector. Let $v = [v_1, v_2, \dots, v_n]^\top \in \mathbb{R}^n$. For a chosen index set $I = \{i_1, i_2, \dots, i_k\}$ with $i_1 < i_2 < \cdots < i_k$, the corresponding subvector is \(\mathcal{D} = [v_{i_1}, v_{i_2},\dots, v_{i_k}]^\top\). A vector is said to be unisigned if every nonzero entry is either nonnegative or nonpositive. Note that a vector with only one nonzero entry is trivially unisigned. The matrix $\mathcal{A} \in \mathbb{R}^{n\times n}$ is nonnegative (respectively positive) if $a_{ij}\geq0$ $(a_{ij}>0)$ for $i,j=\{1,2, \dots n\}$. $\mathcal{A}_{(i,j)}$ represents the $(i,j)^{th}$ entry of the matrix $\mathcal{A}$ and $\mathcal{A}_{(:,j)}$ refers to $j^{th}$  column while $\mathcal{A}_{(i,:)}$ refers to $i^{th}$ row of the matrix, $\mathcal{A}$. The image of the matrix $\mathcal{A}$ is given by $\mathcal{I}m(\mathcal{A} )$ = $\{ y~|~y= \mathcal{A} ~v \}$. Let $\mathfrak{S}(x)$ be the sign pattern of $x$, where $x$ may be a matrix or an entry of a matrix. Let $\lfloor x \rfloor$ denote the greatest integer less than or equal to $x$ (often referred to as the staircase function), and $x \bmod n$ gives the remainder $r$ when $x$ is divided by $n$, where $0 \le r < n$.
\subsection{Graph Theory}\label{Graph theory}
Let $\mathcal{G=(V, E, A)}$ represent the weighted signed digraph of the network, where $\mathcal{V}$ is the set of nodes and $\mathcal{E \subseteq V \times V}$ denotes the edge set. Each element of $\mathcal{E}$ can be represented by $(i,j)$, which means that there is a directed edge from node $ i$ to $ j $ in the digraph $\mathcal{G(A,B)}$.An entry $a_{ij} \ne 0$ indicates that $(j,i) \in \mathcal{E}$, whereas $a_{ij} = 0$ implies that there is no edge between nodes $i$ and $j$. An entry $a_{ij} > 0$ (resp. $a_{ij} < 0$) indicates that $(j, i) \in \mathcal{E}_+$ (resp. $(j, i) \in \mathcal{E}_-)$.
A walk in a digraph $\mathcal{G(A,B)}$ is a sequence of directed edges that successively connects an initial node to the final node. A path is a walk in which no vertex is repeated. A cycle in a digraph is a closed walk in the digraph. Let ${w}^r_{(i,j)}$ be the $r^{\text{th}}$ walk from node $i$ to $j$. Similarly $\{\mathcal{W}^r_{(i,j)}\}$ be the product of the edge weights in  ${w}^r_{(i,j)}$. Let \( \mathcal{V}(X) \) denote the set of nodes in \( X \).
\subsection{Layered Graph}\label{LG}
A layered graph is a tree-structured representation of a digraph $\mathcal{G(A,B)}$. It is isomorphic to $\mathcal{G(A,B)}$ and portrays the distances between the selected root node and all other nodes. A layered graph comprises all possible walks from the root node to any other node in a digraph. The application of layered graphs extends to search algorithms, network flows, Euler paths, and discrete optimization. In a layered graph, $L_p$ represents layer $p$ and $V_{L_p}$ is the set of nodes in layer $p$,~$p\in \mathbb{Z}^+$. If there are no edges between two consecutive layers, the graph is disconnected. The cyclic portion of the graph can also be represented as a layered graph in which the nodes and edges repeat. Hence, the layers can be unfolded into any desired number of iterations.
\subsection{Controllability and Structural Controllability} \label{Str.con}
The state-space representation of an LTI system is given by the following equation:
\begin{equation} \label{sys1:sys1eqn}
    \dot{x}(t) = \mathcal{A} x(t) + \mathcal{B} u(t), \quad t \in \mathbb{R}^+.
\end{equation}
where $x \in \mathbb{R}^n$ and $u \in \mathbb{R}^m $ are the state and input vectors to the system, respectively. The matrices $\mathcal{A}\in\mathbb{R}^{n\times n}$ and $\mathcal{B}\in\mathbb{R}^{n\times m}$ represent the system matrix and the input matrix, respectively. An LTI system is said to be controllable if the columns of the controllability matrix $\mathcal{C}(\mathcal{A,B})$ span the entire $\mathbb{R}^n$ space; thus, it can reach any point in the state space. 
 If the pair $(\mathcal{A},\mathcal{B})$ is not completely controllable, it is structurally controllable if there exists another pair $(\mathcal{\Bar{A}},\mathcal{\Bar{B}})$ of the same structure that is completely controllable \cite{lin1974structural}. If the system is completely controllable for all realizations, then it is called strongly structurally controllable.  
\subsection{Herdability}\label{herd_para}  
As discussed above, analyzing complete controllability is irrelevant in many physical systems. In such systems, reaching a specific subset in the state space is both practical and targeted. Ideally, the analysis focuses on reaching the positive orthant of the state space.
\begin{definition}
A system is said to be herdable if, for every initial condition $x(0)$, there exists an input $u(t)$ such that its states can be brought to the positive orthant in a finite time $t_f$.
\end{definition}

The system $(\mathcal{A},\mathcal{B})$ is said to be (strictly) herdable if for every \( x(0) \) and every \( d > 0 \), there exists a time \( t_f > 0 \) and an input \( u(t), \, t \in [0, t_f) \), that drives the state of the system from \( x(0) \) to  \( x(t_f)\) \(\geq\) \(d \) \(\mathbf{1}_n \). 
For an LTI system, the existence of a strictly positive image in the range space of the controllability matrix guarantees the herdability of the system.
As shown in Fig.~\ref{subspace}, even if the graph is not fully controllable, it is still herdable.
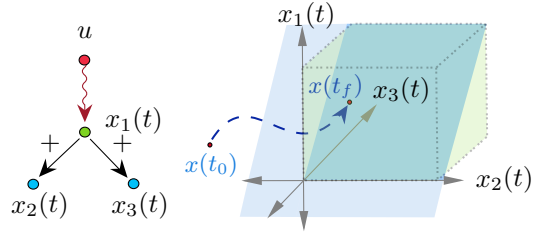
\begin{figure}[ht]
    \centering
\tikzset{every picture/.style={scale=0.75pt}} 

\begin{tikzpicture}[x=0.85pt,y=0.75pt,yscale=-1,xscale=1]

\draw  [fill={rgb, 255:red, 126; green, 211; blue, 33 }  ,fill opacity=1 ] (61.19,91.34) .. controls (61.19,89.52) and (62.58,88.05) .. (64.3,88.05) .. controls (66.02,88.05) and (67.41,89.52) .. (67.41,91.34) .. controls (67.41,93.16) and (66.02,94.63) .. (64.3,94.63) .. controls (62.58,94.63) and (61.19,93.16) .. (61.19,91.34) -- cycle ;
\draw  [fill={rgb, 255:red, 0; green, 192; blue, 248 }  ,fill opacity=1 ] (89.8,125.95) .. controls (89.8,124.14) and (91.19,122.66) .. (92.91,122.66) .. controls (94.63,122.66) and (96.02,124.14) .. (96.02,125.95) .. controls (96.02,127.77) and (94.63,129.25) .. (92.91,129.25) .. controls (91.19,129.25) and (89.8,127.77) .. (89.8,125.95) -- cycle ;
\draw  [fill={rgb, 255:red, 0; green, 192; blue, 248 }  ,fill opacity=1 ] (30.87,126.28) .. controls (30.87,124.47) and (32.26,122.99) .. (33.98,122.99) .. controls (35.7,122.99) and (37.09,124.47) .. (37.09,126.28) .. controls (37.09,128.1) and (35.7,129.57) .. (33.98,129.57) .. controls (32.26,129.57) and (30.87,128.1) .. (30.87,126.28) -- cycle ;
\draw [color={rgb, 255:red, 161; green, 27; blue, 43 }  ,draw opacity=1 ]   (64.02,43.03) .. controls (65.7,44.68) and (65.71,46.35) .. (64.05,48.03) .. controls (62.4,49.71) and (62.41,51.38) .. (64.09,53.03) .. controls (65.76,54.69) and (65.77,56.36) .. (64.12,58.03) .. controls (62.47,59.7) and (62.48,61.37) .. (64.15,63.03) .. controls (65.83,64.68) and (65.84,66.35) .. (64.19,68.03) .. controls (62.54,69.7) and (62.55,71.37) .. (64.22,73.03) -- (64.23,73.76) -- (64.28,81.76) ;
\draw [shift={(64.3,84.76)}, rotate = 269.62] [fill={rgb, 255:red, 161; green, 27; blue, 43 }  ,fill opacity=1 ][line width=0.08]  [draw opacity=0] (10.72,-5.15) -- (0,0) -- (10.72,5.15) -- (7.12,0) -- cycle    ;
\draw    (67.36,95.97) -- (87.45,119.74) ;
\draw [shift={(89.38,122.03)}, rotate = 229.79] [fill={rgb, 255:red, 0; green, 0; blue, 0 }  ][line width=0.08]  [draw opacity=0] (10.72,-5.15) -- (0,0) -- (10.72,5.15) -- (7.12,0) -- cycle    ;
\draw    (61.14,95.7) -- (39.33,119.21) ;
\draw [shift={(37.3,121.41)}, rotate = 312.84] [fill={rgb, 255:red, 0; green, 0; blue, 0 }  ][line width=0.08]  [draw opacity=0] (10.72,-5.15) -- (0,0) -- (10.72,5.15) -- (7.12,0) -- cycle    ;
\draw  [fill={rgb, 255:red, 239; green, 43; blue, 67 }  ,fill opacity=1 ] (60.91,43.03) .. controls (60.91,41.21) and (62.3,39.74) .. (64.02,39.74) .. controls (65.74,39.74) and (67.13,41.21) .. (67.13,43.03) .. controls (67.13,44.85) and (65.74,46.32) .. (64.02,46.32) .. controls (62.3,46.32) and (60.91,44.85) .. (60.91,43.03) -- cycle ;
\draw  [color={rgb, 255:red, 0; green, 0; blue, 0 }  ,draw opacity=0 ][fill={rgb, 255:red, 74; green, 144; blue, 226 }  ,fill opacity=0.19 ] (218.2,18.67) -- (182.7,19.52) -- (159.97,124.59) -- (192.79,123.66) -- cycle ;
\draw  [color={rgb, 255:red, 0; green, 0; blue, 0 }  ,draw opacity=0 ][fill={rgb, 255:red, 74; green, 144; blue, 226 }  ,fill opacity=0.19 ] (192.79,123.66) -- (159.97,124.59) -- (155.24,148.14) -- (188.19,148.29) -- cycle ;
\draw  [color={rgb, 255:red, 0; green, 0; blue, 0 }  ,draw opacity=0 ][fill={rgb, 255:red, 74; green, 144; blue, 226 }  ,fill opacity=0.19 ] (274.83,123.71) -- (192.79,123.66) -- (188.19,148.29) -- (269.56,148.53) -- cycle ;
\draw  [color={rgb, 255:red, 0; green, 0; blue, 0 }  ,draw opacity=0 ][fill={rgb, 255:red, 0; green, 124; blue, 248 }  ,fill opacity=0.37 ] (300.24,18.72) -- (218.2,18.67) -- (192.79,123.66) -- (274.83,123.71) -- cycle ;
\draw  [fill={rgb, 255:red, 208; green, 2; blue, 27 }  ,fill opacity=1 ] (136.56,99.96) .. controls (136.55,99.19) and (137.07,98.57) .. (137.71,98.56) .. controls (138.35,98.55) and (138.88,99.17) .. (138.89,99.93) .. controls (138.89,100.69) and (138.38,101.31) .. (137.74,101.32) .. controls (137.09,101.33) and (136.57,100.72) .. (136.56,99.96) -- cycle ;
\draw [color={rgb, 255:red, 128; green, 128; blue, 128 }  ,draw opacity=1 ]   (158.21,123.92) -- (285.03,123.69) ;
\draw [shift={(287.03,123.69)}, rotate = 179.9] [fill={rgb, 255:red, 128; green, 128; blue, 128 }  ,fill opacity=1 ][line width=0.08]  [draw opacity=0] (12,-3) -- (0,0) -- (12,3) -- cycle    ;
\draw [shift={(156.21,123.92)}, rotate = 359.9] [fill={rgb, 255:red, 128; green, 128; blue, 128 }  ,fill opacity=1 ][line width=0.08]  [draw opacity=0] (12,-3) -- (0,0) -- (12,3) -- cycle    ;
\draw [color={rgb, 255:red, 128; green, 128; blue, 128 }  ,draw opacity=1 ]   (192.91,155.88) -- (191.95,22.04) ;
\draw [shift={(191.93,20.04)}, rotate = 89.59] [fill={rgb, 255:red, 128; green, 128; blue, 128 }  ,fill opacity=1 ][line width=0.08]  [draw opacity=0] (12,-3) -- (0,0) -- (12,3) -- cycle    ;
\draw [shift={(192.93,157.88)}, rotate = 269.59] [fill={rgb, 255:red, 128; green, 128; blue, 128 }  ,fill opacity=1 ][line width=0.08]  [draw opacity=0] (12,-3) -- (0,0) -- (12,3) -- cycle    ;
\draw [color={rgb, 255:red, 128; green, 128; blue, 128 }  ,draw opacity=1 ]   (234.53,74.14) -- (172.69,147.75) ;
\draw [shift={(171.4,149.28)}, rotate = 310.04] [fill={rgb, 255:red, 128; green, 128; blue, 128 }  ,fill opacity=1 ][line width=0.08]  [draw opacity=0] (12,-3) -- (0,0) -- (12,3) -- cycle    ;
\draw [shift={(235.82,72.61)}, rotate = 130.04] [fill={rgb, 255:red, 128; green, 128; blue, 128 }  ,fill opacity=1 ][line width=0.08]  [draw opacity=0] (12,-3) -- (0,0) -- (12,3) -- cycle    ;
\draw [color={rgb, 255:red, 8; green, 39; blue, 164 }  ,draw opacity=1 ][line width=0.75]  [dash pattern={on 4.5pt off 4.5pt}]  (138.9,96.28) .. controls (163.55,52.26) and (185.91,122.08) .. (217.23,75.45) ;
\draw [shift={(218.67,73.24)}, rotate = 122.18] [fill={rgb, 255:red, 8; green, 39; blue, 164 }  ,fill opacity=1 ][line width=0.08]  [draw opacity=0] (10.72,-5.15) -- (0,0) -- (10.72,5.15) -- (7.12,0) -- cycle    ;
\draw  [fill={rgb, 255:red, 208; green, 2; blue, 27 }  ,fill opacity=1 ] (218.58,71.26) .. controls (218.57,70.5) and (219.09,69.88) .. (219.73,69.87) .. controls (220.37,69.86) and (220.9,70.47) .. (220.91,71.23) .. controls (220.91,72) and (220.4,72.62) .. (219.76,72.63) .. controls (219.11,72.64) and (218.59,72.02) .. (218.58,71.26) -- cycle ;
\draw  [color={rgb, 255:red, 0; green, 0; blue, 0 }  ,draw opacity=0.27 ][fill={rgb, 255:red, 184; green, 233; blue, 134 }  ,fill opacity=0.29 ][dash pattern={on 0.75pt off 0.75pt on 0.75pt off 1.5pt}][line width=0.75]  (191.97,48.05) -- (221.06,18.66) -- (299.7,18.71) -- (300.52,94.32) -- (271.43,123.7) -- (192.79,123.65) -- cycle ; \draw  [color={rgb, 255:red, 0; green, 0; blue, 0 }  ,draw opacity=0.27 ][dash pattern={on 0.75pt off 0.75pt on 0.75pt off 1.5pt}][line width=0.75]  (299.7,18.71) -- (270.61,48.1) -- (191.97,48.05) ; \draw  [color={rgb, 255:red, 0; green, 0; blue, 0 }  ,draw opacity=0.27 ][dash pattern={on 0.75pt off 0.75pt on 0.75pt off 1.5pt}][line width=0.75]  (270.61,48.1) -- (271.43,123.7) ;

\draw (58.22,18.93) node [anchor=north west][inner sep=0.75pt]  [font=\normalsize]  {$u$};
\draw (74.6,72.98) node [anchor=north west][inner sep=0.75pt]  [color={rgb, 255:red, 0; green, 0; blue, 0 }  ,opacity=1 ,rotate=-359.39]  {$x_{1}( t)$};
\draw (19.38,130.29) node [anchor=north west][inner sep=0.75pt]  [color={rgb, 255:red, 0; green, 0; blue, 0 }  ,opacity=1 ,rotate=-359.39]  {$x_{2}( t)$};
\draw (78.21,130.28) node [anchor=north west][inner sep=0.75pt]  [color={rgb, 255:red, 0; green, 0; blue, 0 }  ,opacity=1 ,rotate=-359.39]  {$x_{3}( t)$};
\draw (121.4,104.17) node [anchor=north west][inner sep=0.75pt]  [font=\footnotesize,color={rgb, 255:red, 22; green, 125; blue, 223 }  ,opacity=1 ,rotate=-359.39]  {$x( t_{0})$};
\draw (176,2.81) node [anchor=north west][inner sep=0.75pt]  [color={rgb, 255:red, 0; green, 0; blue, 0 }  ,opacity=1 ,rotate=-359.39]  {$x_{1}( t)$};
\draw (292.27,112.93) node [anchor=north west][inner sep=0.75pt]  [color={rgb, 255:red, 0; green, 0; blue, 0 }  ,opacity=1 ,rotate=-359.39]  {$x_{2}( t)$};
\draw (231.79,53.17) node [anchor=north west][inner sep=0.75pt]  [color={rgb, 255:red, 0; green, 0; blue, 0 }  ,opacity=1 ,rotate=-359.39]  {$x_{3}( t)$};
\draw (194.01,51.42) node [anchor=north west][inner sep=0.75pt]  [font=\footnotesize,color={rgb, 255:red, 34; green, 83; blue, 175 }  ,opacity=1 ,rotate=-359.39]  {$x( t_{f})$};
\draw (36.43,91.44) node [anchor=north west][inner sep=0.75pt]    {$+$};
\draw (78.82,91.44) node [anchor=north west][inner sep=0.75pt]    {$+$};

\end{tikzpicture}
    \caption{An uncontrollable graph can still be herdable.}
    \label{subspace}
\end{figure}
 This is possible if a subspace spanned by the columns of the controllability matrix extends to the positive orthant of the state space. Hence, herdability is considered a set-based reachability, which is a more relaxed version of controllability.
   
\subsection{Dilation and Signed Dilation}
\subsubsection{Dilation \cite{lin1974structural}}\label{dil para} A graph becomes structurally uncontrollable if the graph contains a dilation or an inaccessible node, resulting in form I or form II representations, as demonstrated in \cite{lin1974structural}. It is evident that the system cannot be fully controlled when an inaccessible node is involved, as shown in Fig.(\ref{fig:inaccess}). When a graph contains dilation, as shown in Fig. (\ref{fig: dilation}), the leader node cannot independently govern two or more nodes. 
\subsubsection{Signed Dilation \cite{ruf2018herdable}}\label{SD para}
Similarly, a graph is said to have a signed dilation if it consists of a node whose outgoing edges have different signs, as shown in Fig.(\ref{fig: signed dilation}). With a signed dilation, the leader node cannot drive the followers in the same direction simultaneosly, thereby affecting the herdability of the system.
\begin{figure}[ht]
  \centering
   \begin{subfigure}{0.17\textwidth}
  \centering
\tikzset{every picture/.style={scale=0.75pt}} 

\begin{tikzpicture}[x=0.75pt,y=0.75pt,yscale=-1,xscale=1]

\draw    (132.56,65.4) -- (117.02,89.62) ;
\draw [shift={(115.39,92.14)}, rotate = 302.7] [fill={rgb, 255:red, 0; green, 0; blue, 0 }  ][line width=0.08]  [draw opacity=0] (10.72,-5.15) -- (0,0) -- (10.72,5.15) -- (7.12,0) -- cycle    ;
\draw  [draw opacity=0][fill={rgb, 255:red, 126; green, 211; blue, 33 }  ,fill opacity=1 ] (131.6,60.25) .. controls (131.6,58.13) and (133.5,56.41) .. (135.85,56.41) .. controls (138.19,56.41) and (140.1,58.13) .. (140.1,60.25) .. controls (140.1,62.36) and (138.19,64.08) .. (135.85,64.08) .. controls (133.5,64.08) and (131.6,62.36) .. (131.6,60.25) -- cycle ;
\draw  [draw opacity=0][fill={rgb, 255:red, 0; green, 192; blue, 248 }  ,fill opacity=1 ] (155.93,99.71) .. controls (155.93,97.38) and (158.03,95.48) .. (160.62,95.48) .. controls (163.21,95.48) and (165.31,97.38) .. (165.31,99.71) .. controls (165.31,102.04) and (163.21,103.93) .. (160.62,103.93) .. controls (158.03,103.93) and (155.93,102.04) .. (155.93,99.71) -- cycle ;
\draw  [draw opacity=0][fill={rgb, 255:red, 0; green, 192; blue, 248 }  ,fill opacity=1 ] (106.56,99.43) .. controls (106.56,97.1) and (108.65,95.21) .. (111.24,95.21) .. controls (113.83,95.21) and (115.93,97.1) .. (115.93,99.43) .. controls (115.93,101.77) and (113.83,103.66) .. (111.24,103.66) .. controls (108.65,103.66) and (106.56,101.77) .. (106.56,99.43) -- cycle ;

\draw (130.61,38.38) node [anchor=north west][inner sep=0.75pt]  [color={rgb, 255:red, 0; green, 0; blue, 0 }  ,opacity=1 ,rotate=-359.39]  {$1$};
\draw (156.02,106.46) node [anchor=north west][inner sep=0.75pt]  [color={rgb, 255:red, 0; green, 0; blue, 0 }  ,opacity=1 ,rotate=-359.39]  {$3$};
\draw (105.05,106.44) node [anchor=north west][inner sep=0.75pt]  [color={rgb, 255:red, 0; green, 0; blue, 0 }  ,opacity=1 ,rotate=-359.39]  {$2$};

\end{tikzpicture}
    \caption{Inaccessible node}
    \label{fig:inaccess}
  \end{subfigure}
  \begin{subfigure}{0.137\textwidth}
    \centering
\tikzset{every picture/.style={scale=0.75pt}} 

\begin{tikzpicture}[x=0.75pt,y=0.75pt,yscale=-1,xscale=1]

\draw    (139.29,65.81) -- (154.32,90.09) ;
\draw [shift={(155.89,92.64)}, rotate = 238.24] [fill={rgb, 255:red, 0; green, 0; blue, 0 }  ][line width=0.08]  [draw opacity=0] (10.72,-5.15) -- (0,0) -- (10.72,5.15) -- (7.12,0) -- cycle    ;
\draw    (132.56,65.4) -- (117.02,89.62) ;
\draw [shift={(115.39,92.14)}, rotate = 302.7] [fill={rgb, 255:red, 0; green, 0; blue, 0 }  ][line width=0.08]  [draw opacity=0] (10.72,-5.15) -- (0,0) -- (10.72,5.15) -- (7.12,0) -- cycle    ;
\draw  [draw opacity=0][fill={rgb, 255:red, 126; green, 211; blue, 33 }  ,fill opacity=1 ] (131.6,60.25) .. controls (131.6,58.13) and (133.5,56.41) .. (135.85,56.41) .. controls (138.19,56.41) and (140.1,58.13) .. (140.1,60.25) .. controls (140.1,62.36) and (138.19,64.08) .. (135.85,64.08) .. controls (133.5,64.08) and (131.6,62.36) .. (131.6,60.25) -- cycle ;
\draw  [draw opacity=0][fill={rgb, 255:red, 0; green, 192; blue, 248 }  ,fill opacity=1 ] (155.93,99.71) .. controls (155.93,97.38) and (158.03,95.48) .. (160.62,95.48) .. controls (163.21,95.48) and (165.31,97.38) .. (165.31,99.71) .. controls (165.31,102.04) and (163.21,103.93) .. (160.62,103.93) .. controls (158.03,103.93) and (155.93,102.04) .. (155.93,99.71) -- cycle ;
\draw  [draw opacity=0][fill={rgb, 255:red, 0; green, 192; blue, 248 }  ,fill opacity=1 ] (106.56,99.43) .. controls (106.56,97.1) and (108.65,95.21) .. (111.24,95.21) .. controls (113.83,95.21) and (115.93,97.1) .. (115.93,99.43) .. controls (115.93,101.77) and (113.83,103.66) .. (111.24,103.66) .. controls (108.65,103.66) and (106.56,101.77) .. (106.56,99.43) -- cycle ;

\draw (130.61,38.38) node [anchor=north west][inner sep=0.75pt]  [color={rgb, 255:red, 0; green, 0; blue, 0 }  ,opacity=1 ,rotate=-359.39]  {$1$};
\draw (156.02,106.46) node [anchor=north west][inner sep=0.75pt]  [color={rgb, 255:red, 0; green, 0; blue, 0 }  ,opacity=1 ,rotate=-359.39]  {$3$};
\draw (105.05,106.44) node [anchor=north west][inner sep=0.75pt]  [color={rgb, 255:red, 0; green, 0; blue, 0 }  ,opacity=1 ,rotate=-359.39]  {$2$};

\end{tikzpicture}
    \caption{Dilation}
    \label{fig: dilation}
  \end{subfigure}
  \begin{subfigure}{0.16\textwidth}
    \centering 

\tikzset{every picture/.style={scale=0.75pt}} 

\begin{tikzpicture}[x=0.75pt,y=0.75pt,yscale=-1,xscale=1]

\draw    (139.29,65.81) -- (154.32,90.09) ;
\draw [shift={(155.89,92.64)}, rotate = 238.24] [fill={rgb, 255:red, 0; green, 0; blue, 0 }  ][line width=0.08]  [draw opacity=0] (10.72,-5.15) -- (0,0) -- (10.72,5.15) -- (7.12,0) -- cycle    ;
\draw    (132.56,65.4) -- (117.02,89.62) ;
\draw [shift={(115.39,92.14)}, rotate = 302.7] [fill={rgb, 255:red, 0; green, 0; blue, 0 }  ][line width=0.08]  [draw opacity=0] (10.72,-5.15) -- (0,0) -- (10.72,5.15) -- (7.12,0) -- cycle    ;
\draw  [draw opacity=0][fill={rgb, 255:red, 126; green, 211; blue, 33 }  ,fill opacity=1 ] (131.6,60.25) .. controls (131.6,58.13) and (133.5,56.41) .. (135.85,56.41) .. controls (138.19,56.41) and (140.1,58.13) .. (140.1,60.25) .. controls (140.1,62.36) and (138.19,64.08) .. (135.85,64.08) .. controls (133.5,64.08) and (131.6,62.36) .. (131.6,60.25) -- cycle ;
\draw  [draw opacity=0][fill={rgb, 255:red, 0; green, 192; blue, 248 }  ,fill opacity=1 ] (155.93,99.71) .. controls (155.93,97.38) and (158.03,95.48) .. (160.62,95.48) .. controls (163.21,95.48) and (165.31,97.38) .. (165.31,99.71) .. controls (165.31,102.04) and (163.21,103.93) .. (160.62,103.93) .. controls (158.03,103.93) and (155.93,102.04) .. (155.93,99.71) -- cycle ;
\draw  [draw opacity=0][fill={rgb, 255:red, 0; green, 192; blue, 248 }  ,fill opacity=1 ] (106.56,99.43) .. controls (106.56,97.1) and (108.65,95.21) .. (111.24,95.21) .. controls (113.83,95.21) and (115.93,97.1) .. (115.93,99.43) .. controls (115.93,101.77) and (113.83,103.66) .. (111.24,103.66) .. controls (108.65,103.66) and (106.56,101.77) .. (106.56,99.43) -- cycle ;

\draw (130.61,38.38) node [anchor=north west][inner sep=0.75pt]  [color={rgb, 255:red, 0; green, 0; blue, 0 }  ,opacity=1 ,rotate=-359.39]  {$1$};
\draw (156.02,106.46) node [anchor=north west][inner sep=0.75pt]  [color={rgb, 255:red, 0; green, 0; blue, 0 }  ,opacity=1 ,rotate=-359.39]  {$3$};
\draw (105.05,106.44) node [anchor=north west][inner sep=0.75pt]  [color={rgb, 255:red, 0; green, 0; blue, 0 }  ,opacity=1 ,rotate=-359.39]  {$2$};
\draw (106.5,65.4) node [anchor=north west][inner sep=0.75pt]    {$-$};
\draw (152.5,65.4) node [anchor=north west][inner sep=0.75pt]    {$+$};

\end{tikzpicture}
    \caption{Signed dilation}
    \label{fig: signed dilation}
  \end{subfigure}
  \caption{Types of dilation}
\end{figure}
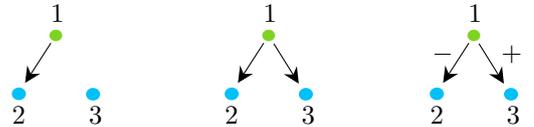
\begin{remark} \label{Remark1}
    Let $(\mathcal{A}, \mathcal{B})$ denote the system matrices corresponding to~\eqref{sys1:sys1eqn}. Assume that node 1 is the leader node that receives external input $u(t)$. Hence $\mathcal{B}\in \mathbb{R}^n$. Then, the controllability matrix of the pair \( (A, B) \) is given by
    \begin{equation} \label{Control_mat}
\mathcal{C}(\mathcal{A,B}) = \begin{bmatrix}
~~\mathcal{B} ~|~\Psi_1~|~\Psi_2~|....~|\Psi_{n-1}
\end{bmatrix}
\end{equation}
\normalfont
where  \( [\Psi_k] =[\mathcal{A}^k\mathcal{B}]\), $k\in \{1,2, \dots n-1]$ and \(\mathcal{B} = \begin{bmatrix} b_1 & 0 & 0 & 0 & \cdots \end{bmatrix}^\top \in \mathbb{R}^n\). Here, $b_ 1$ is the strength of the input signal received by the leader node. Without loss of generality, the sign of $b_1$ is assumed to be positive for the remainder of this study.
 If the $j^{th}$ entry of $[\Psi_k]$, $[\Psi_k]_j$ is not equal to $0$, then there exists at least one walk from the leader to node $j$ of length $k$. 
\end{remark}
\begin{proposition}\cite{ruf2019herdability}\label{prop 1 follo}
    A follower \( j \) is herdable if there exists at least one \( k \) such that \([A^k B]_j \neq 0\). 
    \end{proposition}
\begin{lemma} \cite{ruf2019herdability} \label{lemma 19}
A state $i$ is herdable if there exists a column $c_j$ of the controllability matrix $\mathcal{C}$ such that  \([\mathcal{C}]_{i,j} \neq 0 \quad \text{and} \quad 
\{k \neq i : [\mathcal{C}]_{k,j} \neq 0\} \subseteq H,\), where $H$ is the set of already herdable states.  
 \end{lemma}
\begin{theorem}\label{gordan thm}(Gordan's Theorem {\cite{gordan1873ueber,dantzig2016linear}})
Let $A \in \mathbb{R}^{m \times n}$. Then exactly one of the following statements holds:
\begin{enumerate}
    \item There exists $\bm{x} \in \mathbb{R}^n$ such that $A \bm{x} \gg 0$.
    \item There exists $\bm{y} \in \mathbb{R}^m$, with $\bm{y} \ge 0$, such that $A^{\top} \bm{y} = 0$.
\end{enumerate}
\end{theorem}
\section{Motivation and Problem Statement}\label{Sec:3} 

Consider $\mathit{N}$ agents whose communication is described by a weighted directed signed graph that can be integrated in the form of a state equation (1). $\mathcal{A}$ matrix of the system is chosen to be the adjacency matrix of the associated network. The nodes receiving external input are considered leaders, whereas the other nodes are followers. A graph is said to be input connected if there exists at least one walk to all the followers from the leader. This naturally leads to the question of whether followers can be guided to achieve any desired state through the influence of control input through leader nodes. The sign pattern of the graph is fixed; therefore, the magnitude of the edge weight $a_{ij} \in \mathbb{R}^+$. Although the existing literature uses the sign pattern of the edge weights of a digraph to analyze herdability, the magnitudes of the weights also play a key role, as shown below. 
\begin{eg} \label{two path example}
\normalfont

Consider a digraph with a node $v$ such that there exist at least two distinct paths ${w}^{1}_{(1,v)}~ ,~{w}^{2}_{(1,v)}$ from the leader of the same length $k$. The controllability matrix associated with the digraph shown in Fig.\ref{fig:ss digraph}, is  
{
\small
\begin{equation}
        \mathcal{C}_1= \begin{bmatrix}
          b_1& 0 & 0 & 0 & 0&~0 \\
    0 & b_1a_{21} & 0 & 0 & 0 &~0\\
    0 & 0 & b_1a_{21}a_{32} & 0 & 0&~0 \\
    0 & 0 & b_1a_{21}a_{42} & 0 & 0 &~0\\
    0 & 0 & 0 & b_1a_{21}(a_{42}a_{54}-a_{32}a_{53}) & 0&~0\\
    0 & 0 & 0& -b_1a_{21}a_{42}a_{64} & 0 & ~0
    \end{bmatrix}
    \end{equation} 
    }
    
where each $[\Psi_k]_j~\neq0$, as defined in \eqref{Control_mat}, is the sum of the product weights of paths from leader to follower $j$ with length $k$ equivalently $\{\mathcal{W}^1_{(1,k)}\}+\{\mathcal{W}^2_{(1,k)}\}$.  
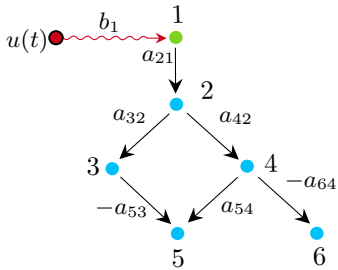
\begin{figure}[ht]
  \centering
 
\tikzset{every picture/.style={scale=0.7pt}} 
       
\begin{tikzpicture}[x=0.75pt,y=0.75pt,yscale=-1,xscale=1]

\draw [color={rgb, 255:red, 0; green, 0; blue, 0 }  ,draw opacity=1 ]   (370.24,102.75) -- (370.38,132.76) ;
\draw [shift={(370.39,135.76)}, rotate = 269.75] [fill={rgb, 255:red, 0; green, 0; blue, 0 }  ,fill opacity=1 ][line width=0.08]  [draw opacity=0] (10.72,-5.15) -- (0,0) -- (10.72,5.15) -- (7.12,0) -- cycle    ;
\draw [color={rgb, 255:red, 0; green, 0; blue, 0 }  ,draw opacity=1 ]   (378.55,149.76) -- (413.28,181.21) ;
\draw [shift={(415.5,183.22)}, rotate = 222.16] [fill={rgb, 255:red, 0; green, 0; blue, 0 }  ,fill opacity=1 ][line width=0.08]  [draw opacity=0] (10.72,-5.15) -- (0,0) -- (10.72,5.15) -- (7.12,0) -- cycle    ;
\draw [color={rgb, 255:red, 0; green, 0; blue, 0 }  ,draw opacity=1 ]   (365.97,150.17) -- (332.51,180.77) ;
\draw [shift={(330.29,182.8)}, rotate = 317.56] [fill={rgb, 255:red, 0; green, 0; blue, 0 }  ,fill opacity=1 ][line width=0.08]  [draw opacity=0] (10.72,-5.15) -- (0,0) -- (10.72,5.15) -- (7.12,0) -- cycle    ;
\draw [color={rgb, 255:red, 0; green, 0; blue, 0 }  ,draw opacity=1 ]   (331.08,195.6) -- (365.08,227.46) ;
\draw [shift={(367.27,229.51)}, rotate = 223.14] [fill={rgb, 255:red, 0; green, 0; blue, 0 }  ,fill opacity=1 ][line width=0.08]  [draw opacity=0] (10.72,-5.15) -- (0,0) -- (10.72,5.15) -- (7.12,0) -- cycle    ;
\draw  [draw opacity=0][fill={rgb, 255:red, 0; green, 192; blue, 248 }  ,fill opacity=1 ] (325.09,184.98) .. controls (327.73,184.93) and (329.92,186.99) .. (329.98,189.59) .. controls (330.04,192.19) and (327.94,194.35) .. (325.3,194.4) .. controls (322.66,194.45) and (320.47,192.38) .. (320.41,189.78) .. controls (320.35,187.18) and (322.45,185.03) .. (325.09,184.98) -- cycle ;
\draw  [draw opacity=0][fill={rgb, 255:red, 0; green, 192; blue, 248 }  ,fill opacity=1 ] (371.67,231.54) .. controls (374.31,231.48) and (376.5,233.55) .. (376.56,236.15) .. controls (376.62,238.75) and (374.52,240.9) .. (371.88,240.96) .. controls (369.24,241.01) and (367.05,238.94) .. (366.99,236.34) .. controls (366.93,233.74) and (369.03,231.59) .. (371.67,231.54) -- cycle ;
\draw  [draw opacity=0][fill={rgb, 255:red, 0; green, 192; blue, 248 }  ,fill opacity=1 ] (370.89,138.7) .. controls (373.54,138.64) and (375.73,140.71) .. (375.79,143.31) .. controls (375.84,145.91) and (373.75,148.06) .. (371.1,148.12) .. controls (368.46,148.17) and (366.27,146.1) .. (366.21,143.5) .. controls (366.16,140.9) and (368.25,138.75) .. (370.89,138.7) -- cycle ;
\draw  [draw opacity=0][fill={rgb, 255:red, 0; green, 192; blue, 248 }  ,fill opacity=1 ] (421.31,183.18) .. controls (423.95,183.13) and (426.14,185.2) .. (426.2,187.8) .. controls (426.26,190.4) and (424.16,192.55) .. (421.52,192.6) .. controls (418.88,192.66) and (416.69,190.59) .. (416.63,187.99) .. controls (416.57,185.39) and (418.67,183.24) .. (421.31,183.18) -- cycle ;
\draw  [draw opacity=0][fill={rgb, 255:red, 0; green, 192; blue, 248 }  ,fill opacity=1 ] (471.02,231.39) .. controls (473.66,231.34) and (475.85,233.4) .. (475.91,236) .. controls (475.97,238.61) and (473.87,240.76) .. (471.23,240.81) .. controls (468.59,240.86) and (466.4,238.8) .. (466.34,236.19) .. controls (466.28,233.59) and (468.38,231.44) .. (471.02,231.39) -- cycle ;
\draw  [draw opacity=0][fill={rgb, 255:red, 126; green, 211; blue, 33 }  ,fill opacity=1 ] (370.49,90.67) .. controls (373.13,90.61) and (375.32,92.68) .. (375.38,95.28) .. controls (375.44,97.88) and (373.34,100.03) .. (370.7,100.09) .. controls (368.05,100.14) and (365.86,98.07) .. (365.81,95.47) .. controls (365.75,92.87) and (367.84,90.72) .. (370.49,90.67) -- cycle ;
\draw [color={rgb, 255:red, 0; green, 0; blue, 0 }  ,draw opacity=1 ]   (414.77,195.77) -- (381.31,226.37) ;
\draw [shift={(379.09,228.4)}, rotate = 317.56] [fill={rgb, 255:red, 0; green, 0; blue, 0 }  ,fill opacity=1 ][line width=0.08]  [draw opacity=0] (10.72,-5.15) -- (0,0) -- (10.72,5.15) -- (7.12,0) -- cycle    ;
\draw [color={rgb, 255:red, 0; green, 0; blue, 0 }  ,draw opacity=1 ]   (429.35,195.76) -- (464.08,227.21) ;
\draw [shift={(466.3,229.22)}, rotate = 222.16] [fill={rgb, 255:red, 0; green, 0; blue, 0 }  ,fill opacity=1 ][line width=0.08]  [draw opacity=0] (10.72,-5.15) -- (0,0) -- (10.72,5.15) -- (7.12,0) -- cycle    ;
\draw [color={rgb, 255:red, 208; green, 2; blue, 27 }  ,draw opacity=1 ]   (285.11,95.58) .. controls (286.78,93.92) and (288.45,93.93) .. (290.11,95.6) .. controls (291.77,97.27) and (293.44,97.28) .. (295.11,95.62) .. controls (296.78,93.97) and (298.45,93.98) .. (300.11,95.65) .. controls (301.77,97.32) and (303.44,97.33) .. (305.11,95.67) .. controls (306.78,94.01) and (308.45,94.02) .. (310.11,95.69) .. controls (311.77,97.36) and (313.44,97.37) .. (315.11,95.71) .. controls (316.78,94.06) and (318.45,94.07) .. (320.11,95.74) .. controls (321.77,97.41) and (323.44,97.42) .. (325.11,95.76) .. controls (326.78,94.1) and (328.45,94.11) .. (330.11,95.78) .. controls (331.77,97.45) and (333.44,97.46) .. (335.11,95.8) .. controls (336.78,94.15) and (338.45,94.16) .. (340.11,95.83) .. controls (341.77,97.5) and (343.44,97.51) .. (345.11,95.85) .. controls (346.78,94.19) and (348.45,94.2) .. (350.11,95.87) -- (353.07,95.88) -- (361.07,95.92) ;
\draw [shift={(364.07,95.93)}, rotate = 180.26] [fill={rgb, 255:red, 208; green, 2; blue, 27 }  ,fill opacity=1 ][line width=0.08]  [draw opacity=0] (8.04,-3.86) -- (0,0) -- (8.04,3.86) -- (5.34,0) -- cycle    ;
\draw  [color={rgb, 255:red, 0; green, 0; blue, 0 }  ,draw opacity=1 ][fill={rgb, 255:red, 208; green, 2; blue, 27 }  ,fill opacity=1 ][line width=0.75]  (285,90.87) .. controls (287.64,90.82) and (289.83,92.88) .. (289.89,95.48) .. controls (289.95,98.08) and (287.85,100.24) .. (285.21,100.29) .. controls (282.57,100.34) and (280.38,98.27) .. (280.32,95.67) .. controls (280.26,93.07) and (282.36,90.92) .. (285,90.87) -- cycle ;

\draw (386.19,125.68) node [anchor=north west][inner sep=0.75pt]  [color={rgb, 255:red, 0; green, 0; blue, 0 }  ,opacity=1 ,rotate=-358.27]  {$2$};
\draw (304.81,180.49) node [anchor=north west][inner sep=0.75pt]  [color={rgb, 255:red, 0; green, 0; blue, 0 }  ,opacity=1 ,rotate=-358.27]  {$3$};
\draw (466.2,244.93) node [anchor=north west][inner sep=0.75pt]  [color={rgb, 255:red, 0; green, 0; blue, 0 }  ,opacity=1 ,rotate=-358.38]  {$6$};
\draw (365.5,246.19) node [anchor=north west][inner sep=0.75pt]  [color={rgb, 255:red, 0; green, 0; blue, 0 }  ,opacity=1 ,rotate=-358.27]  {$5$};
\draw (431.67,177.75) node [anchor=north west][inner sep=0.75pt]  [color={rgb, 255:red, 0; green, 0; blue, 0 }  ,opacity=1 ,rotate=-358.27]  {$4$};
\draw (365.3,70.65) node [anchor=north west][inner sep=0.75pt]  [color={rgb, 255:red, 0; green, 0; blue, 0 }  ,opacity=1 ,rotate=-358.27]  {$1$};
\draw (344.4,102.9) node [anchor=north west][inner sep=0.75pt]  [font=\small]  {$a_{21}$};
\draw (323.8,145.5) node [anchor=north west][inner sep=0.75pt]  [font=\small]  {$a_{32}$};
\draw (399.6,146.7) node [anchor=north west][inner sep=0.75pt]  [font=\small]  {$a_{42}$};
\draw (311.2,210.3) node [anchor=north west][inner sep=0.75pt]  [font=\small]  {$-a_{53}$};
\draw (400.93,211.49) node [anchor=north west][inner sep=0.75pt]  [font=\small]  {$a_{54}$};
\draw (446.4,190.7) node [anchor=north west][inner sep=0.75pt]  [font=\small]  {$-a_{64}$};
\draw (247.83,87.4) node [anchor=north west][inner sep=0.75pt]  [font=\small]  {$u( t)$};
\draw (313.83,75.7) node [anchor=north west][inner sep=0.75pt]  [font=\small,color={rgb, 255:red, 0; green, 0; blue, 0 }  ,opacity=1 ]  {$b_{1}$};

\end{tikzpicture}
    \caption{An $\mathcal{SS}$ herdable digraph}
    \label{fig:ss digraph}
\end{figure}
The linear combination of columns of $\mathcal{C}_1$ i.e. $\mathcal{B}\delta_1 + \Psi_1\delta_2 + \Psi_2\delta_3 + \Psi_3\delta_4$ 
can generate a positive vector $\bm{v}$ only if $a_{42}a_{54} < a_{32}a_{53}$. Hence, it is herdable in a structural sense, that is, only for some magnitudes of the edge weights.
\end{eg}

In Example \ref{two path example}, the dependence on the magnitudes of the edge weights arises because of the existence of multiple paths from node 1 to node 5. This highlights that the study of sign patterns is sufficient in the case of directed trees, where only unique paths exist between any two nodes. In general, the presence of cycles and multiple paths between nodes poses challenges in herdability analysis. This motivates the need to explore the following points.
\begin{itemize}
     \item The existing literature \cite{ruf2018herdable}, \cite{ruf2019herdability} often focuses on the sign pattern of digraphs for herdability analysis. However, as shown in Example \ref{two path example}, the edge weights also affect herdability. Hence, our primary objective is to develop graph-theoretical conditions for structural-sign herdability.
     \item Further, there are limited works on the herdability of cyclic digraphs. This is because of the existence of walks of arbitrarily long lengths, which increases the computational complexity of the existing algorithms. Similarly, only a limited number of studies exist on the herdability of DAGs (which are not directed trees). We are also interested in developing a generalized approach to check herdability in a structural sense for any arbitrary digraph.
\end{itemize}   
\subsection{Test for herdability}
To the authors' best understanding, there is no existing test similar to the Kalman controllability criterion for herdability. Adopting Gordan's theorem of alternatives given in Theorem~\ref{gordan thm}, we formulate the following test for the herdability of a linear system described by \eqref{sys1:sys1eqn}.
\begin{proposition} \label{test for H}
A system \eqref{sys1:sys1eqn} is herdable if and only if there exists no $\bm{y}\geq0;~\bm{y}~\in~\mathbb{R}^n$ such that $\mathcal{C}(\mathcal{A},\mathcal{B})^{\top} \bm{y}=0$, where $\mathcal{C}(\mathcal{A},\mathcal{B})$ is the controllability matrix associated with the pair $(\mathcal{A},\mathcal{B})$. In other words, there is no nonnegative $\bm{y}$ such that $\bm{y}~\in~\textit{Null}~(\mathcal{C}(\mathcal{A},\mathcal{B})^{\top})$.
\end{proposition}
The proof follows immediately from Gordan’s Theorem. If there exists no nonnegative $\bm{y}$ such that $\bm{y}~\in~\textit{Null}~(\mathcal{C}(\mathcal{A},\mathcal{B})^{\top})$ then according to Theorem \ref{gordan thm}, the statement (1) holds i.e. there exists $\bm{x} \in \mathbb{R}^n$ such that $\mathcal{C}(\mathcal{A},\mathcal{B})\cdot \bm{x} \gg0$. This implies the system \ref{sys1:sys1eqn} is herdable.
\section{Structural Sign Herdability in Arbitrary Networks}   
\label{Sec:4}
In this section, we present tools that aid in the analysis of herdability in a structural sense for signed digraphs. To formally explain the concept of herdability in a structural sense, we introduce the following definitions:
\begin{definition}
 System (1) is said to be structurally sign $(\mathcal{SS})$ herdable if there exists a pair $(\Bar{\mathcal{A}}, \Bar{\mathcal{B}})$ with the same structure as $(\mathcal{A,B})$ and $(\Bar{\mathcal{A}}, \Bar{\mathcal{B}})$ herdable for Some different parametric realizations are given by $a_{ij} \in \mathbb{R}^{+}$ and $b_i \in \mathbb{R} \setminus \{0\}$.
\end{definition}

In the case of signed digraphs, there are several challenges in the analysis of $\mathcal{SS}$ herdability, as outlined in Section. \ref{Sec:3}. One of the major challenges is the possibility of the existence of multiple paths between the leader and follower nodes. In such cases, the magnitude of the edge weights of the paths can affect the sign pattern of the controllability matrix, thereby affecting the herdability of the system. To address this issue, we introduce a \textit{sign-structured controllability matrix}, as discussed below.
\subsection{Sign-structured controllability matrix ($\mathfrak{S}(\mathcal{C}(\mathcal{A,B}))$) }\label{SSC-M para}
For the system in \eqref{sys1:sys1eqn}, the sign-structured controllability matrix, $\mathfrak{S}(\mathcal{C}(\mathcal{A,B}))$ or simply $\mathfrak{S}[\mathcal{C}]$ is defined, with each entry given by \eqref{eq:sign col}. It records all possible signs of the nonzero entries of the controllability matrix $\mathcal{C}$, depending on the parametric realizations of $(\mathcal{A},\mathcal{B})$. The $(i,j)^{th}$ entry of $\mathfrak{S}[\mathcal{C}]$ is
  
\vspace{-2mm}
\begin{align}
\mathfrak{S}[\mathcal{C}_{(i,j)}] :=
\begin{cases}
0, & \text{if there is no path from $1$ to $i$}, \\[4pt]
+ (\text{or} -), & \text{if } 
   \operatorname{sign}\!\Big(\sum_r \mathcal{W}^r_{(1,i)} \text{ of length } \\  
      &(j-1)\Big) > 0 \quad(\text{ or }< 0), \\[4pt]  
+/-, & \text{otherwise}.     
\end{cases}
\label{eq:sign col}
\end{align}
where $r \in \mathbb{Z^+}$ and $\mathcal{W}^r_{(1,i)}$ is the product of edge weights of the $r^{th}$ walk from node $1$ to $i$, as defined in Section \ref{Sec:2}.

Let us revisit Example \ref{two path example} in which follower $5$ has two distinct paths of the same length from the leader, with different signs in the path gain. Consequently, the entry $[\Psi_3]_5=b_1a_{21}(a_{42}a_{54}-a_{32}a_{53})$ can take any sign depending on the parameter values. The $\mathfrak{S}[\mathcal{C}]$ matrix corresponding to the digraph is given as follows:
{
\small
\begin{align} 
    \mathfrak{S}(\mathcal{C}_1)= \begin{bmatrix}
    +& 0 & 0 & 0 & 0&0 \\
    0 & + & 0 & 0 & 0 &0\\
    0 & 0 & + & 0 & 0&0 \\
    0 & 0 & + & 0 & 0 &0\\
    0 & 0 & 0 & +/- & 0&0\\
    0 & 0 & 0& + & 0 & 0 
    \end{bmatrix}
\end{align}
}
Note that the sign of the entry associated with follower $5$ is not structurally fixed. Hence, the entry  $\mathfrak{S}([\Psi_3]_5)= +/-$. Thus, the $\mathfrak{S}[\mathcal{C}]$ matrix allows us to highlight the same.

Consider the digraphs $\mathcal{G}_2$ and $\mathcal{G}_3$ shown in Figs. ~\ref{moti eg 1} and \ref{Moti eg2} with their sign-structured controllability matrices $\mathfrak{S}[\mathcal{C}_2]$ and $\mathfrak{S}[\mathcal{C}_3]$. For herdability, we are interested in determining if the columns of $\mathfrak{S}[\mathcal{C}_2]$ and $\mathfrak{S}[\mathcal{C}_3]$ span the positive orthant of $\mathbb{R}^n$. Although such conclusions can be drawn in principle for both matrices, it is relatively easier for $\mathcal{C}_2$ owing to its sparsity. Hence:
\begin{itemize}
    \item[a.] The analysis of herdability is relatively easier for directed trees than for other acyclic digraphs. 
    \item[b.] Similarly, the presence of cycles can make the analysis much more complex than that of acyclic digraphs. 
\end{itemize}

Matrix computations to check ($\mathcal{SS}$) herdability often become tedious for large networks, particularly when they have cycles. Hence, we propose the use of graph-theoretic methods, which are discussed next.
\begin{figure}[ht]
\centering
\begin{minipage}[b]{0.45\columnwidth}   
    \centering
 
\tikzset{every picture/.style={scale=0.75pt}} 

\begin{tikzpicture}[x=0.75pt,y=0.75pt,yscale=-1,xscale=1]

\draw [color={rgb, 255:red, 0; green, 0; blue, 0 }  ,draw opacity=1 ]   (557.76,142.94) -- (582.59,166.17) ;
\draw [shift={(584.78,168.22)}, rotate = 223.09] [fill={rgb, 255:red, 0; green, 0; blue, 0 }  ,fill opacity=1 ][line width=0.08]  [draw opacity=0] (7.14,-3.43) -- (0,0) -- (7.14,3.43) -- (4.74,0) -- cycle    ;
\draw [color={rgb, 255:red, 0; green, 0; blue, 0 }  ,draw opacity=1 ]   (538.93,142.34) -- (514.19,168.06) ;
\draw [shift={(512.11,170.22)}, rotate = 313.89] [fill={rgb, 255:red, 0; green, 0; blue, 0 }  ,fill opacity=1 ][line width=0.08]  [draw opacity=0] (7.14,-3.43) -- (0,0) -- (7.14,3.43) -- (4.74,0) -- cycle    ;
\draw  [draw opacity=0][fill={rgb, 255:red, 0; green, 192; blue, 248 }  ,fill opacity=1 ] (505.79,172.04) .. controls (508.64,171.99) and (510.99,174.16) .. (511.06,176.9) .. controls (511.12,179.64) and (508.86,181.91) .. (506.02,181.96) .. controls (503.17,182.02) and (500.82,179.84) .. (500.75,177.1) .. controls (500.69,174.36) and (502.95,172.1) .. (505.79,172.04) -- cycle ;
\draw  [draw opacity=0][fill={rgb, 255:red, 0; green, 192; blue, 248 }  ,fill opacity=1 ] (590.33,169.08) .. controls (593.18,169.03) and (595.54,171.2) .. (595.6,173.94) .. controls (595.66,176.68) and (593.41,178.95) .. (590.56,179) .. controls (587.72,179.06) and (585.36,176.88) .. (585.3,174.14) .. controls (585.23,171.4) and (587.49,169.14) .. (590.33,169.08) -- cycle ;
\draw  [draw opacity=0][fill={rgb, 255:red, 126; green, 211; blue, 33 }  ,fill opacity=1 ] (548.02,82.71) .. controls (550.86,82.66) and (553.22,84.84) .. (553.28,87.58) .. controls (553.34,90.31) and (551.09,92.58) .. (548.24,92.64) .. controls (545.4,92.69) and (543.04,90.52) .. (542.98,87.78) .. controls (542.92,85.04) and (545.17,82.77) .. (548.02,82.71) -- cycle ;
\draw [color={rgb, 255:red, 208; green, 2; blue, 27 }  ,draw opacity=1 ]   (486.64,88.46) .. controls (488.28,86.77) and (489.95,86.74) .. (491.64,88.38) .. controls (493.33,90.01) and (495,89.98) .. (496.64,88.29) .. controls (498.28,86.6) and (499.95,86.57) .. (501.64,88.2) .. controls (503.33,89.84) and (505,89.81) .. (506.64,88.12) .. controls (508.28,86.43) and (509.95,86.4) .. (511.64,88.03) .. controls (513.33,89.66) and (515,89.63) .. (516.63,87.94) .. controls (518.27,86.25) and (519.94,86.22) .. (521.63,87.86) .. controls (523.32,89.49) and (524.99,89.46) .. (526.63,87.77) -- (529.21,87.72) -- (537.21,87.58) ;
\draw [shift={(540.21,87.53)}, rotate = 179] [fill={rgb, 255:red, 208; green, 2; blue, 27 }  ,fill opacity=1 ][line width=0.08]  [draw opacity=0] (8.04,-3.86) -- (0,0) -- (8.04,3.86) -- (5.34,0) -- cycle    ;
\draw  [color={rgb, 255:red, 0; green, 0; blue, 0 }  ,draw opacity=1 ][fill={rgb, 255:red, 208; green, 2; blue, 27 }  ,fill opacity=1 ][line width=0.75]  (486.29,83.56) .. controls (488.87,83.52) and (491,85.43) .. (491.06,87.83) .. controls (491.12,90.24) and (489.07,92.23) .. (486.5,92.28) .. controls (483.92,92.33) and (481.79,90.42) .. (481.73,88.01) .. controls (481.67,85.6) and (483.72,83.61) .. (486.29,83.56) -- cycle ;
\draw [color={rgb, 255:red, 0; green, 0; blue, 0 }  ,draw opacity=1 ]   (549.37,148.21) -- (551.01,198.56) ;
\draw [shift={(551.11,201.56)}, rotate = 268.13] [fill={rgb, 255:red, 0; green, 0; blue, 0 }  ,fill opacity=1 ][line width=0.08]  [draw opacity=0] (7.14,-3.43) -- (0,0) -- (7.14,3.43) -- (4.74,0) -- cycle    ;
\draw  [draw opacity=0][fill={rgb, 255:red, 0; green, 192; blue, 248 }  ,fill opacity=1 ] (550.38,212.17) .. controls (553.22,212.11) and (555.58,214.29) .. (555.64,217.03) .. controls (555.7,219.77) and (553.45,222.04) .. (550.6,222.09) .. controls (547.76,222.15) and (545.4,219.97) .. (545.34,217.23) .. controls (545.28,214.49) and (547.53,212.22) .. (550.38,212.17) -- cycle ;
\draw [color={rgb, 255:red, 0; green, 0; blue, 0 }  ,draw opacity=1 ]   (514.44,183.89) -- (540.32,209.77) ;
\draw [shift={(542.44,211.89)}, rotate = 225] [fill={rgb, 255:red, 0; green, 0; blue, 0 }  ,fill opacity=1 ][line width=0.08]  [draw opacity=0] (7.14,-3.43) -- (0,0) -- (7.14,3.43) -- (4.74,0) -- cycle    ;
\draw [color={rgb, 255:red, 0; green, 0; blue, 0 }  ,draw opacity=1 ]   (586.63,181) -- (560.18,208.72) ;
\draw [shift={(558.11,210.89)}, rotate = 313.65] [fill={rgb, 255:red, 0; green, 0; blue, 0 }  ,fill opacity=1 ][line width=0.08]  [draw opacity=0] (7.14,-3.43) -- (0,0) -- (7.14,3.43) -- (4.74,0) -- cycle    ;
\draw  [draw opacity=0][fill={rgb, 255:red, 0; green, 192; blue, 248 }  ,fill opacity=1 ] (548.64,132) .. controls (551.48,131.94) and (553.84,134.12) .. (553.9,136.86) .. controls (553.96,139.6) and (551.71,141.87) .. (548.86,141.92) .. controls (546.02,141.98) and (543.66,139.8) .. (543.6,137.06) .. controls (543.54,134.32) and (545.79,132.06) .. (548.64,132) -- cycle ;
\draw [color={rgb, 255:red, 0; green, 0; blue, 0 }  ,draw opacity=1 ]   (547.9,97.35) -- (548.6,122.74) ;
\draw [shift={(548.69,125.74)}, rotate = 268.41] [fill={rgb, 255:red, 0; green, 0; blue, 0 }  ,fill opacity=1 ][line width=0.08]  [draw opacity=0] (7.14,-3.43) -- (0,0) -- (7.14,3.43) -- (4.74,0) -- cycle    ;

\draw (559.08,126.64) node [anchor=north west][inner sep=0.75pt]  [font=\small,color={rgb, 255:red, 0; green, 0; blue, 0 }  ,opacity=1 ,rotate=-358.27]  {$2$};
\draw (481.92,172.48) node [anchor=north west][inner sep=0.75pt]  [font=\small,color={rgb, 255:red, 0; green, 0; blue, 0 }  ,opacity=1 ,rotate=-358.27]  {$3$};
\draw (555.99,72.74) node [anchor=north west][inner sep=0.75pt]  [font=\small,color={rgb, 255:red, 0; green, 0; blue, 0 }  ,opacity=1 ,rotate=-358.27]  {$1$};
\draw (536.75,170.84) node [anchor=north west][inner sep=0.75pt]  [font=\footnotesize,color={rgb, 255:red, 0; green, 0; blue, 0 }  ,opacity=1 ,rotate=-358.88]  {$-$};
\draw (580.32,196.17) node [anchor=north west][inner sep=0.75pt]  [font=\footnotesize,color={rgb, 255:red, 0; green, 0; blue, 0 }  ,opacity=1 ,rotate=-358.88]  {$-$};
\draw (511.94,68.49) node [anchor=north west][inner sep=0.75pt]  [font=\scriptsize,color={rgb, 255:red, 0; green, 0; blue, 0 }  ,opacity=1 ,rotate=-358.88]  {$+$};
\draw (471.24,107.49) node [anchor=north west][inner sep=0.75pt]  [font=\small]  {$u( t)$};
\draw (603.29,166.76) node [anchor=north west][inner sep=0.75pt]  [font=\small,color={rgb, 255:red, 0; green, 0; blue, 0 }  ,opacity=1 ,rotate=-358.27]  {$4$};
\draw (546.9,230.12) node [anchor=north west][inner sep=0.75pt]  [font=\small,color={rgb, 255:red, 0; green, 0; blue, 0 }  ,opacity=1 ,rotate=-358.27]  {$5$};
\draw (531.63,103.39) node [anchor=north west][inner sep=0.75pt]  [font=\scriptsize,color={rgb, 255:red, 0; green, 0; blue, 0 }  ,opacity=1 ,rotate=-358.88]  {$+$};
\draw (513.99,140.79) node [anchor=north west][inner sep=0.75pt]  [font=\scriptsize,color={rgb, 255:red, 0; green, 0; blue, 0 }  ,opacity=1 ,rotate=-358.88]  {$+$};
\draw (578.4,146.94) node [anchor=north west][inner sep=0.75pt]  [font=\scriptsize,color={rgb, 255:red, 0; green, 0; blue, 0 }  ,opacity=1 ,rotate=-358.88]  {$+$};
\draw (512.25,203.13) node [anchor=north west][inner sep=0.75pt]  [font=\scriptsize,color={rgb, 255:red, 0; green, 0; blue, 0 }  ,opacity=1 ,rotate=-358.88]  {$+$};

\end{tikzpicture}
  \caption{A digraph $\mathcal{G}_2 $}
    \label{moti eg 1}
\end{minipage}%
\hfill
\begin{minipage}[b]{0.55\columnwidth}   
    \centering

\tikzset{every picture/.style={scale=0.75pt}} 

\begin{tikzpicture}[x=0.75pt,y=0.75pt,yscale=-1,xscale=1]

\draw [color={rgb, 255:red, 0; green, 0; blue, 0 }  ,draw opacity=1 ]   (122.46,159.02) -- (122.48,191.46) ;
\draw [shift={(122.49,194.46)}, rotate = 269.96] [fill={rgb, 255:red, 0; green, 0; blue, 0 }  ,fill opacity=1 ][line width=0.08]  [draw opacity=0] (7.14,-3.43) -- (0,0) -- (7.14,3.43) -- (4.74,0) -- cycle    ;
\draw [color={rgb, 255:red, 0; green, 0; blue, 0 }  ,draw opacity=1 ]   (167.33,88.67) -- (199.57,137.61) ;
\draw [shift={(201.22,140.11)}, rotate = 236.63] [fill={rgb, 255:red, 0; green, 0; blue, 0 }  ,fill opacity=1 ][line width=0.08]  [draw opacity=0] (7.14,-3.43) -- (0,0) -- (7.14,3.43) -- (4.74,0) -- cycle    ;
\draw [color={rgb, 255:red, 0; green, 0; blue, 0 }  ,draw opacity=1 ]   (158.33,88.67) -- (125.31,137.95) ;
\draw [shift={(123.64,140.45)}, rotate = 303.82] [fill={rgb, 255:red, 0; green, 0; blue, 0 }  ,fill opacity=1 ][line width=0.08]  [draw opacity=0] (7.14,-3.43) -- (0,0) -- (7.14,3.43) -- (4.74,0) -- cycle    ;
\draw  [draw opacity=0][fill={rgb, 255:red, 0; green, 192; blue, 248 }  ,fill opacity=1 ] (121.64,145.28) .. controls (124.48,145.23) and (126.83,147.23) .. (126.89,149.74) .. controls (126.96,152.26) and (124.71,154.34) .. (121.87,154.39) .. controls (119.03,154.44) and (116.67,152.45) .. (116.61,149.93) .. controls (116.55,147.41) and (118.8,145.33) .. (121.64,145.28) -- cycle ;
\draw  [draw opacity=0][fill={rgb, 255:red, 0; green, 192; blue, 248 }  ,fill opacity=1 ] (202.69,195.7) .. controls (205.53,195.65) and (207.88,197.65) .. (207.95,200.16) .. controls (208.01,202.68) and (205.76,204.76) .. (202.92,204.81) .. controls (200.08,204.86) and (197.73,202.87) .. (197.66,200.35) .. controls (197.6,197.83) and (199.85,195.75) .. (202.69,195.7) -- cycle ;
\draw  [draw opacity=0][fill={rgb, 255:red, 0; green, 192; blue, 248 }  ,fill opacity=1 ] (204.74,141.04) .. controls (207.58,140.99) and (209.94,142.99) .. (210,145.5) .. controls (210.06,148.02) and (207.81,150.1) .. (204.97,150.15) .. controls (202.13,150.2) and (199.78,148.21) .. (199.71,145.69) .. controls (199.65,143.17) and (201.9,141.09) .. (204.74,141.04) -- cycle ;
\draw [line width=0.75]    (128.25,142.12) .. controls (152.21,126.94) and (170.63,126.48) .. (194.74,143.51) ;
\draw [shift={(197,145.14)}, rotate = 216.5] [fill={rgb, 255:red, 0; green, 0; blue, 0 }  ][line width=0.08]  [draw opacity=0] (7.14,-3.43) -- (0,0) -- (7.14,3.43) -- (4.74,0) -- cycle    ;
\draw [color={rgb, 255:red, 0; green, 0; blue, 0 }  ,draw opacity=1 ]   (204.08,156.67) -- (204.1,189.61) ;
\draw [shift={(204.1,192.61)}, rotate = 269.96] [fill={rgb, 255:red, 0; green, 0; blue, 0 }  ,fill opacity=1 ][line width=0.08]  [draw opacity=0] (7.14,-3.43) -- (0,0) -- (7.14,3.43) -- (4.74,0) -- cycle    ;
\draw  [draw opacity=0][fill={rgb, 255:red, 0; green, 192; blue, 248 }  ,fill opacity=1 ] (122.36,197.66) .. controls (125.2,197.61) and (127.56,199.61) .. (127.62,202.13) .. controls (127.68,204.64) and (125.43,206.72) .. (122.59,206.78) .. controls (119.75,206.83) and (117.4,204.83) .. (117.34,202.31) .. controls (117.27,199.8) and (119.52,197.71) .. (122.36,197.66) -- cycle ;
\draw  [draw opacity=0][fill={rgb, 255:red, 126; green, 211; blue, 33 }  ,fill opacity=1 ] (162.72,75.22) .. controls (165.56,75.17) and (167.92,77.17) .. (167.98,79.69) .. controls (168.04,82.2) and (165.79,84.29) .. (162.95,84.34) .. controls (160.11,84.39) and (157.76,82.39) .. (157.7,79.87) .. controls (157.63,77.36) and (159.89,75.27) .. (162.72,75.22) -- cycle ;
\draw [color={rgb, 255:red, 0; green, 0; blue, 0 }  ,draw opacity=1 ]   (199.06,158.61) -- (137.37,192.76) ;
\draw [shift={(134.74,194.21)}, rotate = 331.03] [fill={rgb, 255:red, 0; green, 0; blue, 0 }  ,fill opacity=1 ][line width=0.08]  [draw opacity=0] (7.14,-3.43) -- (0,0) -- (7.14,3.43) -- (4.74,0) -- cycle    ;
\draw [color={rgb, 255:red, 0; green, 0; blue, 0 }  ,draw opacity=1 ]   (130.19,160.65) -- (191.57,193.42) ;
\draw [shift={(194.21,194.83)}, rotate = 208.1] [fill={rgb, 255:red, 0; green, 0; blue, 0 }  ,fill opacity=1 ][line width=0.08]  [draw opacity=0] (7.14,-3.43) -- (0,0) -- (7.14,3.43) -- (4.74,0) -- cycle    ;
\draw    (134.31,154.06) .. controls (150.06,160.15) and (174.84,165.41) .. (194.52,152.26) ;
\draw [shift={(131.41,152.88)}, rotate = 22.86] [fill={rgb, 255:red, 0; green, 0; blue, 0 }  ][line width=0.08]  [draw opacity=0] (10.72,-5.15) -- (0,0) -- (10.72,5.15) -- (7.12,0) -- cycle    ;
\draw [color={rgb, 255:red, 208; green, 2; blue, 27 }  ,draw opacity=1 ]   (93.37,81.25) .. controls (95.01,79.56) and (96.68,79.54) .. (98.37,81.18) .. controls (100.06,82.82) and (101.72,82.8) .. (103.37,81.11) .. controls (105.02,79.42) and (106.68,79.4) .. (108.37,81.04) .. controls (110.06,82.68) and (111.72,82.66) .. (113.37,80.97) .. controls (115.02,79.28) and (116.68,79.26) .. (118.37,80.9) .. controls (120.06,82.54) and (121.73,82.51) .. (123.37,80.82) .. controls (125.02,79.13) and (126.68,79.11) .. (128.37,80.75) .. controls (130.06,82.39) and (131.72,82.37) .. (133.37,80.68) .. controls (135.01,78.99) and (136.67,78.97) .. (138.36,80.61) -- (142.76,80.55) -- (150.76,80.43) ;
\draw [shift={(153.76,80.39)}, rotate = 179.18] [fill={rgb, 255:red, 208; green, 2; blue, 27 }  ,fill opacity=1 ][line width=0.08]  [draw opacity=0] (8.04,-3.86) -- (0,0) -- (8.04,3.86) -- (5.34,0) -- cycle    ;
\draw  [color={rgb, 255:red, 0; green, 0; blue, 0 }  ,draw opacity=1 ][fill={rgb, 255:red, 208; green, 2; blue, 27 }  ,fill opacity=1 ][line width=0.75]  (89.71,76.84) .. controls (91.97,76.8) and (93.84,78.71) .. (93.9,81.1) .. controls (93.96,83.5) and (92.17,85.48) .. (89.91,85.52) .. controls (87.65,85.57) and (85.77,83.66) .. (85.72,81.26) .. controls (85.66,78.87) and (87.45,76.89) .. (89.71,76.84) -- cycle ;

\draw (103.46,139.32) node [anchor=north west][inner sep=0.75pt]  [font=\small,color={rgb, 255:red, 0; green, 0; blue, 0 }  ,opacity=1 ,rotate=-358.27]  {$2$};
\draw (213.37,138.41) node [anchor=north west][inner sep=0.75pt]  [font=\small,color={rgb, 255:red, 0; green, 0; blue, 0 }  ,opacity=1 ,rotate=-358.27]  {$3$};
\draw (214.37,192.39) node [anchor=north west][inner sep=0.75pt]  [font=\small,color={rgb, 255:red, 0; green, 0; blue, 0 }  ,opacity=1 ,rotate=-358.27]  {$5$};
\draw (101.43,191.67) node [anchor=north west][inner sep=0.75pt]  [font=\small,color={rgb, 255:red, 0; green, 0; blue, 0 }  ,opacity=1 ,rotate=-358.27]  {$4$};
\draw (185.94,84.23) node [anchor=north west][inner sep=0.75pt]  [font=\small,color={rgb, 255:red, 0; green, 0; blue, 0 }  ,opacity=1 ,rotate=-358.27]  {$1$};
\draw (185.77,104.92) node [anchor=north west][inner sep=0.75pt]  [font=\footnotesize,rotate=-358.88]  {$+$};
\draw (215.54,166.74) node [anchor=north west][inner sep=0.75pt]  [font=\footnotesize,color={rgb, 255:red, 0; green, 0; blue, 0 }  ,opacity=1 ,rotate=-358.88]  {$-$};
\draw (123.91,106.05) node [anchor=north west][inner sep=0.75pt]  [font=\footnotesize,color={rgb, 255:red, 0; green, 0; blue, 0 }  ,opacity=1 ,rotate=-358.88]  {$+$};
\draw (154.52,113.41) node [anchor=north west][inner sep=0.75pt]  [font=\footnotesize,rotate=-358.88]  {$+$};
\draw (157.54,147.35) node [anchor=north west][inner sep=0.75pt]  [font=\footnotesize,color={rgb, 255:red, 0; green, 0; blue, 0 }  ,opacity=1 ]  {$-$};
\draw (99.02,168.29) node [anchor=north west][inner sep=0.75pt]  [font=\footnotesize,color={rgb, 255:red, 0; green, 0; blue, 0 }  ,opacity=1 ]  {$-$};
\draw (136.31,170.28) node [anchor=north west][inner sep=0.75pt]  [font=\footnotesize,rotate=-358.88]  {$+$};
\draw (181.22,170.28) node [anchor=north west][inner sep=0.75pt]  [font=\footnotesize,rotate=-358.88]  {$+$};
\draw (130.69,60.08) node [anchor=north west][inner sep=0.75pt]  [font=\normalsize,rotate=-89.64]  {$+$};
\draw (68.73,98.64) node [anchor=north west][inner sep=0.75pt]  [font=\small]  {$u( t)$};

\end{tikzpicture}

    \caption{A digraph $\mathcal{G}_3 $}
    \label{Moti eg2}
\end{minipage}%

\end{figure} 

\begin{figure}[ht]
\begin{minipage}[b]{0.2\columnwidth}   
    \centering
      {
\small
\[
\mathfrak{S}(\mathcal{C}_2) = 
\begin{bmatrix}
+ & 0 & 0 & 0 & 0 \\
0 & + & 0 & 0 & 0\\
0 & 0 & + & 0 & 0  \\
0 & 0 & + & 0 & 0 \\
0 & 0 & - & +/- & 0\\ 

\end{bmatrix}
\]}
\label{SSc M2}
\end{minipage}
\hfill
\begin{minipage}[b]{1.05\columnwidth}   
    \centering
      {
\small
\[
\mathfrak{S}(\mathcal{C}_3) = 
\begin{bmatrix}
+ & 0 & 0 & 0 & 0 & \\
0 & + & - & - & +& \\
0 & + & + & - & - &  \\
0 & 0 & +/- & + & +/-& \\
0 & 0 & +/- & - & +/-& 
\end{bmatrix} 
\]}
\label{SSc M3}
\end{minipage}
\end{figure} 

\subsection{Signed Layered Graph ($\mathcal{G}_s$)} \label{G_s para}
Given a digraph $\mathcal{G(A,B)}$, we can always generate a corresponding layered graph, which is a tree-structured representation of $\mathcal{G}$ (see Section \ref{LG}  for more details). When $\mathcal{G}$ is a signed network, the corresponding layered graph is termed a signed layered graph, denoted as $\mathcal{G}_s$. 
\begin{itemize}
    \item In the signed layered graph $\mathcal{G}_s$, the set of nodes in the $k^{th}$ layer $L_k$ is denoted as $V_{L_k}$ where the node set $V_{L_1}$ includes only the leader node. 
    \item $V_{L_k}$ includes the set of nodes that can be reached in $k-1$ steps from the leader.
    \item Since there can be multiple paths to a single node $v$ of the same length $m-1$, a node $v$ can appear multiple times in the node set $V_{L_{m}}$.
    \item Since cycles can increase the number of layers infinitely, we extend $\mathcal{G}_s$ to only $n$ layers. This is consistent with Remark \ref{Remark1} in the sense that only $n$ columns are relevant in the controllability matrix.
\end{itemize}
  Consider the digraph $\mathcal{G}_4$, as shown in Fig. \ref{fig:signed  layered graph}, the corresponding signed layer representation is depicted adjacent to it. Although cycles can extend the graph indefinitely, we limit it to $n=7$ layers to match the number of columns of the controllability matrix.

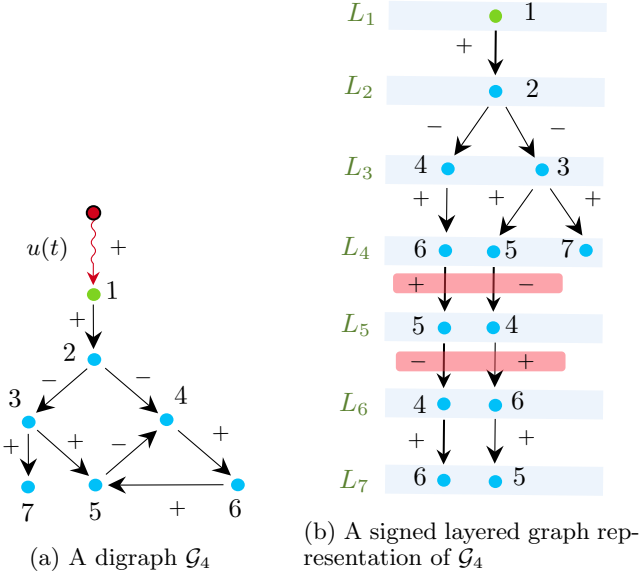
\begin{figure}[ht]
\centering
  \begin{subfigure}{0.18\textwidth}   
    \centering
 
\tikzset{every picture/.style={scale=0.85pt}} 

\begin{tikzpicture}[x=0.75pt,y=0.75pt,yscale=-1,xscale=1]

\draw [color={rgb, 255:red, 0; green, 0; blue, 0 }  ,draw opacity=1 ]   (70.39,110.24) -- (70.49,133.66) ;
\draw [shift={(70.51,136.66)}, rotate = 269.73] [fill={rgb, 255:red, 0; green, 0; blue, 0 }  ,fill opacity=1 ][line width=0.08]  [draw opacity=0] (10.72,-5.15) -- (0,0) -- (10.72,5.15) -- (7.12,0) -- cycle    ;
\draw [color={rgb, 255:red, 0; green, 0; blue, 0 }  ,draw opacity=1 ]   (77.38,147.86) -- (106.21,172.68) ;
\draw [shift={(108.48,174.64)}, rotate = 220.72] [fill={rgb, 255:red, 0; green, 0; blue, 0 }  ,fill opacity=1 ][line width=0.08]  [draw opacity=0] (10.72,-5.15) -- (0,0) -- (10.72,5.15) -- (7.12,0) -- cycle    ;
\draw [color={rgb, 255:red, 0; green, 0; blue, 0 }  ,draw opacity=1 ]   (66.79,148.19) -- (39.02,172.33) ;
\draw [shift={(36.75,174.3)}, rotate = 319] [fill={rgb, 255:red, 0; green, 0; blue, 0 }  ,fill opacity=1 ][line width=0.08]  [draw opacity=0] (10.72,-5.15) -- (0,0) -- (10.72,5.15) -- (7.12,0) -- cycle    ;
\draw [color={rgb, 255:red, 0; green, 0; blue, 0 }  ,draw opacity=1 ]   (37.42,184.54) -- (65.64,209.68) ;
\draw [shift={(67.88,211.68)}, rotate = 221.69] [fill={rgb, 255:red, 0; green, 0; blue, 0 }  ,fill opacity=1 ][line width=0.08]  [draw opacity=0] (10.72,-5.15) -- (0,0) -- (10.72,5.15) -- (7.12,0) -- cycle    ;
\draw  [draw opacity=0][fill={rgb, 255:red, 0; green, 192; blue, 248 }  ,fill opacity=1 ] (31.99,214.45) .. controls (34.22,214.4) and (36.06,216.06) .. (36.11,218.14) .. controls (36.16,220.22) and (34.39,221.94) .. (32.17,221.99) .. controls (29.94,222.03) and (28.1,220.37) .. (28.05,218.29) .. controls (28,216.21) and (29.77,214.49) .. (31.99,214.45) -- cycle ;
\draw  [draw opacity=0][fill={rgb, 255:red, 0; green, 192; blue, 248 }  ,fill opacity=1 ] (32.37,176.04) .. controls (34.6,176) and (36.44,177.66) .. (36.49,179.74) .. controls (36.54,181.82) and (34.77,183.54) .. (32.55,183.58) .. controls (30.32,183.62) and (28.48,181.97) .. (28.43,179.89) .. controls (28.38,177.81) and (30.15,176.08) .. (32.37,176.04) -- cycle ;
\draw  [draw opacity=0][fill={rgb, 255:red, 0; green, 192; blue, 248 }  ,fill opacity=1 ] (71.58,213.3) .. controls (73.81,213.26) and (75.65,214.91) .. (75.7,216.99) .. controls (75.75,219.08) and (73.99,220.8) .. (71.76,220.84) .. controls (69.54,220.88) and (67.69,219.23) .. (67.64,217.15) .. controls (67.6,215.06) and (69.36,213.34) .. (71.58,213.3) -- cycle ;
\draw  [draw opacity=0][fill={rgb, 255:red, 0; green, 192; blue, 248 }  ,fill opacity=1 ] (70.93,139.01) .. controls (73.16,138.97) and (75,140.62) .. (75.05,142.7) .. controls (75.1,144.78) and (73.33,146.51) .. (71.11,146.55) .. controls (68.88,146.59) and (67.04,144.93) .. (66.99,142.85) .. controls (66.94,140.77) and (68.71,139.05) .. (70.93,139.01) -- cycle ;
\draw  [draw opacity=0][fill={rgb, 255:red, 0; green, 192; blue, 248 }  ,fill opacity=1 ] (113.37,174.61) .. controls (115.6,174.56) and (117.44,176.22) .. (117.49,178.3) .. controls (117.54,180.38) and (115.77,182.1) .. (113.55,182.15) .. controls (111.32,182.19) and (109.48,180.53) .. (109.43,178.45) .. controls (109.38,176.37) and (111.15,174.65) .. (113.37,174.61) -- cycle ;
\draw [color={rgb, 255:red, 0; green, 0; blue, 0 }  ,draw opacity=1 ]   (32.2,186.78) -- (32.29,208.95) ;
\draw [shift={(32.31,211.95)}, rotate = 269.75] [fill={rgb, 255:red, 0; green, 0; blue, 0 }  ,fill opacity=1 ][line width=0.08]  [draw opacity=0] (10.72,-5.15) -- (0,0) -- (10.72,5.15) -- (7.12,0) -- cycle    ;
\draw  [draw opacity=0][fill={rgb, 255:red, 0; green, 192; blue, 248 }  ,fill opacity=1 ] (155.22,213.18) .. controls (157.44,213.14) and (159.29,214.79) .. (159.34,216.88) .. controls (159.38,218.96) and (157.62,220.68) .. (155.4,220.72) .. controls (153.17,220.76) and (151.33,219.11) .. (151.28,217.03) .. controls (151.23,214.94) and (152.99,213.22) .. (155.22,213.18) -- cycle ;
\draw  [draw opacity=0][fill={rgb, 255:red, 126; green, 211; blue, 33 }  ,fill opacity=1 ] (70.59,100.57) .. controls (72.81,100.53) and (74.66,102.18) .. (74.7,104.27) .. controls (74.75,106.35) and (72.99,108.07) .. (70.76,108.11) .. controls (68.54,108.15) and (66.7,106.5) .. (66.65,104.42) .. controls (66.6,102.34) and (68.36,100.61) .. (70.59,100.57) -- cycle ;
\draw [color={rgb, 255:red, 0; green, 0; blue, 0 }  ,draw opacity=1 ]   (105.61,186.65) -- (77.83,210.79) ;
\draw [shift={(107.87,184.68)}, rotate = 139] [fill={rgb, 255:red, 0; green, 0; blue, 0 }  ,fill opacity=1 ][line width=0.08]  [draw opacity=0] (10.72,-5.15) -- (0,0) -- (10.72,5.15) -- (7.12,0) -- cycle    ;
\draw [color={rgb, 255:red, 0; green, 0; blue, 0 }  ,draw opacity=1 ]   (120.14,184.67) -- (148.97,209.49) ;
\draw [shift={(151.25,211.45)}, rotate = 220.72] [fill={rgb, 255:red, 0; green, 0; blue, 0 }  ,fill opacity=1 ][line width=0.08]  [draw opacity=0] (10.72,-5.15) -- (0,0) -- (10.72,5.15) -- (7.12,0) -- cycle    ;
\draw [color={rgb, 255:red, 0; green, 0; blue, 0 }  ,draw opacity=1 ]   (146.85,216.87) -- (81.93,217.1) ;
\draw [shift={(78.93,217.11)}, rotate = 359.8] [fill={rgb, 255:red, 0; green, 0; blue, 0 }  ,fill opacity=1 ][line width=0.08]  [draw opacity=0] (10.72,-5.15) -- (0,0) -- (10.72,5.15) -- (7.12,0) -- cycle    ;
\draw [color={rgb, 255:red, 208; green, 2; blue, 27 }  ,draw opacity=1 ]   (70.52,55.98) .. controls (72.18,57.65) and (72.17,59.32) .. (70.5,60.98) .. controls (68.83,62.64) and (68.82,64.31) .. (70.48,65.98) .. controls (72.14,67.65) and (72.13,69.32) .. (70.46,70.98) .. controls (68.79,72.65) and (68.79,74.31) .. (70.45,75.98) .. controls (72.11,77.65) and (72.1,79.32) .. (70.43,80.98) .. controls (68.76,82.64) and (68.75,84.31) .. (70.41,85.98) -- (70.4,87.79) -- (70.37,95.79) ;
\draw [shift={(70.36,98.79)}, rotate = 270.21] [fill={rgb, 255:red, 208; green, 2; blue, 27 }  ,fill opacity=1 ][line width=0.08]  [draw opacity=0] (8.04,-3.86) -- (0,0) -- (8.04,3.86) -- (5.34,0) -- cycle    ;
\draw  [color={rgb, 255:red, 0; green, 0; blue, 0 }  ,draw opacity=1 ][fill={rgb, 255:red, 208; green, 2; blue, 27 }  ,fill opacity=1 ][line width=0.75]  (70.43,52.02) .. controls (72.62,51.97) and (74.43,53.71) .. (74.47,55.9) .. controls (74.52,58.08) and (72.79,59.89) .. (70.6,59.93) .. controls (68.42,59.98) and (66.61,58.24) .. (66.56,56.06) .. controls (66.51,53.87) and (68.25,52.06) .. (70.43,52.02) -- cycle ;

\draw (50.83,131.91) node [anchor=north west][inner sep=0.75pt]  [color={rgb, 255:red, 0; green, 0; blue, 0 }  ,opacity=1 ,rotate=-358.27]  {$2$};
\draw (18.69,159.53) node [anchor=north west][inner sep=0.75pt]  [color={rgb, 255:red, 0; green, 0; blue, 0 }  ,opacity=1 ,rotate=-358.27]  {$3$};
\draw (148.24,224.44) node [anchor=north west][inner sep=0.75pt]  [color={rgb, 255:red, 0; green, 0; blue, 0 }  ,opacity=1 ,rotate=-358.38]  {$6$};
\draw (65.89,225.93) node [anchor=north west][inner sep=0.75pt]  [color={rgb, 255:red, 0; green, 0; blue, 0 }  ,opacity=1 ,rotate=-358.27]  {$5$};
\draw (116.19,157.59) node [anchor=north west][inner sep=0.75pt]  [color={rgb, 255:red, 0; green, 0; blue, 0 }  ,opacity=1 ,rotate=-358.27]  {$4$};
\draw (75.95,95.63) node [anchor=north west][inner sep=0.75pt]  [color={rgb, 255:red, 0; green, 0; blue, 0 }  ,opacity=1 ,rotate=-358.27]  {$1$};
\draw (54.42,113.48) node [anchor=north west][inner sep=0.75pt]  [font=\scriptsize,color={rgb, 255:red, 0; green, 0; blue, 0 }  ,opacity=1 ,rotate=-358.88]  {$+$};
\draw (26.31,226.67) node [anchor=north west][inner sep=0.75pt]    {$7$};
\draw (15.35,188.43) node [anchor=north west][inner sep=0.75pt]  [font=\scriptsize,color={rgb, 255:red, 0; green, 0; blue, 0 }  ,opacity=1 ,rotate=-358.88]  {$+$};
\draw (78.9,187.07) node [anchor=north west][inner sep=0.75pt]  [font=\scriptsize,color={rgb, 255:red, 0; green, 0; blue, 0 }  ,opacity=1 ]  {$-$};
\draw (138.57,180.94) node [anchor=north west][inner sep=0.75pt]  [font=\scriptsize,rotate=-358.88]  {$+$};
\draw (113.04,224.95) node [anchor=north west][inner sep=0.75pt]  [font=\scriptsize,rotate=-358.88]  {$+$};
\draw (37.8,149.24) node [anchor=north west][inner sep=0.75pt]  [font=\scriptsize,color={rgb, 255:red, 0; green, 0; blue, 0 }  ,opacity=1 ]  {$-$};
\draw (93,147.64) node [anchor=north west][inner sep=0.75pt]  [font=\scriptsize,color={rgb, 255:red, 0; green, 0; blue, 0 }  ,opacity=1 ]  {$-$};
\draw (53.18,185.21) node [anchor=north west][inner sep=0.75pt]  [font=\scriptsize,color={rgb, 255:red, 0; green, 0; blue, 0 }  ,opacity=1 ,rotate=-358.88]  {$+$};
\draw (29.61,68.89) node [anchor=north west][inner sep=0.75pt]  [font=\small]  {$u( t)$};
\draw (78.39,70.79) node [anchor=north west][inner sep=0.75pt]  [font=\scriptsize,color={rgb, 255:red, 0; green, 0; blue, 0 }  ,opacity=1 ,rotate=-358.88]  {$+$};

\end{tikzpicture}

    \caption{A digraph $\mathcal{G}_4$ }
    \label{L-dil1}
\end{subfigure}%
\hfill
  \begin{subfigure}{0.25\textwidth}
    \centering

\tikzset{every picture/.style={scale=0.85pt}} 

\begin{tikzpicture}[x=0.75pt,y=0.75pt,yscale=-1,xscale=1]

\draw  [color={rgb, 255:red, 0; green, 0; blue, 0 }  ,draw opacity=0 ][fill={rgb, 255:red, 74; green, 144; blue, 226 }  ,fill opacity=0.1 ] (234.53,140.71) -- (362.58,142.03) -- (362.4,159.17) -- (234.35,157.85) -- cycle ;
\draw  [color={rgb, 255:red, 0; green, 0; blue, 0 }  ,draw opacity=0 ][fill={rgb, 255:red, 74; green, 144; blue, 226 }  ,fill opacity=0.1 ] (236.95,1.44) -- (365,2.76) -- (364.82,19.9) -- (236.77,18.58) -- cycle ;
\draw  [color={rgb, 255:red, 0; green, 0; blue, 0 }  ,draw opacity=0 ][fill={rgb, 255:red, 74; green, 144; blue, 226 }  ,fill opacity=0.1 ] (235.76,45.81) -- (363.8,47.13) -- (363.62,64.27) -- (235.58,62.95) -- cycle ;
\draw  [color={rgb, 255:red, 0; green, 0; blue, 0 }  ,draw opacity=0 ][fill={rgb, 255:red, 74; green, 144; blue, 226 }  ,fill opacity=0.1 ] (235.03,92.7) -- (363.07,94.02) -- (362.89,111.16) -- (234.85,109.84) -- cycle ;
\draw  [color={rgb, 255:red, 0; green, 0; blue, 0 }  ,draw opacity=0 ][fill={rgb, 255:red, 74; green, 144; blue, 226 }  ,fill opacity=0.1 ] (234.56,185.65) -- (362.61,186.97) -- (362.43,204.11) -- (234.38,202.79) -- cycle ;
\draw [color={rgb, 255:red, 0; green, 0; blue, 0 }  ,draw opacity=1 ]   (295.39,64.15) -- (277.13,90.51) ;
\draw [shift={(275.42,92.98)}, rotate = 304.71] [fill={rgb, 255:red, 0; green, 0; blue, 0 }  ,fill opacity=1 ][line width=0.08]  [draw opacity=0] (10.72,-5.15) -- (0,0) -- (10.72,5.15) -- (7.12,0) -- cycle    ;
\draw [color={rgb, 255:red, 0; green, 0; blue, 0 }  ,draw opacity=1 ]   (305.71,64.26) -- (321.14,91.69) ;
\draw [shift={(322.61,94.3)}, rotate = 240.64] [fill={rgb, 255:red, 0; green, 0; blue, 0 }  ,fill opacity=1 ][line width=0.08]  [draw opacity=0] (10.72,-5.15) -- (0,0) -- (10.72,5.15) -- (7.12,0) -- cycle    ;
\draw [color={rgb, 255:red, 0; green, 0; blue, 0 }  ,draw opacity=1 ][line width=0.75]    (299.77,19.05) -- (299.53,42.56) ;
\draw [shift={(299.5,45.56)}, rotate = 270.59] [fill={rgb, 255:red, 0; green, 0; blue, 0 }  ,fill opacity=1 ][line width=0.08]  [draw opacity=0] (9.82,-4.72) -- (0,0) -- (9.82,4.72) -- (6.52,0) -- cycle    ;
\draw [color={rgb, 255:red, 0; green, 0; blue, 0 }  ,draw opacity=1 ]   (321.43,112.99) -- (303.71,137.98) ;
\draw [shift={(301.98,140.43)}, rotate = 305.35] [fill={rgb, 255:red, 0; green, 0; blue, 0 }  ,fill opacity=1 ][line width=0.08]  [draw opacity=0] (10.72,-5.15) -- (0,0) -- (10.72,5.15) -- (7.12,0) -- cycle    ;
\draw [color={rgb, 255:red, 0; green, 0; blue, 0 }  ,draw opacity=1 ]   (331.76,112.54) -- (348.97,137.89) ;
\draw [shift={(350.65,140.37)}, rotate = 235.84] [fill={rgb, 255:red, 0; green, 0; blue, 0 }  ,fill opacity=1 ][line width=0.08]  [draw opacity=0] (10.72,-5.15) -- (0,0) -- (10.72,5.15) -- (7.12,0) -- cycle    ;
\draw [color={rgb, 255:red, 0; green, 0; blue, 0 }  ,draw opacity=1 ]   (271.07,112.2) -- (270.83,135.71) ;
\draw [shift={(270.8,138.71)}, rotate = 270.59] [fill={rgb, 255:red, 0; green, 0; blue, 0 }  ,fill opacity=1 ][line width=0.08]  [draw opacity=0] (10.72,-5.15) -- (0,0) -- (10.72,5.15) -- (7.12,0) -- cycle    ;
\draw [color={rgb, 255:red, 0; green, 0; blue, 0 }  ,draw opacity=1 ][line width=0.75]    (270.04,158.93) -- (269.79,182.45) ;
\draw [shift={(269.76,185.45)}, rotate = 270.59] [fill={rgb, 255:red, 0; green, 0; blue, 0 }  ,fill opacity=1 ][line width=0.08]  [draw opacity=0] (8.93,-4.29) -- (0,0) -- (8.93,4.29) -- (5.93,0) -- cycle    ;
\draw  [draw opacity=0][fill={rgb, 255:red, 0; green, 192; blue, 248 }  ,fill opacity=1 ] (295.75,55) .. controls (295.77,52.9) and (297.56,51.21) .. (299.73,51.23) .. controls (301.91,51.25) and (303.65,52.98) .. (303.63,55.08) .. controls (303.61,57.19) and (301.83,58.88) .. (299.65,58.85) .. controls (297.48,58.83) and (295.73,57.11) .. (295.75,55) -- cycle ;
\draw  [draw opacity=0][fill={rgb, 255:red, 0; green, 192; blue, 248 }  ,fill opacity=1 ] (267.75,100.83) .. controls (267.77,98.72) and (269.55,97.03) .. (271.73,97.06) .. controls (273.9,97.08) and (275.65,98.8) .. (275.63,100.91) .. controls (275.6,103.01) and (273.82,104.7) .. (271.65,104.68) .. controls (269.47,104.66) and (267.73,102.93) .. (267.75,100.83) -- cycle ;
\draw  [draw opacity=0][fill={rgb, 255:red, 0; green, 192; blue, 248 }  ,fill opacity=1 ] (323.06,101.4) .. controls (323.08,99.29) and (324.86,97.6) .. (327.04,97.62) .. controls (329.21,97.65) and (330.96,99.37) .. (330.93,101.48) .. controls (330.91,103.58) and (329.13,105.27) .. (326.96,105.25) .. controls (324.78,105.23) and (323.04,103.5) .. (323.06,101.4) -- cycle ;
\draw  [draw opacity=0][fill={rgb, 255:red, 0; green, 192; blue, 248 }  ,fill opacity=1 ] (266.24,149.23) .. controls (266.26,147.13) and (268.04,145.44) .. (270.22,145.46) .. controls (272.39,145.48) and (274.14,147.21) .. (274.12,149.31) .. controls (274.09,151.42) and (272.31,153.11) .. (270.14,153.08) .. controls (267.96,153.06) and (266.22,151.34) .. (266.24,149.23) -- cycle ;
\draw  [draw opacity=0][fill={rgb, 255:red, 0; green, 192; blue, 248 }  ,fill opacity=1 ] (294.53,149.9) .. controls (294.55,147.79) and (296.33,146.1) .. (298.51,146.13) .. controls (300.68,146.15) and (302.43,147.87) .. (302.4,149.98) .. controls (302.38,152.08) and (300.6,153.77) .. (298.43,153.75) .. controls (296.25,153.73) and (294.51,152) .. (294.53,149.9) -- cycle ;
\draw  [draw opacity=0][fill={rgb, 255:red, 0; green, 192; blue, 248 }  ,fill opacity=1 ] (348.83,149.06) .. controls (348.85,146.95) and (350.64,145.26) .. (352.81,145.28) .. controls (354.99,145.31) and (356.73,147.03) .. (356.71,149.14) .. controls (356.69,151.24) and (354.91,152.93) .. (352.73,152.91) .. controls (350.56,152.89) and (348.81,151.16) .. (348.83,149.06) -- cycle ;
\draw  [draw opacity=0][fill={rgb, 255:red, 0; green, 192; blue, 248 }  ,fill opacity=1 ] (265.67,195.02) .. controls (265.69,192.91) and (267.48,191.23) .. (269.65,191.25) .. controls (271.83,191.27) and (273.57,193) .. (273.55,195.1) .. controls (273.53,197.21) and (271.75,198.89) .. (269.57,198.87) .. controls (267.4,198.85) and (265.65,197.12) .. (265.67,195.02) -- cycle ;
\draw  [color={rgb, 255:red, 0; green, 0; blue, 0 }  ,draw opacity=0 ][fill={rgb, 255:red, 74; green, 144; blue, 226 }  ,fill opacity=0.1 ] (234.85,230.74) -- (362.9,232.06) -- (362.72,249.2) -- (234.67,247.89) -- cycle ;
\draw [color={rgb, 255:red, 0; green, 0; blue, 0 }  ,draw opacity=1 ][line width=0.75]    (269.38,204.17) -- (269.14,227.69) ;
\draw [shift={(269.11,230.69)}, rotate = 270.59] [fill={rgb, 255:red, 0; green, 0; blue, 0 }  ,fill opacity=1 ][line width=0.08]  [draw opacity=0] (8.93,-4.29) -- (0,0) -- (8.93,4.29) -- (5.93,0) -- cycle    ;
\draw  [draw opacity=0][fill={rgb, 255:red, 0; green, 192; blue, 248 }  ,fill opacity=1 ] (265.34,239.95) .. controls (265.36,237.84) and (267.14,236.15) .. (269.31,236.17) .. controls (271.49,236.2) and (273.23,237.92) .. (273.21,240.03) .. controls (273.19,242.13) and (271.41,243.82) .. (269.23,243.8) .. controls (267.06,243.78) and (265.31,242.05) .. (265.34,239.95) -- cycle ;
\draw  [draw opacity=0][fill={rgb, 255:red, 126; green, 211; blue, 33 }  ,fill opacity=1 ] (299.65,6.51) .. controls (301.73,6.49) and (303.43,8.26) .. (303.45,10.46) .. controls (303.48,12.67) and (301.81,14.47) .. (299.74,14.5) .. controls (297.66,14.52) and (295.96,12.75) .. (295.93,10.54) .. controls (295.91,8.34) and (297.57,6.53) .. (299.65,6.51) -- cycle ;
\draw [color={rgb, 255:red, 0; green, 0; blue, 0 }  ,draw opacity=1 ][line width=0.75]    (298.58,158.97) -- (298.33,182.48) ;
\draw [shift={(298.3,185.48)}, rotate = 270.59] [fill={rgb, 255:red, 0; green, 0; blue, 0 }  ,fill opacity=1 ][line width=0.08]  [draw opacity=0] (8.93,-4.29) -- (0,0) -- (8.93,4.29) -- (5.93,0) -- cycle    ;
\draw  [draw opacity=0][fill={rgb, 255:red, 0; green, 192; blue, 248 }  ,fill opacity=1 ] (294.31,194.87) .. controls (294.33,192.76) and (296.11,191.07) .. (298.29,191.09) .. controls (300.46,191.12) and (302.21,192.84) .. (302.19,194.95) .. controls (302.16,197.05) and (300.38,198.74) .. (298.21,198.72) .. controls (296.03,198.7) and (294.29,196.97) .. (294.31,194.87) -- cycle ;
\draw [color={rgb, 255:red, 0; green, 0; blue, 0 }  ,draw opacity=1 ]   (299.43,203.69) -- (299.18,227.21) ;
\draw [shift={(299.15,230.21)}, rotate = 270.59] [fill={rgb, 255:red, 0; green, 0; blue, 0 }  ,fill opacity=1 ][line width=0.08]  [draw opacity=0] (10.72,-5.15) -- (0,0) -- (10.72,5.15) -- (7.12,0) -- cycle    ;
\draw  [draw opacity=0][fill={rgb, 255:red, 0; green, 192; blue, 248 }  ,fill opacity=1 ] (295.6,240.74) .. controls (295.63,238.63) and (297.41,236.94) .. (299.58,236.97) .. controls (301.76,236.99) and (303.5,238.71) .. (303.48,240.82) .. controls (303.46,242.92) and (301.68,244.61) .. (299.5,244.59) .. controls (297.33,244.57) and (295.58,242.84) .. (295.6,240.74) -- cycle ;
\draw  [color={rgb, 255:red, 0; green, 0; blue, 0 }  ,draw opacity=0 ][fill={rgb, 255:red, 74; green, 144; blue, 226 }  ,fill opacity=0.1 ] (234.85,275.91) -- (362.9,277.23) -- (362.72,294.38) -- (234.67,293.06) -- cycle ;
\draw [color={rgb, 255:red, 0; green, 0; blue, 0 }  ,draw opacity=1 ][line width=0.75]    (269.38,249.34) -- (269.14,272.86) ;
\draw [shift={(269.11,275.86)}, rotate = 270.59] [fill={rgb, 255:red, 0; green, 0; blue, 0 }  ,fill opacity=1 ][line width=0.08]  [draw opacity=0] (8.93,-4.29) -- (0,0) -- (8.93,4.29) -- (5.93,0) -- cycle    ;
\draw  [draw opacity=0][fill={rgb, 255:red, 0; green, 192; blue, 248 }  ,fill opacity=1 ] (265.34,285.12) .. controls (265.36,283.01) and (267.14,281.32) .. (269.31,281.35) .. controls (271.49,281.37) and (273.23,283.09) .. (273.21,285.2) .. controls (273.19,287.3) and (271.41,288.99) .. (269.23,288.97) .. controls (267.06,288.95) and (265.31,287.22) .. (265.34,285.12) -- cycle ;
\draw [color={rgb, 255:red, 0; green, 0; blue, 0 }  ,draw opacity=1 ]   (299.43,248.86) -- (299.18,272.38) ;
\draw [shift={(299.15,275.38)}, rotate = 270.59] [fill={rgb, 255:red, 0; green, 0; blue, 0 }  ,fill opacity=1 ][line width=0.08]  [draw opacity=0] (10.72,-5.15) -- (0,0) -- (10.72,5.15) -- (7.12,0) -- cycle    ;
\draw  [draw opacity=0][fill={rgb, 255:red, 0; green, 192; blue, 248 }  ,fill opacity=1 ] (295.6,285.91) .. controls (295.63,283.8) and (297.41,282.11) .. (299.58,282.14) .. controls (301.76,282.16) and (303.5,283.88) .. (303.48,285.99) .. controls (303.46,288.09) and (301.68,289.78) .. (299.5,289.76) .. controls (297.33,289.74) and (295.58,288.01) .. (295.6,285.91) -- cycle ;
\draw  [color={rgb, 255:red, 0; green, 0; blue, 0 }  ,draw opacity=0 ][fill={rgb, 255:red, 250; green, 62; blue, 75 }  ,fill opacity=0.47 ] (241.08,211.3) .. controls (241.08,210) and (242.14,208.94) .. (243.44,208.94) -- (338.37,208.94) .. controls (339.67,208.94) and (340.73,210) .. (340.73,211.3) -- (340.73,218.36) .. controls (340.73,219.65) and (339.67,220.71) .. (338.37,220.71) -- (243.44,220.71) .. controls (242.14,220.71) and (241.08,219.65) .. (241.08,218.36) -- cycle ;
\draw  [color={rgb, 255:red, 0; green, 0; blue, 0 }  ,draw opacity=0 ][fill={rgb, 255:red, 250; green, 62; blue, 75 }  ,fill opacity=0.47 ] (239.43,165.08) .. controls (239.43,163.78) and (240.48,162.72) .. (241.78,162.72) -- (336.72,162.72) .. controls (338.02,162.72) and (339.07,163.78) .. (339.07,165.08) -- (339.07,172.13) .. controls (339.07,173.43) and (338.02,174.49) .. (336.72,174.49) -- (241.78,174.49) .. controls (240.48,174.49) and (239.43,173.43) .. (239.43,172.13) -- cycle ;

\draw (246.34,163.4) node [anchor=north west][inner sep=0.75pt]  [font=\scriptsize,color={rgb, 255:red, 0; green, 0; blue, 0 }  ,opacity=1 ,rotate=-359.47]  {$+$};
\draw (314.75,1.35) node [anchor=north west][inner sep=0.75pt]  [color={rgb, 255:red, 0; green, 0; blue, 0 }  ,opacity=1 ,rotate=-359.98]  {$1$};
\draw (316,46.68) node [anchor=north west][inner sep=0.75pt]  [color={rgb, 255:red, 0; green, 0; blue, 0 }  ,opacity=1 ,rotate=-359.98]  {$2$};
\draw (334.13,92.99) node [anchor=north west][inner sep=0.75pt]  [color={rgb, 255:red, 0; green, 0; blue, 0 }  ,opacity=1 ,rotate=-359.98]  {$3$};
\draw (250.29,91.79) node [anchor=north west][inner sep=0.75pt]  [color={rgb, 255:red, 0; green, 0; blue, 0 }  ,opacity=1 ,rotate=-359.98]  {$4$};
\draw (303.39,141.92) node [anchor=north west][inner sep=0.75pt]  [color={rgb, 255:red, 0; green, 0; blue, 0 }  ,opacity=1 ,rotate=-359.98]  {$5$};
\draw (250.17,141.31) node [anchor=north west][inner sep=0.75pt]  [color={rgb, 255:red, 0; green, 0; blue, 0 }  ,opacity=1 ,rotate=-359.98]  {$6$};
\draw (274.63,21.98) node [anchor=north west][inner sep=0.75pt]  [font=\scriptsize,color={rgb, 255:red, 0; green, 0; blue, 0 }  ,opacity=1 ,rotate=-0.59]  {$+$};
\draw (350.21,112.45) node [anchor=north west][inner sep=0.75pt]  [font=\scriptsize,color={rgb, 255:red, 0; green, 0; blue, 0 }  ,opacity=1 ,rotate=-0.59]  {$+$};
\draw (247.79,208.8) node [anchor=north west][inner sep=0.75pt]  [font=\scriptsize,color={rgb, 255:red, 0; green, 0; blue, 0 }  ,opacity=1 ,rotate=-0.59]  {$-$};
\draw (336.38,141.36) node [anchor=north west][inner sep=0.75pt]  [rotate=-0.59]  {$7$};
\draw (248.88,187.22) node [anchor=north west][inner sep=0.75pt]  [color={rgb, 255:red, 0; green, 0; blue, 0 }  ,opacity=1 ,rotate=-359.98]  {$5$};
\draw (249.84,233.71) node [anchor=north west][inner sep=0.75pt]  [color={rgb, 255:red, 0; green, 0; blue, 0 }  ,opacity=1 ,rotate=-359.98]  {$4$};
\draw (249.37,112.8) node [anchor=north west][inner sep=0.75pt]  [font=\scriptsize,color={rgb, 255:red, 0; green, 0; blue, 0 }  ,opacity=1 ,rotate=-359.47]  {$+$};
\draw (293.61,112.21) node [anchor=north west][inner sep=0.75pt]  [font=\scriptsize,color={rgb, 255:red, 0; green, 0; blue, 0 }  ,opacity=1 ,rotate=-359.47]  {$+$};
\draw (257.25,70.21) node [anchor=north west][inner sep=0.75pt]  [font=\scriptsize,color={rgb, 255:red, 0; green, 0; blue, 0 }  ,opacity=1 ,rotate=-0.59]  {$-$};
\draw (329.7,71.15) node [anchor=north west][inner sep=0.75pt]  [font=\scriptsize,color={rgb, 255:red, 0; green, 0; blue, 0 }  ,opacity=1 ,rotate=-0.59]  {$-$};
\draw (303.87,186.62) node [anchor=north west][inner sep=0.75pt]  [color={rgb, 255:red, 0; green, 0; blue, 0 }  ,opacity=1 ,rotate=-359.98]  {$4$};
\draw (311.21,163.32) node [anchor=north west][inner sep=0.75pt]  [font=\scriptsize,color={rgb, 255:red, 0; green, 0; blue, 0 }  ,opacity=1 ,rotate=-0.59]  {$-$};
\draw (307.33,231.31) node [anchor=north west][inner sep=0.75pt]  [color={rgb, 255:red, 0; green, 0; blue, 0 }  ,opacity=1 ,rotate=-359.98]  {$6$};
\draw (310.79,208.85) node [anchor=north west][inner sep=0.75pt]  [font=\scriptsize,color={rgb, 255:red, 0; green, 0; blue, 0 }  ,opacity=1 ,rotate=-359.47]  {$+$};
\draw (210.81,1.02) node [anchor=north west][inner sep=0.75pt]  [color={rgb, 255:red, 91; green, 130; blue, 49 }  ,opacity=1 ]  {$L_{1}$};
\draw (209.05,44.53) node [anchor=north west][inner sep=0.75pt]  [color={rgb, 255:red, 91; green, 130; blue, 49 }  ,opacity=1 ]  {$L_{2}$};
\draw (207.28,185) node [anchor=north west][inner sep=0.75pt]  [color={rgb, 255:red, 91; green, 130; blue, 49 }  ,opacity=1 ]  {$L_{5}$};
\draw (207.28,139.69) node [anchor=north west][inner sep=0.75pt]  [color={rgb, 255:red, 91; green, 130; blue, 49 }  ,opacity=1 ]  {$L_{4}$};
\draw (209.05,92.56) node [anchor=north west][inner sep=0.75pt]  [color={rgb, 255:red, 91; green, 130; blue, 49 }  ,opacity=1 ]  {$L_{3}$};
\draw (206.4,232.13) node [anchor=north west][inner sep=0.75pt]  [color={rgb, 255:red, 91; green, 130; blue, 49 }  ,opacity=1 ]  {$L_{6}$};
\draw (206.4,277.31) node [anchor=north west][inner sep=0.75pt]  [color={rgb, 255:red, 91; green, 130; blue, 49 }  ,opacity=1 ]  {$L_{7}$};
\draw (249.96,276.48) node [anchor=north west][inner sep=0.75pt]  [color={rgb, 255:red, 0; green, 0; blue, 0 }  ,opacity=1 ,rotate=-359.98]  {$6$};
\draw (308.26,276.56) node [anchor=north west][inner sep=0.75pt]  [color={rgb, 255:red, 0; green, 0; blue, 0 }  ,opacity=1 ,rotate=-359.98]  {$5$};
\draw (311.16,254.03) node [anchor=north west][inner sep=0.75pt]  [font=\scriptsize,color={rgb, 255:red, 0; green, 0; blue, 0 }  ,opacity=1 ,rotate=-359.47]  {$+$};
\draw (246.53,255.72) node [anchor=north west][inner sep=0.75pt]  [font=\scriptsize,color={rgb, 255:red, 0; green, 0; blue, 0 }  ,opacity=1 ,rotate=-359.47]  {$+$};

\end{tikzpicture}
    \caption{A signed layered graph representation of  $\mathcal{G}_4$}
    \label{LD-2}
\end{subfigure}%
 \caption{A digraph exhibiting layer dilations at $L_4$ and $L_5$, shown in its signed layered representation.}
 \label{fig:signed layered graph}
\end{figure} 

The use of a signed layered graph for a given digraph allows us to have an independent representation of every edge; hence, every walk between every pair of nodes. The relationship between a signed layered graph and the controllability matrix is given in the following lemma:

\begin{lemma} \label{lem1:lemma1}
Consider a digraph $\mathcal{G(A,B)}$. If the $(j,k)^{th}$ entry of its controllability matrix is nonzero, then node $j$ belongs to layer $L_k$ in the signed layered graph of the digraph $\mathcal{G(A,B)}$.
\end{lemma}
\begin{proof}
\normalfont
In a signed layered graph, each layer represents the distance from the leader node. $L_1$ contains the leader node; hence, the column $\mathcal{C}_{(:,1)}\neq 0$. By Remark \ref{Remark1}, for $k\in\{1,..,n-1\}$, if $\mathcal{C}_{(j,k+1)}\neq 0$, equivalently, $[\Psi_k]_j\neq0$, then there exists a walk of length $k$ from the leader to follower $j$. Therefore, the node $j \in V_{L_{k}}$ Hence proved. \hfill $\qed$
\end{proof}

On the track to deriving graph-theoretical conditions for $\mathcal{SS}$ herdability, we now begin analyzing the scenarios in which the digraph fails to be herdable. One such scenario, which has been studied in the literature \cite{de2023herdability}, \cite{she2019characterizing} only for acyclic graphs, is layer dilation.
\subsubsection{Layer Dilation} \label{LD para}
As defined in Section \ref{SD para}, a digraph $\mathcal{G}$ is said to have a signed dilation if the outgoing edges from any node have different signs. Similarly, we introduce the notion of a \textit{layer dilation}, wherein the outgoing edges from any layer have different signs. More precisely, a signed layered graph $\mathcal{G}_s$ is said to have a layer dilation in the $k^{th}$ layer if the outgoing edges of $V_{L_k}$ have different signs. For example, there is a layer dilation in layers $L_4$ and $L_5$ for the signed network, as shown in Fig. \ref{fig:signed layered graph}, as highlighted in red. Notably, the presence of a signed dilation in a digraph $\mathcal{G}$ implies the existence of a layer dilation; however, the converse does not necessarily hold, as shown in Fig. \ref{fig:signed layered graph}. Layer dilation adversely affects $\mathcal{SS}$ herdability in directed trees, as discussed next. 

\begin{proposition}\label{prop3:prop3}
 A directed tree with a layer dilation is not $\mathcal{SS}$ herdable.
\end{proposition}

\begin{proof}
\normalfont
By definition, the signed layer graph representation of a directed tree is identical to the tree itself. Hence, the nodes do not repeat in the signed layer graph and have unique paths from the leader. 
Lemma~\ref{lem1:lemma1} establishes that nodes in layer $L_k$ map to the $k^{\text{th}}$ column of the controllability matrix such that $[\Psi_{k-1}]_j\neq 0$. Furthermore, because of the layer dilation in $L_k$ with outgoing edges to $L_{k+1}$, the column $[\Psi_k]$ contains nonzero entries of differing signs. Even if the parameters of $\mathcal{A}$ and $\mathcal{B}$ vary, the column of $\mathcal{C}(\mathcal{A,B})$ corresponding to layer dilation cannot become unisigned. In this case, the span of the columns of $\mathcal{C}(\mathcal{A,B})$ cannot generate a positive vector for any parametric realization of $\mathcal{A,B}$ given by $a_{ij} \in \mathbb{R}^{+}$ and $b_i \in \mathbb{R} \setminus \{0\}$.

Alternatively, we can prove the proposition as follows. Since in a tree graph all the nodes are reached by a unique path, a layer dilation in $L_k$ , the entries appear in $[\Psi_{k-1}]$ with different signs. Hence for every parametric realization, there exists a corresponding $\bm{y}~\in \mathbb{R}^n$ such that $\mathcal{C}(\mathcal{A,B})^\top \bm{y}=0$. By Proposition \ref{test for H}, it is not herdable for any realization. Hence, proved. 
\hfill $\qed$
\end{proof}
In the case of general (acyclic and cyclic) digraphs, the effect of signed dilation on $\mathcal{SS}$ herdability is nontrivial. We explore this in the following section. 
\section{Necessary and Sufficient conditions for $\mathcal{SS}$ Herdability} \label{Sec:5}
We start this section with the following question: Can a digraph be herdable despite having a layer dilation or signed dilation? Although the effect of signed dilation was introduced in \cite{ruf2018herdable}, to the best of our knowledge, no prior study has directly addressed the herdability condition in such cases. In this section, we analyze how these conditions affect $SS$ herdability and investigate the graph-theoretical conditions for $\mathcal{SS}$ herdability in digraphs with arbitrary topologies.

\subsection{Signed dilation set} \label{SD_set para}
Consider a digraph $\mathcal{G(A,B)}$ and its corresponding layered graph $\mathcal{G}_s$. For any node $i \in L_k,~k\in\{1,...,n\}$, we define the following sets:
\begin{subequations}\label{SD-set eqn}
\begin{align}
    \Delta_i &\triangleq \{\, j \in \mathcal{V} : (i,j) \in \mathcal{E} \,\}, \label{eq:Delta_i}\\
    \Delta_i^P &\triangleq \{\, j \in \mathcal{V} : (i,j) \in \mathcal{E}_+ \,\}, \label{eq:Delta_iP}\\
    \Delta_i^N &\triangleq \{\, j \in \mathcal{V} : (i,j) \in \mathcal{E}_- \,\}. \label{eq:Delta_iN}
\end{align}

\end{subequations}

Then, obviously $\Delta_i=\Delta_i^P\cup \Delta_i^N$. We already know that every node in $\Delta_i$ can be herded by node $i$ in the absence of a signed dilation. 

Let us assume that there is a signed dilation associated with the nodes in $\Delta_i$. With Lemma \ref{lem1:lemma1} in view, the column(s) of the controllability matrix associated with $\Delta_i$ are not unisigned. Consequently, the nodes in either $\Delta_i^P$ or $\Delta_i^N$ can be herded. For example, consider the digraph $\mathcal{G}_5$ shown in Fig. ~\ref{eg-SD}(\subref{SD_1}) whose corresponding $\mathcal{G}_s$ is given in Fig. ~\ref{eg-SD}(\subref{SD-2}). 
For node 2, $\Delta_2 = \{4,5,6\}$, and additionally, there exists a signed dilation originating from it at $L_2$. 
The corresponding $\mathfrak{S}(\mathcal{C}(\mathcal{A,B}))$ matrix, $\mathfrak{S}(\mathcal{C}_5)$, given in \eqref{SSc_M5}, shows that the entries of $\Psi_2$ have different signs. Hence, either $\Delta_i^P = \{5\}$ or $\Delta_i^N = \{4,6\}$ can be herded by node $2$, but not both simultaneously.

\begin{figure}[ht]
  \centering
  \begin{subfigure}{0.19\textwidth}
    \centering 
\tikzset{every picture/.style={scale=0.75pt}} 

\begin{tikzpicture}[x=0.75pt,y=0.75pt,yscale=-1,xscale=1]

\draw [color={rgb, 255:red, 0; green, 0; blue, 0 }  ,draw opacity=1 ]   (512.96,317.39) -- (513.75,334.98) ;
\draw [shift={(513.88,337.97)}, rotate = 267.43] [fill={rgb, 255:red, 0; green, 0; blue, 0 }  ,fill opacity=1 ][line width=0.08]  [draw opacity=0] (7.14,-3.43) -- (0,0) -- (7.14,3.43) -- (4.74,0) -- cycle    ;
\draw [color={rgb, 255:red, 0; green, 0; blue, 0 }  ,draw opacity=1 ]   (480.69,264.22) -- (505.22,297.01) ;
\draw [shift={(507.01,299.41)}, rotate = 233.2] [fill={rgb, 255:red, 0; green, 0; blue, 0 }  ,fill opacity=1 ][line width=0.08]  [draw opacity=0] (7.14,-3.43) -- (0,0) -- (7.14,3.43) -- (4.74,0) -- cycle    ;
\draw [color={rgb, 255:red, 0; green, 0; blue, 0 }  ,draw opacity=1 ]   (475.01,265.41) -- (450.82,297.52) ;
\draw [shift={(449.01,299.91)}, rotate = 307] [fill={rgb, 255:red, 0; green, 0; blue, 0 }  ,fill opacity=1 ][line width=0.08]  [draw opacity=0] (7.14,-3.43) -- (0,0) -- (7.14,3.43) -- (4.74,0) -- cycle    ;
\draw  [draw opacity=0][fill={rgb, 255:red, 0; green, 192; blue, 248 }  ,fill opacity=1 ] (443.41,301.07) .. controls (446.33,301.01) and (448.75,303.36) .. (448.82,306.32) .. controls (448.88,309.29) and (446.57,311.74) .. (443.65,311.8) .. controls (440.73,311.85) and (438.31,309.5) .. (438.24,306.54) .. controls (438.18,303.58) and (440.49,301.13) .. (443.41,301.07) -- cycle ;
\draw  [draw opacity=0][fill={rgb, 255:red, 0; green, 192; blue, 248 }  ,fill opacity=1 ] (511.09,301.41) .. controls (514.01,301.36) and (516.43,303.71) .. (516.49,306.67) .. controls (516.56,309.63) and (514.24,312.08) .. (511.32,312.14) .. controls (508.4,312.2) and (505.98,309.85) .. (505.92,306.89) .. controls (505.85,303.92) and (508.17,301.47) .. (511.09,301.41) -- cycle ;
\draw  [draw opacity=0][fill={rgb, 255:red, 126; green, 211; blue, 33 }  ,fill opacity=1 ] (478.13,250.28) .. controls (481.05,250.22) and (483.47,252.58) .. (483.53,255.54) .. controls (483.6,258.5) and (481.28,260.95) .. (478.36,261.01) .. controls (475.44,261.07) and (473.02,258.72) .. (472.96,255.75) .. controls (472.9,252.79) and (475.21,250.34) .. (478.13,250.28) -- cycle ;
\draw [color={rgb, 255:red, 0; green, 0; blue, 0 }  ,draw opacity=1 ]   (437.31,315.91) -- (407.83,354.53) ;
\draw [shift={(406.01,356.91)}, rotate = 307.36] [fill={rgb, 255:red, 0; green, 0; blue, 0 }  ,fill opacity=1 ][line width=0.08]  [draw opacity=0] (7.14,-3.43) -- (0,0) -- (7.14,3.43) -- (4.74,0) -- cycle    ;
\draw  [draw opacity=0][fill={rgb, 255:red, 0; green, 192; blue, 248 }  ,fill opacity=1 ] (403.41,357.68) .. controls (406.33,357.62) and (408.75,359.98) .. (408.82,362.94) .. controls (408.88,365.9) and (406.57,368.35) .. (403.65,368.41) .. controls (400.73,368.47) and (398.31,366.12) .. (398.24,363.15) .. controls (398.18,360.19) and (400.49,357.74) .. (403.41,357.68) -- cycle ;
\draw  [draw opacity=0][fill={rgb, 255:red, 0; green, 192; blue, 248 }  ,fill opacity=1 ] (513.09,343.7) .. controls (516.01,343.64) and (518.43,345.99) .. (518.49,348.95) .. controls (518.56,351.91) and (516.24,354.36) .. (513.32,354.42) .. controls (510.4,354.48) and (507.98,352.13) .. (507.92,349.17) .. controls (507.85,346.21) and (510.17,343.76) .. (513.09,343.7) -- cycle ;
\draw  [draw opacity=0][fill={rgb, 255:red, 0; green, 192; blue, 248 }  ,fill opacity=1 ] (474.75,359.7) .. controls (477.67,359.64) and (480.09,361.99) .. (480.16,364.95) .. controls (480.22,367.91) and (477.91,370.36) .. (474.99,370.42) .. controls (472.07,370.48) and (469.65,368.13) .. (469.59,365.17) .. controls (469.52,362.21) and (471.83,359.76) .. (474.75,359.7) -- cycle ;
\draw [line width=0.75]    (500.33,341.08) .. controls (468.72,321) and (462.15,320.43) .. (452.33,312) ;
\draw [shift={(503.33,343)}, rotate = 212.71] [fill={rgb, 255:red, 0; green, 0; blue, 0 }  ][line width=0.08]  [draw opacity=0] (7.14,-3.43) -- (0,0) -- (7.14,3.43) -- (4.74,0) -- cycle    ;
\draw [color={rgb, 255:red, 0; green, 0; blue, 0 }  ,draw opacity=1 ]   (448.69,317.22) -- (465.98,351.32) ;
\draw [shift={(467.33,354)}, rotate = 243.12] [fill={rgb, 255:red, 0; green, 0; blue, 0 }  ,fill opacity=1 ][line width=0.08]  [draw opacity=0] (7.14,-3.43) -- (0,0) -- (7.14,3.43) -- (4.74,0) -- cycle    ;
\draw [color={rgb, 255:red, 208; green, 2; blue, 27 }  ,draw opacity=1 ]   (419.11,255.58) .. controls (420.74,253.88) and (422.41,253.85) .. (424.11,255.49) .. controls (425.8,257.12) and (427.47,257.09) .. (429.1,255.4) .. controls (430.74,253.71) and (432.41,253.68) .. (434.1,255.32) .. controls (435.79,256.95) and (437.46,256.92) .. (439.1,255.23) .. controls (440.74,253.54) and (442.41,253.51) .. (444.1,255.14) .. controls (445.79,256.77) and (447.46,256.74) .. (449.1,255.05) .. controls (450.74,253.36) and (452.41,253.33) .. (454.1,254.97) .. controls (455.79,256.6) and (457.46,256.57) .. (459.1,254.88) -- (460.33,254.86) -- (468.33,254.72) ;
\draw [shift={(471.33,254.67)}, rotate = 179] [fill={rgb, 255:red, 208; green, 2; blue, 27 }  ,fill opacity=1 ][line width=0.08]  [draw opacity=0] (8.04,-3.86) -- (0,0) -- (8.04,3.86) -- (5.34,0) -- cycle    ;
\draw  [color={rgb, 255:red, 0; green, 0; blue, 0 }  ,draw opacity=1 ][fill={rgb, 255:red, 208; green, 2; blue, 27 }  ,fill opacity=1 ][line width=0.75]  (419,250.87) .. controls (421.64,250.82) and (423.83,252.88) .. (423.89,255.48) .. controls (423.95,258.08) and (421.85,260.24) .. (419.21,260.29) .. controls (416.57,260.34) and (414.38,258.27) .. (414.32,255.67) .. controls (414.26,253.07) and (416.36,250.92) .. (419,250.87) -- cycle ;

\draw (423.97,295.95) node [anchor=north west][inner sep=0.75pt]  [font=\small,color={rgb, 255:red, 0; green, 0; blue, 0 }  ,opacity=1 ,rotate=-358.27]  {$2$};
\draw (524.65,296.41) node [anchor=north west][inner sep=0.75pt]  [font=\small,color={rgb, 255:red, 0; green, 0; blue, 0 }  ,opacity=1 ,rotate=-358.27]  {$3$};
\draw (454.76,374.43) node [anchor=north west][inner sep=0.75pt]  [font=\small,color={rgb, 255:red, 0; green, 0; blue, 0 }  ,opacity=1 ,rotate=-358.27]  {$5$};
\draw (395.1,372.93) node [anchor=north west][inner sep=0.75pt]  [font=\small,color={rgb, 255:red, 0; green, 0; blue, 0 }  ,opacity=1 ,rotate=-358.27]  {$4$};
\draw (486.57,240.38) node [anchor=north west][inner sep=0.75pt]  [font=\small,color={rgb, 255:red, 0; green, 0; blue, 0 }  ,opacity=1 ,rotate=-358.27]  {$1$};
\draw (406.16,323.23) node [anchor=north west][inner sep=0.75pt]  [font=\footnotesize,color={rgb, 255:red, 0; green, 0; blue, 0 }  ,opacity=1 ,rotate=-358.88]  {$-$};
\draw (443.62,235.93) node [anchor=north west][inner sep=0.75pt]  [font=\scriptsize,color={rgb, 255:red, 0; green, 0; blue, 0 }  ,opacity=1 ,rotate=-358.88]  {$+$};
\draw (517.45,318.07) node [anchor=north west][inner sep=0.75pt]  [font=\footnotesize,rotate=-358.88]  {$+$};
\draw (439.25,332) node [anchor=north west][inner sep=0.75pt]  [font=\footnotesize,rotate=-358.88]  {$+$};
\draw (524.09,346.43) node [anchor=north west][inner sep=0.75pt]  [font=\small,color={rgb, 255:red, 0; green, 0; blue, 0 }  ,opacity=1 ,rotate=-358.27]  {$6$};
\draw (447.49,268.57) node [anchor=north west][inner sep=0.75pt]  [font=\footnotesize,color={rgb, 255:red, 0; green, 0; blue, 0 }  ,opacity=1 ,rotate=-358.88]  {$-$};
\draw (502.16,268.57) node [anchor=north west][inner sep=0.75pt]  [font=\footnotesize,color={rgb, 255:red, 0; green, 0; blue, 0 }  ,opacity=1 ,rotate=-358.88]  {$-$};
\draw (477.49,308.9) node [anchor=north west][inner sep=0.75pt]  [font=\footnotesize,color={rgb, 255:red, 0; green, 0; blue, 0 }  ,opacity=1 ,rotate=-358.88]  {$-$};
\draw (381.83,247.4) node [anchor=north west][inner sep=0.75pt]  [font=\small]  {$u( t)$};

\end{tikzpicture}

     \caption{A digraph $\mathcal{G}_5$}
    \label{SD_1}
  \end{subfigure}
  \begin{subfigure}{0.20\textwidth}
    \centering

\tikzset{every picture/.style={scale=0.75pt}} 

\begin{tikzpicture}[x=0.75pt,y=0.75pt,yscale=-1,xscale=1]

\draw  [color={rgb, 255:red, 0; green, 0; blue, 0 }  ,draw opacity=0 ][fill={rgb, 255:red, 74; green, 144; blue, 226 }  ,fill opacity=0.05 ] (252.99,97.51) -- (411.34,97.51) -- (411.34,118.58) -- (252.99,118.58) -- cycle ;
\draw  [color={rgb, 255:red, 0; green, 0; blue, 0 }  ,draw opacity=0 ][fill={rgb, 255:red, 74; green, 144; blue, 226 }  ,fill opacity=0.05 ] (252.49,149.16) -- (410.84,149.16) -- (410.84,170.23) -- (252.49,170.23) -- cycle ;
\draw  [color={rgb, 255:red, 0; green, 0; blue, 0 }  ,draw opacity=0 ][fill={rgb, 255:red, 74; green, 144; blue, 226 }  ,fill opacity=0.05 ] (251.87,46.56) -- (410.21,46.56) -- (410.21,67.63) -- (251.87,67.63) -- cycle ;
\draw  [color={rgb, 255:red, 74; green, 144; blue, 226 }  ,draw opacity=0.24 ][fill={rgb, 255:red, 0; green, 247; blue, 141 }  ,fill opacity=0.21 ] (302.75,97.13) .. controls (309.38,97.33) and (349.91,162.8) .. (343.21,166) .. controls (336.51,169.2) and (275.01,168.47) .. (267.21,166.67) .. controls (259.41,164.87) and (296.11,96.93) .. (302.75,97.13) -- cycle ;
\draw [color={rgb, 255:red, 208; green, 2; blue, 27 }  ,draw opacity=1 ]   (326.11,69.39) -- (308.37,93.96) ;
\draw [shift={(306.61,96.39)}, rotate = 305.84] [fill={rgb, 255:red, 208; green, 2; blue, 27 }  ,fill opacity=1 ][line width=0.08]  [draw opacity=0] (10.72,-5.15) -- (0,0) -- (10.72,5.15) -- (7.12,0) -- cycle    ;
\draw [color={rgb, 255:red, 208; green, 2; blue, 27 }  ,draw opacity=1 ]   (336.51,68.89) -- (354.39,94.44) ;
\draw [shift={(356.11,96.89)}, rotate = 235.01] [fill={rgb, 255:red, 208; green, 2; blue, 27 }  ,fill opacity=1 ][line width=0.08]  [draw opacity=0] (10.72,-5.15) -- (0,0) -- (10.72,5.15) -- (7.12,0) -- cycle    ;
\draw [color={rgb, 255:red, 208; green, 2; blue, 27 }  ,draw opacity=1 ][fill={rgb, 255:red, 183; green, 176; blue, 176 }  ,fill opacity=1 ]   (365.72,116.15) -- (366.14,145.23) ;
\draw [shift={(366.18,148.23)}, rotate = 269.19] [fill={rgb, 255:red, 208; green, 2; blue, 27 }  ,fill opacity=1 ][line width=0.08]  [draw opacity=0] (10.72,-5.15) -- (0,0) -- (10.72,5.15) -- (7.12,0) -- cycle    ;
\draw [color={rgb, 255:red, 128; green, 128; blue, 128 }  ,draw opacity=1 ]   (298.27,116.54) -- (280.74,146.09) ;
\draw [shift={(279.21,148.67)}, rotate = 300.67] [fill={rgb, 255:red, 128; green, 128; blue, 128 }  ,fill opacity=1 ][line width=0.08]  [draw opacity=0] (10.72,-5.15) -- (0,0) -- (10.72,5.15) -- (7.12,0) -- cycle    ;
\draw [color={rgb, 255:red, 128; green, 128; blue, 128 }  ,draw opacity=1 ]   (306.78,116.1) -- (328.1,144.92) ;
\draw [shift={(329.88,147.33)}, rotate = 233.52] [fill={rgb, 255:red, 128; green, 128; blue, 128 }  ,fill opacity=1 ][line width=0.08]  [draw opacity=0] (10.72,-5.15) -- (0,0) -- (10.72,5.15) -- (7.12,0) -- cycle    ;
\draw  [draw opacity=0][fill={rgb, 255:red, 0; green, 192; blue, 248 }  ,fill opacity=1 ] (302.59,102.5) .. controls (305.51,102.44) and (307.93,104.79) .. (307.99,107.75) .. controls (308.06,110.72) and (305.74,113.17) .. (302.82,113.22) .. controls (299.9,113.28) and (297.48,110.93) .. (297.42,107.97) .. controls (297.35,105.01) and (299.67,102.56) .. (302.59,102.5) -- cycle ;
\draw  [draw opacity=0][fill={rgb, 255:red, 0; green, 192; blue, 248 }  ,fill opacity=1 ] (364.59,102) .. controls (367.51,101.94) and (369.93,104.29) .. (369.99,107.25) .. controls (370.06,110.22) and (367.74,112.67) .. (364.82,112.72) .. controls (361.9,112.78) and (359.48,110.43) .. (359.42,107.47) .. controls (359.35,104.51) and (361.67,102.06) .. (364.59,102) -- cycle ;
\draw  [draw opacity=0][fill={rgb, 255:red, 0; green, 192; blue, 248 }  ,fill opacity=1 ] (274.59,153.83) .. controls (277.51,153.77) and (279.93,156.12) .. (279.99,159.09) .. controls (280.06,162.05) and (277.74,164.5) .. (274.82,164.56) .. controls (271.9,164.62) and (269.48,162.26) .. (269.42,159.3) .. controls (269.35,156.34) and (271.67,153.89) .. (274.59,153.83) -- cycle ;
\draw  [draw opacity=0][fill={rgb, 255:red, 0; green, 192; blue, 248 }  ,fill opacity=1 ] (334.26,154.44) .. controls (337.18,154.38) and (339.6,156.74) .. (339.67,159.7) .. controls (339.73,162.66) and (337.42,165.11) .. (334.5,165.17) .. controls (331.58,165.23) and (329.16,162.88) .. (329.09,159.91) .. controls (329.03,156.95) and (331.34,154.5) .. (334.26,154.44) -- cycle ;
\draw  [draw opacity=0][fill={rgb, 255:red, 0; green, 192; blue, 248 }  ,fill opacity=1 ] (366.26,152.94) .. controls (369.18,152.88) and (371.6,155.24) .. (371.67,158.2) .. controls (371.73,161.16) and (369.42,163.61) .. (366.5,163.67) .. controls (363.58,163.73) and (361.16,161.38) .. (361.09,158.41) .. controls (361.03,155.45) and (363.34,153) .. (366.26,152.94) -- cycle ;
\draw [color={rgb, 255:red, 208; green, 2; blue, 27 }  ,draw opacity=1 ]   (302.78,117.6) -- (303.62,144.99) ;
\draw [shift={(303.71,147.99)}, rotate = 268.24] [fill={rgb, 255:red, 208; green, 2; blue, 27 }  ,fill opacity=1 ][line width=0.08]  [draw opacity=0] (10.72,-5.15) -- (0,0) -- (10.72,5.15) -- (7.12,0) -- cycle    ;
\draw  [draw opacity=0][fill={rgb, 255:red, 0; green, 192; blue, 248 }  ,fill opacity=1 ] (306.26,153.44) .. controls (309.18,153.38) and (311.6,155.74) .. (311.67,158.7) .. controls (311.73,161.66) and (309.42,164.11) .. (306.5,164.17) .. controls (303.58,164.23) and (301.16,161.88) .. (301.09,158.91) .. controls (301.03,155.95) and (303.34,153.5) .. (306.26,153.44) -- cycle ;
\draw  [draw opacity=0][fill={rgb, 255:red, 126; green, 211; blue, 33 }  ,fill opacity=1 ] (330.92,51.73) .. controls (333.84,51.67) and (336.26,54.03) .. (336.33,56.99) .. controls (336.39,59.95) and (334.08,62.4) .. (331.16,62.46) .. controls (328.24,62.52) and (325.82,60.17) .. (325.76,57.2) .. controls (325.69,54.24) and (328,51.79) .. (330.92,51.73) -- cycle ;

\draw (273.12,102.06) node [anchor=north west][inner sep=0.75pt]  [font=\small,color={rgb, 255:red, 0; green, 0; blue, 0 }  ,opacity=1 ,rotate=-359.39]  {$2$};
\draw (375.21,101.16) node [anchor=north west][inner sep=0.75pt]  [font=\small,color={rgb, 255:red, 0; green, 0; blue, 0 }  ,opacity=1 ,rotate=-359.39]  {$3$};
\draw (256.53,152.54) node [anchor=north west][inner sep=0.75pt]  [font=\small,color={rgb, 255:red, 0; green, 0; blue, 0 }  ,opacity=1 ,rotate=-359.39]  {$4$};
\draw (289.74,151.7) node [anchor=north west][inner sep=0.75pt]  [font=\small,color={rgb, 255:red, 0; green, 0; blue, 0 }  ,opacity=1 ,rotate=-359.39]  {$5$};
\draw (342.59,48.34) node [anchor=north west][inner sep=0.75pt]  [font=\small,color={rgb, 255:red, 0; green, 0; blue, 0 }  ,opacity=1 ,rotate=-358.27]  {$1$};
\draw (266.49,122.98) node [anchor=north west][inner sep=0.75pt]  [font=\scriptsize,color={rgb, 255:red, 0; green, 0; blue, 0 }  ,opacity=1 ,rotate=-358.88]  {$-$};
\draw (372.29,121.82) node [anchor=north west][inner sep=0.75pt]  [font=\scriptsize,rotate=-358.88]  {$+$};
\draw (344.57,151.7) node [anchor=north west][inner sep=0.75pt]  [font=\small,color={rgb, 255:red, 0; green, 0; blue, 0 }  ,opacity=1 ,rotate=-359.39]  {$6$};
\draw (286.45,126.66) node [anchor=north west][inner sep=0.75pt]  [font=\scriptsize,rotate=-358.88]  {$+$};
\draw (298.83,68.82) node [anchor=north west][inner sep=0.75pt]  [font=\footnotesize,color={rgb, 255:red, 0; green, 0; blue, 0 }  ,opacity=1 ,rotate=-358.88]  {$-$};
\draw (377.91,151.04) node [anchor=north west][inner sep=0.75pt]  [font=\small,color={rgb, 255:red, 0; green, 0; blue, 0 }  ,opacity=1 ,rotate=-359.39]  {$6$};
\draw (356.16,69.48) node [anchor=north west][inner sep=0.75pt]  [font=\footnotesize,color={rgb, 255:red, 0; green, 0; blue, 0 }  ,opacity=1 ,rotate=-358.88]  {$-$};
\draw (330.83,122.15) node [anchor=north west][inner sep=0.75pt]  [font=\footnotesize,color={rgb, 255:red, 0; green, 0; blue, 0 }  ,opacity=1 ,rotate=-358.88]  {$-$};

\end{tikzpicture}

\caption{$\Delta_2$ in $L_2$ of $\mathcal{G}_s$ }
    \label{SD-2}
  \end{subfigure}
  
\caption{Signed dilation set $\Delta_2 = \{4,5,6\}$ from node $2$ in $L_2$.  Either $\{4,6\}$ or $\{5\}$ can be herded in $\Delta_2$. 
Here, node $5$ is herded within $\Delta_2$, while node $6$ is herded outside $\Delta_2$.}

\label{eg-SD}
\end{figure}
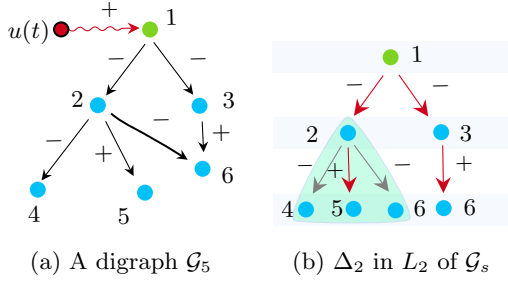
{
\small
\label{SSc_M5}
\begin{align}
\mathfrak{S}(\mathcal{C}_5) =\begin{bNiceMatrix}[margin,columns-width=auto]
    + & ~0 & ~0 & ~0 & ~0 & ~0 \\
    0 & ~- & ~0 & ~0 & ~0 & ~0 \\
    0 & ~- & ~0 & ~0 & ~0 & ~0 \\
    0 & ~0 & ~+ & ~0 & ~0 &~ 0 \\
    0 & ~0 & ~- & ~0 & ~0 & ~0 \\
    0 & ~0 &~+/- & ~0 & ~0 &~ 0\
     \CodeAfter
    \tikz \draw[blue,thick,rounded corners]
    (4-|4)rectangle(7-|3); 
\end{bNiceMatrix}
\end{align}
}
{
\small
\begin{align}\label{SM6} \mathfrak{S}(\mathcal{C}_6)=\begin{bNiceMatrix}[margin,columns-width=auto] ~
+ & 0 & 0 & 0 & 0 & 0\\
0 & - & 0 & + & 0 & -\\
0 & - & 0 & 0 & 0 & 0\\
0 & 0 & + & 0 & - & 0\\
0 & 0 & - & 0 & + & 0\\
0 & 0 & +/- & 0 & - & 0\\
\CodeAfter
\tikz [remember picture, overlay] {
    \draw[blue, thick, rounded corners] (4-|4) rectangle (7-|3);
    \draw[blue, thick, rounded corners] (7-|5) rectangle (4-|6);
    \draw[gray, thick, rounded corners] (3-|4) rectangle (2-|5);
    \draw[gray, thick, rounded corners] (2-|2) rectangle (3-|3);
    \draw[gray, thick, rounded corners] (1-|1) rectangle (2-|2);
    \draw[black, thick, dotted, rounded corners] (4-|4) -- (3-|4);
    \draw[black, thick, dotted, rounded corners] (3-|5) -- (4-|5);
    \draw[black, thick, dotted, rounded corners] (2-|2) -- (2-|2);
    \draw[black, thick, dotted, rounded corners] (3-|3) -- (4-|3);
    \draw[gray, thick, dotted, rounded corners] (2-|2) -- (2-|3);
}
\end{bNiceMatrix}
\end{align}
}

In some digraphs, there is a possibility for a node $i$ with a signed dilation to be reachable through multiple walks of different lengths. Then, it reappears in consecutive layers of $\mathcal{G}_s$. Consequently, the nodes in $\Delta_i$ are jointly associated with multiple columns of the $\mathfrak{S}(\mathcal{C}(\mathcal{A,B}))$ matrix; none of these columns is unisigned. For example, in the digraph $\mathcal{G}_6$ shown in Fig.~\ref{eg-SD2}(\subref{SD_1}), the signed dilation at node $2$ repeats due to the presence of more than one path leading to it. The effect of signed dilation repetition is illustrated in the matrix $\mathfrak{S}(\mathcal{C}(\mathcal{A}, \mathcal{B}))$, associated with the digraph $\mathcal{G}_6$ and denoted by $\mathfrak{S}[\mathcal{C}_6]$ in~\eqref{SM6}. From the highlighted entries corresponding to the walks leading to the dilation set, it is observed that the subvector highlighted at $\Psi_4$ is a scalar multiple of the subvector $\Psi_2$. In such cases, the following result holds.

\begin{figure}[ht]
  \centering
  \begin{subfigure}{0.19\textwidth}
    \centering 
\tikzset{every picture/.style={scale=0.7pt}} 

\begin{tikzpicture}[x=0.75pt,y=0.75pt,yscale=-1,xscale=1]

\draw [color={rgb, 255:red, 0; green, 0; blue, 0 }  ,draw opacity=1 ]   (137.74,152.97) -- (138.53,170.56) ;
\draw [shift={(138.67,173.56)}, rotate = 267.43] [fill={rgb, 255:red, 0; green, 0; blue, 0 }  ,fill opacity=1 ][line width=0.08]  [draw opacity=0] (7.14,-3.43) -- (0,0) -- (7.14,3.43) -- (4.74,0) -- cycle    ;
\draw [color={rgb, 255:red, 0; green, 0; blue, 0 }  ,draw opacity=1 ]   (105.47,99.8) -- (130,132.59) ;
\draw [shift={(131.8,134.99)}, rotate = 233.2] [fill={rgb, 255:red, 0; green, 0; blue, 0 }  ,fill opacity=1 ][line width=0.08]  [draw opacity=0] (7.14,-3.43) -- (0,0) -- (7.14,3.43) -- (4.74,0) -- cycle    ;
\draw [color={rgb, 255:red, 0; green, 0; blue, 0 }  ,draw opacity=1 ]   (99.8,100.99) -- (75.61,133.1) ;
\draw [shift={(73.8,135.49)}, rotate = 307] [fill={rgb, 255:red, 0; green, 0; blue, 0 }  ,fill opacity=1 ][line width=0.08]  [draw opacity=0] (7.14,-3.43) -- (0,0) -- (7.14,3.43) -- (4.74,0) -- cycle    ;
\draw  [draw opacity=0][fill={rgb, 255:red, 0; green, 192; blue, 248 }  ,fill opacity=1 ] (68.2,136.65) .. controls (71.12,136.59) and (73.54,138.94) .. (73.6,141.91) .. controls (73.67,144.87) and (71.35,147.32) .. (68.43,147.38) .. controls (65.51,147.44) and (63.09,145.08) .. (63.03,142.12) .. controls (62.97,139.16) and (65.28,136.71) .. (68.2,136.65) -- cycle ;
\draw  [draw opacity=0][fill={rgb, 255:red, 0; green, 192; blue, 248 }  ,fill opacity=1 ] (135.87,137) .. controls (138.79,136.94) and (141.21,139.29) .. (141.28,142.25) .. controls (141.34,145.22) and (139.03,147.67) .. (136.11,147.72) .. controls (133.19,147.78) and (130.77,145.43) .. (130.71,142.47) .. controls (130.64,139.51) and (132.95,137.06) .. (135.87,137) -- cycle ;
\draw [line width=0.75]    (97.17,191.45) .. controls (104.38,156.72) and (86.36,153.44) .. (75.27,148.84) ;
\draw [shift={(72.56,147.59)}, rotate = 27.79] [fill={rgb, 255:red, 0; green, 0; blue, 0 }  ][line width=0.08]  [draw opacity=0] (7.14,-3.43) -- (0,0) -- (7.14,3.43) -- (4.74,0) -- cycle    ;
\draw  [draw opacity=0][fill={rgb, 255:red, 126; green, 211; blue, 33 }  ,fill opacity=1 ] (102.92,85.86) .. controls (105.84,85.81) and (108.25,88.16) .. (108.32,91.12) .. controls (108.38,94.08) and (106.07,96.53) .. (103.15,96.59) .. controls (100.23,96.65) and (97.81,94.3) .. (97.75,91.34) .. controls (97.68,88.37) and (100,85.92) .. (102.92,85.86) -- cycle ;
\draw [color={rgb, 255:red, 208; green, 2; blue, 27 }  ,draw opacity=1 ]   (103.61,25.03) .. controls (105.27,26.7) and (105.26,28.37) .. (103.59,30.03) .. controls (101.92,31.69) and (101.91,33.36) .. (103.57,35.03) .. controls (105.23,36.7) and (105.22,38.37) .. (103.55,40.03) .. controls (101.88,41.69) and (101.87,43.36) .. (103.53,45.03) .. controls (105.19,46.7) and (105.18,48.37) .. (103.51,50.03) .. controls (101.84,51.7) and (101.84,53.36) .. (103.5,55.03) .. controls (105.16,56.7) and (105.15,58.37) .. (103.48,60.03) -- (103.46,64.97) -- (103.43,72.97) ;
\draw [shift={(103.42,75.97)}, rotate = 270.21] [fill={rgb, 255:red, 208; green, 2; blue, 27 }  ,fill opacity=1 ][line width=0.08]  [draw opacity=0] (8.04,-3.86) -- (0,0) -- (8.04,3.86) -- (5.34,0) -- cycle    ;
\draw  [color={rgb, 255:red, 0; green, 0; blue, 0 }  ,draw opacity=1 ][fill={rgb, 255:red, 208; green, 2; blue, 27 }  ,fill opacity=1 ][line width=0.75]  (103.5,20.32) .. controls (106.14,20.26) and (108.33,22.33) .. (108.39,24.93) .. controls (108.45,27.53) and (106.35,29.68) .. (103.71,29.74) .. controls (101.07,29.79) and (98.88,27.72) .. (98.82,25.12) .. controls (98.76,22.52) and (100.86,20.37) .. (103.5,20.32) -- cycle ;
\draw [color={rgb, 255:red, 0; green, 0; blue, 0 }  ,draw opacity=1 ]   (62.1,151.49) -- (32.62,190.11) ;
\draw [shift={(30.8,192.49)}, rotate = 307.36] [fill={rgb, 255:red, 0; green, 0; blue, 0 }  ,fill opacity=1 ][line width=0.08]  [draw opacity=0] (7.14,-3.43) -- (0,0) -- (7.14,3.43) -- (4.74,0) -- cycle    ;
\draw  [draw opacity=0][fill={rgb, 255:red, 0; green, 192; blue, 248 }  ,fill opacity=1 ] (28.2,193.26) .. controls (31.12,193.2) and (33.54,195.56) .. (33.6,198.52) .. controls (33.67,201.48) and (31.35,203.93) .. (28.43,203.99) .. controls (25.51,204.05) and (23.09,201.7) .. (23.03,198.74) .. controls (22.97,195.77) and (25.28,193.32) .. (28.2,193.26) -- cycle ;
\draw  [draw opacity=0][fill={rgb, 255:red, 0; green, 192; blue, 248 }  ,fill opacity=1 ] (137.87,179.28) .. controls (140.79,179.22) and (143.21,181.57) .. (143.28,184.53) .. controls (143.34,187.5) and (141.03,189.95) .. (138.11,190.01) .. controls (135.19,190.07) and (132.77,187.71) .. (132.71,184.75) .. controls (132.64,181.79) and (134.95,179.34) .. (137.87,179.28) -- cycle ;
\draw  [draw opacity=0][fill={rgb, 255:red, 0; green, 192; blue, 248 }  ,fill opacity=1 ] (99.54,195.28) .. controls (102.46,195.22) and (104.88,197.57) .. (104.94,200.53) .. controls (105.01,203.5) and (102.69,205.95) .. (99.77,206.01) .. controls (96.86,206.07) and (94.44,203.71) .. (94.37,200.75) .. controls (94.31,197.79) and (96.62,195.34) .. (99.54,195.28) -- cycle ;
\draw [line width=0.75]    (90.17,193.57) .. controls (61.91,186.78) and (64.52,162.25) .. (68.76,150.38) ;
\draw [shift={(93.37,194.23)}, rotate = 190.13] [fill={rgb, 255:red, 0; green, 0; blue, 0 }  ][line width=0.08]  [draw opacity=0] (7.14,-3.43) -- (0,0) -- (7.14,3.43) -- (4.74,0) -- cycle    ;
\draw [line width=0.75]    (131.05,170.93) .. controls (107.11,143.44) and (103.28,141.54) .. (80,140.89) ;
\draw [shift={(133.33,173.56)}, rotate = 229.09] [fill={rgb, 255:red, 0; green, 0; blue, 0 }  ][line width=0.08]  [draw opacity=0] (7.14,-3.43) -- (0,0) -- (7.14,3.43) -- (4.74,0) -- cycle    ;

\draw (48.76,131.53) node [anchor=north west][inner sep=0.75pt]  [font=\small,color={rgb, 255:red, 0; green, 0; blue, 0 }  ,opacity=1 ,rotate=-358.27]  {$2$};
\draw (149.43,131.99) node [anchor=north west][inner sep=0.75pt]  [font=\small,color={rgb, 255:red, 0; green, 0; blue, 0 }  ,opacity=1 ,rotate=-358.27]  {$3$};
\draw (79.54,210.01) node [anchor=north west][inner sep=0.75pt]  [font=\small,color={rgb, 255:red, 0; green, 0; blue, 0 }  ,opacity=1 ,rotate=-358.27]  {$5$};
\draw (19.89,208.51) node [anchor=north west][inner sep=0.75pt]  [font=\small,color={rgb, 255:red, 0; green, 0; blue, 0 }  ,opacity=1 ,rotate=-358.27]  {$4$};
\draw (111.35,75.96) node [anchor=north west][inner sep=0.75pt]  [font=\small,color={rgb, 255:red, 0; green, 0; blue, 0 }  ,opacity=1 ,rotate=-358.27]  {$1$};
\draw (30.95,158.82) node [anchor=north west][inner sep=0.75pt]  [font=\footnotesize,color={rgb, 255:red, 0; green, 0; blue, 0 }  ,opacity=1 ,rotate=-358.88]  {$-$};
\draw (66.07,16.98) node [anchor=north west][inner sep=0.75pt]  [font=\small]  {$u( t)$};
\draw (107.41,45.51) node [anchor=north west][inner sep=0.75pt]  [font=\scriptsize,color={rgb, 255:red, 0; green, 0; blue, 0 }  ,opacity=1 ,rotate=-358.88]  {$+$};
\draw (142.24,153.66) node [anchor=north west][inner sep=0.75pt]  [font=\footnotesize,rotate=-358.88]  {$+$};
\draw (55.04,176.58) node [anchor=north west][inner sep=0.75pt]  [font=\footnotesize,rotate=-358.88]  {$+$};
\draw (148.88,182.01) node [anchor=north west][inner sep=0.75pt]  [font=\small,color={rgb, 255:red, 0; green, 0; blue, 0 }  ,opacity=1 ,rotate=-358.27]  {$6$};
\draw (72.28,104.15) node [anchor=north west][inner sep=0.75pt]  [font=\footnotesize,color={rgb, 255:red, 0; green, 0; blue, 0 }  ,opacity=1 ,rotate=-358.88]  {$-$};
\draw (126.95,104.15) node [anchor=north west][inner sep=0.75pt]  [font=\footnotesize,color={rgb, 255:red, 0; green, 0; blue, 0 }  ,opacity=1 ,rotate=-358.88]  {$-$};
\draw (97.61,153.48) node [anchor=north west][inner sep=0.75pt]  [font=\footnotesize,color={rgb, 255:red, 0; green, 0; blue, 0 }  ,opacity=1 ,rotate=-358.88]  {$-$};
\draw (102.28,131.48) node [anchor=north west][inner sep=0.75pt]  [font=\footnotesize,color={rgb, 255:red, 0; green, 0; blue, 0 }  ,opacity=1 ,rotate=-358.88]  {$-$};

\end{tikzpicture}
     \caption{ A digraph $\mathcal{G}_6$}
    \label{SD_21}
  \end{subfigure}
  \begin{subfigure}{0.20\textwidth}
    \centering
\tikzset{every picture/.style={scale=0.7pt}} 

\begin{tikzpicture}[x=0.75pt,y=0.75pt,yscale=-1,xscale=1]

\draw  [color={rgb, 255:red, 0; green, 0; blue, 0 }  ,draw opacity=0 ][fill={rgb, 255:red, 74; green, 144; blue, 226 }  ,fill opacity=0.05 ] (301.78,59.51) -- (460.13,59.51) -- (460.13,80.58) -- (301.78,80.58) -- cycle ;
\draw  [color={rgb, 255:red, 0; green, 0; blue, 0 }  ,draw opacity=0 ][fill={rgb, 255:red, 74; green, 144; blue, 226 }  ,fill opacity=0.05 ] (301.28,111.16) -- (459.63,111.16) -- (459.63,132.23) -- (301.28,132.23) -- cycle ;
\draw  [color={rgb, 255:red, 0; green, 0; blue, 0 }  ,draw opacity=0 ][fill={rgb, 255:red, 74; green, 144; blue, 226 }  ,fill opacity=0.05 ] (300.66,8.56) -- (459,8.56) -- (459,29.63) -- (300.66,29.63) -- cycle ;
\draw  [color={rgb, 255:red, 0; green, 0; blue, 0 }  ,draw opacity=0 ][fill={rgb, 255:red, 74; green, 144; blue, 226 }  ,fill opacity=0.05 ] (302.45,160.51) -- (460.79,160.51) -- (460.79,181.58) -- (302.45,181.58) -- cycle ;
\draw  [color={rgb, 255:red, 0; green, 0; blue, 0 }  ,draw opacity=0 ][fill={rgb, 255:red, 74; green, 144; blue, 226 }  ,fill opacity=0.05 ] (301.95,212.16) -- (460.29,212.16) -- (460.29,233.23) -- (301.95,233.23) -- cycle ;
\draw  [color={rgb, 255:red, 74; green, 144; blue, 226 }  ,draw opacity=0.24 ][fill={rgb, 255:red, 0; green, 247; blue, 141 }  ,fill opacity=0.21 ] (355.53,159.46) .. controls (362.17,159.66) and (402.7,225.13) .. (396,228.33) .. controls (389.3,231.53) and (327.8,230.8) .. (320,229) .. controls (312.2,227.2) and (348.9,159.26) .. (355.53,159.46) -- cycle ;
\draw  [color={rgb, 255:red, 74; green, 144; blue, 226 }  ,draw opacity=0.24 ][fill={rgb, 255:red, 0; green, 247; blue, 141 }  ,fill opacity=0.21 ] (351.53,59.13) .. controls (358.17,59.33) and (398.7,124.8) .. (392,128) .. controls (385.3,131.2) and (323.8,130.47) .. (316,128.67) .. controls (308.2,126.87) and (344.9,58.93) .. (351.53,59.13) -- cycle ;
\draw [color={rgb, 255:red, 208; green, 2; blue, 27 }  ,draw opacity=1 ]   (374.9,31.39) -- (357.16,55.96) ;
\draw [shift={(355.4,58.39)}, rotate = 305.84] [fill={rgb, 255:red, 208; green, 2; blue, 27 }  ,fill opacity=1 ][line width=0.08]  [draw opacity=0] (10.72,-5.15) -- (0,0) -- (10.72,5.15) -- (7.12,0) -- cycle    ;
\draw [color={rgb, 255:red, 208; green, 2; blue, 27 }  ,draw opacity=1 ]   (385.3,30.89) -- (403.18,56.44) ;
\draw [shift={(404.9,58.89)}, rotate = 235.01] [fill={rgb, 255:red, 208; green, 2; blue, 27 }  ,fill opacity=1 ][line width=0.08]  [draw opacity=0] (10.72,-5.15) -- (0,0) -- (10.72,5.15) -- (7.12,0) -- cycle    ;
\draw [color={rgb, 255:red, 208; green, 2; blue, 27 }  ,draw opacity=1 ][fill={rgb, 255:red, 183; green, 176; blue, 176 }  ,fill opacity=1 ]   (414.51,78.15) -- (414.92,107.23) ;
\draw [shift={(414.97,110.23)}, rotate = 269.19] [fill={rgb, 255:red, 208; green, 2; blue, 27 }  ,fill opacity=1 ][line width=0.08]  [draw opacity=0] (10.72,-5.15) -- (0,0) -- (10.72,5.15) -- (7.12,0) -- cycle    ;
\draw [color={rgb, 255:red, 128; green, 128; blue, 128 }  ,draw opacity=1 ]   (347.06,78.54) -- (329.53,108.09) ;
\draw [shift={(328,110.67)}, rotate = 300.67] [fill={rgb, 255:red, 128; green, 128; blue, 128 }  ,fill opacity=1 ][line width=0.08]  [draw opacity=0] (10.72,-5.15) -- (0,0) -- (10.72,5.15) -- (7.12,0) -- cycle    ;
\draw [color={rgb, 255:red, 128; green, 128; blue, 128 }  ,draw opacity=1 ]   (355.57,78.1) -- (376.88,106.92) ;
\draw [shift={(378.67,109.33)}, rotate = 233.52] [fill={rgb, 255:red, 128; green, 128; blue, 128 }  ,fill opacity=1 ][line width=0.08]  [draw opacity=0] (10.72,-5.15) -- (0,0) -- (10.72,5.15) -- (7.12,0) -- cycle    ;
\draw  [draw opacity=0][fill={rgb, 255:red, 0; green, 192; blue, 248 }  ,fill opacity=1 ] (351.37,64.5) .. controls (354.29,64.44) and (356.71,66.79) .. (356.78,69.75) .. controls (356.84,72.72) and (354.53,75.17) .. (351.61,75.22) .. controls (348.69,75.28) and (346.27,72.93) .. (346.21,69.97) .. controls (346.14,67.01) and (348.45,64.56) .. (351.37,64.5) -- cycle ;
\draw  [draw opacity=0][fill={rgb, 255:red, 0; green, 192; blue, 248 }  ,fill opacity=1 ] (413.37,64) .. controls (416.29,63.94) and (418.71,66.29) .. (418.78,69.25) .. controls (418.84,72.22) and (416.53,74.67) .. (413.61,74.72) .. controls (410.69,74.78) and (408.27,72.43) .. (408.21,69.47) .. controls (408.14,66.51) and (410.45,64.06) .. (413.37,64) -- cycle ;
\draw  [draw opacity=0][fill={rgb, 255:red, 0; green, 192; blue, 248 }  ,fill opacity=1 ] (323.37,115.83) .. controls (326.29,115.77) and (328.71,118.12) .. (328.78,121.09) .. controls (328.84,124.05) and (326.53,126.5) .. (323.61,126.56) .. controls (320.69,126.62) and (318.27,124.26) .. (318.21,121.3) .. controls (318.14,118.34) and (320.45,115.89) .. (323.37,115.83) -- cycle ;
\draw  [draw opacity=0][fill={rgb, 255:red, 0; green, 192; blue, 248 }  ,fill opacity=1 ] (383.05,116.44) .. controls (385.97,116.38) and (388.39,118.74) .. (388.45,121.7) .. controls (388.52,124.66) and (386.2,127.11) .. (383.28,127.17) .. controls (380.36,127.23) and (377.95,124.88) .. (377.88,121.91) .. controls (377.82,118.95) and (380.13,116.5) .. (383.05,116.44) -- cycle ;
\draw  [draw opacity=0][fill={rgb, 255:red, 0; green, 192; blue, 248 }  ,fill opacity=1 ] (415.05,114.94) .. controls (417.97,114.88) and (420.39,117.24) .. (420.45,120.2) .. controls (420.52,123.16) and (418.2,125.61) .. (415.28,125.67) .. controls (412.36,125.73) and (409.95,123.38) .. (409.88,120.41) .. controls (409.82,117.45) and (412.13,115) .. (415.05,114.94) -- cycle ;
\draw [color={rgb, 255:red, 208; green, 2; blue, 27 }  ,draw opacity=1 ][fill={rgb, 255:red, 183; green, 176; blue, 176 }  ,fill opacity=1 ]   (355.53,130.79) -- (355.53,155.46) ;
\draw [shift={(355.53,158.46)}, rotate = 270] [fill={rgb, 255:red, 208; green, 2; blue, 27 }  ,fill opacity=1 ][line width=0.08]  [draw opacity=0] (10.72,-5.15) -- (0,0) -- (10.72,5.15) -- (7.12,0) -- cycle    ;
\draw [color={rgb, 255:red, 208; green, 2; blue, 27 }  ,draw opacity=1 ]   (351.57,79.6) -- (352.41,106.99) ;
\draw [shift={(352.5,109.99)}, rotate = 268.24] [fill={rgb, 255:red, 208; green, 2; blue, 27 }  ,fill opacity=1 ][line width=0.08]  [draw opacity=0] (10.72,-5.15) -- (0,0) -- (10.72,5.15) -- (7.12,0) -- cycle    ;
\draw  [draw opacity=0][fill={rgb, 255:red, 0; green, 192; blue, 248 }  ,fill opacity=1 ] (355.05,115.44) .. controls (357.97,115.38) and (360.39,117.74) .. (360.45,120.7) .. controls (360.52,123.66) and (358.2,126.11) .. (355.28,126.17) .. controls (352.36,126.23) and (349.95,123.88) .. (349.88,120.91) .. controls (349.82,117.95) and (352.13,115.5) .. (355.05,115.44) -- cycle ;
\draw [color={rgb, 255:red, 128; green, 128; blue, 128 }  ,draw opacity=1 ]   (351.72,179.2) -- (334.2,208.75) ;
\draw [shift={(332.67,211.33)}, rotate = 300.67] [fill={rgb, 255:red, 128; green, 128; blue, 128 }  ,fill opacity=1 ][line width=0.08]  [draw opacity=0] (10.72,-5.15) -- (0,0) -- (10.72,5.15) -- (7.12,0) -- cycle    ;
\draw [color={rgb, 255:red, 128; green, 128; blue, 128 }  ,draw opacity=1 ]   (360.23,178.76) -- (381.55,207.59) ;
\draw [shift={(383.33,210)}, rotate = 233.52] [fill={rgb, 255:red, 128; green, 128; blue, 128 }  ,fill opacity=1 ][line width=0.08]  [draw opacity=0] (10.72,-5.15) -- (0,0) -- (10.72,5.15) -- (7.12,0) -- cycle    ;
\draw  [draw opacity=0][fill={rgb, 255:red, 0; green, 192; blue, 248 }  ,fill opacity=1 ] (356.04,165.16) .. controls (358.96,165.1) and (361.38,167.46) .. (361.44,170.42) .. controls (361.51,173.38) and (359.19,175.83) .. (356.27,175.89) .. controls (353.36,175.95) and (350.94,173.6) .. (350.87,170.64) .. controls (350.81,167.67) and (353.12,165.22) .. (356.04,165.16) -- cycle ;
\draw  [draw opacity=0][fill={rgb, 255:red, 0; green, 192; blue, 248 }  ,fill opacity=1 ] (387.72,217.11) .. controls (390.64,217.05) and (393.06,219.4) .. (393.12,222.36) .. controls (393.18,225.33) and (390.87,227.78) .. (387.95,227.84) .. controls (385.03,227.9) and (382.61,225.54) .. (382.55,222.58) .. controls (382.48,219.62) and (384.8,217.17) .. (387.72,217.11) -- cycle ;
\draw [color={rgb, 255:red, 208; green, 2; blue, 27 }  ,draw opacity=1 ]   (356.23,180.26) -- (357.07,207.66) ;
\draw [shift={(357.17,210.66)}, rotate = 268.24] [fill={rgb, 255:red, 208; green, 2; blue, 27 }  ,fill opacity=1 ][line width=0.08]  [draw opacity=0] (10.72,-5.15) -- (0,0) -- (10.72,5.15) -- (7.12,0) -- cycle    ;
\draw  [draw opacity=0][fill={rgb, 255:red, 0; green, 192; blue, 248 }  ,fill opacity=1 ] (359.72,216.11) .. controls (362.64,216.05) and (365.06,218.4) .. (365.12,221.36) .. controls (365.18,224.33) and (362.87,226.78) .. (359.95,226.84) .. controls (357.03,226.9) and (354.61,224.54) .. (354.55,221.58) .. controls (354.48,218.62) and (356.8,216.17) .. (359.72,216.11) -- cycle ;
\draw  [draw opacity=0][fill={rgb, 255:red, 0; green, 192; blue, 248 }  ,fill opacity=1 ] (328.71,217.16) .. controls (331.63,217.1) and (334.05,219.46) .. (334.11,222.42) .. controls (334.18,225.38) and (331.86,227.83) .. (328.94,227.89) .. controls (326.02,227.95) and (323.6,225.6) .. (323.54,222.64) .. controls (323.47,219.67) and (325.79,217.22) .. (328.71,217.16) -- cycle ;
\draw  [draw opacity=0][fill={rgb, 255:red, 126; green, 211; blue, 33 }  ,fill opacity=1 ] (379.71,13.73) .. controls (382.63,13.67) and (385.05,16.03) .. (385.11,18.99) .. controls (385.18,21.95) and (382.86,24.4) .. (379.95,24.46) .. controls (377.03,24.52) and (374.61,22.17) .. (374.54,19.2) .. controls (374.48,16.24) and (376.79,13.79) .. (379.71,13.73) -- cycle ;

\draw (321.91,64.06) node [anchor=north west][inner sep=0.75pt]  [font=\small,color={rgb, 255:red, 0; green, 0; blue, 0 }  ,opacity=1 ,rotate=-359.39]  {$2$};
\draw (424,63.16) node [anchor=north west][inner sep=0.75pt]  [font=\small,color={rgb, 255:red, 0; green, 0; blue, 0 }  ,opacity=1 ,rotate=-359.39]  {$3$};
\draw (305.32,114.54) node [anchor=north west][inner sep=0.75pt]  [font=\small,color={rgb, 255:red, 0; green, 0; blue, 0 }  ,opacity=1 ,rotate=-359.39]  {$4$};
\draw (338.53,113.7) node [anchor=north west][inner sep=0.75pt]  [font=\small,color={rgb, 255:red, 0; green, 0; blue, 0 }  ,opacity=1 ,rotate=-359.39]  {$5$};
\draw (391.38,10.34) node [anchor=north west][inner sep=0.75pt]  [font=\small,color={rgb, 255:red, 0; green, 0; blue, 0 }  ,opacity=1 ,rotate=-358.27]  {$1$};
\draw (315.28,84.98) node [anchor=north west][inner sep=0.75pt]  [font=\scriptsize,color={rgb, 255:red, 0; green, 0; blue, 0 }  ,opacity=1 ,rotate=-358.88]  {$-$};
\draw (421.07,83.82) node [anchor=north west][inner sep=0.75pt]  [font=\scriptsize,rotate=-358.88]  {$+$};
\draw (393.36,113.7) node [anchor=north west][inner sep=0.75pt]  [font=\small,color={rgb, 255:red, 0; green, 0; blue, 0 }  ,opacity=1 ,rotate=-359.39]  {$6$};
\draw (338.24,88.66) node [anchor=north west][inner sep=0.75pt]  [font=\scriptsize,rotate=-358.88]  {$+$};
\draw (326.58,164.73) node [anchor=north west][inner sep=0.75pt]  [font=\small,color={rgb, 255:red, 0; green, 0; blue, 0 }  ,opacity=1 ,rotate=-359.39]  {$2$};
\draw (309.98,215.21) node [anchor=north west][inner sep=0.75pt]  [font=\small,color={rgb, 255:red, 0; green, 0; blue, 0 }  ,opacity=1 ,rotate=-359.39]  {$4$};
\draw (343.19,214.37) node [anchor=north west][inner sep=0.75pt]  [font=\small,color={rgb, 255:red, 0; green, 0; blue, 0 }  ,opacity=1 ,rotate=-359.39]  {$5$};
\draw (316.95,188.65) node [anchor=north west][inner sep=0.75pt]  [font=\scriptsize,color={rgb, 255:red, 0; green, 0; blue, 0 }  ,opacity=1 ,rotate=-358.88]  {$-$};
\draw (396.69,214.37) node [anchor=north west][inner sep=0.75pt]  [font=\small,color={rgb, 255:red, 0; green, 0; blue, 0 }  ,opacity=1 ,rotate=-359.39]  {$6$};
\draw (340.91,189.32) node [anchor=north west][inner sep=0.75pt]  [font=\scriptsize,rotate=-358.88]  {$+$};
\draw (347.61,30.82) node [anchor=north west][inner sep=0.75pt]  [font=\footnotesize,color={rgb, 255:red, 0; green, 0; blue, 0 }  ,opacity=1 ,rotate=-358.88]  {$-$};
\draw (363.95,137.65) node [anchor=north west][inner sep=0.75pt]  [font=\scriptsize,color={rgb, 255:red, 0; green, 0; blue, 0 }  ,opacity=1 ,rotate=-358.88]  {$-$};
\draw (426.69,113.04) node [anchor=north west][inner sep=0.75pt]  [font=\small,color={rgb, 255:red, 0; green, 0; blue, 0 }  ,opacity=1 ,rotate=-359.39]  {$6$};
\draw (404.95,31.48) node [anchor=north west][inner sep=0.75pt]  [font=\footnotesize,color={rgb, 255:red, 0; green, 0; blue, 0 }  ,opacity=1 ,rotate=-358.88]  {$-$};
\draw (379.61,84.15) node [anchor=north west][inner sep=0.75pt]  [font=\footnotesize,color={rgb, 255:red, 0; green, 0; blue, 0 }  ,opacity=1 ,rotate=-358.88]  {$-$};
\draw (384.28,182.82) node [anchor=north west][inner sep=0.75pt]  [font=\footnotesize,color={rgb, 255:red, 0; green, 0; blue, 0 }  ,opacity=1 ,rotate=-358.88]  {$-$};

\end{tikzpicture}

\caption{$\Delta_2$ repeating in $L_2$ and $L_4$ of $\mathcal{G}_s$ }
    \label{SD-22}
  \end{subfigure}
  
\caption{The set $\Delta_2$ in $L_2$ reappears in $L_4$ of $\mathcal{G}_s$ due to the edge $(5,2)$; 
hence, the dilation set is repeated in $\mathcal{G}_s$.
}
\label{eg-SD2}
\end{figure}
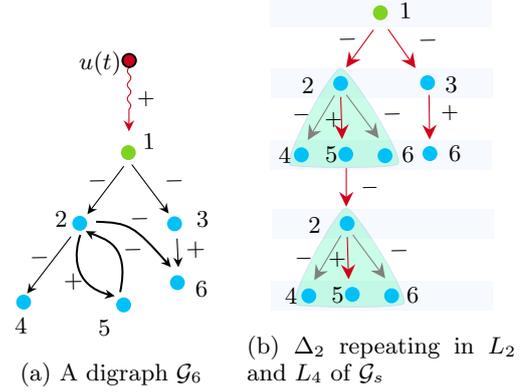
\begin{lemma}
    
 \label{Prop:SD set}

Let $\mathcal{G}(\mathcal{A},\mathcal{B})$ be a digraph associated with system~(\ref{sys1:sys1eqn}), and let $\mathcal{G}_s$ denote its corresponding signed layered graph. For a given node $i$, let $\Delta_i$ denote one of the sets of nodes exhibiting a signed dilation in $\mathcal{G(A,B)}$. Then, any set of nodes in $\Delta_i$ that is unherded by node $i$ remains unherded by it, even if there exist multiple walks to node $i$.
\end{lemma}

\begin{proof}
\normalfont

Let $\mathcal{D}^1_1$ denote the subvector of $[\Psi_p]$ consisting of the entries of $\Delta_i$ in their natural index order, 
and let $\mathcal{D}^2_1$ denote the corresponding subvector of $[\Psi_{p+m}]$ arising from the reappearance of $\Delta_i$ after $m$ layers. 
Since the graph is input connected, $\mathcal{D}^2_1$ is a scalar multiple of $\mathcal{D}^1_1$, determined by the edge weights along the paths. 
Any linear combination of $\mathcal{D}^1_1$ and $\mathcal{D}^2_1$ yields a contradictory system of inequalities for a positive image; 
equivalently, the intersection of the associated half-spaces is empty. Hence, the proposition follows.
\hfill $\qed$
\end{proof}
 

\begin{lemma}\label{SD-herd test}
Let $\mathcal{G}(\mathcal{A},\mathcal{B})$ be a digraph, and let $\mathcal{G}_s$ denote its corresponding signed layered graph. If there exists a signed dilation in $\mathcal{G}(\mathcal{A},\mathcal{B})$, then $\mathcal{G}(\mathcal{A},\mathcal{B})$ is not $\mathcal{SS}$ herdable unless the unherded node in signed dilation set $\Delta_i$ is sign matched outside $\Delta_i$.
\end{lemma}

In the $\mathcal{G}_s$ shown in Fig.~\ref{eg-SD2}(\subref{SD-22}), the nodes of $\Delta_2$ appear in $\Phi_2$ and later in $\Phi_4$ due to the reappearance of the same dilation set in $\mathcal{G}_s$. Either $\Delta_2^N = \{4,6\}$ or $\Delta_2^P = \{5\}$ can be herded. With $6$ already herded by node $3$, any choice leaves at least one node unherded: $4 \notin \Delta_2^P$ and $5 \notin \Delta_2^N$ and therefore the digraph is not $\mathcal{SS}$ herdable.

A digraph $\mathcal{G(A,B)}$ with a signed dilation is $\mathcal{SS}$ herdable only if the unherded nodes in $\Delta_i$ are herded outside the dilation set by a node $j\neq i$. To arrive at graph-theoretic conditions for the $\mathcal{SS}$ herdability of a network even in the presence of such nodes, we introduce the notion of layerwise unisigned graphs in the next section.
\subsection{ Layerwise Unisigned Graphs and Sign Matching}\label{LUG and SM para} 
It is well known that a directed graph is structurally controllable if all the nodes are spanned by a cactus, which is a disjoint union of stems and buds~\cite{lin1974structural}. In this subsection, we investigate a similar graph-theoretic structure whose existence can guarantee $\mathcal{SS}$ herdability. Towards the same, we propose the following definition.
\begin{definition}
\label{def_LUG}
    Consider a signed digraph $\mathcal{G(A,B)}$ and the corresponding signed layered graph $\mathcal{G}_s$. A layerwise unisigned graph $\mathcal{LUG}(\mathcal{G}_s)$ is a subgraph of $\mathcal{G}_s$ in which the signs of all incoming edges of any given layer $L_k,~k\in\{2,..,n\}$ are the same. A node $i \in \mathcal{V}$ is said to be spanned by $\mathcal{LUG}(\mathcal{G}_s)$ if $i \in \mathcal{V}(\mathcal{LUG}(\mathcal{G}_s))$.
\end{definition}

\begin{figure}[ht]
  \centering
   \begin{subfigure}{0.15\textwidth} 
    \centering
\tikzset{every picture/.style={scale=0.75pt}} 

\begin{tikzpicture}[x=0.75pt,y=0.75pt,yscale=-1,xscale=1]

\draw [color={rgb, 255:red, 0; green, 0; blue, 0 }  ,draw opacity=1 ]   (55.3,239.71) -- (55.32,277.32) ;
\draw [shift={(55.33,280.32)}, rotate = 269.97] [fill={rgb, 255:red, 0; green, 0; blue, 0 }  ,fill opacity=1 ][line width=0.08]  [draw opacity=0] (7.14,-3.43) -- (0,0) -- (7.14,3.43) -- (4.74,0) -- cycle    ;
\draw [color={rgb, 255:red, 0; green, 0; blue, 0 }  ,draw opacity=1 ]   (95.47,178.86) -- (126.24,215.75) ;
\draw [shift={(128.16,218.06)}, rotate = 230.18] [fill={rgb, 255:red, 0; green, 0; blue, 0 }  ,fill opacity=1 ][line width=0.08]  [draw opacity=0] (7.14,-3.43) -- (0,0) -- (7.14,3.43) -- (4.74,0) -- cycle    ;
\draw [color={rgb, 255:red, 0; green, 0; blue, 0 }  ,draw opacity=1 ]   (88.01,178.86) -- (58.26,216.1) ;
\draw [shift={(56.39,218.44)}, rotate = 308.62] [fill={rgb, 255:red, 0; green, 0; blue, 0 }  ,fill opacity=1 ][line width=0.08]  [draw opacity=0] (7.14,-3.43) -- (0,0) -- (7.14,3.43) -- (4.74,0) -- cycle    ;
\draw  [draw opacity=0][fill={rgb, 255:red, 0; green, 192; blue, 248 }  ,fill opacity=1 ] (55.7,221.77) .. controls (58.62,221.71) and (61.04,224.33) .. (61.1,227.64) .. controls (61.17,230.95) and (58.85,233.69) .. (55.93,233.75) .. controls (53.01,233.82) and (50.59,231.19) .. (50.53,227.88) .. controls (50.46,224.58) and (52.78,221.84) .. (55.7,221.77) -- cycle ;
\draw  [draw opacity=0][fill={rgb, 255:red, 0; green, 192; blue, 248 }  ,fill opacity=1 ] (130.82,281.38) .. controls (133.74,281.32) and (136.16,283.94) .. (136.22,287.25) .. controls (136.29,290.56) and (133.98,293.3) .. (131.06,293.36) .. controls (128.14,293.43) and (125.72,290.8) .. (125.65,287.49) .. controls (125.59,284.19) and (127.9,281.45) .. (130.82,281.38) -- cycle ;
\draw  [draw opacity=0][fill={rgb, 255:red, 0; green, 192; blue, 248 }  ,fill opacity=1 ] (130.87,221.6) .. controls (133.79,221.54) and (136.21,224.16) .. (136.28,227.47) .. controls (136.34,230.78) and (134.03,233.52) .. (131.11,233.58) .. controls (128.19,233.65) and (125.77,231.02) .. (125.71,227.71) .. controls (125.64,224.41) and (127.95,221.67) .. (130.87,221.6) -- cycle ;
\draw [line width=0.75]    (60.66,220.36) .. controls (81.68,199.49) and (103.04,201.75) .. (122.46,221.9) ;
\draw [shift={(124.25,223.82)}, rotate = 227.75] [fill={rgb, 255:red, 0; green, 0; blue, 0 }  ][line width=0.08]  [draw opacity=0] (7.14,-3.43) -- (0,0) -- (7.14,3.43) -- (4.74,0) -- cycle    ;
\draw    (124.48,233.91) .. controls (121.82,237.67) and (107.54,244.76) .. (94.28,244.42) .. controls (82.27,244.12) and (73.23,238.76) .. (65,234.1) ;
\draw [shift={(62.43,232.66)}, rotate = 28.84] [fill={rgb, 255:red, 0; green, 0; blue, 0 }  ][line width=0.08]  [draw opacity=0] (8.04,-3.86) -- (0,0) -- (8.04,3.86) -- (5.34,0) -- cycle    ;
\draw [color={rgb, 255:red, 0; green, 0; blue, 0 }  ,draw opacity=1 ]   (130.8,237.02) -- (130.82,275.2) ;
\draw [shift={(130.83,278.2)}, rotate = 269.97] [fill={rgb, 255:red, 0; green, 0; blue, 0 }  ,fill opacity=1 ][line width=0.08]  [draw opacity=0] (7.14,-3.43) -- (0,0) -- (7.14,3.43) -- (4.74,0) -- cycle    ;
\draw  [draw opacity=0][fill={rgb, 255:red, 0; green, 192; blue, 248 }  ,fill opacity=1 ] (54.41,282.62) .. controls (57.33,282.56) and (59.75,285.19) .. (59.82,288.49) .. controls (59.88,291.8) and (57.57,294.54) .. (54.65,294.61) .. controls (51.73,294.67) and (49.31,292.04) .. (49.24,288.74) .. controls (49.18,285.43) and (51.49,282.69) .. (54.41,282.62) -- cycle ;
\draw  [draw opacity=0][fill={rgb, 255:red, 126; green, 211; blue, 33 }  ,fill opacity=1 ] (91.79,163.76) .. controls (94.71,163.7) and (97.13,166.32) .. (97.2,169.63) .. controls (97.26,172.94) and (94.95,175.68) .. (92.03,175.74) .. controls (89.11,175.81) and (86.69,173.18) .. (86.62,169.87) .. controls (86.56,166.57) and (88.87,163.83) .. (91.79,163.76) -- cycle ;
\draw [color={rgb, 255:red, 208; green, 2; blue, 27 }  ,draw opacity=1 ]   (92.11,111.58) .. controls (93.77,113.25) and (93.76,114.92) .. (92.09,116.58) .. controls (90.42,118.24) and (90.41,119.91) .. (92.07,121.58) .. controls (93.73,123.25) and (93.72,124.92) .. (92.05,126.58) .. controls (90.38,128.24) and (90.37,129.91) .. (92.03,131.58) .. controls (93.69,133.25) and (93.68,134.92) .. (92.01,136.58) .. controls (90.34,138.25) and (90.34,139.91) .. (92,141.58) .. controls (93.66,143.25) and (93.65,144.92) .. (91.98,146.58) -- (91.96,151.52) -- (91.93,159.52) ;
\draw [shift={(91.92,162.52)}, rotate = 270.21] [fill={rgb, 255:red, 208; green, 2; blue, 27 }  ,fill opacity=1 ][line width=0.08]  [draw opacity=0] (8.04,-3.86) -- (0,0) -- (8.04,3.86) -- (5.34,0) -- cycle    ;
\draw  [color={rgb, 255:red, 0; green, 0; blue, 0 }  ,draw opacity=1 ][fill={rgb, 255:red, 208; green, 2; blue, 27 }  ,fill opacity=1 ][line width=0.75]  (92,106.87) .. controls (94.64,106.82) and (96.83,108.88) .. (96.89,111.48) .. controls (96.95,114.08) and (94.85,116.24) .. (92.21,116.29) .. controls (89.57,116.34) and (87.38,114.27) .. (87.32,111.67) .. controls (87.26,109.07) and (89.36,106.92) .. (92,106.87) -- cycle ;

\draw (37.26,218.08) node [anchor=north west][inner sep=0.75pt]  [font=\small,color={rgb, 255:red, 0; green, 0; blue, 0 }  ,opacity=1 ,rotate=-358.27]  {$2$};
\draw (138.93,217.04) node [anchor=north west][inner sep=0.75pt]  [font=\small,color={rgb, 255:red, 0; green, 0; blue, 0 }  ,opacity=1 ,rotate=-358.27]  {$3$};
\draw (139.86,278.88) node [anchor=north west][inner sep=0.75pt]  [font=\small,color={rgb, 255:red, 0; green, 0; blue, 0 }  ,opacity=1 ,rotate=-358.27]  {$5$};
\draw (35.39,278.06) node [anchor=north west][inner sep=0.75pt]  [font=\small,color={rgb, 255:red, 0; green, 0; blue, 0 }  ,opacity=1 ,rotate=-358.27]  {$4$};
\draw (111.85,153.52) node [anchor=north west][inner sep=0.75pt]  [font=\small,color={rgb, 255:red, 0; green, 0; blue, 0 }  ,opacity=1 ,rotate=-358.27]  {$1$};
\draw (113.41,178.54) node [anchor=north west][inner sep=0.75pt]  [font=\footnotesize,rotate=-358.88]  {$+$};
\draw (140.95,249.37) node [anchor=north west][inner sep=0.75pt]  [font=\footnotesize,color={rgb, 255:red, 0; green, 0; blue, 0 }  ,opacity=1 ,rotate=-358.88]  {$-$};
\draw (56.18,179.84) node [anchor=north west][inner sep=0.75pt]  [font=\footnotesize,color={rgb, 255:red, 0; green, 0; blue, 0 }  ,opacity=1 ,rotate=-358.88]  {$+$};
\draw (87.87,186.13) node [anchor=north west][inner sep=0.75pt]  [font=\footnotesize,rotate=-358.88]  {$+$};
\draw (87.3,248.89) node [anchor=north west][inner sep=0.75pt]  [font=\footnotesize,color={rgb, 255:red, 0; green, 0; blue, 0 }  ,opacity=1 ]  {$-$};
\draw (33.17,251.16) node [anchor=north west][inner sep=0.75pt]  [font=\footnotesize,color={rgb, 255:red, 0; green, 0; blue, 0 }  ,opacity=1 ]  {$-$};
\draw (55.23,128.2) node [anchor=north west][inner sep=0.75pt]  [font=\small]  {$u( t)$};
\draw (102.91,130.06) node [anchor=north west][inner sep=0.75pt]  [font=\scriptsize,color={rgb, 255:red, 0; green, 0; blue, 0 }  ,opacity=1 ,rotate=-358.88]  {$+$};

\end{tikzpicture}

    \caption{A digraph $\mathcal{G}_7$} 
    \label{Fig:5a}
  \end{subfigure}
  \begin{subfigure}{0.16\textwidth}
    \centering 

\tikzset{every picture/.style={scale=0.75pt}} 

\begin{tikzpicture}[x=0.75pt,y=0.75pt,yscale=-1,xscale=1]

\draw [color={rgb, 255:red, 208; green, 2; blue, 27 }  ,draw opacity=0.68 ]   (220.74,66.94) -- (211.3,89.37) ;
\draw [shift={(210.14,92.14)}, rotate = 292.8] [fill={rgb, 255:red, 208; green, 2; blue, 27 }  ,fill opacity=0.68 ][line width=0.08]  [draw opacity=0] (10.72,-5.15) -- (0,0) -- (10.72,5.15) -- (7.12,0) -- cycle    ;
\draw [color={rgb, 255:red, 208; green, 2; blue, 27 }  ,draw opacity=0.68 ]   (226.67,66.6) -- (236.76,89.48) ;
\draw [shift={(237.97,92.22)}, rotate = 246.21] [fill={rgb, 255:red, 208; green, 2; blue, 27 }  ,fill opacity=0.68 ][line width=0.08]  [draw opacity=0] (10.72,-5.15) -- (0,0) -- (10.72,5.15) -- (7.12,0) -- cycle    ;
\draw [color={rgb, 255:red, 208; green, 2; blue, 27 }  ,draw opacity=1 ][fill={rgb, 255:red, 183; green, 176; blue, 176 }  ,fill opacity=1 ]   (203.02,110.42) -- (189.47,133.78) ;
\draw [shift={(187.97,136.37)}, rotate = 300.11] [fill={rgb, 255:red, 208; green, 2; blue, 27 }  ,fill opacity=1 ][line width=0.08]  [draw opacity=0] (10.72,-5.15) -- (0,0) -- (10.72,5.15) -- (7.12,0) -- cycle    ;
\draw [color={rgb, 255:red, 126; green, 211; blue, 33 }  ,draw opacity=1 ]   (208.7,111.24) -- (208.61,131.76) ;
\draw [shift={(208.6,134.76)}, rotate = 270.23] [fill={rgb, 255:red, 126; green, 211; blue, 33 }  ,fill opacity=1 ][line width=0.08]  [draw opacity=0] (10.72,-5.15) -- (0,0) -- (10.72,5.15) -- (7.12,0) -- cycle    ;
\draw  [draw opacity=0][fill={rgb, 255:red, 0; green, 192; blue, 248 }  ,fill opacity=1 ] (205.03,100.63) .. controls (205.03,98.69) and (206.73,97.11) .. (208.82,97.11) .. controls (210.9,97.11) and (212.6,98.69) .. (212.6,100.63) .. controls (212.6,102.57) and (210.9,104.15) .. (208.82,104.15) .. controls (206.73,104.15) and (205.03,102.57) .. (205.03,100.63) -- cycle ;
\draw  [draw opacity=0][fill={rgb, 255:red, 0; green, 192; blue, 248 }  ,fill opacity=1 ] (235.61,100.63) .. controls (235.61,98.69) and (237.3,97.11) .. (239.39,97.11) .. controls (241.48,97.11) and (243.17,98.69) .. (243.17,100.63) .. controls (243.17,102.57) and (241.48,104.15) .. (239.39,104.15) .. controls (237.3,104.15) and (235.61,102.57) .. (235.61,100.63) -- cycle ;
\draw  [draw opacity=0][fill={rgb, 255:red, 0; green, 192; blue, 248 }  ,fill opacity=1 ] (205.37,142.51) .. controls (205.37,140.56) and (207.06,138.99) .. (209.15,138.99) .. controls (211.24,138.99) and (212.93,140.56) .. (212.93,142.51) .. controls (212.93,144.45) and (211.24,146.02) .. (209.15,146.02) .. controls (207.06,146.02) and (205.37,144.45) .. (205.37,142.51) -- cycle ;
\draw  [draw opacity=0][fill={rgb, 255:red, 0; green, 192; blue, 248 }  ,fill opacity=1 ] (235.28,142.74) .. controls (235.28,140.8) and (236.97,139.23) .. (239.06,139.23) .. controls (241.15,139.23) and (242.85,140.8) .. (242.85,142.74) .. controls (242.85,144.69) and (241.15,146.26) .. (239.06,146.26) .. controls (236.97,146.26) and (235.28,144.69) .. (235.28,142.74) -- cycle ;
\draw  [draw opacity=0][fill={rgb, 255:red, 0; green, 192; blue, 248 }  ,fill opacity=1 ] (256.14,142.68) .. controls (256.14,140.74) and (257.84,139.17) .. (259.93,139.17) .. controls (262.02,139.17) and (263.71,140.74) .. (263.71,142.68) .. controls (263.71,144.62) and (262.02,146.2) .. (259.93,146.2) .. controls (257.84,146.2) and (256.14,144.62) .. (256.14,142.68) -- cycle ;
\draw  [draw opacity=0][fill={rgb, 255:red, 126; green, 211; blue, 33 }  ,fill opacity=1 ] (223.05,53.94) .. controls (225.05,53.9) and (226.7,55.51) .. (226.74,57.55) .. controls (226.79,59.58) and (225.21,61.26) .. (223.21,61.31) .. controls (221.22,61.35) and (219.57,59.73) .. (219.52,57.7) .. controls (219.48,55.66) and (221.06,53.98) .. (223.05,53.94) -- cycle ;
\draw  [draw opacity=0][fill={rgb, 255:red, 0; green, 192; blue, 248 }  ,fill opacity=1 ] (182.53,142.98) .. controls (182.53,141.04) and (184.23,139.46) .. (186.32,139.46) .. controls (188.41,139.46) and (190.1,141.04) .. (190.1,142.98) .. controls (190.1,144.92) and (188.41,146.49) .. (186.32,146.49) .. controls (184.23,146.49) and (182.53,144.92) .. (182.53,142.98) -- cycle ;
\draw [color={rgb, 255:red, 208; green, 2; blue, 27 }  ,draw opacity=1 ][fill={rgb, 255:red, 183; green, 176; blue, 176 }  ,fill opacity=1 ]   (239.18,110.89) -- (238.75,132.7) ;
\draw [shift={(238.69,135.7)}, rotate = 271.12] [fill={rgb, 255:red, 208; green, 2; blue, 27 }  ,fill opacity=1 ][line width=0.08]  [draw opacity=0] (10.72,-5.15) -- (0,0) -- (10.72,5.15) -- (7.12,0) -- cycle    ;
\draw [color={rgb, 255:red, 208; green, 2; blue, 27 }  ,draw opacity=1 ]   (244.27,110.18) -- (256.94,134.11) ;
\draw [shift={(258.34,136.77)}, rotate = 242.11] [fill={rgb, 255:red, 208; green, 2; blue, 27 }  ,fill opacity=1 ][line width=0.08]  [draw opacity=0] (10.72,-5.15) -- (0,0) -- (10.72,5.15) -- (7.12,0) -- cycle    ;
\draw [color={rgb, 255:red, 208; green, 2; blue, 27 }  ,draw opacity=1 ]   (205.4,196.51) -- (189.34,219.27) ;
\draw [shift={(187.61,221.72)}, rotate = 305.22] [fill={rgb, 255:red, 208; green, 2; blue, 27 }  ,fill opacity=1 ][line width=0.08]  [draw opacity=0] (10.72,-5.15) -- (0,0) -- (10.72,5.15) -- (7.12,0) -- cycle    ;
\draw [color={rgb, 255:red, 126; green, 211; blue, 33 }  ,draw opacity=1 ][fill={rgb, 255:red, 183; green, 176; blue, 176 }  ,fill opacity=1 ]   (209.16,196.51) -- (209.37,218.64) ;
\draw [shift={(209.4,221.64)}, rotate = 269.45] [fill={rgb, 255:red, 126; green, 211; blue, 33 }  ,fill opacity=1 ][line width=0.08]  [draw opacity=0] (10.72,-5.15) -- (0,0) -- (10.72,5.15) -- (7.12,0) -- cycle    ;
\draw [color={rgb, 255:red, 208; green, 2; blue, 27 }  ,draw opacity=1 ][fill={rgb, 255:red, 183; green, 176; blue, 176 }  ,fill opacity=1 ]   (238.16,197.3) -- (238.2,219.82) ;
\draw [shift={(238.2,222.82)}, rotate = 269.91] [fill={rgb, 255:red, 208; green, 2; blue, 27 }  ,fill opacity=1 ][line width=0.08]  [draw opacity=0] (10.72,-5.15) -- (0,0) -- (10.72,5.15) -- (7.12,0) -- cycle    ;
\draw [color={rgb, 255:red, 208; green, 2; blue, 27 }  ,draw opacity=1 ][fill={rgb, 255:red, 183; green, 176; blue, 176 }  ,fill opacity=1 ]   (242.61,197.1) -- (256.48,220.25) ;
\draw [shift={(258.02,222.82)}, rotate = 239.07] [fill={rgb, 255:red, 208; green, 2; blue, 27 }  ,fill opacity=1 ][line width=0.08]  [draw opacity=0] (10.72,-5.15) -- (0,0) -- (10.72,5.15) -- (7.12,0) -- cycle    ;
\draw  [draw opacity=0][fill={rgb, 255:red, 0; green, 192; blue, 248 }  ,fill opacity=1 ] (205.65,228.95) .. controls (205.65,227.01) and (207.34,225.44) .. (209.43,225.44) .. controls (211.52,225.44) and (213.22,227.01) .. (213.22,228.95) .. controls (213.22,230.9) and (211.52,232.47) .. (209.43,232.47) .. controls (207.34,232.47) and (205.65,230.9) .. (205.65,228.95) -- cycle ;
\draw  [draw opacity=0][fill={rgb, 255:red, 0; green, 192; blue, 248 }  ,fill opacity=1 ] (235.16,228.4) .. controls (235.16,226.46) and (236.85,224.89) .. (238.94,224.89) .. controls (241.03,224.89) and (242.73,226.46) .. (242.73,228.4) .. controls (242.73,230.35) and (241.03,231.92) .. (238.94,231.92) .. controls (236.85,231.92) and (235.16,230.35) .. (235.16,228.4) -- cycle ;
\draw  [draw opacity=0][fill={rgb, 255:red, 0; green, 192; blue, 248 }  ,fill opacity=1 ] (258.85,228.74) .. controls (258.85,226.8) and (260.55,225.22) .. (262.64,225.22) .. controls (264.73,225.22) and (266.42,226.8) .. (266.42,228.74) .. controls (266.42,230.68) and (264.73,232.25) .. (262.64,232.25) .. controls (260.55,232.25) and (258.85,230.68) .. (258.85,228.74) -- cycle ;
\draw  [draw opacity=0][fill={rgb, 255:red, 0; green, 192; blue, 248 }  ,fill opacity=1 ] (180.79,229.62) .. controls (180.79,227.68) and (182.49,226.11) .. (184.58,226.11) .. controls (186.67,226.11) and (188.36,227.68) .. (188.36,229.62) .. controls (188.36,231.57) and (186.67,233.14) .. (184.58,233.14) .. controls (182.49,233.14) and (180.79,231.57) .. (180.79,229.62) -- cycle ;
\draw [color={rgb, 255:red, 208; green, 2; blue, 27 }  ,draw opacity=1 ][fill={rgb, 255:red, 183; green, 176; blue, 176 }  ,fill opacity=1 ]   (204.62,154.28) -- (187.82,179.89) ;
\draw [shift={(186.18,182.4)}, rotate = 303.26] [fill={rgb, 255:red, 208; green, 2; blue, 27 }  ,fill opacity=1 ][line width=0.08]  [draw opacity=0] (10.72,-5.15) -- (0,0) -- (10.72,5.15) -- (7.12,0) -- cycle    ;
\draw [color={rgb, 255:red, 208; green, 2; blue, 27 }  ,draw opacity=1 ]   (208.62,156.26) -- (208.54,176.78) ;
\draw [shift={(208.53,179.78)}, rotate = 270.23] [fill={rgb, 255:red, 208; green, 2; blue, 27 }  ,fill opacity=1 ][line width=0.08]  [draw opacity=0] (10.72,-5.15) -- (0,0) -- (10.72,5.15) -- (7.12,0) -- cycle    ;
\draw  [draw opacity=0][fill={rgb, 255:red, 0; green, 192; blue, 248 }  ,fill opacity=1 ] (205.12,188.27) .. controls (205.12,186.33) and (206.82,184.75) .. (208.91,184.75) .. controls (211,184.75) and (212.69,186.33) .. (212.69,188.27) .. controls (212.69,190.21) and (211,191.79) .. (208.91,191.79) .. controls (206.82,191.79) and (205.12,190.21) .. (205.12,188.27) -- cycle ;
\draw  [draw opacity=0][fill={rgb, 255:red, 0; green, 192; blue, 248 }  ,fill opacity=1 ] (235.64,188.7) .. controls (235.64,186.76) and (237.34,185.19) .. (239.43,185.19) .. controls (241.52,185.19) and (243.21,186.76) .. (243.21,188.7) .. controls (243.21,190.65) and (241.52,192.22) .. (239.43,192.22) .. controls (237.34,192.22) and (235.64,190.65) .. (235.64,188.7) -- cycle ;
\draw  [draw opacity=0][fill={rgb, 255:red, 0; green, 192; blue, 248 }  ,fill opacity=1 ] (257.6,188.21) .. controls (257.6,186.27) and (259.29,184.69) .. (261.38,184.69) .. controls (263.47,184.69) and (265.17,186.27) .. (265.17,188.21) .. controls (265.17,190.15) and (263.47,191.73) .. (261.38,191.73) .. controls (259.29,191.73) and (257.6,190.15) .. (257.6,188.21) -- cycle ;
\draw  [draw opacity=0][fill={rgb, 255:red, 0; green, 192; blue, 248 }  ,fill opacity=1 ] (180.07,187.68) .. controls (180.07,185.74) and (181.76,184.16) .. (183.85,184.16) .. controls (185.94,184.16) and (187.63,185.74) .. (187.63,187.68) .. controls (187.63,189.62) and (185.94,191.2) .. (183.85,191.2) .. controls (181.76,191.2) and (180.07,189.62) .. (180.07,187.68) -- cycle ;
\draw [color={rgb, 255:red, 126; green, 211; blue, 33 }  ,draw opacity=1 ][fill={rgb, 255:red, 183; green, 176; blue, 176 }  ,fill opacity=1 ]   (239.42,156.23) -- (238.99,178.04) ;
\draw [shift={(238.94,181.04)}, rotate = 271.12] [fill={rgb, 255:red, 126; green, 211; blue, 33 }  ,fill opacity=1 ][line width=0.08]  [draw opacity=0] (10.72,-5.15) -- (0,0) -- (10.72,5.15) -- (7.12,0) -- cycle    ;
\draw [color={rgb, 255:red, 208; green, 2; blue, 27 }  ,draw opacity=1 ][fill={rgb, 255:red, 183; green, 176; blue, 176 }  ,fill opacity=1 ]   (242.47,154.75) -- (256.31,179.78) ;
\draw [shift={(257.76,182.4)}, rotate = 241.06] [fill={rgb, 255:red, 208; green, 2; blue, 27 }  ,fill opacity=1 ][line width=0.08]  [draw opacity=0] (10.72,-5.15) -- (0,0) -- (10.72,5.15) -- (7.12,0) -- cycle    ;
\draw  [color={rgb, 255:red, 0; green, 0; blue, 0 }  ,draw opacity=0 ][fill={rgb, 255:red, 74; green, 144; blue, 226 }  ,fill opacity=0.1 ] (167.94,49.51) -- (278.33,49.51) -- (278.33,65.73) -- (167.94,65.73) -- cycle ;
\draw  [color={rgb, 255:red, 0; green, 0; blue, 0 }  ,draw opacity=0 ][fill={rgb, 255:red, 74; green, 144; blue, 226 }  ,fill opacity=0.1 ] (166.98,92.96) -- (277.37,92.96) -- (277.37,109.17) -- (166.98,109.17) -- cycle ;
\draw  [color={rgb, 255:red, 0; green, 0; blue, 0 }  ,draw opacity=0 ][fill={rgb, 255:red, 74; green, 144; blue, 226 }  ,fill opacity=0.1 ] (166.98,133.99) -- (277.37,133.99) -- (277.37,150.2) -- (166.98,150.2) -- cycle ;
\draw  [color={rgb, 255:red, 0; green, 0; blue, 0 }  ,draw opacity=0 ][fill={rgb, 255:red, 74; green, 144; blue, 226 }  ,fill opacity=0.1 ] (166.98,179.36) -- (277.37,179.36) -- (277.37,195.57) -- (166.98,195.57) -- cycle ;
\draw  [color={rgb, 255:red, 0; green, 0; blue, 0 }  ,draw opacity=0 ][fill={rgb, 255:red, 74; green, 144; blue, 226 }  ,fill opacity=0.1 ] (166.98,219.91) -- (277.37,219.91) -- (277.37,236.12) -- (166.98,236.12) -- cycle ;
\draw [color={rgb, 255:red, 126; green, 211; blue, 33 }  ,draw opacity=0.54 ][line width=0.75]    (220.74,66.94) -- (213.1,86.21) ;
\draw [shift={(212,89)}, rotate = 291.6] [fill={rgb, 255:red, 126; green, 211; blue, 33 }  ,fill opacity=0.54 ][line width=0.08]  [draw opacity=0] (11.61,-5.58) -- (0,0) -- (11.61,5.58) -- (7.71,0) -- cycle    ;
\draw [color={rgb, 255:red, 126; green, 211; blue, 33 }  ,draw opacity=0.54 ][line width=0.75]    (226.67,66.6) -- (234.85,86.23) ;
\draw [shift={(236,89)}, rotate = 247.39] [fill={rgb, 255:red, 126; green, 211; blue, 33 }  ,fill opacity=0.54 ][line width=0.08]  [draw opacity=0] (11.61,-5.58) -- (0,0) -- (11.61,5.58) -- (7.71,0) -- cycle    ;

\draw (192.42,92.26) node [anchor=north west][inner sep=0.75pt]  [font=\small,color={rgb, 255:red, 0; green, 0; blue, 0 }  ,opacity=1 ,rotate=-359.39]  {$2$};
\draw (245.12,93.1) node [anchor=north west][inner sep=0.75pt]  [font=\small,color={rgb, 255:red, 0; green, 0; blue, 0 }  ,opacity=1 ,rotate=-359.39]  {$3$};
\draw (172.14,135.16) node [anchor=north west][inner sep=0.75pt]  [font=\small,color={rgb, 255:red, 0; green, 0; blue, 0 }  ,opacity=1 ,rotate=-359.39]  {$4$};
\draw (262.22,135.59) node [anchor=north west][inner sep=0.75pt]  [font=\small,color={rgb, 255:red, 0; green, 0; blue, 0 }  ,opacity=1 ,rotate=-359.39]  {$5$};
\draw (196.6,68.37) node [anchor=north west][inner sep=0.75pt]  [font=\scriptsize,color={rgb, 255:red, 0; green, 0; blue, 0 }  ,opacity=1 ]  {$+$};
\draw (238.7,68.19) node [anchor=north west][inner sep=0.75pt]  [font=\scriptsize,color={rgb, 255:red, 0; green, 0; blue, 0 }  ,opacity=1 ]  {$+$};
\draw (211.48,113.81) node [anchor=north west][inner sep=0.75pt]  [font=\scriptsize,color={rgb, 255:red, 0; green, 0; blue, 0 }  ,opacity=1 ]  {$+$};
\draw (179.24,114.66) node [anchor=north west][inner sep=0.75pt]  [font=\scriptsize,color={rgb, 255:red, 0; green, 0; blue, 0 }  ,opacity=1 ]  {$-$};
\draw (229.31,49.4) node [anchor=north west][inner sep=0.75pt]  [font=\small,color={rgb, 255:red, 0; green, 0; blue, 0 }  ,opacity=1 ,rotate=-358.27]  {$1$};
\draw (196.17,221.19) node [anchor=north west][inner sep=0.75pt]  [font=\small,color={rgb, 255:red, 0; green, 0; blue, 0 }  ,opacity=1 ,rotate=-1.04]  {$3$};
\draw (242.55,135.76) node [anchor=north west][inner sep=0.75pt]  [font=\small,color={rgb, 255:red, 0; green, 0; blue, 0 }  ,opacity=1 ,rotate=-359.39]  {$2$};
\draw (224.23,114.59) node [anchor=north west][inner sep=0.75pt]  [font=\scriptsize,color={rgb, 255:red, 0; green, 0; blue, 0 }  ,opacity=1 ]  {$-$};
\draw (169.17,221.17) node [anchor=north west][inner sep=0.75pt]  [font=\small,color={rgb, 255:red, 0; green, 0; blue, 0 }  ,opacity=1 ,rotate=-359.39]  {$4$};
\draw (266.26,220.81) node [anchor=north west][inner sep=0.75pt]  [font=\small,color={rgb, 255:red, 0; green, 0; blue, 0 }  ,opacity=1 ,rotate=-359.39]  {$5$};
\draw (195.65,134.85) node [anchor=north west][inner sep=0.75pt]  [font=\small,color={rgb, 255:red, 0; green, 0; blue, 0 }  ,opacity=1 ,rotate=-359.39]  {$3$};
\draw (243.56,220.58) node [anchor=north west][inner sep=0.75pt]  [font=\small,color={rgb, 255:red, 0; green, 0; blue, 0 }  ,opacity=1 ,rotate=-359.39]  {$2$};
\draw (176.96,201.35) node [anchor=north west][inner sep=0.75pt]  [font=\scriptsize,color={rgb, 255:red, 0; green, 0; blue, 0 }  ,opacity=1 ]  {$-$};
\draw (212.45,201.18) node [anchor=north west][inner sep=0.75pt]  [font=\scriptsize,color={rgb, 255:red, 0; green, 0; blue, 0 }  ,opacity=1 ]  {$+$};
\draw (224.14,201.61) node [anchor=north west][inner sep=0.75pt]  [font=\scriptsize,color={rgb, 255:red, 0; green, 0; blue, 0 }  ,opacity=1 ]  {$-$};
\draw (265.64,181.36) node [anchor=north west][inner sep=0.75pt]  [font=\small,color={rgb, 255:red, 0; green, 0; blue, 0 }  ,opacity=1 ,rotate=-359.39]  {$4$};
\draw (168.21,180.39) node [anchor=north west][inner sep=0.75pt]  [font=\small,color={rgb, 255:red, 0; green, 0; blue, 0 }  ,opacity=1 ,rotate=-359.39]  {$5$};
\draw (223.96,158.45) node [anchor=north west][inner sep=0.75pt]  [font=\scriptsize,color={rgb, 255:red, 0; green, 0; blue, 0 }  ,opacity=1 ,rotate=-358.88]  {$+$};
\draw (194.26,180.07) node [anchor=north west][inner sep=0.75pt]  [font=\small,color={rgb, 255:red, 0; green, 0; blue, 0 }  ,opacity=1 ,rotate=-359.39]  {$2$};
\draw (257.2,159.32) node [anchor=north west][inner sep=0.75pt]  [font=\scriptsize,color={rgb, 255:red, 0; green, 0; blue, 0 }  ,opacity=1 ]  {$-$};
\draw (243.18,180.54) node [anchor=north west][inner sep=0.75pt]  [font=\small,color={rgb, 255:red, 0; green, 0; blue, 0 }  ,opacity=1 ,rotate=-359.39]  {$3$};
\draw (210.61,158.38) node [anchor=north west][inner sep=0.75pt]  [font=\scriptsize,color={rgb, 255:red, 0; green, 0; blue, 0 }  ,opacity=1 ]  {$-$};
\draw (177.91,158.81) node [anchor=north west][inner sep=0.75pt]  [font=\scriptsize,color={rgb, 255:red, 0; green, 0; blue, 0 }  ,opacity=1 ,rotate=-358.88]  {$-$};
\draw (257.08,200.87) node [anchor=north west][inner sep=0.75pt]  [font=\scriptsize,color={rgb, 255:red, 0; green, 0; blue, 0 }  ,opacity=1 ,rotate=-358.88]  {$-$};
\draw (254.81,114.3) node [anchor=north west][inner sep=0.75pt]  [font=\scriptsize,color={rgb, 255:red, 0; green, 0; blue, 0 }  ,opacity=1 ]  {$-$};

\end{tikzpicture}

     \caption {$\mathcal{LUG}~(\mathcal{G}_s)$ }
    \label{LUG2}
  \end{subfigure}
  \begin{subfigure}{0.145\textwidth}
    \centering

\tikzset{every picture/.style={scale=0.75pt}} 

\begin{tikzpicture}[x=0.75pt,y=0.75pt,yscale=-1,xscale=1]

\draw [color={rgb, 255:red, 208; green, 2; blue, 27 }  ,draw opacity=1 ]   (157,39.42) -- (147.57,61.85) ;
\draw [shift={(146.41,64.62)}, rotate = 292.8] [fill={rgb, 255:red, 208; green, 2; blue, 27 }  ,fill opacity=1 ][line width=0.08]  [draw opacity=0] (10.72,-5.15) -- (0,0) -- (10.72,5.15) -- (7.12,0) -- cycle    ;
\draw [color={rgb, 255:red, 208; green, 2; blue, 27 }  ,draw opacity=1 ]   (162.93,39.08) -- (173.02,61.96) ;
\draw [shift={(174.23,64.71)}, rotate = 246.21] [fill={rgb, 255:red, 208; green, 2; blue, 27 }  ,fill opacity=1 ][line width=0.08]  [draw opacity=0] (10.72,-5.15) -- (0,0) -- (10.72,5.15) -- (7.12,0) -- cycle    ;
\draw [color={rgb, 255:red, 208; green, 2; blue, 27 }  ,draw opacity=1 ][fill={rgb, 255:red, 183; green, 176; blue, 176 }  ,fill opacity=1 ]   (138.96,81.9) -- (125.42,105.26) ;
\draw [shift={(123.91,107.86)}, rotate = 300.11] [fill={rgb, 255:red, 208; green, 2; blue, 27 }  ,fill opacity=1 ][line width=0.08]  [draw opacity=0] (10.72,-5.15) -- (0,0) -- (10.72,5.15) -- (7.12,0) -- cycle    ;
\draw [color={rgb, 255:red, 176; green, 169; blue, 169 }  ,draw opacity=0.59 ]   (144.64,82.73) -- (144.56,103.24) ;
\draw [shift={(144.55,106.24)}, rotate = 270.23] [fill={rgb, 255:red, 176; green, 169; blue, 169 }  ,fill opacity=0.59 ][line width=0.08]  [draw opacity=0] (10.72,-5.15) -- (0,0) -- (10.72,5.15) -- (7.12,0) -- cycle    ;
\draw  [draw opacity=0][fill={rgb, 255:red, 0; green, 192; blue, 248 }  ,fill opacity=1 ] (141.29,73.11) .. controls (141.29,71.17) and (142.99,69.6) .. (145.08,69.6) .. controls (147.17,69.6) and (148.86,71.17) .. (148.86,73.11) .. controls (148.86,75.05) and (147.17,76.63) .. (145.08,76.63) .. controls (142.99,76.63) and (141.29,75.05) .. (141.29,73.11) -- cycle ;
\draw  [draw opacity=0][fill={rgb, 255:red, 0; green, 192; blue, 248 }  ,fill opacity=1 ] (171.87,73.11) .. controls (171.87,71.17) and (173.56,69.6) .. (175.65,69.6) .. controls (177.74,69.6) and (179.44,71.17) .. (179.44,73.11) .. controls (179.44,75.05) and (177.74,76.63) .. (175.65,76.63) .. controls (173.56,76.63) and (171.87,75.05) .. (171.87,73.11) -- cycle ;
\draw  [draw opacity=0][fill={rgb, 255:red, 0; green, 192; blue, 248 }  ,fill opacity=1 ] (141.63,114.99) .. controls (141.63,113.05) and (143.32,111.47) .. (145.41,111.47) .. controls (147.5,111.47) and (149.2,113.05) .. (149.2,114.99) .. controls (149.2,116.93) and (147.5,118.5) .. (145.41,118.5) .. controls (143.32,118.5) and (141.63,116.93) .. (141.63,114.99) -- cycle ;
\draw  [draw opacity=0][fill={rgb, 255:red, 0; green, 192; blue, 248 }  ,fill opacity=1 ] (171.54,115.23) .. controls (171.54,113.28) and (173.24,111.71) .. (175.33,111.71) .. controls (177.42,111.71) and (179.11,113.28) .. (179.11,115.23) .. controls (179.11,117.17) and (177.42,118.74) .. (175.33,118.74) .. controls (173.24,118.74) and (171.54,117.17) .. (171.54,115.23) -- cycle ;
\draw  [draw opacity=0][fill={rgb, 255:red, 0; green, 192; blue, 248 }  ,fill opacity=1 ] (192.41,115.17) .. controls (192.41,113.22) and (194.1,111.65) .. (196.19,111.65) .. controls (198.28,111.65) and (199.97,113.22) .. (199.97,115.17) .. controls (199.97,117.11) and (198.28,118.68) .. (196.19,118.68) .. controls (194.1,118.68) and (192.41,117.11) .. (192.41,115.17) -- cycle ;
\draw  [draw opacity=0][fill={rgb, 255:red, 126; green, 211; blue, 33 }  ,fill opacity=1 ] (159.32,26.42) .. controls (161.31,26.38) and (162.96,28) .. (163.01,30.03) .. controls (163.05,32.06) and (161.47,33.75) .. (159.47,33.79) .. controls (157.48,33.83) and (155.83,32.21) .. (155.79,30.18) .. controls (155.74,28.15) and (157.32,26.46) .. (159.32,26.42) -- cycle ;
\draw  [draw opacity=0][fill={rgb, 255:red, 0; green, 192; blue, 248 }  ,fill opacity=1 ] (118.8,115.46) .. controls (118.8,113.52) and (120.49,111.94) .. (122.58,111.94) .. controls (124.67,111.94) and (126.36,113.52) .. (126.36,115.46) .. controls (126.36,117.4) and (124.67,118.98) .. (122.58,118.98) .. controls (120.49,118.98) and (118.8,117.4) .. (118.8,115.46) -- cycle ;
\draw [color={rgb, 255:red, 208; green, 2; blue, 27 }  ,draw opacity=1 ][fill={rgb, 255:red, 183; green, 176; blue, 176 }  ,fill opacity=1 ]   (175.12,82.37) -- (174.7,104.19) ;
\draw [shift={(174.64,107.19)}, rotate = 271.12] [fill={rgb, 255:red, 208; green, 2; blue, 27 }  ,fill opacity=1 ][line width=0.08]  [draw opacity=0] (10.72,-5.15) -- (0,0) -- (10.72,5.15) -- (7.12,0) -- cycle    ;
\draw [color={rgb, 255:red, 208; green, 2; blue, 27 }  ,draw opacity=1 ]   (180.54,81.67) -- (193.2,105.6) ;
\draw [shift={(194.6,108.25)}, rotate = 242.11] [fill={rgb, 255:red, 208; green, 2; blue, 27 }  ,fill opacity=1 ][line width=0.08]  [draw opacity=0] (10.72,-5.15) -- (0,0) -- (10.72,5.15) -- (7.12,0) -- cycle    ;
\draw [color={rgb, 255:red, 208; green, 2; blue, 27 }  ,draw opacity=1 ]   (141.99,169) -- (125.92,191.75) ;
\draw [shift={(124.19,194.2)}, rotate = 305.22] [fill={rgb, 255:red, 208; green, 2; blue, 27 }  ,fill opacity=1 ][line width=0.08]  [draw opacity=0] (10.72,-5.15) -- (0,0) -- (10.72,5.15) -- (7.12,0) -- cycle    ;
\draw [color={rgb, 255:red, 176; green, 169; blue, 169 }  ,draw opacity=0.77 ][fill={rgb, 255:red, 183; green, 176; blue, 176 }  ,fill opacity=1 ]   (145.42,169) -- (145.64,191.12) ;
\draw [shift={(145.67,194.12)}, rotate = 269.45] [fill={rgb, 255:red, 176; green, 169; blue, 169 }  ,fill opacity=0.77 ][line width=0.08]  [draw opacity=0] (10.72,-5.15) -- (0,0) -- (10.72,5.15) -- (7.12,0) -- cycle    ;
\draw [color={rgb, 255:red, 208; green, 2; blue, 27 }  ,draw opacity=1 ][fill={rgb, 255:red, 183; green, 176; blue, 176 }  ,fill opacity=1 ]   (174.75,169.78) -- (174.78,192.3) ;
\draw [shift={(174.79,195.3)}, rotate = 269.91] [fill={rgb, 255:red, 208; green, 2; blue, 27 }  ,fill opacity=1 ][line width=0.08]  [draw opacity=0] (10.72,-5.15) -- (0,0) -- (10.72,5.15) -- (7.12,0) -- cycle    ;
\draw [color={rgb, 255:red, 208; green, 2; blue, 27 }  ,draw opacity=1 ][fill={rgb, 255:red, 183; green, 176; blue, 176 }  ,fill opacity=1 ]   (179.19,169.59) -- (193.06,192.73) ;
\draw [shift={(194.6,195.3)}, rotate = 239.07] [fill={rgb, 255:red, 208; green, 2; blue, 27 }  ,fill opacity=1 ][line width=0.08]  [draw opacity=0] (10.72,-5.15) -- (0,0) -- (10.72,5.15) -- (7.12,0) -- cycle    ;
\draw  [draw opacity=0][fill={rgb, 255:red, 0; green, 192; blue, 248 }  ,fill opacity=1 ] (141.91,201.44) .. controls (141.91,199.49) and (143.6,197.92) .. (145.69,197.92) .. controls (147.78,197.92) and (149.48,199.49) .. (149.48,201.44) .. controls (149.48,203.38) and (147.78,204.95) .. (145.69,204.95) .. controls (143.6,204.95) and (141.91,203.38) .. (141.91,201.44) -- cycle ;
\draw  [draw opacity=0][fill={rgb, 255:red, 0; green, 192; blue, 248 }  ,fill opacity=1 ] (171.42,200.89) .. controls (171.42,198.94) and (173.12,197.37) .. (175.21,197.37) .. controls (177.3,197.37) and (178.99,198.94) .. (178.99,200.89) .. controls (178.99,202.83) and (177.3,204.4) .. (175.21,204.4) .. controls (173.12,204.4) and (171.42,202.83) .. (171.42,200.89) -- cycle ;
\draw  [draw opacity=0][fill={rgb, 255:red, 0; green, 192; blue, 248 }  ,fill opacity=1 ] (195.12,201.22) .. controls (195.12,199.28) and (196.81,197.7) .. (198.9,197.7) .. controls (200.99,197.7) and (202.68,199.28) .. (202.68,201.22) .. controls (202.68,203.16) and (200.99,204.74) .. (198.9,204.74) .. controls (196.81,204.74) and (195.12,203.16) .. (195.12,201.22) -- cycle ;
\draw  [draw opacity=0][fill={rgb, 255:red, 0; green, 192; blue, 248 }  ,fill opacity=1 ] (117.06,202.11) .. controls (117.06,200.16) and (118.75,198.59) .. (120.84,198.59) .. controls (122.93,198.59) and (124.62,200.16) .. (124.62,202.11) .. controls (124.62,204.05) and (122.93,205.62) .. (120.84,205.62) .. controls (118.75,205.62) and (117.06,204.05) .. (117.06,202.11) -- cycle ;
\draw [color={rgb, 255:red, 208; green, 2; blue, 27 }  ,draw opacity=1 ][fill={rgb, 255:red, 183; green, 176; blue, 176 }  ,fill opacity=1 ]   (141.2,126.76) -- (124.4,152.38) ;
\draw [shift={(122.76,154.88)}, rotate = 303.26] [fill={rgb, 255:red, 208; green, 2; blue, 27 }  ,fill opacity=1 ][line width=0.08]  [draw opacity=0] (10.72,-5.15) -- (0,0) -- (10.72,5.15) -- (7.12,0) -- cycle    ;
\draw [color={rgb, 255:red, 208; green, 2; blue, 27 }  ,draw opacity=1 ]   (145.2,128.74) -- (145.12,149.26) ;
\draw [shift={(145.11,152.26)}, rotate = 270.23] [fill={rgb, 255:red, 208; green, 2; blue, 27 }  ,fill opacity=1 ][line width=0.08]  [draw opacity=0] (10.72,-5.15) -- (0,0) -- (10.72,5.15) -- (7.12,0) -- cycle    ;
\draw  [draw opacity=0][fill={rgb, 255:red, 0; green, 192; blue, 248 }  ,fill opacity=1 ] (141.39,160.75) .. controls (141.39,158.81) and (143.08,157.24) .. (145.17,157.24) .. controls (147.26,157.24) and (148.95,158.81) .. (148.95,160.75) .. controls (148.95,162.69) and (147.26,164.27) .. (145.17,164.27) .. controls (143.08,164.27) and (141.39,162.69) .. (141.39,160.75) -- cycle ;
\draw  [draw opacity=0][fill={rgb, 255:red, 0; green, 192; blue, 248 }  ,fill opacity=1 ] (171.91,161.19) .. controls (171.91,159.25) and (173.6,157.67) .. (175.69,157.67) .. controls (177.78,157.67) and (179.47,159.25) .. (179.47,161.19) .. controls (179.47,163.13) and (177.78,164.7) .. (175.69,164.7) .. controls (173.6,164.7) and (171.91,163.13) .. (171.91,161.19) -- cycle ;
\draw  [draw opacity=0][fill={rgb, 255:red, 0; green, 192; blue, 248 }  ,fill opacity=1 ] (193.86,160.69) .. controls (193.86,158.75) and (195.56,157.18) .. (197.65,157.18) .. controls (199.74,157.18) and (201.43,158.75) .. (201.43,160.69) .. controls (201.43,162.64) and (199.74,164.21) .. (197.65,164.21) .. controls (195.56,164.21) and (193.86,162.64) .. (193.86,160.69) -- cycle ;
\draw  [draw opacity=0][fill={rgb, 255:red, 0; green, 192; blue, 248 }  ,fill opacity=1 ] (116.33,160.16) .. controls (116.33,158.22) and (118.02,156.65) .. (120.11,156.65) .. controls (122.2,156.65) and (123.9,158.22) .. (123.9,160.16) .. controls (123.9,162.1) and (122.2,163.68) .. (120.11,163.68) .. controls (118.02,163.68) and (116.33,162.1) .. (116.33,160.16) -- cycle ;
\draw [color={rgb, 255:red, 176; green, 169; blue, 169 }  ,draw opacity=0.59 ][fill={rgb, 255:red, 183; green, 176; blue, 176 }  ,fill opacity=1 ]   (175.68,127.39) -- (175.26,149.2) ;
\draw [shift={(175.2,152.2)}, rotate = 271.12] [fill={rgb, 255:red, 176; green, 169; blue, 169 }  ,fill opacity=0.59 ][line width=0.08]  [draw opacity=0] (10.72,-5.15) -- (0,0) -- (10.72,5.15) -- (7.12,0) -- cycle    ;
\draw [color={rgb, 255:red, 176; green, 169; blue, 169 }  ,draw opacity=0.5 ][fill={rgb, 255:red, 183; green, 176; blue, 176 }  ,fill opacity=1 ]   (179.06,127.24) -- (192.89,152.26) ;
\draw [shift={(194.35,154.88)}, rotate = 241.06] [fill={rgb, 255:red, 176; green, 169; blue, 169 }  ,fill opacity=0.5 ][line width=0.08]  [draw opacity=0] (10.72,-5.15) -- (0,0) -- (10.72,5.15) -- (7.12,0) -- cycle    ;
\draw  [color={rgb, 255:red, 0; green, 0; blue, 0 }  ,draw opacity=0 ][fill={rgb, 255:red, 74; green, 144; blue, 226 }  ,fill opacity=0.1 ] (104.2,22) -- (214.59,22) -- (214.59,38.21) -- (104.2,38.21) -- cycle ;
\draw  [color={rgb, 255:red, 0; green, 0; blue, 0 }  ,draw opacity=0 ][fill={rgb, 255:red, 74; green, 144; blue, 226 }  ,fill opacity=0.1 ] (103.24,65.44) -- (213.63,65.44) -- (213.63,81.65) -- (103.24,81.65) -- cycle ;
\draw  [color={rgb, 255:red, 0; green, 0; blue, 0 }  ,draw opacity=0 ][fill={rgb, 255:red, 74; green, 144; blue, 226 }  ,fill opacity=0.1 ] (103.24,106.47) -- (213.63,106.47) -- (213.63,122.68) -- (103.24,122.68) -- cycle ;
\draw  [color={rgb, 255:red, 0; green, 0; blue, 0 }  ,draw opacity=0 ][fill={rgb, 255:red, 74; green, 144; blue, 226 }  ,fill opacity=0.1 ] (103.24,151.84) -- (213.63,151.84) -- (213.63,168.06) -- (103.24,168.06) -- cycle ;
\draw  [color={rgb, 255:red, 0; green, 0; blue, 0 }  ,draw opacity=0 ][fill={rgb, 255:red, 74; green, 144; blue, 226 }  ,fill opacity=0.1 ] (103.24,192.39) -- (213.63,192.39) -- (213.63,208.6) -- (103.24,208.6) -- cycle ;

\draw (128.68,64.74) node [anchor=north west][inner sep=0.75pt]  [font=\small,color={rgb, 255:red, 0; green, 0; blue, 0 }  ,opacity=1 ,rotate=-359.39]  {$2$};
\draw (181.38,65.59) node [anchor=north west][inner sep=0.75pt]  [font=\small,color={rgb, 255:red, 0; green, 0; blue, 0 }  ,opacity=1 ,rotate=-359.39]  {$3$};
\draw (108.4,107.64) node [anchor=north west][inner sep=0.75pt]  [font=\small,color={rgb, 255:red, 0; green, 0; blue, 0 }  ,opacity=1 ,rotate=-359.39]  {$4$};
\draw (198.48,108.07) node [anchor=north west][inner sep=0.75pt]  [font=\small,color={rgb, 255:red, 0; green, 0; blue, 0 }  ,opacity=1 ,rotate=-359.39]  {$5$};
\draw (165.57,21.89) node [anchor=north west][inner sep=0.75pt]  [font=\small,color={rgb, 255:red, 0; green, 0; blue, 0 }  ,opacity=1 ,rotate=-358.27]  {$1$};
\draw (132.43,193.68) node [anchor=north west][inner sep=0.75pt]  [font=\small,color={rgb, 255:red, 0; green, 0; blue, 0 }  ,opacity=1 ,rotate=-1.04]  {$3$};
\draw (178.81,108.24) node [anchor=north west][inner sep=0.75pt]  [font=\small,color={rgb, 255:red, 0; green, 0; blue, 0 }  ,opacity=1 ,rotate=-359.39]  {$2$};
\draw (105.43,193.66) node [anchor=north west][inner sep=0.75pt]  [font=\small,color={rgb, 255:red, 0; green, 0; blue, 0 }  ,opacity=1 ,rotate=-359.39]  {$4$};
\draw (202.52,193.29) node [anchor=north west][inner sep=0.75pt]  [font=\small,color={rgb, 255:red, 0; green, 0; blue, 0 }  ,opacity=1 ,rotate=-359.39]  {$5$};
\draw (131.92,107.33) node [anchor=north west][inner sep=0.75pt]  [font=\small,color={rgb, 255:red, 0; green, 0; blue, 0 }  ,opacity=1 ,rotate=-359.39]  {$3$};
\draw (179.82,193.07) node [anchor=north west][inner sep=0.75pt]  [font=\small,color={rgb, 255:red, 0; green, 0; blue, 0 }  ,opacity=1 ,rotate=-359.39]  {$2$};
\draw (201.91,153.84) node [anchor=north west][inner sep=0.75pt]  [font=\small,color={rgb, 255:red, 0; green, 0; blue, 0 }  ,opacity=1 ,rotate=-359.39]  {$4$};
\draw (104.47,152.88) node [anchor=north west][inner sep=0.75pt]  [font=\small,color={rgb, 255:red, 0; green, 0; blue, 0 }  ,opacity=1 ,rotate=-359.39]  {$5$};
\draw (130.52,152.55) node [anchor=north west][inner sep=0.75pt]  [font=\small,color={rgb, 255:red, 0; green, 0; blue, 0 }  ,opacity=1 ,rotate=-359.39]  {$2$};
\draw (179.44,153.02) node [anchor=north west][inner sep=0.75pt]  [font=\small,color={rgb, 255:red, 0; green, 0; blue, 0 }  ,opacity=1 ,rotate=-359.39]  {$3$};
\draw (132.1,41.06) node [anchor=north west][inner sep=0.75pt]  [font=\scriptsize,color={rgb, 255:red, 0; green, 0; blue, 0 }  ,opacity=1 ]  {$+$};
\draw (174.2,40.88) node [anchor=north west][inner sep=0.75pt]  [font=\scriptsize,color={rgb, 255:red, 0; green, 0; blue, 0 }  ,opacity=1 ]  {$+$};
\draw (146.98,86.5) node [anchor=north west][inner sep=0.75pt]  [font=\scriptsize,color={rgb, 255:red, 0; green, 0; blue, 0 }  ,opacity=1 ]  {$+$};
\draw (114.74,87.35) node [anchor=north west][inner sep=0.75pt]  [font=\scriptsize,color={rgb, 255:red, 0; green, 0; blue, 0 }  ,opacity=1 ]  {$-$};
\draw (159.73,87.28) node [anchor=north west][inner sep=0.75pt]  [font=\scriptsize,color={rgb, 255:red, 0; green, 0; blue, 0 }  ,opacity=1 ]  {$-$};
\draw (112.46,174.04) node [anchor=north west][inner sep=0.75pt]  [font=\scriptsize,color={rgb, 255:red, 0; green, 0; blue, 0 }  ,opacity=1 ]  {$-$};
\draw (147.95,173.87) node [anchor=north west][inner sep=0.75pt]  [font=\scriptsize,color={rgb, 255:red, 0; green, 0; blue, 0 }  ,opacity=1 ]  {$+$};
\draw (159.64,174.3) node [anchor=north west][inner sep=0.75pt]  [font=\scriptsize,color={rgb, 255:red, 0; green, 0; blue, 0 }  ,opacity=1 ]  {$-$};
\draw (159.46,131.14) node [anchor=north west][inner sep=0.75pt]  [font=\scriptsize,color={rgb, 255:red, 0; green, 0; blue, 0 }  ,opacity=1 ,rotate=-358.88]  {$+$};
\draw (192.7,132.01) node [anchor=north west][inner sep=0.75pt]  [font=\scriptsize,color={rgb, 255:red, 0; green, 0; blue, 0 }  ,opacity=1 ]  {$-$};
\draw (146.11,131.07) node [anchor=north west][inner sep=0.75pt]  [font=\scriptsize,color={rgb, 255:red, 0; green, 0; blue, 0 }  ,opacity=1 ]  {$-$};
\draw (113.41,131.5) node [anchor=north west][inner sep=0.75pt]  [font=\scriptsize,color={rgb, 255:red, 0; green, 0; blue, 0 }  ,opacity=1 ,rotate=-358.88]  {$-$};
\draw (192.58,173.56) node [anchor=north west][inner sep=0.75pt]  [font=\scriptsize,color={rgb, 255:red, 0; green, 0; blue, 0 }  ,opacity=1 ,rotate=-358.88]  {$-$};
\draw (190.31,86.99) node [anchor=north west][inner sep=0.75pt]  [font=\scriptsize,color={rgb, 255:red, 0; green, 0; blue, 0 }  ,opacity=1 ]  {$-$};

\end{tikzpicture}
\caption{$\mathcal{LUG^{H}}(\mathcal{G}_s) $}
    \label{LUG H}
  \end{subfigure}
  \caption{Signed layered graph of $\mathcal{G}_6$ with two $\mathcal{LUG}(\mathcal{G}_s)$, highlighted in green and red in Fig~\ref{3fig}(\subref{LUG2}). This shows that a digraph may admit multiple $\mathcal{LUG}(\mathcal{G}_s)$. Fig~\ref{3fig}(\subref{LUG H}) depicts $\mathcal{LUG}^H(\mathcal{G}_s)$, where all nodes are sign-matched.}

  \label{3fig}
\end{figure}
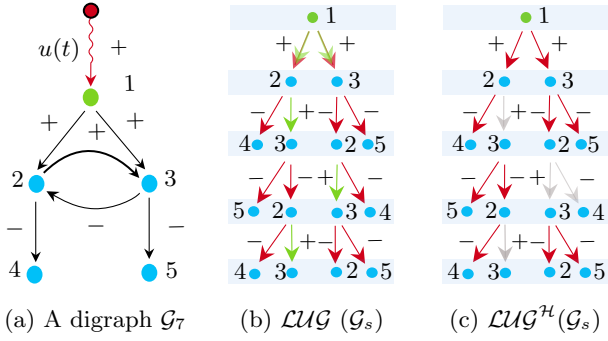

For example, consider the digraph $\mathcal{G}_7$ and its signed layered graph $\mathcal{G}_s$ given in Fig.~\ref{3fig}. From Def. \ref{def_LUG}, it follows that multiple $\mathcal{LUG}(\mathcal{G}_s)$ can exist within $\mathcal{G}_s$, as shown in Fig.~\ref{3fig}(\subref{LUG2}) (highlighted in green and red). Although $L_{m+1}$ is unisigned, there exist cases where the column $\Psi_m$ is not unisigned, since each entry of $\Psi_m$ is obtained as the sum of products of the corresponding path weights, which can be of different sign from other entries.
Thus, an $\mathcal{LUG}(\mathcal{G}_s)$ alone is not sufficient to derive the conditions for $\mathcal{SS}$ herdability. To overcome this issue, we propose to use sign matching in $\mathcal{LUG}(\mathcal{G}_s)$, motivated by \cite{liu2011controllability}.

\begin{definition}\label{sign M def}
\textit{Sign matching:} A set of nodes $S$ in $L_p$ is said to be sign-matched if the following conditions are satisfied:  
\begin{enumerate}
    \item[i.]  All incoming edges from $L_{p-1}$ to nodes in $S$ have the same sign.  

\item[ii.]  Every $i \in S$, there must exist walks of the same signs from the leader to all the nodes in $S$. Equivalently, with the same signs. 
\vspace{-2mm}
\[
\forall i,j \in S, \ \exists r:  \operatorname{sign}\!\big(\mathcal{W}^{r}_{(1,i)}\big)=\operatorname{sign}\!\big(\mathcal{W}^{r}_{(1,j)}\big) 
\]
\item[ii.] Let $\Delta_i$ represent a set of nodes that exhibit a signed dilation in $\mathcal{G(A,B)}$. If within $\Delta_i$, either $\Delta_i^N$ or $\Delta_i^P$ is sign-matched, then the other set will be sign-matched solely outside of $\Delta_i$.

\end{enumerate}
\end{definition}

\begin{definition}\label{lug h def}
\textit{$\mathcal{LUG}^{H} (\mathcal{G}_s)$:} An  $\mathcal{LUG}^{H} (\mathcal{G}_S)$ is an $\mathcal{LUG} (\mathcal{G}_S)$ where every node $\mathcal{G}(\mathcal{A,B})$ in $(\mathcal{G}_S)$ is sign-matched. In the single-leader case, only one edge connects $u(t)$ to $L_1$, so node $1$ is naturally sign-matched.

\end{definition}
\subsection{Properties of $\mathcal{LUG}^{H} (\mathcal{G}_S)$:}

\begin{itemize}
    \item An $\mathcal{LUG}^{H}(\mathcal{G}_S)$ is a subgraph of the signed layer graph $\mathcal{G}_S$ that may appear disconnected, but all its nodes are inherently input-connected in $\mathcal{G}_S$.
   \item In every layer $L_k$ of $\mathcal{LUG}^{H}(\mathcal{G}_S)$, only sign-matched nodes have incoming edges. Consequently, all such nodes satisfy the properties of sign-matched nodes given in Definition \ref{sign M def}.
\item Let $\Delta_i$ represent a signed dilation set in $\mathcal{G(A,B)}$. Then an $\mathcal{LUG^H}(\mathcal{G}_s)$  either $\Delta_i^N$ or $\Delta_i^P$ is sign-matched within $\Delta_i$, and the nodes in other set is sign-matched outside of $\Delta_i$.
    
\end{itemize}

The $\mathcal{LUG} (\mathcal{G}_S)$ highlighted in Fig.~\ref{3fig}(\subref{LUG H}) is an $\mathcal{LUG}^{H} (\mathcal{G}_S)$. When a digraph is spanned by an $\mathcal{LUG}^{\mathcal{H}}(\mathcal{G}_s)$, the following scenarios can occur:

\begin{figure}[ht]
    \centering
    
    \begin{subfigure}{0.14\textwidth}
        \centering

\tikzset{every picture/.style={scale=0.85pt}} 

\begin{tikzpicture}[x=0.75pt,y=0.75pt,yscale=-1,xscale=1]

\draw  [fill={rgb, 255:red, 0; green, 192; blue, 248 }  ,fill opacity=1 ] (70.97,278.72) .. controls (70.97,276.93) and (72.43,275.47) .. (74.22,275.47) .. controls (76.01,275.47) and (77.46,276.93) .. (77.46,278.72) .. controls (77.46,280.51) and (76.01,281.96) .. (74.22,281.96) .. controls (72.43,281.96) and (70.97,280.51) .. (70.97,278.72) -- cycle ;
\draw [color={rgb, 255:red, 161; green, 27; blue, 43 }  ,draw opacity=1 ]   (74.13,229.25) .. controls (75.8,230.92) and (75.79,232.58) .. (74.12,234.25) .. controls (72.45,235.92) and (72.45,237.58) .. (74.12,239.25) -- (74.12,239.9) -- (74.1,247.9) ;
\draw [shift={(74.1,250.9)}, rotate = 270.1] [fill={rgb, 255:red, 161; green, 27; blue, 43 }  ,fill opacity=1 ][line width=0.08]  [draw opacity=0] (10.72,-5.15) -- (0,0) -- (10.72,5.15) -- (7.12,0) -- cycle    ;
\draw  [fill={rgb, 255:red, 239; green, 43; blue, 67 }  ,fill opacity=1 ] (70.89,226) .. controls (70.89,224.21) and (72.34,222.76) .. (74.13,222.76) .. controls (75.93,222.76) and (77.38,224.21) .. (77.38,226) .. controls (77.38,227.79) and (75.93,229.25) .. (74.13,229.25) .. controls (72.34,229.25) and (70.89,227.79) .. (70.89,226) -- cycle ;
\draw  [fill={rgb, 255:red, 0; green, 192; blue, 248 }  ,fill opacity=1 ] (70.7,302.83) .. controls (70.7,301.04) and (72.16,299.58) .. (73.95,299.58) .. controls (75.74,299.58) and (77.19,301.04) .. (77.19,302.83) .. controls (77.19,304.62) and (75.74,306.07) .. (73.95,306.07) .. controls (72.16,306.07) and (70.7,304.62) .. (70.7,302.83) -- cycle ;
\draw  [fill={rgb, 255:red, 0; green, 192; blue, 248 }  ,fill opacity=1 ] (70.57,326.96) .. controls (70.57,325.17) and (72.02,323.72) .. (73.81,323.72) .. controls (75.61,323.72) and (77.06,325.17) .. (77.06,326.96) .. controls (77.06,328.75) and (75.61,330.21) .. (73.81,330.21) .. controls (72.02,330.21) and (70.57,328.75) .. (70.57,326.96) -- cycle ;
\draw  [fill={rgb, 255:red, 0; green, 192; blue, 248 }  ,fill opacity=1 ] (70.85,254.15) .. controls (70.85,252.36) and (72.3,250.9) .. (74.1,250.9) .. controls (75.89,250.9) and (77.34,252.36) .. (77.34,254.15) .. controls (77.34,255.94) and (75.89,257.39) .. (74.1,257.39) .. controls (72.3,257.39) and (70.85,255.94) .. (70.85,254.15) -- cycle ;
\draw    (74.06,306.21) -- (73.83,320.83) ;
\draw [shift={(73.79,323.83)}, rotate = 270.88] [fill={rgb, 255:red, 0; green, 0; blue, 0 }  ][line width=0.08]  [draw opacity=0] (10.72,-5.15) -- (0,0) -- (10.72,5.15) -- (7.12,0) -- cycle    ;
\draw    (74.49,257.85) -- (74.26,272.47) ;
\draw [shift={(74.22,275.47)}, rotate = 270.88] [fill={rgb, 255:red, 0; green, 0; blue, 0 }  ][line width=0.08]  [draw opacity=0] (10.72,-5.15) -- (0,0) -- (10.72,5.15) -- (7.12,0) -- cycle    ;
\draw    (74.22,281.96) -- (73.99,296.58) ;
\draw [shift={(73.95,299.58)}, rotate = 270.88] [fill={rgb, 255:red, 0; green, 0; blue, 0 }  ][line width=0.08]  [draw opacity=0] (10.72,-5.15) -- (0,0) -- (10.72,5.15) -- (7.12,0) -- cycle    ;

\draw (70.67,213.48) node [anchor=north west][inner sep=0.75pt]  [font=\tiny]  {$u$};
\draw (57.2,256.73) node [anchor=north west][inner sep=0.75pt]  [font=\scriptsize]  {$+$};
\draw (58.66,280.2) node [anchor=north west][inner sep=0.75pt]  [font=\scriptsize]  {$-$};
\draw (56.94,304.74) node [anchor=north west][inner sep=0.75pt]  [font=\scriptsize]  {$+$};
\draw (83.5,252.61) node [anchor=north west][inner sep=0.75pt]  [font=\tiny]  {$1$};
\draw (85.3,276.01) node [anchor=north west][inner sep=0.75pt]  [font=\tiny]  {$2$};
\draw (86.2,298.61) node [anchor=north west][inner sep=0.75pt]  [font=\tiny]  {$3$};
\draw (85.6,322.31) node [anchor=north west][inner sep=0.75pt]  [font=\tiny]  {$4$};

\end{tikzpicture}

        \caption{}
        \label{5fig a}
    \end{subfigure}
    \begin{subfigure}{0.14\textwidth}
        \centering
\tikzset{every picture/.style={scale=0.85pt}} 

\begin{tikzpicture}[x=0.75pt,y=0.75pt,yscale=-1,xscale=1]

\draw  [fill={rgb, 255:red, 0; green, 192; blue, 248 }  ,fill opacity=1 ] (80.74,86.09) .. controls (80.74,84.3) and (82.19,82.85) .. (83.98,82.85) .. controls (85.77,82.85) and (87.23,84.3) .. (87.23,86.09) .. controls (87.23,87.88) and (85.77,89.33) .. (83.98,89.33) .. controls (82.19,89.33) and (80.74,87.88) .. (80.74,86.09) -- cycle ;
\draw  [fill={rgb, 255:red, 0; green, 192; blue, 248 }  ,fill opacity=1 ] (43.58,136.2) .. controls (43.58,134.41) and (45.03,132.96) .. (46.82,132.96) .. controls (48.61,132.96) and (50.07,134.41) .. (50.07,136.2) .. controls (50.07,137.99) and (48.61,139.44) .. (46.82,139.44) .. controls (45.03,139.44) and (43.58,137.99) .. (43.58,136.2) -- cycle ;
\draw [color={rgb, 255:red, 161; green, 27; blue, 43 }  ,draw opacity=1 ]   (83.89,36.62) .. controls (85.56,38.29) and (85.56,39.95) .. (83.89,41.62) .. controls (82.22,43.29) and (82.22,44.95) .. (83.88,46.62) -- (83.88,47.28) -- (83.86,55.28) ;
\draw [shift={(83.86,58.28)}, rotate = 270.1] [fill={rgb, 255:red, 161; green, 27; blue, 43 }  ,fill opacity=1 ][line width=0.08]  [draw opacity=0] (10.72,-5.15) -- (0,0) -- (10.72,5.15) -- (7.12,0) -- cycle    ;
\draw  [fill={rgb, 255:red, 239; green, 43; blue, 67 }  ,fill opacity=1 ] (80.65,33.37) .. controls (80.65,31.58) and (82.1,30.13) .. (83.89,30.13) .. controls (85.69,30.13) and (87.14,31.58) .. (87.14,33.37) .. controls (87.14,35.17) and (85.69,36.62) .. (83.89,36.62) .. controls (82.1,36.62) and (80.65,35.17) .. (80.65,33.37) -- cycle ;
\draw  [fill={rgb, 255:red, 248; green, 231; blue, 28 }  ,fill opacity=1 ] (117.14,111.12) .. controls (117.14,109.33) and (118.59,107.87) .. (120.38,107.87) .. controls (122.17,107.87) and (123.63,109.33) .. (123.63,111.12) .. controls (123.63,112.91) and (122.17,114.36) .. (120.38,114.36) .. controls (118.59,114.36) and (117.14,112.91) .. (117.14,111.12) -- cycle ;
\draw  [fill={rgb, 255:red, 248; green, 231; blue, 28 }  ,fill opacity=1 ] (43.85,112.09) .. controls (43.85,110.3) and (45.3,108.85) .. (47.09,108.85) .. controls (48.88,108.85) and (50.34,110.3) .. (50.34,112.09) .. controls (50.34,113.88) and (48.88,115.33) .. (47.09,115.33) .. controls (45.3,115.33) and (43.85,113.88) .. (43.85,112.09) -- cycle ;
\draw  [fill={rgb, 255:red, 248; green, 231; blue, 28 }  ,fill opacity=1 ] (80.25,110.47) .. controls (80.25,108.68) and (81.7,107.23) .. (83.49,107.23) .. controls (85.29,107.23) and (86.74,108.68) .. (86.74,110.47) .. controls (86.74,112.26) and (85.29,113.71) .. (83.49,113.71) .. controls (81.7,113.71) and (80.25,112.26) .. (80.25,110.47) -- cycle ;
\draw    (86.93,87.77) -- (115.39,106.16) ;
\draw [shift={(117.91,107.79)}, rotate = 212.88] [fill={rgb, 255:red, 0; green, 0; blue, 0 }  ][line width=0.08]  [draw opacity=0] (10.72,-5.15) -- (0,0) -- (10.72,5.15) -- (7.12,0) -- cycle    ;
\draw  [fill={rgb, 255:red, 0; green, 192; blue, 248 }  ,fill opacity=1 ] (80.61,61.52) .. controls (80.61,59.73) and (82.07,58.28) .. (83.86,58.28) .. controls (85.65,58.28) and (87.1,59.73) .. (87.1,61.52) .. controls (87.1,63.31) and (85.65,64.76) .. (83.86,64.76) .. controls (82.07,64.76) and (80.61,63.31) .. (80.61,61.52) -- cycle ;
\draw    (81,87.92) -- (51.65,107.37) ;
\draw [shift={(49.15,109.02)}, rotate = 326.47] [fill={rgb, 255:red, 0; green, 0; blue, 0 }  ][line width=0.08]  [draw opacity=0] (10.72,-5.15) -- (0,0) -- (10.72,5.15) -- (7.12,0) -- cycle    ;
\draw    (47.09,115.33) -- (46.87,129.96) ;
\draw [shift={(46.82,132.96)}, rotate = 270.88] [fill={rgb, 255:red, 0; green, 0; blue, 0 }  ][line width=0.08]  [draw opacity=0] (10.72,-5.15) -- (0,0) -- (10.72,5.15) -- (7.12,0) -- cycle    ;
\draw    (83.98,89.33) -- (83.76,103.96) ;
\draw [shift={(83.71,106.96)}, rotate = 270.88] [fill={rgb, 255:red, 0; green, 0; blue, 0 }  ][line width=0.08]  [draw opacity=0] (10.72,-5.15) -- (0,0) -- (10.72,5.15) -- (7.12,0) -- cycle    ;
\draw    (84.25,65.22) -- (84.03,79.85) ;
\draw [shift={(83.98,82.85)}, rotate = 270.88] [fill={rgb, 255:red, 0; green, 0; blue, 0 }  ][line width=0.08]  [draw opacity=0] (10.72,-5.15) -- (0,0) -- (10.72,5.15) -- (7.12,0) -- cycle    ;

\draw (81.64,22.07) node [anchor=north west][inner sep=0.75pt]  [font=\tiny]  {$u$};
\draw (59.57,84.51) node [anchor=north west][inner sep=0.75pt]  [font=\scriptsize]  {$+$};
\draw (103.92,85.59) node [anchor=north west][inner sep=0.75pt]  [font=\scriptsize]  {$+$};
\draw (69.84,94.78) node [anchor=north west][inner sep=0.75pt]  [font=\scriptsize]  {$+$};
\draw (30.9,117.49) node [anchor=north west][inner sep=0.75pt]  [font=\scriptsize]  {$-$};
\draw (68.22,63.97) node [anchor=north west][inner sep=0.75pt]  [font=\scriptsize]  {$-$};
\draw (93.1,59.81) node [anchor=north west][inner sep=0.75pt]  [font=\tiny]  {$1$};
\draw (92.7,80.01) node [anchor=north west][inner sep=0.75pt]  [font=\tiny]  {$2$};
\draw (31.8,109.01) node [anchor=north west][inner sep=0.75pt]  [font=\tiny]  {$3$};
\draw (80.8,121.51) node [anchor=north west][inner sep=0.75pt]  [font=\tiny]  {$4$};
\draw (115.3,120.61) node [anchor=north west][inner sep=0.75pt]  [font=\tiny]  {$5$};
\draw (31.9,135.01) node [anchor=north west][inner sep=0.75pt]  [font=\tiny]  {$6$};

\end{tikzpicture}

        \caption{}
        \label{5fig b}
    \end{subfigure}
    \hfill
    \begin{subfigure}{0.14\textwidth}
        \centering
       
\tikzset{every picture/.style={scale=0.85pt}} 

\begin{tikzpicture}[x=0.75pt,y=0.75pt,yscale=-1,xscale=1]

\draw  [fill={rgb, 255:red, 0; green, 192; blue, 248 }  ,fill opacity=1 ] (144.62,198.17) .. controls (144.62,196.38) and (146.07,194.93) .. (147.86,194.93) .. controls (149.66,194.93) and (151.11,196.38) .. (151.11,198.17) .. controls (151.11,199.96) and (149.66,201.41) .. (147.86,201.41) .. controls (146.07,201.41) and (144.62,199.96) .. (144.62,198.17) -- cycle ;
\draw    (148.13,177.3) -- (147.91,191.93) ;
\draw [shift={(147.86,194.93)}, rotate = 270.88] [fill={rgb, 255:red, 0; green, 0; blue, 0 }  ][line width=0.08]  [draw opacity=0] (10.72,-5.15) -- (0,0) -- (10.72,5.15) -- (7.12,0) -- cycle    ;
\draw  [fill={rgb, 255:red, 0; green, 192; blue, 248 }  ,fill opacity=1 ] (109.49,149.03) .. controls (109.49,147.24) and (110.94,145.79) .. (112.73,145.79) .. controls (114.53,145.79) and (115.98,147.24) .. (115.98,149.03) .. controls (115.98,150.82) and (114.53,152.28) .. (112.73,152.28) .. controls (110.94,152.28) and (109.49,150.82) .. (109.49,149.03) -- cycle ;
\draw  [fill={rgb, 255:red, 0; green, 192; blue, 248 }  ,fill opacity=1 ] (108.02,200.03) .. controls (108.02,198.23) and (109.47,196.78) .. (111.26,196.78) .. controls (113.06,196.78) and (114.51,198.23) .. (114.51,200.03) .. controls (114.51,201.82) and (113.06,203.27) .. (111.26,203.27) .. controls (109.47,203.27) and (108.02,201.82) .. (108.02,200.03) -- cycle ;
\draw [color={rgb, 255:red, 161; green, 27; blue, 43 }  ,draw opacity=1 ]   (112.65,99.56) .. controls (114.31,101.23) and (114.31,102.89) .. (112.64,104.56) .. controls (110.97,106.23) and (110.97,107.89) .. (112.63,109.56) -- (112.63,110.22) -- (112.62,118.22) ;
\draw [shift={(112.61,121.22)}, rotate = 270.1] [fill={rgb, 255:red, 161; green, 27; blue, 43 }  ,fill opacity=1 ][line width=0.08]  [draw opacity=0] (10.72,-5.15) -- (0,0) -- (10.72,5.15) -- (7.12,0) -- cycle    ;
\draw  [fill={rgb, 255:red, 239; green, 43; blue, 67 }  ,fill opacity=1 ] (109.4,96.32) .. controls (109.4,94.53) and (110.86,93.07) .. (112.65,93.07) .. controls (114.44,93.07) and (115.89,94.53) .. (115.89,96.32) .. controls (115.89,98.11) and (114.44,99.56) .. (112.65,99.56) .. controls (110.86,99.56) and (109.4,98.11) .. (109.4,96.32) -- cycle ;
\draw  [fill={rgb, 255:red, 0; green, 192; blue, 248 }  ,fill opacity=1 ] (144.89,174.06) .. controls (144.89,172.27) and (146.34,170.82) .. (148.13,170.82) .. controls (149.93,170.82) and (151.38,172.27) .. (151.38,174.06) .. controls (151.38,175.85) and (149.93,177.3) .. (148.13,177.3) .. controls (146.34,177.3) and (144.89,175.85) .. (144.89,174.06) -- cycle ;
\draw  [fill={rgb, 255:red, 0; green, 192; blue, 248 }  ,fill opacity=1 ] (72.6,175.03) .. controls (72.6,173.24) and (74.05,171.79) .. (75.85,171.79) .. controls (77.64,171.79) and (79.09,173.24) .. (79.09,175.03) .. controls (79.09,176.83) and (77.64,178.28) .. (75.85,178.28) .. controls (74.05,178.28) and (72.6,176.83) .. (72.6,175.03) -- cycle ;
\draw    (115.68,150.71) -- (144.15,169.1) ;
\draw [shift={(146.66,170.73)}, rotate = 212.88] [fill={rgb, 255:red, 0; green, 0; blue, 0 }  ][line width=0.08]  [draw opacity=0] (10.72,-5.15) -- (0,0) -- (10.72,5.15) -- (7.12,0) -- cycle    ;
\draw  [fill={rgb, 255:red, 0; green, 192; blue, 248 }  ,fill opacity=1 ] (109.37,124.46) .. controls (109.37,122.67) and (110.82,121.22) .. (112.61,121.22) .. controls (114.41,121.22) and (115.86,122.67) .. (115.86,124.46) .. controls (115.86,126.25) and (114.41,127.71) .. (112.61,127.71) .. controls (110.82,127.71) and (109.37,126.25) .. (109.37,124.46) -- cycle ;
\draw    (109.75,150.87) -- (80.41,170.31) ;
\draw [shift={(77.91,171.97)}, rotate = 326.47] [fill={rgb, 255:red, 0; green, 0; blue, 0 }  ][line width=0.08]  [draw opacity=0] (10.72,-5.15) -- (0,0) -- (10.72,5.15) -- (7.12,0) -- cycle    ;
\draw    (113,128.17) -- (112.78,142.79) ;
\draw [shift={(112.73,145.79)}, rotate = 270.88] [fill={rgb, 255:red, 0; green, 0; blue, 0 }  ][line width=0.08]  [draw opacity=0] (10.72,-5.15) -- (0,0) -- (10.72,5.15) -- (7.12,0) -- cycle    ;
\draw    (145.7,176.49) -- (116.36,195.94) ;
\draw [shift={(113.86,197.59)}, rotate = 326.47] [fill={rgb, 255:red, 0; green, 0; blue, 0 }  ][line width=0.08]  [draw opacity=0] (10.72,-5.15) -- (0,0) -- (10.72,5.15) -- (7.12,0) -- cycle    ;
\draw    (77.85,177.16) -- (106.31,195.56) ;
\draw [shift={(108.83,197.19)}, rotate = 212.88] [fill={rgb, 255:red, 0; green, 0; blue, 0 }  ][line width=0.08]  [draw opacity=0] (10.72,-5.15) -- (0,0) -- (10.72,5.15) -- (7.12,0) -- cycle    ;

\draw (110.4,85.01) node [anchor=north west][inner sep=0.75pt]  [font=\tiny]  {$u$};
\draw (75.34,187.74) node [anchor=north west][inner sep=0.75pt]  [font=\scriptsize]  {$+$};
\draw (126.72,187.86) node [anchor=north west][inner sep=0.75pt]  [font=\scriptsize]  {$-$};
\draw (98.6,128.4) node [anchor=north west][inner sep=0.75pt]  [font=\scriptsize]  {$+$};
\draw (84.54,148.4) node [anchor=north west][inner sep=0.75pt]  [font=\scriptsize]  {$+$};
\draw (131.05,149.48) node [anchor=north west][inner sep=0.75pt]  [font=\scriptsize]  {$+$};
\draw (155.39,177.59) node [anchor=north west][inner sep=0.75pt]  [font=\scriptsize]  {$-$};
\draw (124.4,121.01) node [anchor=north west][inner sep=0.75pt]  [font=\tiny]  {$1$};
\draw (124.4,146.01) node [anchor=north west][inner sep=0.75pt]  [font=\tiny]  {$2$};
\draw (58.4,173.01) node [anchor=north west][inner sep=0.75pt]  [font=\tiny]  {$3$};
\draw (158.4,170.01) node [anchor=north west][inner sep=0.75pt]  [font=\tiny]  {$4$};
\draw (109.4,183.01) node [anchor=north west][inner sep=0.75pt]  [font=\tiny]  {$5$};
\draw (158.4,195.01) node [anchor=north west][inner sep=0.75pt]  [font=\tiny]  {$6$};

\end{tikzpicture}
        \caption{}
        \label{5fig c}
    \end{subfigure}
    
    \vskip\baselineskip 
    \begin{subfigure}{0.14\textwidth}
        \centering
       
\tikzset{every picture/.style={scale=0.85pt}} 

\begin{tikzpicture}[x=0.75pt,y=0.75pt,yscale=-1,xscale=1]

\draw  [fill={rgb, 255:red, 0; green, 192; blue, 248 }  ,fill opacity=1 ] (378.31,87.91) .. controls (378.31,86.11) and (379.77,84.66) .. (381.56,84.66) .. controls (383.35,84.66) and (384.8,86.11) .. (384.8,87.91) .. controls (384.8,89.7) and (383.35,91.15) .. (381.56,91.15) .. controls (379.77,91.15) and (378.31,89.7) .. (378.31,87.91) -- cycle ;
\draw [color={rgb, 255:red, 161; green, 27; blue, 43 }  ,draw opacity=1 ]   (381.47,38.43) .. controls (383.14,40.1) and (383.13,41.77) .. (381.46,43.43) .. controls (379.79,45.1) and (379.79,46.76) .. (381.46,48.43) -- (381.45,49.09) -- (381.44,57.09) ;
\draw [shift={(381.44,60.09)}, rotate = 270.1] [fill={rgb, 255:red, 161; green, 27; blue, 43 }  ,fill opacity=1 ][line width=0.08]  [draw opacity=0] (10.72,-5.15) -- (0,0) -- (10.72,5.15) -- (7.12,0) -- cycle    ;
\draw  [fill={rgb, 255:red, 239; green, 43; blue, 67 }  ,fill opacity=1 ] (378.23,35.19) .. controls (378.23,33.4) and (379.68,31.95) .. (381.47,31.95) .. controls (383.26,31.95) and (384.72,33.4) .. (384.72,35.19) .. controls (384.72,36.98) and (383.26,38.43) .. (381.47,38.43) .. controls (379.68,38.43) and (378.23,36.98) .. (378.23,35.19) -- cycle ;
\draw  [fill={rgb, 255:red, 0; green, 192; blue, 248 }  ,fill opacity=1 ] (413.71,112.93) .. controls (413.71,111.14) and (415.17,109.69) .. (416.96,109.69) .. controls (418.75,109.69) and (420.2,111.14) .. (420.2,112.93) .. controls (420.2,114.73) and (418.75,116.18) .. (416.96,116.18) .. controls (415.17,116.18) and (413.71,114.73) .. (413.71,112.93) -- cycle ;
\draw  [fill={rgb, 255:red, 0; green, 192; blue, 248 }  ,fill opacity=1 ] (341.42,113.91) .. controls (341.42,112.12) and (342.88,110.66) .. (344.67,110.66) .. controls (346.46,110.66) and (347.92,112.12) .. (347.92,113.91) .. controls (347.92,115.7) and (346.46,117.15) .. (344.67,117.15) .. controls (342.88,117.15) and (341.42,115.7) .. (341.42,113.91) -- cycle ;
\draw    (384.51,89.58) -- (412.97,107.98) ;
\draw [shift={(415.49,109.61)}, rotate = 212.88] [fill={rgb, 255:red, 0; green, 0; blue, 0 }  ][line width=0.08]  [draw opacity=0] (10.72,-5.15) -- (0,0) -- (10.72,5.15) -- (7.12,0) -- cycle    ;
\draw  [fill={rgb, 255:red, 0; green, 192; blue, 248 }  ,fill opacity=1 ] (378.19,63.34) .. controls (378.19,61.55) and (379.64,60.09) .. (381.44,60.09) .. controls (383.23,60.09) and (384.68,61.55) .. (384.68,63.34) .. controls (384.68,65.13) and (383.23,66.58) .. (381.44,66.58) .. controls (379.64,66.58) and (378.19,65.13) .. (378.19,63.34) -- cycle ;
\draw    (378.57,89.74) -- (349.23,109.18) ;
\draw [shift={(346.73,110.84)}, rotate = 326.47] [fill={rgb, 255:red, 0; green, 0; blue, 0 }  ][line width=0.08]  [draw opacity=0] (10.72,-5.15) -- (0,0) -- (10.72,5.15) -- (7.12,0) -- cycle    ;
\draw    (381.83,67.04) -- (381.6,81.66) ;
\draw [shift={(381.56,84.66)}, rotate = 270.88] [fill={rgb, 255:red, 0; green, 0; blue, 0 }  ][line width=0.08]  [draw opacity=0] (10.72,-5.15) -- (0,0) -- (10.72,5.15) -- (7.12,0) -- cycle    ;
\draw    (347.92,113.91) -- (410.71,112.98) ;
\draw [shift={(413.71,112.93)}, rotate = 179.15] [fill={rgb, 255:red, 0; green, 0; blue, 0 }  ][line width=0.08]  [draw opacity=0] (10.72,-5.15) -- (0,0) -- (10.72,5.15) -- (7.12,0) -- cycle    ;

\draw (379.22,23.88) node [anchor=north west][inner sep=0.75pt]  [font=\tiny]  {$u$};
\draw (350.11,89.44) node [anchor=north west][inner sep=0.75pt]  [font=\scriptsize]  {$+$};
\draw (398.79,87.4) node [anchor=north west][inner sep=0.75pt]  [font=\scriptsize]  {$+$};
\draw (365.5,67.27) node [anchor=north west][inner sep=0.75pt]  [font=\scriptsize]  {$+$};
\draw (376.08,119.3) node [anchor=north west][inner sep=0.75pt]  [font=\scriptsize]  {$-$};
\draw (388.3,60.61) node [anchor=north west][inner sep=0.75pt]  [font=\tiny]  {$1$};
\draw (390.3,82.41) node [anchor=north west][inner sep=0.75pt]  [font=\tiny]  {$2$};
\draw (342.2,123.41) node [anchor=north west][inner sep=0.75pt]  [font=\tiny]  {$3$};
\draw (417.6,123.91) node [anchor=north west][inner sep=0.75pt]  [font=\tiny]  {$4$};

\end{tikzpicture}
        \caption{}
        \label{5fig d}
    \end{subfigure}
    \hfill
    \begin{subfigure}{0.14\textwidth}
        \centering
    
\tikzset{every picture/.style={scale=0.85pt}} 

\begin{tikzpicture}[x=0.75pt,y=0.75pt,yscale=-1,xscale=1]

\draw  [fill={rgb, 255:red, 0; green, 192; blue, 248 }  ,fill opacity=1 ] (224.33,89.69) .. controls (224.33,87.9) and (225.78,86.45) .. (227.58,86.45) .. controls (229.37,86.45) and (230.82,87.9) .. (230.82,89.69) .. controls (230.82,91.48) and (229.37,92.93) .. (227.58,92.93) .. controls (225.78,92.93) and (224.33,91.48) .. (224.33,89.69) -- cycle ;
\draw [color={rgb, 255:red, 161; green, 27; blue, 43 }  ,draw opacity=1 ]   (227.49,40.22) .. controls (229.16,41.89) and (229.15,43.55) .. (227.48,45.22) .. controls (225.81,46.89) and (225.81,48.55) .. (227.47,50.22) -- (227.47,50.88) -- (227.46,58.88) ;
\draw [shift={(227.45,61.88)}, rotate = 270.1] [fill={rgb, 255:red, 161; green, 27; blue, 43 }  ,fill opacity=1 ][line width=0.08]  [draw opacity=0] (10.72,-5.15) -- (0,0) -- (10.72,5.15) -- (7.12,0) -- cycle    ;
\draw  [fill={rgb, 255:red, 239; green, 43; blue, 67 }  ,fill opacity=1 ] (224.25,36.98) .. controls (224.25,35.18) and (225.7,33.73) .. (227.49,33.73) .. controls (229.28,33.73) and (230.74,35.18) .. (230.74,36.98) .. controls (230.74,38.77) and (229.28,40.22) .. (227.49,40.22) .. controls (225.7,40.22) and (224.25,38.77) .. (224.25,36.98) -- cycle ;
\draw  [fill={rgb, 255:red, 0; green, 192; blue, 248 }  ,fill opacity=1 ] (259.73,114.72) .. controls (259.73,112.93) and (261.18,111.48) .. (262.98,111.48) .. controls (264.77,111.48) and (266.22,112.93) .. (266.22,114.72) .. controls (266.22,116.51) and (264.77,117.96) .. (262.98,117.96) .. controls (261.18,117.96) and (259.73,116.51) .. (259.73,114.72) -- cycle ;
\draw  [fill={rgb, 255:red, 0; green, 192; blue, 248 }  ,fill opacity=1 ] (187.44,115.69) .. controls (187.44,113.9) and (188.9,112.45) .. (190.69,112.45) .. controls (192.48,112.45) and (193.93,113.9) .. (193.93,115.69) .. controls (193.93,117.48) and (192.48,118.94) .. (190.69,118.94) .. controls (188.9,118.94) and (187.44,117.48) .. (187.44,115.69) -- cycle ;
\draw    (233.04,92.99) -- (261.51,111.39) ;
\draw [shift={(230.52,91.37)}, rotate = 32.88] [fill={rgb, 255:red, 0; green, 0; blue, 0 }  ][line width=0.08]  [draw opacity=0] (10.72,-5.15) -- (0,0) -- (10.72,5.15) -- (7.12,0) -- cycle    ;
\draw  [fill={rgb, 255:red, 0; green, 192; blue, 248 }  ,fill opacity=1 ] (224.21,65.12) .. controls (224.21,63.33) and (225.66,61.88) .. (227.45,61.88) .. controls (229.25,61.88) and (230.7,63.33) .. (230.7,65.12) .. controls (230.7,66.91) and (229.25,68.36) .. (227.45,68.36) .. controls (225.66,68.36) and (224.21,66.91) .. (224.21,65.12) -- cycle ;
\draw    (224.59,91.53) -- (195.25,110.97) ;
\draw [shift={(192.75,112.63)}, rotate = 326.47] [fill={rgb, 255:red, 0; green, 0; blue, 0 }  ][line width=0.08]  [draw opacity=0] (10.72,-5.15) -- (0,0) -- (10.72,5.15) -- (7.12,0) -- cycle    ;
\draw    (227.85,68.82) -- (227.62,83.45) ;
\draw [shift={(227.58,86.45)}, rotate = 270.88] [fill={rgb, 255:red, 0; green, 0; blue, 0 }  ][line width=0.08]  [draw opacity=0] (10.72,-5.15) -- (0,0) -- (10.72,5.15) -- (7.12,0) -- cycle    ;
\draw    (193.93,115.69) -- (256.73,114.76) ;
\draw [shift={(259.73,114.72)}, rotate = 179.15] [fill={rgb, 255:red, 0; green, 0; blue, 0 }  ][line width=0.08]  [draw opacity=0] (10.72,-5.15) -- (0,0) -- (10.72,5.15) -- (7.12,0) -- cycle    ;

\draw (225.24,25.67) node [anchor=north west][inner sep=0.75pt]  [font=\tiny]  {$u$};
\draw (204.25,88.1) node [anchor=north west][inner sep=0.75pt]  [font=\scriptsize]  {$+$};
\draw (221.55,120.54) node [anchor=north west][inner sep=0.75pt]  [font=\scriptsize]  {$-$};
\draw (250.22,93.51) node [anchor=north west][inner sep=0.75pt]  [font=\scriptsize]  {$-$};
\draw (209.93,67.4) node [anchor=north west][inner sep=0.75pt]  [font=\scriptsize]  {$+$};
\draw (236.3,63.01) node [anchor=north west][inner sep=0.75pt]  [font=\tiny]  {$1$};
\draw (239.9,83.21) node [anchor=north west][inner sep=0.75pt]  [font=\tiny]  {$2$};
\draw (187,126.61) node [anchor=north west][inner sep=0.75pt]  [font=\tiny]  {$3$};
\draw (261.6,124.71) node [anchor=north west][inner sep=0.75pt]  [font=\tiny]  {$4$};

\end{tikzpicture}
        \caption{}
        \label{5fig e}
    \end{subfigure}
    \hfill
    \begin{subfigure}{0.14\textwidth}

\tikzset{every picture/.style={scale=0.85pt}} 

\begin{tikzpicture}[x=0.75pt,y=0.75pt,yscale=-1,xscale=1]

\draw  [fill={rgb, 255:red, 0; green, 192; blue, 248 }  ,fill opacity=1 ] (185.75,402.14) .. controls (185.75,400.34) and (187.2,398.89) .. (188.99,398.89) .. controls (190.79,398.89) and (192.24,400.34) .. (192.24,402.14) .. controls (192.24,403.93) and (190.79,405.38) .. (188.99,405.38) .. controls (187.2,405.38) and (185.75,403.93) .. (185.75,402.14) -- cycle ;
\draw    (189.26,381.27) -- (189.04,395.89) ;
\draw [shift={(188.99,398.89)}, rotate = 270.88] [fill={rgb, 255:red, 0; green, 0; blue, 0 }  ][line width=0.08]  [draw opacity=0] (10.72,-5.15) -- (0,0) -- (10.72,5.15) -- (7.12,0) -- cycle    ;
\draw  [fill={rgb, 255:red, 0; green, 192; blue, 248 }  ,fill opacity=1 ] (186.49,345.03) .. controls (186.49,343.24) and (187.94,341.79) .. (189.73,341.79) .. controls (191.53,341.79) and (192.98,343.24) .. (192.98,345.03) .. controls (192.98,346.82) and (191.53,348.28) .. (189.73,348.28) .. controls (187.94,348.28) and (186.49,346.82) .. (186.49,345.03) -- cycle ;
\draw  [fill={rgb, 255:red, 0; green, 192; blue, 248 }  ,fill opacity=1 ] (186.02,378.03) .. controls (186.02,376.23) and (187.47,374.78) .. (189.26,374.78) .. controls (191.06,374.78) and (192.51,376.23) .. (192.51,378.03) .. controls (192.51,379.82) and (191.06,381.27) .. (189.26,381.27) .. controls (187.47,381.27) and (186.02,379.82) .. (186.02,378.03) -- cycle ;
\draw [color={rgb, 255:red, 161; green, 27; blue, 43 }  ,draw opacity=1 ]   (189.65,295.56) .. controls (191.31,297.23) and (191.31,298.89) .. (189.64,300.56) .. controls (187.97,302.23) and (187.97,303.89) .. (189.63,305.56) -- (189.63,306.22) -- (189.62,314.22) ;
\draw [shift={(189.61,317.22)}, rotate = 270.1] [fill={rgb, 255:red, 161; green, 27; blue, 43 }  ,fill opacity=1 ][line width=0.08]  [draw opacity=0] (10.72,-5.15) -- (0,0) -- (10.72,5.15) -- (7.12,0) -- cycle    ;
\draw  [fill={rgb, 255:red, 239; green, 43; blue, 67 }  ,fill opacity=1 ] (186.4,292.32) .. controls (186.4,290.53) and (187.86,289.07) .. (189.65,289.07) .. controls (191.44,289.07) and (192.89,290.53) .. (192.89,292.32) .. controls (192.89,294.11) and (191.44,295.56) .. (189.65,295.56) .. controls (187.86,295.56) and (186.4,294.11) .. (186.4,292.32) -- cycle ;
\draw  [fill={rgb, 255:red, 0; green, 192; blue, 248 }  ,fill opacity=1 ] (215.33,366) .. controls (215.33,364.21) and (216.79,362.76) .. (218.58,362.76) .. controls (220.37,362.76) and (221.82,364.21) .. (221.82,366) .. controls (221.82,367.79) and (220.37,369.24) .. (218.58,369.24) .. controls (216.79,369.24) and (215.33,367.79) .. (215.33,366) -- cycle ;
\draw  [fill={rgb, 255:red, 0; green, 192; blue, 248 }  ,fill opacity=1 ] (156.84,365) .. controls (156.84,363.21) and (158.3,361.76) .. (160.09,361.76) .. controls (161.88,361.76) and (163.33,363.21) .. (163.33,365) .. controls (163.33,366.79) and (161.88,368.24) .. (160.09,368.24) .. controls (158.3,368.24) and (156.84,366.79) .. (156.84,365) -- cycle ;
\draw [color={rgb, 255:red, 128; green, 128; blue, 128 }  ,draw opacity=0.52 ]   (192.68,346.71) -- (214.14,362.01) ;
\draw [shift={(216.58,363.76)}, rotate = 215.51] [fill={rgb, 255:red, 128; green, 128; blue, 128 }  ,fill opacity=0.52 ][line width=0.08]  [draw opacity=0] (10.72,-5.15) -- (0,0) -- (10.72,5.15) -- (7.12,0) -- cycle    ;
\draw  [fill={rgb, 255:red, 0; green, 192; blue, 248 }  ,fill opacity=1 ] (186.37,320.46) .. controls (186.37,318.67) and (187.82,317.22) .. (189.61,317.22) .. controls (191.41,317.22) and (192.86,318.67) .. (192.86,320.46) .. controls (192.86,322.25) and (191.41,323.71) .. (189.61,323.71) .. controls (187.82,323.71) and (186.37,322.25) .. (186.37,320.46) -- cycle ;
\draw    (186.75,346.87) -- (165.8,361.3) ;
\draw [shift={(163.33,363)}, rotate = 325.43] [fill={rgb, 255:red, 0; green, 0; blue, 0 }  ][line width=0.08]  [draw opacity=0] (10.72,-5.15) -- (0,0) -- (10.72,5.15) -- (7.12,0) -- cycle    ;
\draw    (190,324.17) -- (189.78,338.79) ;
\draw [shift={(189.73,341.79)}, rotate = 270.88] [fill={rgb, 255:red, 0; green, 0; blue, 0 }  ][line width=0.08]  [draw opacity=0] (10.72,-5.15) -- (0,0) -- (10.72,5.15) -- (7.12,0) -- cycle    ;
\draw [color={rgb, 255:red, 128; green, 128; blue, 128 }  ,draw opacity=0.41 ]   (189.68,349.71) -- (189.31,371.78) ;
\draw [shift={(189.26,374.78)}, rotate = 270.96] [fill={rgb, 255:red, 128; green, 128; blue, 128 }  ,fill opacity=0.41 ][line width=0.08]  [draw opacity=0] (10.72,-5.15) -- (0,0) -- (10.72,5.15) -- (7.12,0) -- cycle    ;
\draw    (191.68,404.71) -- (213.14,420.01) ;
\draw [shift={(215.58,421.76)}, rotate = 215.51] [fill={rgb, 255:red, 0; green, 0; blue, 0 }  ][line width=0.08]  [draw opacity=0] (10.72,-5.15) -- (0,0) -- (10.72,5.15) -- (7.12,0) -- cycle    ;
\draw [color={rgb, 255:red, 0; green, 0; blue, 0 }  ,draw opacity=1 ]   (186.49,404.88) -- (165.55,419.31) ;
\draw [shift={(163.08,421.01)}, rotate = 325.43] [fill={rgb, 255:red, 0; green, 0; blue, 0 }  ,fill opacity=1 ][line width=0.08]  [draw opacity=0] (10.72,-5.15) -- (0,0) -- (10.72,5.15) -- (7.12,0) -- cycle    ;
\draw  [fill={rgb, 255:red, 0; green, 192; blue, 248 }  ,fill opacity=1 ] (156.59,423.01) .. controls (156.59,421.22) and (158.04,419.77) .. (159.83,419.77) .. controls (161.62,419.77) and (163.08,421.22) .. (163.08,423.01) .. controls (163.08,424.8) and (161.62,426.26) .. (159.83,426.26) .. controls (158.04,426.26) and (156.59,424.8) .. (156.59,423.01) -- cycle ;
\draw  [fill={rgb, 255:red, 0; green, 192; blue, 248 }  ,fill opacity=1 ] (215.58,423.76) .. controls (215.58,421.97) and (217.03,420.51) .. (218.82,420.51) .. controls (220.62,420.51) and (222.07,421.97) .. (222.07,423.76) .. controls (222.07,425.55) and (220.62,427) .. (218.82,427) .. controls (217.03,427) and (215.58,425.55) .. (215.58,423.76) -- cycle ;

\draw (187.4,281.01) node [anchor=north west][inner sep=0.75pt]  [font=\tiny]  {$u$};
\draw (175.22,379.86) node [anchor=north west][inner sep=0.75pt]  [font=\scriptsize]  {$-$};
\draw (175.6,324.4) node [anchor=north west][inner sep=0.75pt]  [font=\scriptsize]  {$+$};
\draw (161.54,344.4) node [anchor=north west][inner sep=0.75pt]  [font=\scriptsize]  {$+$};
\draw (201.4,317.01) node [anchor=north west][inner sep=0.75pt]  [font=\tiny]  {$1$};
\draw (199.4,342.01) node [anchor=north west][inner sep=0.75pt]  [font=\tiny]  {$2$};
\draw (148.9,363.01) node [anchor=north west][inner sep=0.75pt]  [font=\tiny]  {$3$};
\draw (225.9,363.51) node [anchor=north west][inner sep=0.75pt]  [font=\tiny]  {$4$};
\draw (199.4,376.01) node [anchor=north west][inner sep=0.75pt]  [font=\tiny]  {$5$};
\draw (199.4,397.01) node [anchor=north west][inner sep=0.75pt]  [font=\tiny]  {$6$};
\draw (177.04,355.33) node [anchor=north west][inner sep=0.75pt]  [font=\scriptsize]  {$-$};
\draw (147.9,420.51) node [anchor=north west][inner sep=0.75pt]  [font=\tiny]  {$5$};
\draw (225.4,422.01) node [anchor=north west][inner sep=0.75pt]  [font=\tiny]  {$4$};
\draw (204.1,401.4) node [anchor=north west][inner sep=0.75pt]  [font=\scriptsize]  {$+$};
\draw (208.72,342.86) node [anchor=north west][inner sep=0.75pt]  [font=\scriptsize]  {$-$};
\draw (168.6,399.4) node [anchor=north west][inner sep=0.75pt]  [font=\scriptsize]  {$+$};

\end{tikzpicture}
            \caption{}
            \label{5fig f}   
    \end{subfigure}
   \caption{The figures shown here are $\mathcal{LUG}^{\mathcal{H}}(\mathcal{G}_s)$ for different digraphs. In each of them, it can be seen that all the nodes of the digraph get spanned by the respective $\mathcal{LUG}^{\mathcal{H}}(\mathcal{G}_s)$. They differ in the nature of the paths from node $1$.}
    \label{Fig_5_thm}
\end{figure}
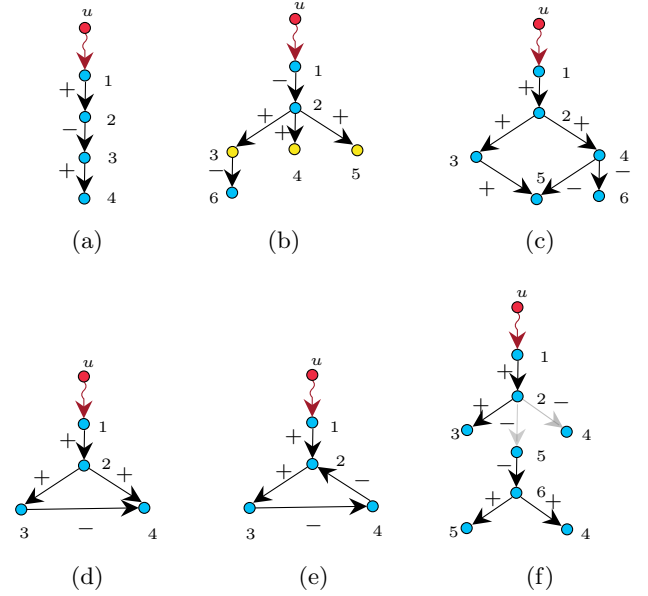

\begin{enumerate}
  \item[i.] All the nodes in $\{2,3,..,n\}$ have unique walks from node $1$. For a given set of nodes in layer \( L_p \) of \( \mathcal{LUG}^{\mathcal{H}}(\mathcal{G}_s) \), the corresponding column in \( \mathfrak{S}[\mathcal{C}(\mathcal{A,B})] \), denoted by \( [\Psi]_{p-1} \), contains nonzero entries. Moreover, all of them have the same sign owing to sign matching.
\begin{enumerate}
\item If all the walks have different lengths as shown in Fig~\ref{Fig_5_thm} (\subref{5fig a}), then $\mathcal{LUG}^{\mathcal{H}}(\mathcal{G}_s)$ reduces to a path graph, which is structurally controllable.
\item If no walk longer than length $m$ exists, with $2 \leq m < n$, then columns $(m+2)$ to $n$ of $\mathfrak{S}[\mathcal{C}(\mathcal{A,B})]$ contain only zero entries, and all nonzero entries correspond to nodes are in the first $m$ columns. Collectively, these columns ensure that the states of all nodes can reach the positive orthant. An illustrative example is shown in Fig.~\ref{Fig_5_thm}(\subref{5fig b}).

\end{enumerate}
Hence, for all parametric realizations with $a_{ij} \in \mathbb{R}^{+},~\forall (i,j)\in\mathcal{E}$, the states of all the nodes can be driven to the positive orthant of $\mathbb{R}^n$.

\item[ii.] In some cases, there may be multiple walks of the same length $p$ from node $1$ to node $j$ to be herded, with different signs in the product of edge weights. Consequently, the entry $\mathfrak{S}[\mathcal{C}_{(j,p+1)}]$ takes the value $(+/-)$.It can be positive or negative, depending upon the parametric realization of $(\mathcal{A,B})$. An instance of which is shown in Fig~\ref{Fig_5_thm} (\subref{5fig c}). The following possibilities emerge in this context.
\begin{enumerate}
\item If the entry $\mathcal{C}_{(j,p+1)}$ is the only nonzero entry in its column, then the corresponding node is herdable, regardless of its sign.

\item If $\mathcal{C}_{(j,p+1)}$ corresponds to the only node that needs to be herded in that column, while other nodes in the column are herded in different columns, it is herdable irrespective of its sign.

\item If the $j^{\text{th}}$ node must be herded along with other nodes in the same column, then for some choice of parameters in $(\mathcal{A},\mathcal{B})$, the sign of $\mathfrak{S}[\mathcal{C}_{(j,p+1)}]$ can be realized such that all these nodes are herded together with the same sign.
\end{enumerate}

\item[iii.] Every column of $\mathfrak{S}[\mathcal{C}(\mathcal{A,B})]$ is nonzero for at least one parametric realization if and only if:
\begin{enumerate}
    \item the digraph is acyclic and contains a walk of length $n$ from the leader to at least one node in the network, as illustrated in Fig.~\ref{Fig_5_thm}(\subref{5fig d}), or
    \item the digraph is cyclic, in which case each node in the cycle is reachable through multiple walks, as shown in Fig.~\ref{Fig_5_thm}(\subref{5fig e}).
\end{enumerate}

In both cases, every nonzero column has entries corresponding to sign-matched nodes with the same sign. Therefore, all the states can be driven to the positive orthant.

\item[iv.] In all other cases, where the nodes that are sign-matched at corresponding layers, as shown in Fig.~\ref{Fig_5_thm}(\subref{5fig f}), are herded in that respective column. Thereby, the states of all nodes are herdable to the positive orthant.
\end{enumerate}

In the matching process, nodes that do not repeat are given higher priority and matched first, whereas nodes belonging to cycles or those repeating due to multiple paths are considered subsequently.
\begin{lemma}\label{lemma C herd}
Consider a system \eqref{sys1:sys1eqn} characterized by structured matrices $\mathcal{\bar{A},\bar{B}}$, where the parameters vary such that $a_{i,j}~\in~\mathbb{R^+}$ and $b_i \in \mathbb{R}$. Define a constant $d\in \mathbb{R^+}$ such that 
\[
a_{i,j} 
\begin{cases}
> d \gg 0, & \text{if $(j,i)$ is an edge that sign matches $i$,} \\[6pt]
\in (0,\, d), & \text{otherwise.}
\end{cases}
\]
\normalfont
 As the result, realization of $\mathcal{\bar{A},\bar{B}}$ is such that the sign-matched nodes in the columns of $\mathcal{C}(\mathcal{\bar{A},\bar{B}})$ have the same sign. The entry $(+/-)$  is realized to the sign that corresponds to sign matching.
 If there is at least one realization of $\mathcal{\bar{A},\bar{B}}$ such that, for every $i \in \{1,2..n\}$, there exists at least one $x$ satisfying 
\begin{equation}\label{C herd satis}
\hspace{-1em}
|\mathcal{C}_{(i,x)}| >\sum_{j=1; j\neq x }^{n} | \mathcal{C}_{(i,j)}~|
~~ \left |\raisebox{-0.4em}{\scriptsize$\begin{array}{c}
x \in \{1, \dots n\}
\end{array}$}
\right.\vspace{-1mm}
\end{equation} 
and $C_{(i,x)}$ is part of the nodes that are herdable in the respective column, which are unisigned, then there exists a suitable $\bm{\delta}$ such that $\mathcal{C}(\mathcal{\bar{A},\bar{B}})\cdot\bm{\delta}=v>0$. Therefore \eqref{sys1:sys1eqn} is $\mathcal{SS}$ herdable.
\end{lemma}

Building on the previous definitions, we introduce our main theorem regarding the $\mathcal{SS}$ herdability of a digraph. 

\begin{theorem}\label{theorem1}
 Let $\mathcal{G(A,B)}$ be a digraph described by (\ref{sys1:sys1eqn}) and $\mathcal{G}_s$ be the signed layered graph associated with it. Then $\mathcal{G(A,B)}$ is $\mathcal{SS}$ herdable if and only if  $\mathcal{G}_s$  has an $\mathcal{LUG}^{\mathcal{H}}(\mathcal{G}_s)$ that spans all the nodes of $\mathcal{G(A,B)}$ in at most $n$ layers.
\end{theorem}
The proof is presented in Appendix~\ref{proof1}.

\begin{remark} \label{delta remark}
Consider the $\mathcal{SSC}$ matrix of this realization where $\mathfrak{S}[\mathcal{C}_{(i,j)}]$ is given by (\ref{eq:sign col}). Let $\boldsymbol{\delta} = [\delta_1, \delta_2, \dots, \delta_n]$ be a vector in $\mathbb{R}^n$, where $\delta_i \in \mathbb{R}$ for all $i \in \{1, 2, \dots, n\}$. If there exists a positive image in the range space of $\mathcal{C}(\mathcal{A}, \mathcal{B})$, then the system described by \eqref{sys1:sys1eqn} is herdable. For a parametric realization proposed in Lemma \ref{lemma C herd}, a suitable preimage $\boldsymbol{\delta}$, that can guide the system described by \eqref{sys1:sys1eqn} toward the positive orthant is constructed using the $\mathcal{LUG}^{\mathcal{H}}(\mathcal{G}_s)$ framework as follows:

\begin{enumerate}
    \item \textbf{Sign pattern of $\delta$:}
    The sign of each entry of the $\delta$ is given by the following expression:
    \vspace{-2mm}
    \begin{equation}\label{delta}
\operatorname{sign}~(\delta_{p}) \Big|_{p\in\{1,2.. n\}}=  \operatorname{sign}~(\mathcal{W}_{(1,i)})=~\mathfrak{S}[\mathcal{C}_{(1,i)}] \vspace{-1.5mm}
\end{equation} 
where $i$ denotes a node in the layer $L_p$ of $\mathcal{LUG}^{\mathcal{H}}(\mathcal{G}_s)$. 
\item  \textbf{Magnitude of each $\delta_i$:} The magnitude condition requires that for every node $i$ in the layer $L_x$ of $\mathcal{LUG}^{\mathcal{H}}(\mathcal{G}_s)$, there exists a $\delta_x$ and the magnitude of $\delta_x$ is chosen such that
\vspace{-2mm}
\begin{equation}\label{delta satis}
\hspace{-1em}
\delta_x~ \mathcal{C}_{(i,x)} >\sum_{j=1; j\neq x }^{n} |~\delta_j~ \mathcal{C}_{(i,j)}~|
~~ \left |\raisebox{-0.4em}{\scriptsize$\begin{array}{c}
x \in \{1, \dots n\}
\end{array}$}
\right.\vspace{-1mm}
\end{equation} 
\end{enumerate}
where $\delta_x$ and $\mathcal{C}_{(i,x)}$ have the same sign, which is evident from \eqref{delta}. If a node $i$ repeated multiple times in $\mathcal{LUG}^{\mathcal{H}}$ as a result of recurring sign matching, then condition \eqref{delta satis} can be satisfied with any $\delta_x, x \in\{1,2, \dots n\}$ for the corresponding node $i$. Hence, each $\delta_x$ is chosen such that \eqref{delta satis} is satisfied for each node $i$ of $\mathcal{G(A,B)}$.

A vector $\boldsymbol{\delta}$ with every entry satisfying \eqref{delta} and \eqref{delta satis} serves as a preimage under $\mathcal{C}(\mathcal{A}, \mathcal{B})$, generating a vector $\mathbf{v} \in \mathcal{I}m(\mathcal{C}(\mathcal{A}, \mathcal{B}))$ such that $\mathbf{v} > 0$. Hence if all the nodes of a digraph $\mathcal{G}(\mathcal{A}, \mathcal{B})$ is spanned by an $\mathcal{LUG}^{\mathcal{H}}(\mathcal{G}_s)$, then $(\mathcal{C}(\mathcal{A}, \mathcal{B})$  associated with $\mathcal{G}(\mathcal{A}, \mathcal{B})$ will always has a subspace spanning inside the positive orthant.

\end{remark}
\begin{remark}
If all follower nodes are spanned within $m$ layers of $\mathcal{LUG}^{\mathcal{H}}(\mathcal{G}_s)$, where $m < n$, then a straightforward construction of $\bm{\delta}$ can be given as follows:

\[
\delta_l =
\begin{cases}
\text{as in Remark}~\ref{delta remark}, & l \leq m, \\[6pt]
0, & m < l \leq n,
\end{cases}
\quad l \in \{1,2..n\}
\]

\end{remark}

\begin{eg}

\normalfont
Consider the digraph $\mathcal{G}_8$ given in Fig.~\ref{Dig with sds}(\subref{SD digraph}). The signed layered graph $\mathcal{G}_s$ associated with the digraph is shown in Fig.~\ref{Dig with sds}(\subref{lug H in sd dig}).  
    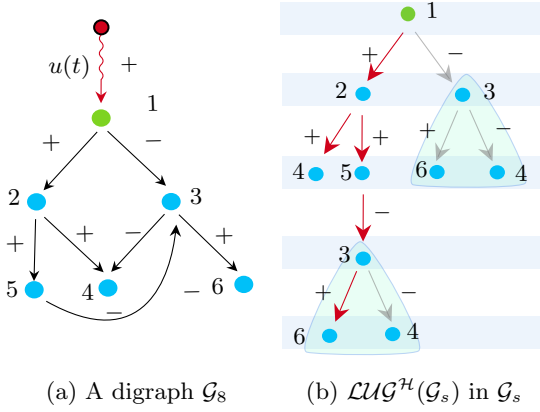
\begin{figure}[ht]
        \centering
       
      \begin{subfigure}{0.2\textwidth}
\tikzset{every picture/.style={scale=0.75pt}} 

\begin{tikzpicture}[x=0.75pt,y=0.75pt,yscale=-1,xscale=1]

\draw [color={rgb, 255:red, 0; green, 0; blue, 0 }  ,draw opacity=1 ]   (487.36,182.22) -- (486.38,220.8) ;
\draw [shift={(486.3,223.8)}, rotate = 271.46] [fill={rgb, 255:red, 0; green, 0; blue, 0 }  ,fill opacity=1 ][line width=0.08]  [draw opacity=0] (7.14,-3.43) -- (0,0) -- (7.14,3.43) -- (4.74,0) -- cycle    ;
\draw [color={rgb, 255:red, 0; green, 0; blue, 0 }  ,draw opacity=1 ]   (535.26,123.24) -- (571.98,157.84) ;
\draw [shift={(574.17,159.89)}, rotate = 223.29] [fill={rgb, 255:red, 0; green, 0; blue, 0 }  ,fill opacity=1 ][line width=0.08]  [draw opacity=0] (7.14,-3.43) -- (0,0) -- (7.14,3.43) -- (4.74,0) -- cycle    ;
\draw [color={rgb, 255:red, 0; green, 0; blue, 0 }  ,draw opacity=1 ]   (528.36,123.74) -- (494.69,160.68) ;
\draw [shift={(492.67,162.89)}, rotate = 312.35] [fill={rgb, 255:red, 0; green, 0; blue, 0 }  ,fill opacity=1 ][line width=0.08]  [draw opacity=0] (7.14,-3.43) -- (0,0) -- (7.14,3.43) -- (4.74,0) -- cycle    ;
\draw  [draw opacity=0][fill={rgb, 255:red, 0; green, 192; blue, 248 }  ,fill opacity=1 ] (487.83,164.83) .. controls (491.31,164.77) and (494.2,167.31) .. (494.28,170.52) .. controls (494.35,173.73) and (491.59,176.38) .. (488.11,176.44) .. controls (484.63,176.51) and (481.75,173.96) .. (481.67,170.75) .. controls (481.59,167.55) and (484.35,164.9) .. (487.83,164.83) -- cycle ;
\draw  [draw opacity=0][fill={rgb, 255:red, 0; green, 192; blue, 248 }  ,fill opacity=1 ] (535.4,222.6) .. controls (538.88,222.53) and (541.77,225.08) .. (541.84,228.29) .. controls (541.92,231.49) and (539.16,234.14) .. (535.68,234.21) .. controls (532.2,234.27) and (529.31,231.73) .. (529.24,228.52) .. controls (529.16,225.32) and (531.92,222.66) .. (535.4,222.6) -- cycle ;
\draw  [draw opacity=0][fill={rgb, 255:red, 0; green, 192; blue, 248 }  ,fill opacity=1 ] (577.46,164.67) .. controls (580.94,164.6) and (583.83,167.15) .. (583.91,170.35) .. controls (583.98,173.56) and (581.22,176.21) .. (577.74,176.28) .. controls (574.26,176.34) and (571.38,173.8) .. (571.3,170.59) .. controls (571.22,167.38) and (573.98,164.73) .. (577.46,164.67) -- cycle ;
\draw    (580.69,189.33) .. controls (577.36,221.8) and (559.44,274.15) .. (494.46,237.28) ;
\draw [shift={(581,185.88)}, rotate = 94.29] [fill={rgb, 255:red, 0; green, 0; blue, 0 }  ][line width=0.08]  [draw opacity=0] (8.04,-3.86) -- (0,0) -- (8.04,3.86) -- (5.34,0) -- cycle    ;
\draw  [draw opacity=0][fill={rgb, 255:red, 0; green, 192; blue, 248 }  ,fill opacity=1 ] (486.3,223.8) .. controls (489.78,223.74) and (492.66,226.28) .. (492.74,229.49) .. controls (492.82,232.7) and (490.06,235.35) .. (486.58,235.41) .. controls (483.1,235.48) and (480.21,232.93) .. (480.14,229.72) .. controls (480.06,226.52) and (482.82,223.87) .. (486.3,223.8) -- cycle ;
\draw  [draw opacity=0][fill={rgb, 255:red, 126; green, 211; blue, 33 }  ,fill opacity=1 ] (530.87,108.61) .. controls (534.35,108.55) and (537.23,111.1) .. (537.31,114.3) .. controls (537.39,117.51) and (534.63,120.16) .. (531.15,120.22) .. controls (527.66,120.29) and (524.78,117.74) .. (524.7,114.54) .. controls (524.63,111.33) and (527.39,108.68) .. (530.87,108.61) -- cycle ;
\draw [color={rgb, 255:red, 0; green, 0; blue, 0 }  ,draw opacity=1 ]   (582.76,179.24) -- (619.58,215.14) ;
\draw [shift={(621.73,217.23)}, rotate = 224.27] [fill={rgb, 255:red, 0; green, 0; blue, 0 }  ,fill opacity=1 ][line width=0.08]  [draw opacity=0] (7.14,-3.43) -- (0,0) -- (7.14,3.43) -- (4.74,0) -- cycle    ;
\draw [color={rgb, 255:red, 0; green, 0; blue, 0 }  ,draw opacity=1 ]   (573.17,178.89) -- (539.69,215.39) ;
\draw [shift={(537.66,217.6)}, rotate = 312.53] [fill={rgb, 255:red, 0; green, 0; blue, 0 }  ,fill opacity=1 ][line width=0.08]  [draw opacity=0] (7.14,-3.43) -- (0,0) -- (7.14,3.43) -- (4.74,0) -- cycle    ;
\draw [color={rgb, 255:red, 0; green, 0; blue, 0 }  ,draw opacity=1 ]   (493.17,180.89) -- (530.48,215.84) ;
\draw [shift={(532.67,217.89)}, rotate = 223.13] [fill={rgb, 255:red, 0; green, 0; blue, 0 }  ,fill opacity=1 ][line width=0.08]  [draw opacity=0] (7.14,-3.43) -- (0,0) -- (7.14,3.43) -- (4.74,0) -- cycle    ;
\draw  [draw opacity=0][fill={rgb, 255:red, 0; green, 192; blue, 248 }  ,fill opacity=1 ] (625.8,220.2) .. controls (629.28,220.13) and (632.17,222.68) .. (632.24,225.89) .. controls (632.32,229.09) and (629.56,231.74) .. (626.08,231.81) .. controls (622.6,231.87) and (619.71,229.33) .. (619.64,226.12) .. controls (619.56,222.92) and (622.32,220.26) .. (625.8,220.2) -- cycle ;
\draw [color={rgb, 255:red, 208; green, 2; blue, 27 }  ,draw opacity=1 ]   (531.11,53.58) .. controls (532.77,55.25) and (532.76,56.92) .. (531.09,58.58) .. controls (529.42,60.24) and (529.41,61.91) .. (531.07,63.58) .. controls (532.73,65.25) and (532.72,66.92) .. (531.05,68.58) .. controls (529.38,70.24) and (529.37,71.91) .. (531.03,73.58) .. controls (532.69,75.25) and (532.68,76.92) .. (531.01,78.58) .. controls (529.34,80.25) and (529.34,81.91) .. (531,83.58) .. controls (532.66,85.25) and (532.65,86.92) .. (530.98,88.58) -- (530.96,93.52) -- (530.93,101.52) ;
\draw [shift={(530.92,104.52)}, rotate = 270.21] [fill={rgb, 255:red, 208; green, 2; blue, 27 }  ,fill opacity=1 ][line width=0.08]  [draw opacity=0] (8.04,-3.86) -- (0,0) -- (8.04,3.86) -- (5.34,0) -- cycle    ;
\draw  [color={rgb, 255:red, 0; green, 0; blue, 0 }  ,draw opacity=1 ][fill={rgb, 255:red, 208; green, 2; blue, 27 }  ,fill opacity=1 ][line width=0.75]  (531,48.87) .. controls (533.64,48.82) and (535.83,50.88) .. (535.89,53.48) .. controls (535.95,56.08) and (533.85,58.24) .. (531.21,58.29) .. controls (528.57,58.34) and (526.38,56.27) .. (526.32,53.67) .. controls (526.26,51.07) and (528.36,48.92) .. (531,48.87) -- cycle ;

\draw (467.04,161.06) node [anchor=north west][inner sep=0.75pt]  [font=\small,color={rgb, 255:red, 0; green, 0; blue, 0 }  ,opacity=1 ,rotate=-358.27]  {$2$};
\draw (588.26,160.05) node [anchor=north west][inner sep=0.75pt]  [font=\small,color={rgb, 255:red, 0; green, 0; blue, 0 }  ,opacity=1 ,rotate=-358.27]  {$3$};
\draw (465.77,222.38) node [anchor=north west][inner sep=0.75pt]  [font=\small,color={rgb, 255:red, 0; green, 0; blue, 0 }  ,opacity=1 ,rotate=-358.27]  {$5$};
\draw (516.01,224.78) node [anchor=north west][inner sep=0.75pt]  [font=\small,color={rgb, 255:red, 0; green, 0; blue, 0 }  ,opacity=1 ,rotate=-358.27]  {$4$};
\draw (559.37,97.89) node [anchor=north west][inner sep=0.75pt]  [font=\small,color={rgb, 255:red, 0; green, 0; blue, 0 }  ,opacity=1 ,rotate=-358.27]  {$1$};
\draw (557.81,122.77) node [anchor=north west][inner sep=0.75pt]  [font=\footnotesize,rotate=-358.88]  {$-$};
\draw (583.15,224) node [anchor=north west][inner sep=0.75pt]  [font=\footnotesize,color={rgb, 255:red, 0; green, 0; blue, 0 }  ,opacity=1 ,rotate=-358.88]  {$-$};
\draw (489.58,124.02) node [anchor=north west][inner sep=0.75pt]  [font=\footnotesize,color={rgb, 255:red, 0; green, 0; blue, 0 }  ,opacity=1 ,rotate=-358.88]  {$+$};
\draw (531.66,237.73) node [anchor=north west][inner sep=0.75pt]  [font=\footnotesize,color={rgb, 255:red, 0; green, 0; blue, 0 }  ,opacity=1 ]  {$-$};
\draw (544.21,184.37) node [anchor=north west][inner sep=0.75pt]  [font=\footnotesize,rotate=-358.88]  {$-$};
\draw (465.58,192.02) node [anchor=north west][inner sep=0.75pt]  [font=\footnotesize,color={rgb, 255:red, 0; green, 0; blue, 0 }  ,opacity=1 ,rotate=-358.88]  {$+$};
\draw (511.98,184.82) node [anchor=north west][inner sep=0.75pt]  [font=\footnotesize,color={rgb, 255:red, 0; green, 0; blue, 0 }  ,opacity=1 ,rotate=-358.88]  {$+$};
\draw (602.57,220.78) node [anchor=north west][inner sep=0.75pt]  [font=\small,color={rgb, 255:red, 0; green, 0; blue, 0 }  ,opacity=1 ,rotate=-358.27]  {$6$};
\draw (604.78,186.42) node [anchor=north west][inner sep=0.75pt]  [font=\footnotesize,color={rgb, 255:red, 0; green, 0; blue, 0 }  ,opacity=1 ,rotate=-358.88]  {$+$};
\draw (494.23,70.2) node [anchor=north west][inner sep=0.75pt]  [font=\small]  {$u( t)$};
\draw (541.91,72.06) node [anchor=north west][inner sep=0.75pt]  [font=\scriptsize,color={rgb, 255:red, 0; green, 0; blue, 0 }  ,opacity=1 ,rotate=-358.88]  {$+$};

\end{tikzpicture}
    \caption{A digraph $\mathcal{G}_8$}
    \label{SD digraph}

   \end{subfigure}
  \vspace{0.5mm}
   \begin{subfigure}{0.2\textwidth}
       
\tikzset{every picture/.style={scale=0.75pt}} 
\tikzset{every picture/.style={scale=0.8pt}} 

\begin{tikzpicture}[x=0.75pt,y=0.75pt,yscale=-1,xscale=1]

\draw  [color={rgb, 255:red, 74; green, 144; blue, 226 }  ,draw opacity=0.38 ][fill={rgb, 255:red, 0; green, 247; blue, 141 }  ,fill opacity=0.09 ] (379.24,318.47) .. controls (386.72,318.7) and (422.63,384.01) .. (415.37,387.22) .. controls (408.11,390.42) and (354.45,389.69) .. (345.99,387.88) .. controls (337.54,386.08) and (371.76,318.23) .. (379.24,318.47) -- cycle ;
\draw  [color={rgb, 255:red, 74; green, 144; blue, 226 }  ,draw opacity=0.38 ][fill={rgb, 255:red, 0; green, 247; blue, 141 }  ,fill opacity=0.09 ] (312.91,423.7) .. controls (320.39,423.94) and (356.31,489.25) .. (349.04,492.46) .. controls (341.78,495.66) and (288.12,494.93) .. (279.67,493.12) .. controls (271.21,491.32) and (305.43,423.47) .. (312.91,423.7) -- cycle ;
\draw  [color={rgb, 255:red, 0; green, 0; blue, 0 }  ,draw opacity=0 ][fill={rgb, 255:red, 74; green, 144; blue, 226 }  ,fill opacity=0.1 ] (264.51,318.04) -- (426.93,318.04) -- (426.93,338.5) -- (264.51,338.5) -- cycle ;
\draw  [color={rgb, 255:red, 0; green, 0; blue, 0 }  ,draw opacity=0 ][fill={rgb, 255:red, 74; green, 144; blue, 226 }  ,fill opacity=0.1 ] (264.51,369.8) -- (426.93,369.8) -- (426.93,390.26) -- (264.51,390.26) -- cycle ;
\draw  [color={rgb, 255:red, 0; green, 0; blue, 0 }  ,draw opacity=0 ][fill={rgb, 255:red, 74; green, 144; blue, 226 }  ,fill opacity=0.1 ] (264.51,420.05) -- (426.93,420.05) -- (426.93,440.51) -- (264.51,440.51) -- cycle ;
\draw [color={rgb, 255:red, 208; green, 2; blue, 27 }  ,draw opacity=1 ]   (339.68,292.21) -- (319.11,320.25) ;
\draw [shift={(317.33,322.67)}, rotate = 306.27] [fill={rgb, 255:red, 208; green, 2; blue, 27 }  ,fill opacity=1 ][line width=0.08]  [draw opacity=0] (10.72,-5.15) -- (0,0) -- (10.72,5.15) -- (7.12,0) -- cycle    ;
\draw [color={rgb, 255:red, 128; green, 128; blue, 128 }  ,draw opacity=0.55 ]   (347.17,291.78) -- (368.57,321.24) ;
\draw [shift={(370.33,323.67)}, rotate = 234] [fill={rgb, 255:red, 128; green, 128; blue, 128 }  ,fill opacity=0.55 ][line width=0.08]  [draw opacity=0] (10.72,-5.15) -- (0,0) -- (10.72,5.15) -- (7.12,0) -- cycle    ;
\draw [color={rgb, 255:red, 208; green, 2; blue, 27 }  ,draw opacity=1 ][fill={rgb, 255:red, 183; green, 176; blue, 176 }  ,fill opacity=1 ]   (306.93,343.07) -- (291.97,366.15) ;
\draw [shift={(290.33,368.67)}, rotate = 302.96] [fill={rgb, 255:red, 208; green, 2; blue, 27 }  ,fill opacity=1 ][line width=0.08]  [draw opacity=0] (10.72,-5.15) -- (0,0) -- (10.72,5.15) -- (7.12,0) -- cycle    ;
\draw [color={rgb, 255:red, 208; green, 2; blue, 27 }  ,draw opacity=1 ]   (314.09,344.11) -- (313.98,370.78) ;
\draw [shift={(313.97,373.78)}, rotate = 270.23] [fill={rgb, 255:red, 208; green, 2; blue, 27 }  ,fill opacity=1 ][line width=0.08]  [draw opacity=0] (10.72,-5.15) -- (0,0) -- (10.72,5.15) -- (7.12,0) -- cycle    ;
\draw  [draw opacity=0][fill={rgb, 255:red, 0; green, 192; blue, 248 }  ,fill opacity=1 ] (309.87,330.72) .. controls (309.87,328.27) and (312.01,326.28) .. (314.64,326.28) .. controls (317.28,326.28) and (319.42,328.27) .. (319.42,330.72) .. controls (319.42,333.17) and (317.28,335.16) .. (314.64,335.16) .. controls (312.01,335.16) and (309.87,333.17) .. (309.87,330.72) -- cycle ;
\draw  [draw opacity=0][fill={rgb, 255:red, 0; green, 192; blue, 248 }  ,fill opacity=1 ] (309.29,380.55) .. controls (309.29,378.1) and (311.43,376.12) .. (314.06,376.12) .. controls (316.7,376.12) and (318.84,378.1) .. (318.84,380.55) .. controls (318.84,383) and (316.7,384.99) .. (314.06,384.99) .. controls (311.43,384.99) and (309.29,383) .. (309.29,380.55) -- cycle ;
\draw  [draw opacity=0][fill={rgb, 255:red, 0; green, 192; blue, 248 }  ,fill opacity=1 ] (356.03,379.85) .. controls (356.03,377.4) and (358.17,375.42) .. (360.8,375.42) .. controls (363.44,375.42) and (365.57,377.4) .. (365.57,379.85) .. controls (365.57,382.3) and (363.44,384.29) .. (360.8,384.29) .. controls (358.17,384.29) and (356.03,382.3) .. (356.03,379.85) -- cycle ;
\draw  [draw opacity=0][fill={rgb, 255:red, 0; green, 192; blue, 248 }  ,fill opacity=1 ] (393.64,380.05) .. controls (393.64,377.6) and (395.78,375.61) .. (398.42,375.61) .. controls (401.05,375.61) and (403.19,377.6) .. (403.19,380.05) .. controls (403.19,382.5) and (401.05,384.49) .. (398.42,384.49) .. controls (395.78,384.49) and (393.64,382.5) .. (393.64,380.05) -- cycle ;
\draw  [draw opacity=0][fill={rgb, 255:red, 126; green, 211; blue, 33 }  ,fill opacity=1 ] (342.6,275.81) .. controls (345.12,275.76) and (347.2,277.79) .. (347.26,280.36) .. controls (347.31,282.93) and (345.32,285.05) .. (342.8,285.1) .. controls (340.29,285.15) and (338.21,283.12) .. (338.15,280.55) .. controls (338.09,277.98) and (340.09,275.86) .. (342.6,275.81) -- cycle ;
\draw  [draw opacity=0][fill={rgb, 255:red, 0; green, 192; blue, 248 }  ,fill opacity=1 ] (280.49,381.15) .. controls (280.49,378.7) and (282.63,376.71) .. (285.26,376.71) .. controls (287.9,376.71) and (290.04,378.7) .. (290.04,381.15) .. controls (290.04,383.6) and (287.9,385.59) .. (285.26,385.59) .. controls (282.63,385.59) and (280.49,383.6) .. (280.49,381.15) -- cycle ;
\draw [color={rgb, 255:red, 208; green, 2; blue, 27 }  ,draw opacity=1 ]   (314.8,393.91) -- (314.69,420.58) ;
\draw [shift={(314.68,423.58)}, rotate = 270.23] [fill={rgb, 255:red, 208; green, 2; blue, 27 }  ,fill opacity=1 ][line width=0.08]  [draw opacity=0] (10.72,-5.15) -- (0,0) -- (10.72,5.15) -- (7.12,0) -- cycle    ;
\draw  [draw opacity=0][fill={rgb, 255:red, 0; green, 192; blue, 248 }  ,fill opacity=1 ] (309.98,434.29) .. controls (309.98,431.84) and (312.12,429.86) .. (314.76,429.86) .. controls (317.39,429.86) and (319.53,431.84) .. (319.53,434.29) .. controls (319.53,436.74) and (317.39,438.73) .. (314.76,438.73) .. controls (312.12,438.73) and (309.98,436.74) .. (309.98,434.29) -- cycle ;
\draw  [color={rgb, 255:red, 0; green, 0; blue, 0 }  ,draw opacity=0 ][fill={rgb, 255:red, 74; green, 144; blue, 226 }  ,fill opacity=0.1 ] (264.33,470.55) -- (426.76,470.55) -- (426.76,491.01) -- (264.33,491.01) -- cycle ;
\draw  [draw opacity=0][fill={rgb, 255:red, 0; green, 192; blue, 248 }  ,fill opacity=1 ] (289.03,482.85) .. controls (289.03,480.4) and (291.17,478.42) .. (293.8,478.42) .. controls (296.44,478.42) and (298.57,480.4) .. (298.57,482.85) .. controls (298.57,485.3) and (296.44,487.29) .. (293.8,487.29) .. controls (291.17,487.29) and (289.03,485.3) .. (289.03,482.85) -- cycle ;
\draw  [draw opacity=0][fill={rgb, 255:red, 0; green, 192; blue, 248 }  ,fill opacity=1 ] (328.35,481.78) .. controls (328.35,479.33) and (330.48,477.34) .. (333.12,477.34) .. controls (335.75,477.34) and (337.89,479.33) .. (337.89,481.78) .. controls (337.89,484.23) and (335.75,486.21) .. (333.12,486.21) .. controls (330.48,486.21) and (328.35,484.23) .. (328.35,481.78) -- cycle ;
\draw [color={rgb, 255:red, 208; green, 2; blue, 27 }  ,draw opacity=1 ]   (310.68,442.21) -- (298.48,471.24) ;
\draw [shift={(297.32,474)}, rotate = 292.8] [fill={rgb, 255:red, 208; green, 2; blue, 27 }  ,fill opacity=1 ][line width=0.08]  [draw opacity=0] (10.72,-5.15) -- (0,0) -- (10.72,5.15) -- (7.12,0) -- cycle    ;
\draw [color={rgb, 255:red, 128; green, 128; blue, 128 }  ,draw opacity=0.55 ]   (318.17,441.78) -- (331.21,471.37) ;
\draw [shift={(332.42,474.11)}, rotate = 246.21] [fill={rgb, 255:red, 128; green, 128; blue, 128 }  ,fill opacity=0.55 ][line width=0.08]  [draw opacity=0] (10.72,-5.15) -- (0,0) -- (10.72,5.15) -- (7.12,0) -- cycle    ;
\draw  [draw opacity=0][fill={rgb, 255:red, 0; green, 192; blue, 248 }  ,fill opacity=1 ] (371.98,331.29) .. controls (371.98,328.84) and (374.12,326.86) .. (376.76,326.86) .. controls (379.39,326.86) and (381.53,328.84) .. (381.53,331.29) .. controls (381.53,333.74) and (379.39,335.73) .. (376.76,335.73) .. controls (374.12,335.73) and (371.98,333.74) .. (371.98,331.29) -- cycle ;
\draw [color={rgb, 255:red, 128; green, 128; blue, 128 }  ,draw opacity=0.55 ]   (373.68,341.21) -- (361.48,370.24) ;
\draw [shift={(360.32,373)}, rotate = 292.8] [fill={rgb, 255:red, 128; green, 128; blue, 128 }  ,fill opacity=0.55 ][line width=0.08]  [draw opacity=0] (10.72,-5.15) -- (0,0) -- (10.72,5.15) -- (7.12,0) -- cycle    ;
\draw [color={rgb, 255:red, 128; green, 128; blue, 128 }  ,draw opacity=0.55 ]   (381.17,340.78) -- (394.21,370.37) ;
\draw [shift={(395.42,373.11)}, rotate = 246.21] [fill={rgb, 255:red, 128; green, 128; blue, 128 }  ,fill opacity=0.55 ][line width=0.08]  [draw opacity=0] (10.72,-5.15) -- (0,0) -- (10.72,5.15) -- (7.12,0) -- cycle    ;
\draw  [color={rgb, 255:red, 0; green, 0; blue, 0 }  ,draw opacity=0 ][fill={rgb, 255:red, 74; green, 144; blue, 226 }  ,fill opacity=0.1 ] (262.91,272.98) -- (425.33,272.98) -- (425.33,293.44) -- (262.91,293.44) -- cycle ;

\draw (295.55,321.87) node [anchor=north west][inner sep=0.75pt]  [font=\small,color={rgb, 255:red, 0; green, 0; blue, 0 }  ,opacity=1 ,rotate=-359.39]  {$2$};
\draw (268.96,373) node [anchor=north west][inner sep=0.75pt]  [font=\small,color={rgb, 255:red, 0; green, 0; blue, 0 }  ,opacity=1 ,rotate=-359.39]  {$4$};
\draw (298.59,373.54) node [anchor=north west][inner sep=0.75pt]  [font=\small,color={rgb, 255:red, 0; green, 0; blue, 0 }  ,opacity=1 ,rotate=-359.39]  {$5$};
\draw (352.12,271.76) node [anchor=north west][inner sep=0.75pt]  [font=\small,color={rgb, 255:red, 0; green, 0; blue, 0 }  ,opacity=1 ,rotate=-358.27]  {$1$};
\draw (340.21,473.52) node [anchor=north west][inner sep=0.75pt]  [font=\small,color={rgb, 255:red, 0; green, 0; blue, 0 }  ,opacity=1 ,rotate=-359.39]  {$4$};
\draw (345.63,371.6) node [anchor=north west][inner sep=0.75pt]  [font=\small,color={rgb, 255:red, 0; green, 0; blue, 0 }  ,opacity=1 ,rotate=-359.39]  {$6$};
\draw (407.92,374.29) node [anchor=north west][inner sep=0.75pt]  [font=\small,color={rgb, 255:red, 0; green, 0; blue, 0 }  ,opacity=1 ,rotate=-359.39]  {$4$};
\draw (297.58,425.25) node [anchor=north west][inner sep=0.75pt]  [font=\small,color={rgb, 255:red, 0; green, 0; blue, 0 }  ,opacity=1 ,rotate=-359.39]  {$3$};
\draw (311.84,299.48) node [anchor=north west][inner sep=0.75pt]  [font=\scriptsize,color={rgb, 255:red, 0; green, 0; blue, 0 }  ,opacity=1 ]  {$+$};
\draw (318.61,351.82) node [anchor=north west][inner sep=0.75pt]  [font=\scriptsize,color={rgb, 255:red, 0; green, 0; blue, 0 }  ,opacity=1 ]  {$+$};
\draw (319.51,399.05) node [anchor=north west][inner sep=0.75pt]  [font=\scriptsize,color={rgb, 255:red, 0; green, 0; blue, 0 }  ,opacity=1 ]  {$-$};
\draw (365.26,299.43) node [anchor=north west][inner sep=0.75pt]  [font=\scriptsize,color={rgb, 255:red, 0; green, 0; blue, 0 }  ,opacity=1 ]  {$-$};
\draw (269.63,476.6) node [anchor=north west][inner sep=0.75pt]  [font=\small,color={rgb, 255:red, 0; green, 0; blue, 0 }  ,opacity=1 ,rotate=-359.39]  {$6$};
\draw (282.84,449.48) node [anchor=north west][inner sep=0.75pt]  [font=\scriptsize,color={rgb, 255:red, 0; green, 0; blue, 0 }  ,opacity=1 ]  {$+$};
\draw (336.26,449.43) node [anchor=north west][inner sep=0.75pt]  [font=\scriptsize,color={rgb, 255:red, 0; green, 0; blue, 0 }  ,opacity=1 ]  {$-$};
\draw (347.84,347.48) node [anchor=north west][inner sep=0.75pt]  [font=\scriptsize,color={rgb, 255:red, 0; green, 0; blue, 0 }  ,opacity=1 ]  {$+$};
\draw (395.26,346.43) node [anchor=north west][inner sep=0.75pt]  [font=\scriptsize,color={rgb, 255:red, 0; green, 0; blue, 0 }  ,opacity=1 ]  {$-$};
\draw (275.84,351.48) node [anchor=north west][inner sep=0.75pt]  [font=\scriptsize,color={rgb, 255:red, 0; green, 0; blue, 0 }  ,opacity=1 ]  {$+$};
\draw (388.08,322.75) node [anchor=north west][inner sep=0.75pt]  [font=\small,color={rgb, 255:red, 0; green, 0; blue, 0 }  ,opacity=1 ,rotate=-359.39]  {$3$};

\end{tikzpicture}
    \caption{$\mathcal{LUG^H}(\mathcal{G}_s)$ in $\mathcal{G}_s$  }
    \label{lug H in sd dig}
   \end{subfigure}
  \caption{A digraph with a signed dilation that is $\mathcal{SS}$ herdable }
        \label{Dig with sds}
    \end{figure}
\end{eg}
The digraph contains a single signed dilation set, $\Delta_3 = \{4,6\}$, originating from node $3$. Although $\Delta_3$ reappears in $\mathcal{G}_s$ due to an alternative path to node $3$, either $\Delta_3^P = \{6\}$ or $\Delta_3^N = \{4\}$ can be sign-matched. Notably, no nodes of $\Delta_3$ in $L_2$ are matched, as they are either unmatched or excluded from $\mathcal{LUG}$. A spanning $\mathcal{LUG}^H(\mathcal{G}_s)$ is highlighted, where node $4$ is sign-matched outside $\Delta_3$, and node $6$ is matched within $\Delta_3$ at $L_5$. By Theorem \ref{theorem1}, the digraph given in Fig. \ref{Dig with sds} is $\mathcal{SS}$ herdable.

Note that $\mathcal{LUG^H}(\mathcal{G}_s)$ is not unique; in general, multiple $\mathcal{LUG^H}(\mathcal{G}_s)$ may exist in $\mathcal{G}_s$. If a digraph is not herdable under a single input, additional inputs are introduced so that the resulting digraph is completely herdable. The following section analyzes the $\mathcal{SS}$ herdability of digraphs with multiple inputs.

\section{$\mathcal{SS}$ herdability of digraph with multiple leader nodes} \label{Sec:6}
While \cite{liu2011controllability} shows that adding an external input to an unmatched node enhances controllability, a similar principle applies to herdabiltiy. We derive conditions under which $\mathcal{SS}$ herdability in networks with multiple inputs. If the same external input is applied to multiple nodes, those nodes are referred to as \emph{leader nodes}, whereas nodes receiving independent external inputs are referred to as \emph{driver nodes}.

The following proposition provides the graph-theoretic condition for $\mathcal{SS}$ herdability of a digraph having multiple leader nodes.
\begin{proposition}\label{mul-L}
Let $\mathcal{G(A,B)}$ be a digraph with multiple leader nodes receiving the same input $u(t)$ described by (\ref{sys1:sys1eqn}) and $\mathcal{G}_s$ be the signed layered graph $\mathcal{G}_s$ associated with $\mathcal{G(A,B)}$. Then, $\mathcal{G(A,B)}$  is $\mathcal{SS}$ herdable if and only if  $\mathcal{G}_s$  has an $\mathcal{LUG}^{\mathcal{H}}(\mathcal{G}_s)$ that spans all the nodes of $\mathcal{G(A,B)}$.
\end{proposition}

\begin{proof}
\normalfont
The proof is analogous to that of Theorem \ref{theorem1} with minor modifications. Hence, we omit the details. 
Note that the difference in the setting of Proposition~\ref{mul-L} is that multiple leader nodes are present in $L_1$, and all the leaders can be spanned together in layer $L_1$. Thus, the sign of edges coming from the input signal must be considered while constructing the $\mathcal{LUG^{H}}(\mathcal{G}_s)$ for $\mathcal{SS}$ herdability of a digraph with multiple leaders. 
\hfill$\qed$
\end{proof}
\begin{eg}
\normalfont
Let us consider a $\mathcal{G}_9\mathcal{(A,B)}$ given in Fig.\ref{multi L}(\subref{mul dia}) with two leader nodes receiving external input u(t). Here, nodes $1$ and $2$ receive the input $u(t)$ and are considered as leaders.

   \begin{figure}[ht]
        \centering
       
      \begin{subfigure}{0.2\textwidth}
\tikzset{every picture/.style={scale=0.75pt}} 

\begin{tikzpicture}[x=0.75pt,y=0.75pt,yscale=-1,xscale=1]

\draw [color={rgb, 255:red, 0; green, 0; blue, 0 }  ,draw opacity=1 ]   (54.41,116.17) -- (54.74,134.72) ;
\draw [shift={(54.79,137.72)}, rotate = 268.98] [fill={rgb, 255:red, 0; green, 0; blue, 0 }  ,fill opacity=1 ][line width=0.08]  [draw opacity=0] (7.14,-3.43) -- (0,0) -- (7.14,3.43) -- (4.74,0) -- cycle    ;
\draw  [draw opacity=0][fill={rgb, 255:red, 0; green, 192; blue, 248 }  ,fill opacity=1 ] (54.79,140.41) .. controls (57.73,140.35) and (60.16,142.71) .. (60.23,145.68) .. controls (60.29,148.65) and (57.96,151.11) .. (55.03,151.17) .. controls (52.09,151.23) and (49.66,148.87) .. (49.59,145.9) .. controls (49.53,142.93) and (51.86,140.47) .. (54.79,140.41) -- cycle ;
\draw [color={rgb, 255:red, 0; green, 0; blue, 0 }  ,draw opacity=1 ]   (50,154) -- (37.02,175.45) ;
\draw [shift={(35.47,178.02)}, rotate = 301.18] [fill={rgb, 255:red, 0; green, 0; blue, 0 }  ,fill opacity=1 ][line width=0.08]  [draw opacity=0] (7.14,-3.43) -- (0,0) -- (7.14,3.43) -- (4.74,0) -- cycle    ;
\draw  [draw opacity=0][fill={rgb, 255:red, 0; green, 192; blue, 248 }  ,fill opacity=1 ] (32.13,180.02) .. controls (35.07,179.96) and (37.5,182.32) .. (37.57,185.29) .. controls (37.63,188.26) and (35.3,190.72) .. (32.37,190.78) .. controls (29.43,190.84) and (27,188.48) .. (26.93,185.51) .. controls (26.87,182.54) and (29.2,180.08) .. (32.13,180.02) -- cycle ;
\draw  [draw opacity=0][fill={rgb, 255:red, 126; green, 211; blue, 33 }  ,fill opacity=1 ] (54.79,100.55) .. controls (57.73,100.49) and (60.16,102.85) .. (60.23,105.82) .. controls (60.29,108.79) and (57.96,111.25) .. (55.03,111.3) .. controls (52.09,111.36) and (49.66,109.01) .. (49.59,106.03) .. controls (49.53,103.06) and (51.86,100.61) .. (54.79,100.55) -- cycle ;
\draw  [draw opacity=0][fill={rgb, 255:red, 126; green, 211; blue, 33 }  ,fill opacity=1 ] (147.12,97.1) .. controls (150.05,97.04) and (152.48,99.4) .. (152.55,102.37) .. controls (152.61,105.34) and (150.29,107.8) .. (147.35,107.86) .. controls (144.41,107.92) and (141.98,105.56) .. (141.92,102.59) .. controls (141.85,99.62) and (144.18,97.16) .. (147.12,97.1) -- cycle ;
\draw    (155.89,164.98) .. controls (172.57,143.1) and (165.37,118.61) .. (157,111.68) ;
\draw [shift={(153.98,167.35)}, rotate = 310.41] [fill={rgb, 255:red, 0; green, 0; blue, 0 }  ][line width=0.08]  [draw opacity=0] (10.72,-5.15) -- (0,0) -- (10.72,5.15) -- (7.12,0) -- cycle    ;
\draw  [draw opacity=0][fill={rgb, 255:red, 0; green, 192; blue, 248 }  ,fill opacity=1 ] (142.12,176.22) .. controls (142.12,173.18) and (144.64,170.72) .. (147.74,170.72) .. controls (150.84,170.72) and (153.36,173.18) .. (153.36,176.22) .. controls (153.36,179.25) and (150.84,181.71) .. (147.74,181.71) .. controls (144.64,181.71) and (142.12,179.25) .. (142.12,176.22) -- cycle ;
\draw [color={rgb, 255:red, 208; green, 2; blue, 27 }  ,draw opacity=1 ]   (102.82,53.25) .. controls (102.72,55.6) and (101.49,56.73) .. (99.14,56.63) .. controls (96.79,56.53) and (95.56,57.66) .. (95.45,60.01) .. controls (95.35,62.36) and (94.12,63.49) .. (91.77,63.39) .. controls (89.42,63.29) and (88.19,64.42) .. (88.09,66.77) .. controls (87.99,69.12) and (86.76,70.25) .. (84.41,70.15) .. controls (82.05,70.05) and (80.82,71.18) .. (80.72,73.54) .. controls (80.62,75.89) and (79.39,77.02) .. (77.04,76.92) .. controls (74.69,76.82) and (73.46,77.95) .. (73.36,80.3) .. controls (73.25,82.65) and (72.02,83.78) .. (69.67,83.68) .. controls (67.32,83.58) and (66.09,84.71) .. (65.99,87.06) -- (65.82,87.23) -- (59.92,92.64) ;
\draw [shift={(57.71,94.67)}, rotate = 317.44] [fill={rgb, 255:red, 208; green, 2; blue, 27 }  ,fill opacity=1 ][line width=0.08]  [draw opacity=0] (8.04,-3.86) -- (0,0) -- (8.04,3.86) -- (5.34,0) -- cycle    ;
\draw [color={rgb, 255:red, 208; green, 2; blue, 27 }  ,draw opacity=1 ]   (102.82,53.25) .. controls (105.18,53.2) and (106.38,54.36) .. (106.42,56.72) .. controls (106.47,59.07) and (107.67,60.23) .. (110.02,60.19) .. controls (112.37,60.15) and (113.57,61.31) .. (113.62,63.66) .. controls (113.67,66.01) and (114.87,67.17) .. (117.22,67.13) .. controls (119.57,67.09) and (120.77,68.25) .. (120.82,70.6) .. controls (120.87,72.95) and (122.07,74.11) .. (124.42,74.07) .. controls (126.77,74.03) and (127.97,75.19) .. (128.02,77.54) .. controls (128.07,79.89) and (129.27,81.05) .. (131.62,81.01) .. controls (133.97,80.97) and (135.17,82.13) .. (135.22,84.48) -- (135.79,85.03) -- (141.55,90.58) ;
\draw [shift={(143.71,92.67)}, rotate = 223.95] [fill={rgb, 255:red, 208; green, 2; blue, 27 }  ,fill opacity=1 ][line width=0.08]  [draw opacity=0] (8.04,-3.86) -- (0,0) -- (8.04,3.86) -- (5.34,0) -- cycle    ;
\draw  [color={rgb, 255:red, 0; green, 0; blue, 0 }  ,draw opacity=1 ][fill={rgb, 255:red, 208; green, 2; blue, 27 }  ,fill opacity=1 ][line width=0.75]  (102.71,48.54) .. controls (105.36,48.48) and (107.55,50.55) .. (107.6,53.15) .. controls (107.66,55.75) and (105.57,57.9) .. (102.92,57.96) .. controls (100.28,58.01) and (98.09,55.94) .. (98.03,53.34) .. controls (97.98,50.74) and (100.07,48.59) .. (102.71,48.54) -- cycle ;
\draw    (138.9,168.25) .. controls (126.41,153.67) and (124.93,126.81) .. (137.09,111.13) ;
\draw [shift={(138.9,108.99)}, rotate = 132.69] [fill={rgb, 255:red, 0; green, 0; blue, 0 }  ][line width=0.08]  [draw opacity=0] (10.72,-5.15) -- (0,0) -- (10.72,5.15) -- (7.12,0) -- cycle    ;
\draw    (130.87,175.48) .. controls (106.59,174.92) and (70.25,159.92) .. (61.33,153.56) ;
\draw [shift={(134.22,175.46)}, rotate = 177.87] [fill={rgb, 255:red, 0; green, 0; blue, 0 }  ][line width=0.08]  [draw opacity=0] (10.72,-5.15) -- (0,0) -- (10.72,5.15) -- (7.12,0) -- cycle    ;

\draw (159.69,93.14) node [anchor=north west][inner sep=0.75pt]  [font=\small,color={rgb, 255:red, 0; green, 0; blue, 0 }  ,opacity=1 ,rotate=-358.27]  {$2$};
\draw (66.83,138) node [anchor=north west][inner sep=0.75pt]  [font=\small,color={rgb, 255:red, 0; green, 0; blue, 0 }  ,opacity=1 ,rotate=-358.27]  {$3$};
\draw (157.77,174.73) node [anchor=north west][inner sep=0.75pt]  [font=\small,color={rgb, 255:red, 0; green, 0; blue, 0 }  ,opacity=1 ,rotate=-358.27]  {$4$};
\draw (32.57,97.52) node [anchor=north west][inner sep=0.75pt]  [font=\small,color={rgb, 255:red, 0; green, 0; blue, 0 }  ,opacity=1 ,rotate=-358.27]  {$1$};
\draw (24.14,154.6) node [anchor=north west][inner sep=0.75pt]  [font=\small,color={rgb, 255:red, 0; green, 0; blue, 0 }  ,opacity=1 ,rotate=-358.88]  {$+$};
\draw (170.01,138.89) node [anchor=north west][inner sep=0.75pt]  [font=\small,color={rgb, 255:red, 0; green, 0; blue, 0 }  ,opacity=1 ]  {$-$};
\draw (26.68,124.89) node [anchor=north west][inner sep=0.75pt]  [font=\small,color={rgb, 255:red, 0; green, 0; blue, 0 }  ,opacity=1 ]  {$-$};
\draw (12.83,181.98) node [anchor=north west][inner sep=0.75pt]  [font=\small,color={rgb, 255:red, 0; green, 0; blue, 0 }  ,opacity=1 ,rotate=-358.27]  {$5$};
\draw (90.95,21.87) node [anchor=north west][inner sep=0.75pt]  [font=\small]  {$u( t)$};
\draw (62.62,58.73) node [anchor=north west][inner sep=0.75pt]  [font=\small,color={rgb, 255:red, 0; green, 0; blue, 0 }  ,opacity=1 ,rotate=-358.88]  {$+$};
\draw (130.01,60.55) node [anchor=north west][inner sep=0.75pt]  [font=\small,color={rgb, 255:red, 0; green, 0; blue, 0 }  ,opacity=1 ]  {$-$};
\draw (110.01,130.22) node [anchor=north west][inner sep=0.75pt]  [font=\small,color={rgb, 255:red, 0; green, 0; blue, 0 }  ,opacity=1 ]  {$-$};
\draw (81.35,170.89) node [anchor=north west][inner sep=0.75pt]  [font=\small,color={rgb, 255:red, 0; green, 0; blue, 0 }  ,opacity=1 ]  {$-$};

\end{tikzpicture}
\caption{A digraph $\mathcal{G}_9 (\mathcal{A,B})$}
\label{mul dia}
   \end{subfigure}
  \vspace{0.5mm}
   \begin{subfigure}{0.2\textwidth}
       
\tikzset{every picture/.style={scale=0.75pt}} 

\begin{tikzpicture}[x=0.75pt,y=0.75pt,yscale=-1,xscale=1]

\draw  [color={rgb, 255:red, 0; green, 0; blue, 0 }  ,draw opacity=0 ][fill={rgb, 255:red, 74; green, 144; blue, 226 }  ,fill opacity=0.1 ] (287.13,252.69) -- (464.21,252.69) -- (464.21,268.67) -- (287.13,268.67) -- cycle ;
\draw  [color={rgb, 255:red, 0; green, 0; blue, 0 }  ,draw opacity=0 ][fill={rgb, 255:red, 74; green, 144; blue, 226 }  ,fill opacity=0.1 ] (287.13,297.42) -- (464.21,297.42) -- (464.21,313.4) -- (287.13,313.4) -- cycle ;
\draw  [color={rgb, 255:red, 0; green, 0; blue, 0 }  ,draw opacity=0 ][fill={rgb, 255:red, 74; green, 144; blue, 226 }  ,fill opacity=0.1 ] (288.49,167.25) -- (465.57,167.25) -- (465.57,183.23) -- (288.49,183.23) -- cycle ;
\draw  [color={rgb, 255:red, 0; green, 0; blue, 0 }  ,draw opacity=0 ][fill={rgb, 255:red, 74; green, 144; blue, 226 }  ,fill opacity=0.1 ] (287.76,209.45) -- (464.84,209.45) -- (464.84,225.42) -- (287.76,225.42) -- cycle ;
\draw  [color={rgb, 255:red, 0; green, 0; blue, 0 }  ,draw opacity=0 ][fill={rgb, 255:red, 74; green, 144; blue, 226 }  ,fill opacity=0.1 ] (287.81,336.29) -- (464.89,336.29) -- (464.89,352.26) -- (287.81,352.26) -- cycle ;
\draw [color={rgb, 255:red, 208; green, 2; blue, 27 }  ,draw opacity=1 ]   (334.95,224.59) -- (317.7,249.18) ;
\draw [shift={(315.98,251.63)}, rotate = 305.05] [fill={rgb, 255:red, 208; green, 2; blue, 27 }  ,fill opacity=1 ][line width=0.08]  [draw opacity=0] (10.72,-5.15) -- (0,0) -- (10.72,5.15) -- (7.12,0) -- cycle    ;
\draw  [draw opacity=0][fill={rgb, 255:red, 126; green, 211; blue, 33 }  ,fill opacity=1 ][line width=0.75]  (333.57,174.71) .. controls (333.57,171.51) and (336.22,168.92) .. (339.49,168.92) .. controls (342.75,168.92) and (345.4,171.51) .. (345.4,174.71) .. controls (345.4,177.91) and (342.75,180.5) .. (339.49,180.5) .. controls (336.22,180.5) and (333.57,177.91) .. (333.57,174.71) -- cycle ;
\draw [color={rgb, 255:red, 176; green, 169; blue, 169 }  ,draw opacity=1 ]   (344.91,224.59) -- (359.98,249.84) ;
\draw [shift={(361.51,252.41)}, rotate = 239.18] [fill={rgb, 255:red, 176; green, 169; blue, 169 }  ,fill opacity=1 ][line width=0.08]  [draw opacity=0] (10.72,-5.15) -- (0,0) -- (10.72,5.15) -- (7.12,0) -- cycle    ;
\draw [color={rgb, 255:red, 208; green, 2; blue, 27 }  ,draw opacity=1 ][line width=0.75]    (339.22,183.24) -- (339.22,204.94) ;
\draw [shift={(339.22,207.94)}, rotate = 270] [fill={rgb, 255:red, 208; green, 2; blue, 27 }  ,fill opacity=1 ][line width=0.08]  [draw opacity=0] (10.72,-5.15) -- (0,0) -- (10.72,5.15) -- (7.12,0) -- cycle    ;
\draw  [color={rgb, 255:red, 74; green, 144; blue, 226 }  ,draw opacity=1 ][fill={rgb, 255:red, 0; green, 192; blue, 248 }  ,fill opacity=1 ] (334.73,216.5) .. controls (334.73,214.03) and (336.74,212.02) .. (339.21,212.02) .. controls (341.69,212.02) and (343.69,214.03) .. (343.69,216.5) .. controls (343.69,218.98) and (341.69,220.98) .. (339.21,220.98) .. controls (336.74,220.98) and (334.73,218.98) .. (334.73,216.5) -- cycle ;
\draw  [color={rgb, 255:red, 74; green, 144; blue, 226 }  ,draw opacity=1 ][fill={rgb, 255:red, 0; green, 192; blue, 248 }  ,fill opacity=1 ] (360.73,260.5) .. controls (360.73,258.03) and (362.74,256.02) .. (365.21,256.02) .. controls (367.69,256.02) and (369.69,258.03) .. (369.69,260.5) .. controls (369.69,262.98) and (367.69,264.98) .. (365.21,264.98) .. controls (362.74,264.98) and (360.73,262.98) .. (360.73,260.5) -- cycle ;
\draw  [color={rgb, 255:red, 74; green, 144; blue, 226 }  ,draw opacity=1 ][fill={rgb, 255:red, 0; green, 192; blue, 248 }  ,fill opacity=1 ] (306.73,258.5) .. controls (306.73,256.03) and (308.74,254.02) .. (311.21,254.02) .. controls (313.69,254.02) and (315.69,256.03) .. (315.69,258.5) .. controls (315.69,260.98) and (313.69,262.98) .. (311.21,262.98) .. controls (308.74,262.98) and (306.73,260.98) .. (306.73,258.5) -- cycle ;
\draw [color={rgb, 255:red, 176; green, 169; blue, 169 }  ,draw opacity=1 ][line width=0.75]    (365.05,269.8) -- (365.05,291.5) ;
\draw [shift={(365.05,294.5)}, rotate = 270] [fill={rgb, 255:red, 176; green, 169; blue, 169 }  ,fill opacity=1 ][line width=0.08]  [draw opacity=0] (10.72,-5.15) -- (0,0) -- (10.72,5.15) -- (7.12,0) -- cycle    ;
\draw  [draw opacity=0][fill={rgb, 255:red, 126; green, 211; blue, 33 }  ,fill opacity=1 ][line width=0.75]  (359.17,303.11) .. controls (359.17,299.91) and (361.82,297.32) .. (365.09,297.32) .. controls (368.35,297.32) and (371,299.91) .. (371,303.11) .. controls (371,306.31) and (368.35,308.9) .. (365.09,308.9) .. controls (361.82,308.9) and (359.17,306.31) .. (359.17,303.11) -- cycle ;
\draw [color={rgb, 255:red, 176; green, 169; blue, 169 }  ,draw opacity=1 ]   (365,312.38) -- (365,334.08) ;
\draw [shift={(365,337.08)}, rotate = 270] [fill={rgb, 255:red, 176; green, 169; blue, 169 }  ,fill opacity=1 ][line width=0.08]  [draw opacity=0] (10.72,-5.15) -- (0,0) -- (10.72,5.15) -- (7.12,0) -- cycle    ;
\draw  [color={rgb, 255:red, 74; green, 144; blue, 226 }  ,draw opacity=1 ][fill={rgb, 255:red, 0; green, 192; blue, 248 }  ,fill opacity=1 ] (360.07,345.97) .. controls (360.07,343.49) and (362.07,341.49) .. (364.55,341.49) .. controls (367.02,341.49) and (369.03,343.49) .. (369.03,345.97) .. controls (369.03,348.44) and (367.02,350.45) .. (364.55,350.45) .. controls (362.07,350.45) and (360.07,348.44) .. (360.07,345.97) -- cycle ;
\draw  [draw opacity=0][fill={rgb, 255:red, 126; green, 211; blue, 33 }  ,fill opacity=1 ][line width=0.75]  (413.57,173.57) .. controls (413.57,170.37) and (416.22,167.78) .. (419.49,167.78) .. controls (422.75,167.78) and (425.4,170.37) .. (425.4,173.57) .. controls (425.4,176.77) and (422.75,179.36) .. (419.49,179.36) .. controls (416.22,179.36) and (413.57,176.77) .. (413.57,173.57) -- cycle ;
\draw  [color={rgb, 255:red, 74; green, 144; blue, 226 }  ,draw opacity=1 ][fill={rgb, 255:red, 0; green, 192; blue, 248 }  ,fill opacity=1 ] (417.4,259.36) .. controls (417.4,256.89) and (419.4,254.88) .. (421.88,254.88) .. controls (424.35,254.88) and (426.36,256.89) .. (426.36,259.36) .. controls (426.36,261.84) and (424.35,263.84) .. (421.88,263.84) .. controls (419.4,263.84) and (417.4,261.84) .. (417.4,259.36) -- cycle ;
\draw [color={rgb, 255:red, 208; green, 2; blue, 27 }  ,draw opacity=1 ][line width=0.75]    (422.39,270.66) -- (422.39,292.36) ;
\draw [shift={(422.39,295.36)}, rotate = 270] [fill={rgb, 255:red, 208; green, 2; blue, 27 }  ,fill opacity=1 ][line width=0.08]  [draw opacity=0] (10.72,-5.15) -- (0,0) -- (10.72,5.15) -- (7.12,0) -- cycle    ;
\draw  [draw opacity=0][fill={rgb, 255:red, 126; green, 211; blue, 33 }  ,fill opacity=1 ][line width=0.75]  (416.5,305.3) .. controls (416.5,302.11) and (419.15,299.52) .. (422.42,299.52) .. controls (425.69,299.52) and (428.34,302.11) .. (428.34,305.3) .. controls (428.34,308.5) and (425.69,311.09) .. (422.42,311.09) .. controls (419.15,311.09) and (416.5,308.5) .. (416.5,305.3) -- cycle ;
\draw [color={rgb, 255:red, 208; green, 2; blue, 27 }  ,draw opacity=1 ]   (339.39,133.13) .. controls (341.05,134.8) and (341.04,136.47) .. (339.37,138.13) .. controls (337.7,139.79) and (337.69,141.46) .. (339.34,143.13) .. controls (341,144.8) and (340.99,146.47) .. (339.32,148.13) -- (339.3,152.6) -- (339.27,160.6) ;
\draw [shift={(339.25,163.6)}, rotate = 270.26] [fill={rgb, 255:red, 208; green, 2; blue, 27 }  ,fill opacity=1 ][line width=0.08]  [draw opacity=0] (10.72,-5.15) -- (0,0) -- (10.72,5.15) -- (7.12,0) -- cycle    ;
\draw [color={rgb, 255:red, 176; green, 169; blue, 169 }  ,draw opacity=1 ]   (419.39,133.13) .. controls (421.05,134.8) and (421.04,136.47) .. (419.37,138.13) .. controls (417.7,139.79) and (417.69,141.46) .. (419.34,143.13) -- (419.33,146.67) -- (419.33,146.67) .. controls (420.99,148.34) and (420.98,150.01) .. (419.31,151.67) -- (419.3,152.6) -- (419.27,160.6) ;
\draw [shift={(419.25,163.6)}, rotate = 270.26] [fill={rgb, 255:red, 176; green, 169; blue, 169 }  ,fill opacity=1 ][line width=0.08]  [draw opacity=0] (10.72,-5.15) -- (0,0) -- (10.72,5.15) -- (7.12,0) -- cycle    ;
\draw  [color={rgb, 255:red, 74; green, 144; blue, 226 }  ,draw opacity=1 ][fill={rgb, 255:red, 0; green, 192; blue, 248 }  ,fill opacity=1 ] (417.4,344.03) .. controls (417.4,341.55) and (419.4,339.55) .. (421.88,339.55) .. controls (424.35,339.55) and (426.36,341.55) .. (426.36,344.03) .. controls (426.36,346.5) and (424.35,348.51) .. (421.88,348.51) .. controls (419.4,348.51) and (417.4,346.5) .. (417.4,344.03) -- cycle ;
\draw  [color={rgb, 255:red, 74; green, 144; blue, 226 }  ,draw opacity=1 ][fill={rgb, 255:red, 0; green, 192; blue, 248 }  ,fill opacity=1 ] (416.73,215.5) .. controls (416.73,213.03) and (418.74,211.02) .. (421.21,211.02) .. controls (423.69,211.02) and (425.69,213.03) .. (425.69,215.5) .. controls (425.69,217.98) and (423.69,219.98) .. (421.21,219.98) .. controls (418.74,219.98) and (416.73,217.98) .. (416.73,215.5) -- cycle ;
\draw [color={rgb, 255:red, 176; green, 169; blue, 169 }  ,draw opacity=1 ][line width=0.75]    (420.05,181.8) -- (420.05,203.5) ;
\draw [shift={(420.05,206.5)}, rotate = 270] [fill={rgb, 255:red, 176; green, 169; blue, 169 }  ,fill opacity=1 ][line width=0.08]  [draw opacity=0] (10.72,-5.15) -- (0,0) -- (10.72,5.15) -- (7.12,0) -- cycle    ;
\draw [color={rgb, 255:red, 176; green, 169; blue, 169 }  ,draw opacity=1 ][line width=0.75]    (421.72,225.8) -- (421.72,247.5) ;
\draw [shift={(421.72,250.5)}, rotate = 270] [fill={rgb, 255:red, 176; green, 169; blue, 169 }  ,fill opacity=1 ][line width=0.08]  [draw opacity=0] (10.72,-5.15) -- (0,0) -- (10.72,5.15) -- (7.12,0) -- cycle    ;
\draw [color={rgb, 255:red, 208; green, 2; blue, 27 }  ,draw opacity=1 ][line width=0.75]    (422.42,312.09) -- (422.42,333.79) ;
\draw [shift={(422.42,336.79)}, rotate = 270] [fill={rgb, 255:red, 208; green, 2; blue, 27 }  ,fill opacity=1 ][line width=0.08]  [draw opacity=0] (10.72,-5.15) -- (0,0) -- (10.72,5.15) -- (7.12,0) -- cycle    ;
\draw  [color={rgb, 255:red, 0; green, 0; blue, 0 }  ,draw opacity=0 ][fill={rgb, 255:red, 74; green, 144; blue, 226 }  ,fill opacity=0.1 ] (289.25,118.25) -- (466.33,118.25) -- (466.33,134.23) -- (289.25,134.23) -- cycle ;
\draw  [draw opacity=0][fill={rgb, 255:red, 208; green, 2; blue, 27 }  ,fill opacity=1 ][line width=0.75]  (334.57,125.71) .. controls (334.57,122.51) and (337.22,119.92) .. (340.49,119.92) .. controls (343.75,119.92) and (346.4,122.51) .. (346.4,125.71) .. controls (346.4,128.91) and (343.75,131.5) .. (340.49,131.5) .. controls (337.22,131.5) and (334.57,128.91) .. (334.57,125.71) -- cycle ;
\draw  [draw opacity=0][fill={rgb, 255:red, 208; green, 2; blue, 27 }  ,fill opacity=1 ][line width=0.75]  (414.57,124.57) .. controls (414.57,121.37) and (417.22,118.78) .. (420.49,118.78) .. controls (423.75,118.78) and (426.4,121.37) .. (426.4,124.57) .. controls (426.4,127.77) and (423.75,130.36) .. (420.49,130.36) .. controls (417.22,130.36) and (414.57,127.77) .. (414.57,124.57) -- cycle ;

\draw (353.19,165.38) node [anchor=north west][inner sep=0.75pt]  [color={rgb, 255:red, 0; green, 0; blue, 0 }  ,opacity=1 ,rotate=-359.39]  {$1$};
\draw (355.54,208.57) node [anchor=north west][inner sep=0.75pt]  [color={rgb, 255:red, 0; green, 0; blue, 0 }  ,opacity=1 ,rotate=-359.39]  {$3$};
\draw (290.76,251.24) node [anchor=north west][inner sep=0.75pt]  [color={rgb, 255:red, 0; green, 0; blue, 0 }  ,opacity=1 ,rotate=-359.39]  {$5$};
\draw (303.4,225.68) node [anchor=north west][inner sep=0.75pt]  [font=\small,color={rgb, 255:red, 0; green, 0; blue, 0 }  ,opacity=1 ]  {$+$};
\draw (375.12,296.53) node [anchor=north west][inner sep=0.75pt]  [color={rgb, 255:red, 0; green, 0; blue, 0 }  ,opacity=1 ,rotate=-359.39]  {$2$};
\draw (374.05,251.08) node [anchor=north west][inner sep=0.75pt]  [color={rgb, 255:red, 0; green, 0; blue, 0 }  ,opacity=1 ,rotate=-359.39]  {$4$};
\draw (435.12,252.05) node [anchor=north west][inner sep=0.75pt]  [color={rgb, 255:red, 0; green, 0; blue, 0 }  ,opacity=1 ,rotate=-359.39]  {$2$};
\draw (432.45,164.79) node [anchor=north west][inner sep=0.75pt]  [color={rgb, 255:red, 0; green, 0; blue, 0 }  ,opacity=1 ,rotate=-359.39]  {$2$};
\draw (371.74,274) node [anchor=north west][inner sep=0.75pt]  [font=\footnotesize,color={rgb, 255:red, 0; green, 0; blue, 0 }  ,opacity=1 ]  {$-$};
\draw (433.6,275.73) node [anchor=north west][inner sep=0.75pt]  [font=\footnotesize,color={rgb, 255:red, 0; green, 0; blue, 0 }  ,opacity=1 ]  {$-$};
\draw (314.01,139.27) node [anchor=north west][inner sep=0.75pt]  [font=\small,color={rgb, 255:red, 0; green, 0; blue, 0 }  ,opacity=1 ]  {$+$};
\draw (428.35,136.22) node [anchor=north west][inner sep=0.75pt]  [font=\small,color={rgb, 255:red, 0; green, 0; blue, 0 }  ,opacity=1 ]  {$-$};
\draw (436.27,318.4) node [anchor=north west][inner sep=0.75pt]  [font=\footnotesize,color={rgb, 255:red, 0; green, 0; blue, 0 }  ,opacity=1 ]  {$-$};
\draw (313.74,188) node [anchor=north west][inner sep=0.75pt]  [font=\footnotesize,color={rgb, 255:red, 0; green, 0; blue, 0 }  ,opacity=1 ]  {$-$};
\draw (435.38,207.55) node [anchor=north west][inner sep=0.75pt]  [color={rgb, 255:red, 0; green, 0; blue, 0 }  ,opacity=1 ,rotate=-359.39]  {$4$};
\draw (436.05,296.88) node [anchor=north west][inner sep=0.75pt]  [color={rgb, 255:red, 0; green, 0; blue, 0 }  ,opacity=1 ,rotate=-359.39]  {$4$};
\draw (436.45,336.05) node [anchor=north west][inner sep=0.75pt]  [color={rgb, 255:red, 0; green, 0; blue, 0 }  ,opacity=1 ,rotate=-359.39]  {$2$};
\draw (433.01,185.4) node [anchor=north west][inner sep=0.75pt]  [font=\small,color={rgb, 255:red, 0; green, 0; blue, 0 }  ,opacity=1 ]  {$-$};
\draw (435.68,226.73) node [anchor=north west][inner sep=0.75pt]  [font=\small,color={rgb, 255:red, 0; green, 0; blue, 0 }  ,opacity=1 ]  {$-$};
\draw (373.38,337.08) node [anchor=north west][inner sep=0.75pt]  [color={rgb, 255:red, 0; green, 0; blue, 0 }  ,opacity=1 ,rotate=-359.39]  {$4$};
\draw (373.07,316.67) node [anchor=north west][inner sep=0.75pt]  [font=\footnotesize,color={rgb, 255:red, 0; green, 0; blue, 0 }  ,opacity=1 ]  {$-$};
\draw (362.4,230.67) node [anchor=north west][inner sep=0.75pt]  [font=\footnotesize,color={rgb, 255:red, 0; green, 0; blue, 0 }  ,opacity=1 ]  {$-$};
\draw (366.95,117.87) node [anchor=north west][inner sep=0.75pt]  [font=\small]  {$u( t)$};
\draw (477.7,166.4) node [anchor=north west][inner sep=0.75pt]  [color={rgb, 255:red, 91; green, 130; blue, 49 }  ,opacity=1 ]  {$L_{1}$};
\draw (476.37,208.09) node [anchor=north west][inner sep=0.75pt]  [color={rgb, 255:red, 91; green, 130; blue, 49 }  ,opacity=1 ]  {$L_{2}$};
\draw (478.03,334.65) node [anchor=north west][inner sep=0.75pt]  [color={rgb, 255:red, 91; green, 130; blue, 49 }  ,opacity=1 ]  {$L_{5}$};
\draw (478.03,294.6) node [anchor=north west][inner sep=0.75pt]  [color={rgb, 255:red, 91; green, 130; blue, 49 }  ,opacity=1 ]  {$L_{4}$};
\draw (476.37,252.18) node [anchor=north west][inner sep=0.75pt]  [color={rgb, 255:red, 91; green, 130; blue, 49 }  ,opacity=1 ]  {$L_{3}$};
\draw (476.37,118) node [anchor=north west][inner sep=0.75pt]  [color={rgb, 255:red, 91; green, 130; blue, 49 }  ,opacity=1 ]  {$L_{0}$};

\end{tikzpicture}
\caption{$\mathcal{LUG^{H}}(\mathcal{G}_s)$ in $(\mathcal{G}_s)$}
\label{gs mul}
   \end{subfigure}
  \caption{A digraph $\mathcal{G}_9 (\mathcal{A,B})$ with multiple leaders; $\mathcal{G}_s$ contains $\mathcal{LUG}^H(\mathcal{G}_s)$ from leaders $\{1,2\}$, making the digraph $\mathcal{SS}$ herdable.}

        \label{multi L}
    \end{figure}
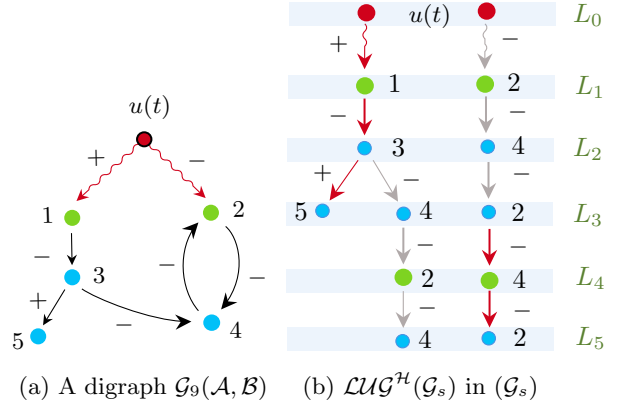 
   The $(\mathcal{G}_s)$ of $\mathcal{G}(A,B)$ is shown in Fig.~\ref{multi L}(\subref{gs mul}). The highlighted $\mathcal{LUG^{H}}(\mathcal{G}_s)$ spans all nodes; hence, the digraph is $\mathcal{SS}$ herdable. If the edge $(4,2)$ is removed, node $2$ is no longer spanned by $\mathcal{LUG^{H}}(\mathcal{G}_s)$, despite being a leader, and the digraph is not  $\mathcal{SS}$ herdable due to layer dilation at $L_0$.
\end{eg}

\subsection{$\mathcal{SS}$ herdability of with multiple driver nodes} \label{Sec:6.1}

In the following results, we provide the graph-theoretic condition for $\mathcal{SS}$ herdability of a digraph having more than one driver node. Hence, the input matrix is denoted by $\mathcal{B} \in \mathbb{R}^{n \times m}$, where $m$ represents the number of independent input signals.

 \begin{remark} \label{Remark multi}
    Consider a pair $(\mathcal{A}, \mathcal{B})$, where $\mathcal{A} \in \mathbb{R}^{n \times n}$ and $\mathcal{B} \in \mathbb{R}^{n \times m}$. 
    Assume that there are $m$ independent inputs such that a set of $m$ nodes are leader nodes. 
    Then, the controllability matrix of the pair $(\mathcal{A}, \mathcal{B})$ is given by
    \begin{equation} \label{Control_mat mul}
        \mathcal{C}(\mathcal{A},\mathcal{B}) = \begin{bmatrix}
       ~~\Psi_0 &\mid~~\Psi_1&\mid~~ \Psi_2\mid & \cdots &\mid \Psi_{n-1}
        \end{bmatrix},
    \end{equation}
    \normalfont
 where $\Psi_k \in \mathbb{R}^{n \times m}$ is defined as
\[
    \Psi_k = \big[\, \Psi_k^{1} ~~ \Psi_k^{2} ~~ \cdots ~~ \Psi_k^{m} \,\big],
\]
with each column given by
\[
\Psi_k^{\,i} = \mathcal{A}^{k}\mathcal{B}_{i}, 
\qquad k \in \{ 0,\dots,n-1\},\quad i \in \{ 1,\dots,m\},
\]

Where, $\mathcal{B}_{i}$ denotes the $i^{th}$ column of $\mathcal{B}$, whose entries are defined as
\[
    (\mathcal{B}_{i})_{j} =
    \begin{cases}
        b_{i}, & \text{if } j = i, \\
        0, & \text{otherwise}.
    \end{cases}
\]

    Here, $b_i$ is the strength of the input signal $u_i(t)$ applied to the leader node $i$. 
    If the $j^{\text{th}}$ entry of $\Psi^i_k$, denoted $[\Psi^i_k]_j$, is nonzero, then there exists at least one walk of length $k$ from the leader node $i$ to node $j$. 
\end{remark}
\begin{lemma}\label{multi lemma}
   Let $\mathcal{G(A,B)}$ be a digraph with multiple driver nodes described by (\ref{sys1:sys1eqn}) and $\mathcal{G}_s$ be the signed layered graph associated with $\mathcal{G(A,B)}$. Then each entry in $\Psi_k^i$ of $\Psi_k$ in \eqref{Control_mat mul} corresponds to the nodes in layer $L_{k+1} $ of $\mathcal{G}_s$. 
\end{lemma}
\begin{proof}
\normalfont
     Consider $\mathcal{C}(\mathcal{A,B})$ given in \eqref{Control_mat mul} associated to the system \eqref{sys1:sys1eqn} having $m$ inputs. For all \(i \in \{1,2,\dots,n\}\) and \(j \in \{1,2,\dots,nm\}\), the \((i,j)\)-th entry of the controllability matrix $\mathcal{C}_{(i,j)}$ is given by \([\mathcal{A}^k \mathcal{B}_\ell]_i\) = \( [\Psi^\ell_k]_i\), where \(k = \left\lfloor~\frac{j-1}{m}~\right\rfloor\), \( \ell=\big((j-1) \bmod m\big)+1\). Consider $[\Psi^\ell_k]_j \ne0$, then there exists a walk of length $k$ from the leader node $\ell \in \{1,2,\dots,m\}$ to the follower node $j$.
     
     Since each layer $L_{k}$, $k\in\{2,3..n\}$ in $\mathcal{G}_s$ consists of nodes that are reachable from a leader node $\ell$ in $L_1$ with a walk of length $k-1$, A node $j$ in $L_{k}$ corresponds to the entry $[\mathcal{A}^{p-1}\mathcal{B}_{\ell}]_j=[\Psi_{p-1}^{\ell}]_j$ in controllability matrix $\mathcal{C}(\mathcal{A,B})$. Hence, proved. \hfill $\qed$ 
\end{proof}
\begin{theorem} \label{Thm multi Driver}
Let $\mathcal{G(A,B)}$ be a digraph with multiple driver nodes described by (\ref{sys1:sys1eqn}) and $\mathcal{G}_s$ be the signed layered graph associated with $\mathcal{G(A,B)}$. Let $\mathcal{LUG}_i^{\mathcal{H}}(\mathcal{G}_s)$ be an $\mathcal{LUG}^{\mathcal{H}}(\mathcal{G}_s)$ with unique leader node $i$. Then $\mathcal{G(A,B)}$ is $\mathcal{SS}$ herdable if and only if all the nodes of $\mathcal{G}(\mathcal{A,B})$ are spanned by $~\bigcup_{i \in \mathbb{Z}^+}~ \mathcal{LUG}_i^{\mathcal{H}}(\mathcal{G}_s)$.

\end{theorem}
Proof is presented in Appendix \ref{proof Thm multi Driver}.
\normalfont

\begin{eg}
\normalfont
Let us consider a $\mathcal{G}_{10}(\mathcal{A,B})$ given in Fig.\ref{multi driver full}(\subref{multi G (a,b)}) with driver nodes $1$ and $2$ receiving external inputs $u_1(t)$ and $u_2(t)$ respectively. Let $\mathcal{G}_s$ be the signed layer graph of $\mathcal{G(A,B)}$ such that the leader nodes $\{1,2\}~ \in L_1$ of $\mathcal{G}_s$ as shown in Fig. \ref{multi driver full} (\subref{LUG multi d1}).

   \begin{figure}[ht]
        \centering
       
      \begin{subfigure}{0.2\textwidth}
\tikzset{every picture/.style={scale=1.1pt}} 

\begin{tikzpicture}[x=0.75pt,y=0.75pt,yscale=-1,xscale=1]

\draw [color={rgb, 255:red, 0; green, 0; blue, 0 }  ,draw opacity=1 ]   (107.65,85.3) -- (107.52,101.85) ;
\draw [shift={(107.51,103.85)}, rotate = 270.45] [color={rgb, 255:red, 0; green, 0; blue, 0 }  ,draw opacity=1 ][line width=0.75]    (6.56,-2.94) .. controls (4.17,-1.38) and (1.99,-0.4) .. (0,0) .. controls (1.99,0.4) and (4.17,1.38) .. (6.56,2.94)   ;
\draw [color={rgb, 255:red, 0; green, 0; blue, 0 }  ,draw opacity=1 ]   (107.65,110.62) -- (96.7,127.76) ;
\draw [shift={(95.62,129.44)}, rotate = 302.58] [color={rgb, 255:red, 0; green, 0; blue, 0 }  ,draw opacity=1 ][line width=0.75]    (6.56,-2.94) .. controls (4.17,-1.38) and (1.99,-0.4) .. (0,0) .. controls (1.99,0.4) and (4.17,1.38) .. (6.56,2.94)   ;
\draw [color={rgb, 255:red, 208; green, 2; blue, 27 }  ,draw opacity=1 ]   (106.97,53.22) .. controls (108.67,54.85) and (108.7,56.52) .. (107.07,58.22) .. controls (105.44,59.92) and (105.47,61.59) .. (107.17,63.22) -- (107.24,67.22) -- (107.4,75.22) ;
\draw [shift={(107.46,78.22)}, rotate = 268.89] [fill={rgb, 255:red, 208; green, 2; blue, 27 }  ,fill opacity=1 ][line width=0.08]  [draw opacity=0] (7.14,-3.43) -- (0,0) -- (7.14,3.43) -- (4.74,0) -- cycle    ;
\draw  [color={rgb, 255:red, 0; green, 0; blue, 0 }  ,draw opacity=1 ][fill={rgb, 255:red, 208; green, 2; blue, 27 }  ,fill opacity=1 ][line width=0.75]  (106.84,47.29) .. controls (108.47,47.26) and (109.81,48.56) .. (109.85,50.19) .. controls (109.89,51.83) and (108.6,53.19) .. (106.97,53.22) .. controls (105.35,53.25) and (104,51.95) .. (103.97,50.31) .. controls (103.93,48.68) and (105.22,47.32) .. (106.84,47.29) -- cycle ;
\draw [color={rgb, 255:red, 208; green, 2; blue, 27 }  ,draw opacity=1 ]   (164.77,52.59) .. controls (166.46,54.22) and (166.49,55.89) .. (164.86,57.59) .. controls (163.23,59.29) and (163.26,60.96) .. (164.96,62.59) -- (165.04,66.6) -- (165.19,74.59) ;
\draw [shift={(165.25,77.59)}, rotate = 268.89] [fill={rgb, 255:red, 208; green, 2; blue, 27 }  ,fill opacity=1 ][line width=0.08]  [draw opacity=0] (7.14,-3.43) -- (0,0) -- (7.14,3.43) -- (4.74,0) -- cycle    ;
\draw  [color={rgb, 255:red, 0; green, 0; blue, 0 }  ,draw opacity=1 ][fill={rgb, 255:red, 208; green, 2; blue, 27 }  ,fill opacity=1 ][line width=0.75]  (164.64,46.66) .. controls (166.26,46.63) and (167.61,47.93) .. (167.64,49.56) .. controls (167.68,51.2) and (166.39,52.56) .. (164.77,52.59) .. controls (163.14,52.62) and (161.79,51.32) .. (161.76,49.69) .. controls (161.72,48.05) and (163.01,46.69) .. (164.64,46.66) -- cycle ;
\draw [color={rgb, 255:red, 0; green, 0; blue, 0 }  ,draw opacity=1 ]   (164.41,84.31) -- (150.13,104.84) ;
\draw [shift={(148.99,106.48)}, rotate = 304.83] [color={rgb, 255:red, 0; green, 0; blue, 0 }  ,draw opacity=1 ][line width=0.75]    (6.56,-2.94) .. controls (4.17,-1.38) and (1.99,-0.4) .. (0,0) .. controls (1.99,0.4) and (4.17,1.38) .. (6.56,2.94)   ;
\draw [color={rgb, 255:red, 0; green, 0; blue, 0 }  ,draw opacity=1 ]   (145.01,111.85) -- (141.34,136.04) ;
\draw [shift={(141.04,138.02)}, rotate = 278.63] [color={rgb, 255:red, 0; green, 0; blue, 0 }  ,draw opacity=1 ][line width=0.75]    (6.56,-2.94) .. controls (4.17,-1.38) and (1.99,-0.4) .. (0,0) .. controls (1.99,0.4) and (4.17,1.38) .. (6.56,2.94)   ;
\draw [color={rgb, 255:red, 0; green, 0; blue, 0 }  ,draw opacity=1 ]   (164.41,84.31) -- (176.63,105.51) ;
\draw [shift={(177.63,107.25)}, rotate = 240.05] [color={rgb, 255:red, 0; green, 0; blue, 0 }  ,draw opacity=1 ][line width=0.75]    (6.56,-2.94) .. controls (4.17,-1.38) and (1.99,-0.4) .. (0,0) .. controls (1.99,0.4) and (4.17,1.38) .. (6.56,2.94)   ;
\draw [color={rgb, 255:red, 0; green, 0; blue, 0 }  ,draw opacity=1 ]   (180.02,113.51) -- (180.81,135.63) ;
\draw [shift={(180.88,137.63)}, rotate = 267.97] [color={rgb, 255:red, 0; green, 0; blue, 0 }  ,draw opacity=1 ][line width=0.75]    (6.56,-2.94) .. controls (4.17,-1.38) and (1.99,-0.4) .. (0,0) .. controls (1.99,0.4) and (4.17,1.38) .. (6.56,2.94)   ;
\draw    (143.38,113.03) .. controls (134.91,118.93) and (127.04,120.87) .. (120.18,118.74) .. controls (115.99,117.44) and (112.17,114.61) .. (108.84,110.24) ;
\draw [shift={(145.01,111.85)}, rotate = 143.12] [color={rgb, 255:red, 0; green, 0; blue, 0 }  ][line width=0.75]    (6.56,-2.94) .. controls (4.17,-1.38) and (1.99,-0.4) .. (0,0) .. controls (1.99,0.4) and (4.17,1.38) .. (6.56,2.94)   ;
\draw  [color={rgb, 255:red, 0; green, 0; blue, 0 }  ,draw opacity=1 ][fill={rgb, 255:red, 126; green, 211; blue, 33 }  ,fill opacity=1 ] (107.51,78.75) .. controls (109.31,78.71) and (110.81,80.19) .. (110.85,82.07) .. controls (110.89,83.94) and (109.46,85.48) .. (107.65,85.52) .. controls (105.85,85.56) and (104.35,84.07) .. (104.31,82.2) .. controls (104.27,80.33) and (105.7,78.78) .. (107.51,78.75) -- cycle ;
\draw  [color={rgb, 255:red, 0; green, 0; blue, 0 }  ,draw opacity=1 ][fill={rgb, 255:red, 126; green, 211; blue, 33 }  ,fill opacity=1 ] (164.27,77.54) .. controls (166.07,77.5) and (167.57,78.98) .. (167.61,80.86) .. controls (167.65,82.73) and (166.22,84.27) .. (164.41,84.31) .. controls (162.61,84.35) and (161.11,82.86) .. (161.07,80.99) .. controls (161.03,79.12) and (162.46,77.57) .. (164.27,77.54) -- cycle ;
\draw  [color={rgb, 255:red, 0; green, 0; blue, 0 }  ,draw opacity=1 ][fill={rgb, 255:red, 0; green, 192; blue, 248 }  ,fill opacity=1 ] (146.47,106.08) .. controls (148.28,106.05) and (149.78,107.53) .. (149.82,109.4) .. controls (149.86,111.27) and (148.42,112.82) .. (146.62,112.86) .. controls (144.81,112.89) and (143.32,111.41) .. (143.28,109.54) .. controls (143.24,107.67) and (144.67,106.12) .. (146.47,106.08) -- cycle ;
\draw  [color={rgb, 255:red, 0; green, 0; blue, 0 }  ,draw opacity=1 ][fill={rgb, 255:red, 0; green, 192; blue, 248 }  ,fill opacity=1 ] (93.57,128.78) .. controls (95.38,128.75) and (96.88,130.23) .. (96.92,132.1) .. controls (96.95,133.97) and (95.52,135.52) .. (93.72,135.56) .. controls (91.91,135.59) and (90.42,134.11) .. (90.38,132.24) .. controls (90.34,130.37) and (91.77,128.82) .. (93.57,128.78) -- cycle ;
\draw  [color={rgb, 255:red, 0; green, 0; blue, 0 }  ,draw opacity=1 ][fill={rgb, 255:red, 0; green, 192; blue, 248 }  ,fill opacity=1 ] (175.57,111.01) .. controls (175.57,109.1) and (177.12,107.55) .. (179.03,107.55) .. controls (180.94,107.55) and (182.48,109.1) .. (182.48,111.01) .. controls (182.48,112.92) and (180.94,114.47) .. (179.03,114.47) .. controls (177.12,114.47) and (175.57,112.92) .. (175.57,111.01) -- cycle ;
\draw  [color={rgb, 255:red, 0; green, 0; blue, 0 }  ,draw opacity=1 ][fill={rgb, 255:red, 0; green, 192; blue, 248 }  ,fill opacity=1 ] (137.83,142.16) .. controls (137.83,140.25) and (139.38,138.7) .. (141.29,138.7) .. controls (143.2,138.7) and (144.74,140.25) .. (144.74,142.16) .. controls (144.74,144.07) and (143.2,145.62) .. (141.29,145.62) .. controls (139.38,145.62) and (137.83,144.07) .. (137.83,142.16) -- cycle ;
\draw  [color={rgb, 255:red, 0; green, 0; blue, 0 }  ,draw opacity=1 ][fill={rgb, 255:red, 0; green, 192; blue, 248 }  ,fill opacity=1 ] (180.88,137.63) .. controls (182.68,137.59) and (184.18,139.08) .. (184.22,140.95) .. controls (184.26,142.82) and (182.83,144.36) .. (181.02,144.4) .. controls (179.22,144.44) and (177.72,142.95) .. (177.68,141.08) .. controls (177.64,139.21) and (179.07,137.67) .. (180.88,137.63) -- cycle ;
\draw    (146.47,106.08) .. controls (135.39,95.34) and (124.73,92.59) .. (112.25,102.76) ;
\draw [shift={(110.88,103.92)}, rotate = 318.63] [color={rgb, 255:red, 0; green, 0; blue, 0 }  ][line width=0.75]    (6.56,-2.94) .. controls (4.17,-1.38) and (1.99,-0.4) .. (0,0) .. controls (1.99,0.4) and (4.17,1.38) .. (6.56,2.94)   ;
\draw  [color={rgb, 255:red, 0; green, 0; blue, 0 }  ,draw opacity=1 ][fill={rgb, 255:red, 0; green, 192; blue, 248 }  ,fill opacity=1 ] (107.51,103.85) .. controls (109.31,103.81) and (110.81,105.29) .. (110.85,107.16) .. controls (110.89,109.03) and (109.46,110.58) .. (107.65,110.62) .. controls (105.85,110.66) and (104.35,109.17) .. (104.31,107.3) .. controls (104.27,105.43) and (105.7,103.88) .. (107.51,103.85) -- cycle ;

\draw (172,77.34) node [anchor=north west][inner sep=0.75pt]  [font=\tiny,color={rgb, 255:red, 0; green, 0; blue, 0 }  ,opacity=1 ,rotate=-358.27]  {$2$};
\draw (94.19,106.08) node [anchor=north west][inner sep=0.75pt]  [font=\tiny,color={rgb, 255:red, 0; green, 0; blue, 0 }  ,opacity=1 ,rotate=-358.27]  {$3$};
\draw (89.87,138.8) node [anchor=north west][inner sep=0.75pt]  [font=\tiny,color={rgb, 255:red, 0; green, 0; blue, 0 }  ,opacity=1 ,rotate=-358.27]  {$4$};
\draw (95.85,78.64) node [anchor=north west][inner sep=0.75pt]  [font=\tiny,color={rgb, 255:red, 0; green, 0; blue, 0 }  ,opacity=1 ,rotate=-358.27]  {$1$};
\draw (90.93,115.5) node [anchor=north west][inner sep=0.75pt]  [font=\tiny,color={rgb, 255:red, 0; green, 0; blue, 0 }  ,opacity=1 ,rotate=-358.88]  {$+$};
\draw (92.83,88.76) node [anchor=north west][inner sep=0.75pt]  [font=\tiny,color={rgb, 255:red, 0; green, 0; blue, 0 }  ,opacity=1 ]  {$-$};
\draw (154.67,108.42) node [anchor=north west][inner sep=0.75pt]  [font=\tiny,color={rgb, 255:red, 0; green, 0; blue, 0 }  ,opacity=1 ,rotate=-358.27]  {$5$};
\draw (98.76,32.24) node [anchor=north west][inner sep=0.75pt]  [font=\tiny]  {$u_{1}( t)$};
\draw (93.36,57.86) node [anchor=north west][inner sep=0.75pt]  [font=\tiny,color={rgb, 255:red, 0; green, 0; blue, 0 }  ,opacity=1 ,rotate=-358.88]  {$+$};
\draw (150.83,55.19) node [anchor=north west][inner sep=0.75pt]  [font=\tiny,color={rgb, 255:red, 0; green, 0; blue, 0 }  ,opacity=1 ]  {$+$};
\draw (118.14,120.23) node [anchor=north west][inner sep=0.75pt]  [font=\tiny,color={rgb, 255:red, 0; green, 0; blue, 0 }  ,opacity=1 ]  {$-$};
\draw (184.57,106.99) node [anchor=north west][inner sep=0.75pt]  [font=\tiny,color={rgb, 255:red, 0; green, 0; blue, 0 }  ,opacity=1 ,rotate=-358.27]  {$6$};
\draw (130.23,141.26) node [anchor=north west][inner sep=0.75pt]  [font=\tiny,color={rgb, 255:red, 0; green, 0; blue, 0 }  ,opacity=1 ,rotate=-358.27]  {$7$};
\draw (145.17,87.28) node [anchor=north west][inner sep=0.75pt]  [font=\tiny,color={rgb, 255:red, 0; green, 0; blue, 0 }  ,opacity=1 ,rotate=-358.88]  {$+$};
\draw (155.75,32.58) node [anchor=north west][inner sep=0.75pt]  [font=\tiny]  {$u_{2}( t)$};
\draw (168.87,138.63) node [anchor=north west][inner sep=0.75pt]  [font=\tiny,color={rgb, 255:red, 0; green, 0; blue, 0 }  ,opacity=1 ,rotate=-358.27]  {$8$};
\draw (133.47,122.99) node [anchor=north west][inner sep=0.75pt]  [font=\tiny,color={rgb, 255:red, 0; green, 0; blue, 0 }  ,opacity=1 ]  {$-$};
\draw (169.4,120.87) node [anchor=north west][inner sep=0.75pt]  [font=\tiny,color={rgb, 255:red, 0; green, 0; blue, 0 }  ,opacity=1 ,rotate=-358.88]  {$+$};
\draw (175.01,86.25) node [anchor=north west][inner sep=0.75pt]  [font=\tiny,color={rgb, 255:red, 0; green, 0; blue, 0 }  ,opacity=1 ]  {$-$};
\draw (123.97,86.28) node [anchor=north west][inner sep=0.75pt]  [font=\tiny,color={rgb, 255:red, 0; green, 0; blue, 0 }  ,opacity=1 ,rotate=-358.88]  {$+$};

\end{tikzpicture}
\caption{A digraph $\mathcal{G}_{10} (\mathcal{A,B})$}
\label{multi G (a,b)}
   \end{subfigure}
  \vspace{0.5mm}
   \begin{subfigure}{0.2\textwidth}
       
\tikzset{every picture/.style={scale=1.1pt}} 

\begin{tikzpicture}[x=0.75pt,y=0.75pt,yscale=-1,xscale=1]

\draw [color={rgb, 255:red, 208; green, 2; blue, 27 }  ,draw opacity=1 ]   (336.27,130.72) -- (336.24,152.08) ;
\draw [shift={(336.24,154.08)}, rotate = 270.07] [color={rgb, 255:red, 208; green, 2; blue, 27 }  ,draw opacity=1 ][line width=0.75]    (6.56,-2.94) .. controls (4.17,-1.38) and (1.99,-0.4) .. (0,0) .. controls (1.99,0.4) and (4.17,1.38) .. (6.56,2.94)   ;
\draw  [color={rgb, 255:red, 0; green, 0; blue, 0 }  ,draw opacity=0 ][fill={rgb, 255:red, 74; green, 144; blue, 226 }  ,fill opacity=0.1 ] (236.95,94.4) -- (299.67,94.4) -- (299.67,105.45) -- (236.95,105.45) -- cycle ;
\draw  [color={rgb, 255:red, 0; green, 0; blue, 0 }  ,draw opacity=0 ][fill={rgb, 255:red, 74; green, 144; blue, 226 }  ,fill opacity=0.1 ] (304.11,94.07) -- (369.64,94.07) -- (369.64,105.12) -- (304.11,105.12) -- cycle ;
\draw  [color={rgb, 255:red, 0; green, 0; blue, 0 }  ,draw opacity=0 ][fill={rgb, 255:red, 74; green, 144; blue, 226 }  ,fill opacity=0.1 ] (303.88,120.52) -- (369.4,120.52) -- (369.4,131.57) -- (303.88,131.57) -- cycle ;
\draw  [color={rgb, 255:red, 0; green, 0; blue, 0 }  ,draw opacity=0 ][fill={rgb, 255:red, 74; green, 144; blue, 226 }  ,fill opacity=0.1 ] (303.51,154.9) -- (369.03,154.9) -- (369.03,165.95) -- (303.51,165.95) -- cycle ;
\draw  [color={rgb, 255:red, 0; green, 0; blue, 0 }  ,draw opacity=0 ][fill={rgb, 255:red, 74; green, 144; blue, 226 }  ,fill opacity=0.1 ] (304.38,64.9) -- (369.91,64.9) -- (369.91,75.95) -- (304.38,75.95) -- cycle ;
\draw  [color={rgb, 255:red, 0; green, 0; blue, 0 }  ,draw opacity=0 ][fill={rgb, 255:red, 74; green, 144; blue, 226 }  ,fill opacity=0.1 ] (304.66,31.02) -- (370.19,31.02) -- (370.19,42.07) -- (304.66,42.07) -- cycle ;
\draw  [color={rgb, 255:red, 0; green, 0; blue, 0 }  ,draw opacity=0 ][fill={rgb, 255:red, 74; green, 144; blue, 226 }  ,fill opacity=0.1 ] (236.9,120.85) -- (299.62,120.85) -- (299.62,131.89) -- (236.9,131.89) -- cycle ;
\draw  [color={rgb, 255:red, 0; green, 0; blue, 0 }  ,draw opacity=0 ][fill={rgb, 255:red, 74; green, 144; blue, 226 }  ,fill opacity=0.1 ] (236.54,155.23) -- (299.26,155.23) -- (299.26,166.27) -- (236.54,166.27) -- cycle ;
\draw  [color={rgb, 255:red, 0; green, 0; blue, 0 }  ,draw opacity=0 ][fill={rgb, 255:red, 74; green, 144; blue, 226 }  ,fill opacity=0.1 ] (237.2,65.23) -- (299.93,65.23) -- (299.93,76.27) -- (237.2,76.27) -- cycle ;
\draw [color={rgb, 255:red, 128; green, 128; blue, 128 }  ,draw opacity=0.44 ][line width=0.75]    (317.38,103.07) -- (317.28,122.24) ;
\draw [shift={(317.27,124.24)}, rotate = 270.28] [color={rgb, 255:red, 128; green, 128; blue, 128 }  ,draw opacity=0.44 ][line width=0.75]    (6.56,-2.94) .. controls (4.17,-1.38) and (1.99,-0.4) .. (0,0) .. controls (1.99,0.4) and (4.17,1.38) .. (6.56,2.94)   ;
\draw  [color={rgb, 255:red, 0; green, 0; blue, 0 }  ,draw opacity=0 ][fill={rgb, 255:red, 74; green, 144; blue, 226 }  ,fill opacity=0.1 ] (237.47,31.35) -- (300.2,31.35) -- (300.2,42.39) -- (237.47,42.39) -- cycle ;
\draw [color={rgb, 255:red, 128; green, 128; blue, 128 }  ,draw opacity=0.44 ]   (331.33,73.79) -- (318.41,95.16) ;
\draw [shift={(317.38,96.87)}, rotate = 301.16] [color={rgb, 255:red, 128; green, 128; blue, 128 }  ,draw opacity=0.44 ][line width=0.75]    (6.56,-2.94) .. controls (4.17,-1.38) and (1.99,-0.4) .. (0,0) .. controls (1.99,0.4) and (4.17,1.38) .. (6.56,2.94)   ;
\draw [color={rgb, 255:red, 208; green, 2; blue, 27 }  ,draw opacity=1 ]   (331.33,73.79) -- (341.9,95.08) ;
\draw [shift={(342.78,96.87)}, rotate = 243.61] [color={rgb, 255:red, 208; green, 2; blue, 27 }  ,draw opacity=1 ][line width=0.75]    (6.56,-2.94) .. controls (4.17,-1.38) and (1.99,-0.4) .. (0,0) .. controls (1.99,0.4) and (4.17,1.38) .. (6.56,2.94)   ;
\draw [color={rgb, 255:red, 128; green, 128; blue, 128 }  ,draw opacity=0.44 ]   (317.38,103.07) -- (334.95,123.03) ;
\draw [shift={(336.27,124.53)}, rotate = 228.64] [color={rgb, 255:red, 128; green, 128; blue, 128 }  ,draw opacity=0.44 ][line width=0.75]    (6.56,-2.94) .. controls (4.17,-1.38) and (1.99,-0.4) .. (0,0) .. controls (1.99,0.4) and (4.17,1.38) .. (6.56,2.94)   ;
\draw [color={rgb, 255:red, 208; green, 2; blue, 27 }  ,draw opacity=1 ]   (342.78,103.07) -- (353.96,122.79) ;
\draw [shift={(354.95,124.53)}, rotate = 240.46] [color={rgb, 255:red, 208; green, 2; blue, 27 }  ,draw opacity=1 ][line width=0.75]    (6.56,-2.94) .. controls (4.17,-1.38) and (1.99,-0.4) .. (0,0) .. controls (1.99,0.4) and (4.17,1.38) .. (6.56,2.94)   ;
\draw [color={rgb, 255:red, 128; green, 128; blue, 128 }  ,draw opacity=0.46 ]   (271.12,102.38) -- (259.56,121.44) ;
\draw [shift={(258.53,123.15)}, rotate = 301.23] [color={rgb, 255:red, 128; green, 128; blue, 128 }  ,draw opacity=0.46 ][line width=0.75]    (6.56,-2.94) .. controls (4.17,-1.38) and (1.99,-0.4) .. (0,0) .. controls (1.99,0.4) and (4.17,1.38) .. (6.56,2.94)   ;
\draw [color={rgb, 255:red, 208; green, 2; blue, 27 }  ,draw opacity=1 ][fill={rgb, 255:red, 208; green, 2; blue, 27 }  ,fill opacity=1 ]   (271.12,102.38) -- (282.27,121.42) ;
\draw [shift={(283.28,123.15)}, rotate = 239.65] [color={rgb, 255:red, 208; green, 2; blue, 27 }  ,draw opacity=1 ][line width=0.75]    (6.56,-2.94) .. controls (4.17,-1.38) and (1.99,-0.4) .. (0,0) .. controls (1.99,0.4) and (4.17,1.38) .. (6.56,2.94)   ;
\draw [color={rgb, 255:red, 208; green, 2; blue, 27 }  ,draw opacity=1 ]   (270.18,42.27) .. controls (271.87,43.9) and (271.9,45.57) .. (270.27,47.27) .. controls (268.64,48.97) and (268.67,50.64) .. (270.37,52.27) -- (270.47,56.94) -- (270.63,64.94) ;
\draw [shift={(270.68,67.94)}, rotate = 268.86] [fill={rgb, 255:red, 208; green, 2; blue, 27 }  ,fill opacity=1 ][line width=0.08]  [draw opacity=0] (7.14,-3.43) -- (0,0) -- (7.14,3.43) -- (4.74,0) -- cycle    ;
\draw  [color={rgb, 255:red, 0; green, 0; blue, 0 }  ,draw opacity=1 ][fill={rgb, 255:red, 208; green, 2; blue, 27 }  ,fill opacity=1 ][line width=0.75]  (270.04,36.19) .. controls (271.75,36.15) and (273.16,37.49) .. (273.2,39.17) .. controls (273.24,40.85) and (271.88,42.24) .. (270.18,42.27) .. controls (268.47,42.31) and (267.05,40.97) .. (267.02,39.29) .. controls (266.98,37.61) and (268.33,36.22) .. (270.04,36.19) -- cycle ;
\draw [color={rgb, 255:red, 208; green, 2; blue, 27 }  ,draw opacity=1 ]   (330.66,41.23) .. controls (332.36,42.86) and (332.39,44.53) .. (330.76,46.23) .. controls (329.13,47.93) and (329.16,49.6) .. (330.86,51.23) -- (330.95,55.9) -- (331.11,63.9) ;
\draw [shift={(331.17,66.9)}, rotate = 268.86] [fill={rgb, 255:red, 208; green, 2; blue, 27 }  ,fill opacity=1 ][line width=0.08]  [draw opacity=0] (7.14,-3.43) -- (0,0) -- (7.14,3.43) -- (4.74,0) -- cycle    ;
\draw  [color={rgb, 255:red, 0; green, 0; blue, 0 }  ,draw opacity=1 ][fill={rgb, 255:red, 208; green, 2; blue, 27 }  ,fill opacity=1 ][line width=0.75]  (330.52,35.15) .. controls (332.23,35.11) and (333.64,36.45) .. (333.68,38.13) .. controls (333.72,39.81) and (332.37,41.2) .. (330.66,41.23) .. controls (328.95,41.27) and (327.54,39.93) .. (327.5,38.25) .. controls (327.46,36.57) and (328.81,35.18) .. (330.52,35.15) -- cycle ;
\draw [color={rgb, 255:red, 208; green, 2; blue, 27 }  ,draw opacity=1 ][fill={rgb, 255:red, 208; green, 2; blue, 27 }  ,fill opacity=1 ]   (270.89,76.43) -- (271.02,93.08) ;
\draw [shift={(271.04,95.08)}, rotate = 269.53] [color={rgb, 255:red, 208; green, 2; blue, 27 }  ,draw opacity=1 ][line width=0.75]    (6.56,-2.94) .. controls (4.17,-1.38) and (1.99,-0.4) .. (0,0) .. controls (1.99,0.4) and (4.17,1.38) .. (6.56,2.94)   ;
\draw [color={rgb, 255:red, 208; green, 2; blue, 27 }  ,draw opacity=1 ][fill={rgb, 255:red, 208; green, 2; blue, 27 }  ,fill opacity=1 ]   (283.28,130.34) -- (283.06,153.08) ;
\draw [shift={(283.04,155.08)}, rotate = 270.56] [color={rgb, 255:red, 208; green, 2; blue, 27 }  ,draw opacity=1 ][line width=0.75]    (6.56,-2.94) .. controls (4.17,-1.38) and (1.99,-0.4) .. (0,0) .. controls (1.99,0.4) and (4.17,1.38) .. (6.56,2.94)   ;
\draw  [color={rgb, 255:red, 0; green, 0; blue, 0 }  ,draw opacity=1 ][fill={rgb, 255:red, 0; green, 192; blue, 248 }  ,fill opacity=1 ] (280.36,159.57) .. controls (280.36,157.86) and (281.67,156.47) .. (283.28,156.47) .. controls (284.89,156.47) and (286.2,157.86) .. (286.2,159.57) .. controls (286.2,161.28) and (284.89,162.66) .. (283.28,162.66) .. controls (281.67,162.66) and (280.36,161.28) .. (280.36,159.57) -- cycle ;
\draw  [color={rgb, 255:red, 0; green, 0; blue, 0 }  ,draw opacity=1 ][fill={rgb, 255:red, 0; green, 192; blue, 248 }  ,fill opacity=1 ] (314.46,99.97) .. controls (314.46,98.26) and (315.76,96.87) .. (317.38,96.87) .. controls (318.99,96.87) and (320.3,98.26) .. (320.3,99.97) .. controls (320.3,101.68) and (318.99,103.07) .. (317.38,103.07) .. controls (315.76,103.07) and (314.46,101.68) .. (314.46,99.97) -- cycle ;
\draw  [color={rgb, 255:red, 0; green, 0; blue, 0 }  ,draw opacity=1 ][fill={rgb, 255:red, 0; green, 192; blue, 248 }  ,fill opacity=1 ] (333.35,127.63) .. controls (333.35,125.92) and (334.66,124.53) .. (336.27,124.53) .. controls (337.88,124.53) and (339.19,125.92) .. (339.19,127.63) .. controls (339.19,129.34) and (337.88,130.72) .. (336.27,130.72) .. controls (334.66,130.72) and (333.35,129.34) .. (333.35,127.63) -- cycle ;
\draw  [color={rgb, 255:red, 0; green, 0; blue, 0 }  ,draw opacity=1 ][fill={rgb, 255:red, 0; green, 192; blue, 248 }  ,fill opacity=1 ] (352.03,127.63) .. controls (352.03,125.92) and (353.33,124.53) .. (354.95,124.53) .. controls (356.56,124.53) and (357.86,125.92) .. (357.86,127.63) .. controls (357.86,129.34) and (356.56,130.72) .. (354.95,130.72) .. controls (353.33,130.72) and (352.03,129.34) .. (352.03,127.63) -- cycle ;
\draw  [color={rgb, 255:red, 0; green, 0; blue, 0 }  ,draw opacity=1 ][fill={rgb, 255:red, 0; green, 192; blue, 248 }  ,fill opacity=1 ] (255.61,126.24) .. controls (255.61,124.53) and (256.91,123.15) .. (258.53,123.15) .. controls (260.14,123.15) and (261.44,124.53) .. (261.44,126.24) .. controls (261.44,127.96) and (260.14,129.34) .. (258.53,129.34) .. controls (256.91,129.34) and (255.61,127.96) .. (255.61,126.24) -- cycle ;
\draw  [color={rgb, 255:red, 0; green, 0; blue, 0 }  ,draw opacity=1 ][fill={rgb, 255:red, 0; green, 192; blue, 248 }  ,fill opacity=1 ] (280.36,126.24) .. controls (280.36,124.53) and (281.67,123.15) .. (283.28,123.15) .. controls (284.89,123.15) and (286.2,124.53) .. (286.2,126.24) .. controls (286.2,127.96) and (284.89,129.34) .. (283.28,129.34) .. controls (281.67,129.34) and (280.36,127.96) .. (280.36,126.24) -- cycle ;
\draw  [color={rgb, 255:red, 0; green, 0; blue, 0 }  ,draw opacity=1 ][fill={rgb, 255:red, 0; green, 192; blue, 248 }  ,fill opacity=1 ] (333.35,160.26) .. controls (333.35,158.55) and (334.66,157.16) .. (336.27,157.16) .. controls (337.88,157.16) and (339.19,158.55) .. (339.19,160.26) .. controls (339.19,161.97) and (337.88,163.36) .. (336.27,163.36) .. controls (334.66,163.36) and (333.35,161.97) .. (333.35,160.26) -- cycle ;
\draw  [color={rgb, 255:red, 0; green, 0; blue, 0 }  ,draw opacity=1 ][fill={rgb, 255:red, 126; green, 211; blue, 33 }  ,fill opacity=1 ] (270.74,68.48) .. controls (272.63,68.44) and (274.2,69.96) .. (274.24,71.88) .. controls (274.29,73.8) and (272.78,75.39) .. (270.89,75.43) .. controls (268.99,75.47) and (267.42,73.94) .. (267.38,72.02) .. controls (267.34,70.1) and (268.84,68.52) .. (270.74,68.48) -- cycle ;
\draw  [color={rgb, 255:red, 0; green, 0; blue, 0 }  ,draw opacity=1 ][fill={rgb, 255:red, 126; green, 211; blue, 33 }  ,fill opacity=1 ] (331.18,66.84) .. controls (333.08,66.8) and (334.65,68.33) .. (334.69,70.25) .. controls (334.73,72.16) and (333.23,73.75) .. (331.33,73.79) .. controls (329.44,73.83) and (327.87,72.31) .. (327.82,70.39) .. controls (327.78,68.47) and (329.29,66.88) .. (331.18,66.84) -- cycle ;
\draw  [color={rgb, 255:red, 0; green, 0; blue, 0 }  ,draw opacity=1 ][fill={rgb, 255:red, 0; green, 192; blue, 248 }  ,fill opacity=1 ] (339.87,99.97) .. controls (339.87,98.26) and (341.17,96.87) .. (342.78,96.87) .. controls (344.4,96.87) and (345.7,98.26) .. (345.7,99.97) .. controls (345.7,101.68) and (344.4,103.07) .. (342.78,103.07) .. controls (341.17,103.07) and (339.87,101.68) .. (339.87,99.97) -- cycle ;
\draw  [color={rgb, 255:red, 0; green, 0; blue, 0 }  ,draw opacity=1 ][fill={rgb, 255:red, 0; green, 192; blue, 248 }  ,fill opacity=1 ] (268.2,99.28) .. controls (268.2,97.57) and (269.51,96.18) .. (271.12,96.18) .. controls (272.73,96.18) and (274.04,97.57) .. (274.04,99.28) .. controls (274.04,100.99) and (272.73,102.38) .. (271.12,102.38) .. controls (269.51,102.38) and (268.2,100.99) .. (268.2,99.28) -- cycle ;
\draw  [color={rgb, 255:red, 0; green, 0; blue, 0 }  ,draw opacity=1 ][fill={rgb, 255:red, 0; green, 192; blue, 248 }  ,fill opacity=1 ] (314.35,127.34) .. controls (314.35,125.63) and (315.66,124.24) .. (317.27,124.24) .. controls (318.89,124.24) and (320.19,125.63) .. (320.19,127.34) .. controls (320.19,129.05) and (318.89,130.44) .. (317.27,130.44) .. controls (315.66,130.44) and (314.35,129.05) .. (314.35,127.34) -- cycle ;

\draw (311.78,79.8) node [anchor=north west][inner sep=0.75pt]  [font=\tiny,color={rgb, 255:red, 0; green, 0; blue, 0 }  ,opacity=1 ,rotate=-358.88]  {$+$};
\draw (278.94,68.38) node [anchor=north west][inner sep=0.75pt]  [font=\tiny,color={rgb, 255:red, 0; green, 0; blue, 0 }  ,opacity=1 ,rotate=-359.39]  {$1$};
\draw (278.71,98.24) node [anchor=north west][inner sep=0.75pt]  [font=\tiny,color={rgb, 255:red, 0; green, 0; blue, 0 }  ,opacity=1 ,rotate=-359.39]  {$3$};
\draw (288.69,123.77) node [anchor=north west][inner sep=0.75pt]  [font=\tiny,color={rgb, 255:red, 0; green, 0; blue, 0 }  ,opacity=1 ,rotate=-359.39]  {$5$};
\draw (245.58,123.91) node [anchor=north west][inner sep=0.75pt]  [font=\tiny,color={rgb, 255:red, 0; green, 0; blue, 0 }  ,opacity=1 ,rotate=-359.39]  {$4$};
\draw (340.67,68.39) node [anchor=north west][inner sep=0.75pt]  [font=\tiny,color={rgb, 255:red, 0; green, 0; blue, 0 }  ,opacity=1 ,rotate=-359.39]  {$2$};
\draw (257.27,79.86) node [anchor=north west][inner sep=0.75pt]  [font=\tiny,color={rgb, 255:red, 0; green, 0; blue, 0 }  ,opacity=1 ]  {$-$};
\draw (244.36,32.89) node [anchor=north west][inner sep=0.75pt]  [font=\tiny]  {$u_{1}( t)$};
\draw (375,67.71) node [anchor=north west][inner sep=0.75pt]  [font=\tiny,color={rgb, 255:red, 91; green, 130; blue, 49 }  ,opacity=1 ]  {$L_{1}$};
\draw (374.12,96.94) node [anchor=north west][inner sep=0.75pt]  [font=\tiny,color={rgb, 255:red, 91; green, 130; blue, 49 }  ,opacity=1 ]  {$L_{2}$};
\draw (375.21,157.75) node [anchor=north west][inner sep=0.75pt]  [font=\tiny,color={rgb, 255:red, 91; green, 130; blue, 49 }  ,opacity=1 ]  {$L_{4}$};
\draw (374.12,122.43) node [anchor=north west][inner sep=0.75pt]  [font=\tiny,color={rgb, 255:red, 91; green, 130; blue, 49 }  ,opacity=1 ]  {$L_{3}$};
\draw (374.11,33.65) node [anchor=north west][inner sep=0.75pt]  [font=\tiny,color={rgb, 255:red, 91; green, 130; blue, 49 }  ,opacity=1 ]  {$L_{0}$};
\draw (339.67,32.69) node [anchor=north west][inner sep=0.75pt]  [font=\tiny]  {$u_{2}( t)$};
\draw (305.07,123.62) node [anchor=north west][inner sep=0.75pt]  [font=\tiny,color={rgb, 255:red, 0; green, 0; blue, 0 }  ,opacity=1 ,rotate=-358.27]  {$7$};
\draw (304.64,96.41) node [anchor=north west][inner sep=0.75pt]  [font=\tiny,color={rgb, 255:red, 0; green, 0; blue, 0 }  ,opacity=1 ,rotate=-359.39]  {$5$};
\draw (350.97,96.67) node [anchor=north west][inner sep=0.75pt]  [font=\tiny,color={rgb, 255:red, 0; green, 0; blue, 0 }  ,opacity=1 ,rotate=-358.27]  {$6$};
\draw (290.34,158.41) node [anchor=north west][inner sep=0.75pt]  [font=\tiny,color={rgb, 255:red, 0; green, 0; blue, 0 }  ,opacity=1 ,rotate=-358.27]  {$7$};
\draw (287.61,138.4) node [anchor=north west][inner sep=0.75pt]  [font=\tiny,color={rgb, 255:red, 0; green, 0; blue, 0 }  ,opacity=1 ]  {$-$};
\draw (305.38,108.9) node [anchor=north west][inner sep=0.75pt]  [font=\tiny,color={rgb, 255:red, 0; green, 0; blue, 0 }  ,opacity=1 ]  {$-$};
\draw (354.14,107.7) node [anchor=north west][inner sep=0.75pt]  [font=\tiny,color={rgb, 255:red, 0; green, 0; blue, 0 }  ,opacity=1 ]  {$+$};
\draw (359.4,123.43) node [anchor=north west][inner sep=0.75pt]  [font=\tiny,color={rgb, 255:red, 0; green, 0; blue, 0 }  ,opacity=1 ,rotate=-358.27]  {$8$};
\draw (251.25,107.3) node [anchor=north west][inner sep=0.75pt]  [font=\tiny,color={rgb, 255:red, 0; green, 0; blue, 0 }  ,opacity=1 ,rotate=-358.88]  {$+$};
\draw (340.17,138.22) node [anchor=north west][inner sep=0.75pt]  [font=\tiny,color={rgb, 255:red, 0; green, 0; blue, 0 }  ,opacity=1 ]  {$+$};
\draw (342.48,78.95) node [anchor=north west][inner sep=0.75pt]  [font=\tiny,color={rgb, 255:red, 0; green, 0; blue, 0 }  ,opacity=1 ]  {$-$};
\draw (283.27,107.98) node [anchor=north west][inner sep=0.75pt]  [font=\tiny,color={rgb, 255:red, 0; green, 0; blue, 0 }  ,opacity=1 ]  {$-$};
\draw (330.39,107.97) node [anchor=north west][inner sep=0.75pt]  [font=\tiny,color={rgb, 255:red, 0; green, 0; blue, 0 }  ,opacity=1 ,rotate=-358.88]  {$+$};
\draw (340.94,123.18) node [anchor=north west][inner sep=0.75pt]  [font=\tiny,color={rgb, 255:red, 0; green, 0; blue, 0 }  ,opacity=1 ,rotate=-359.39]  {$3$};
\draw (323.08,157.47) node [anchor=north west][inner sep=0.75pt]  [font=\tiny,color={rgb, 255:red, 0; green, 0; blue, 0 }  ,opacity=1 ,rotate=-359.39]  {$4$};
\draw (259.11,46.76) node [anchor=north west][inner sep=0.75pt]  [font=\tiny,color={rgb, 255:red, 0; green, 0; blue, 0 }  ,opacity=1 ,rotate=-358.88]  {$+$};
\draw (317.62,47.21) node [anchor=north west][inner sep=0.75pt]  [font=\tiny,color={rgb, 255:red, 0; green, 0; blue, 0 }  ,opacity=1 ]  {$+$};
\draw (255.81,15.35) node [anchor=north west][inner sep=0.75pt]  [font=\tiny]  {$\mathcal{LUG}_{1}^{\mathcal{H}}$};
\draw (320.84,16.82) node [anchor=north west][inner sep=0.75pt]  [font=\tiny]  {$\mathcal{LUG}_{2}^{\mathcal{H}}$};

\end{tikzpicture}
\caption{$\mathcal{LUG^{H}}(\mathcal{G}_s)$ in $(\mathcal{G}_s)$}
\label{LUG multi d1}
   \end{subfigure}
  \caption{$\mathcal{SS}$ herdability of digraph with multiple leaders}
        \label{multi driver full}
    \end{figure}
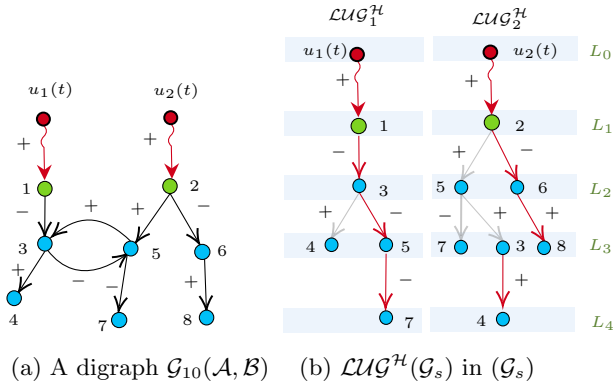
From the Fig. \ref{multi driver full} (\subref{LUG multi d1}), it is evident that all the nodes of $\mathcal{G(A,B)}$ are spanned by $~\bigcup_{i \in \{1,2\}}~ \mathcal{LUG}_i^{\mathcal{H}}(\mathcal{G}_s)$. Hence digraph $\mathcal{G}_{10}(\mathcal{A,B})$ is $\mathcal{SS}$ herdable. 
\end{eg}
  
\section{Conclusion}   \label{Sec:7}
In this paper, the necessary and sufficient conditions for structural sign $(\mathcal{SS})$ herdability of an arbitrary digraph are studied. For this purpose, we used a signed layered graph representation of the digraphs, which facilitates the study of the impact of signed and layer dilations on $(\mathcal{SS})$ herdability. We introduced the signed dilation set to analyze these effects on $(\mathcal{SS})$ herdability, and to address these challenges, we proposed sign-matching within the signed layered graph. In this study, $\mathcal{LUG^{H}}(\mathcal{G}_s)$, a subgraph of the signed layered graph, was formulated using sign-matching, forming the foundation for establishing the necessary and sufficient condition for $(\mathcal{SS})$ herdability.
 We show that the presence of an $\mathcal{LUG^{H}}(\mathcal{G}_s)$ that covers all the nodes within the digraph guarantees $\mathcal{SS}$ herdability, despite the presence of signed dilation and layer dilation in the digraph. We build upon this framework to study $(\mathcal{SS})$ herdability in digraphs with multiple leader nodes that share a common control input, using $\mathcal{LUG^{H}}(\mathcal{G}_s)$ as the basis. Subsequently, we investigate an extended scenario in which multiple driver nodes are each associated with distinct control inputs. To examine herdability, we utilize Gordan's Theorem of alternatives, akin to Kalman’s rank condition, as a fundamental tool for testing and validating our conclusions. In future studies, the aim will be to expand these results to time-varying and temporal networks.

\section*{Appendix}
\setcounter{figure}{13} 

\section{Proof of Theorem~\ref{theorem1}}\label{proof1}
\begin{proof}
\normalfont
\textit{Sufficiency}:

Let $\mathcal{G}(\mathcal{A,B})$ be a digraph where all the nodes are spanned by an $\mathcal{LUG^{H}}(\mathcal{G}_s)$. Then the following statements are true:
\begin{enumerate}
    \item[(i)] For all $i\in\{1,2..n\}$, the $i^{th}$ row of $\mathcal{C}(\mathcal{A},\mathcal{B})$ , i.e  $\mathcal{C}_{(i,:)}$ is nonzero.
    \item [(ii)] For all the sign-matched nodes in $L_{p+1}$, their corresponding entries in $[\Psi_{p}]$ are unisigned. 
    \item[(iii)] Other than the entries corresponding to sign-matched nodes, there can be other nonzero entries in $\Psi_p$. They are matched in the other columns of $ \mathcal{C}(\mathcal{A},\mathcal{B})$. 
    \item[(iv)] It follows from Lemma \ref{lemma C herd} that all (+/-) entries in $ \mathcal{C}(\mathcal{A},\mathcal{B})$ can be realized to the sign for which the corresponding node is sign-matched.
 \end{enumerate}

It follows from Proposition \ref{test for H}, that if for every realization of $(\mathcal{A,B})$ there exists a $\mathbf{y}\geq0$ such that $\mathcal{C}(\mathcal{A,B})^\top \mathbf{y}=0$, then it is not $\mathcal{SS}$ herdable. Let $\mathcal{C}_1,\mathcal{C}_2 \ldots, \mathcal{C}_n$ denote the rows of $ \mathcal{C}(\mathcal{A},\mathcal{B})$, which are equivalently the column vectors of $ \mathcal{C}(\mathcal{A},\mathcal{B})^\top$. The existence of $\mathbf{y}\geq0$ that satisfies the above expression implies that there exist some rows in $\mathcal{C}(\mathcal{A},\mathcal{B})$ that are nonnegatively linearly dependent, irrespective of any realization of $(\mathcal{A,B})$. 

Let us assume that for every realization of $(\mathcal{A},\mathcal{B})$ there exists a $\mathbf{y}\geq0$ such that $\mathcal{C}(\mathcal{A,B})^\top \mathbf{y}=0$. This implies there exists some $j,k~\in~\{1,2 \dots n\}$ such that
\begin{align}\label{c satisfy}
     \mathcal{C}_j = -\alpha ~ \mathcal{C}_k~;~ \alpha \in\mathbb{R}^+
     \end{align}
The expression \([\Psi_p]_j= - \alpha ~[\Psi_p]_k~~; \forall ~p ~\in \{1,2, \dots n-1\}\) implies that for any non-zero $[\Psi_p]_j$ and $[\Psi_p]_k$, nodes $j$ and $k$ are part of the same layer. Given that all nodes are sign-matched by an $\mathcal{LUG^{H}}(\mathcal{G}_s)$, this leads to the following possible scenarios.
\begin{enumerate}
    \item [($i$)] If node $j$ and $k$ are sign-matched at the same layer $L_p$, then $\text{sign}~[\Psi_p]_j= \text{sign}~[\Psi_p]_k$. This is in contradiction with the above relation, given that $\alpha~\in \mathbb{R}^+$; consequently, condition~\eqref{c satisfy} fails to hold.
     \item [($ii$)] If the node $j$ is sign-matched at $L_p$ and node $k$ is sign-matched at a $L_{m \ne p}$, then there exists a walk $\mathcal{W}_{(1,k)}$ of length $m-1$ and the entry $[\Psi_m]_k \ne 0$  in $\mathcal{C}_k$. This leads to  $\mathcal{C}_j \ne \beta ~ \mathcal{C}_k~;~ \beta \in\mathbb{R}$ for some realization of $(\mathcal{A},\mathcal{B})$ given by Lemma \ref{lemma C herd}.
\end{enumerate}
From the previous cases, it can be inferred that there is a realization of $(\mathcal{A},\mathcal{B})$ where \ref{c satisfy} does not hold. This results in the absence of a non-negative vector $\mathbf{y}$ such that $\mathcal{C}(\mathcal{A,B})^\top \mathbf{y}=0$ for the corresponding realization of $(\mathcal{A},\mathcal{B})$, thereby contradicting the initial assumption.

For a digraph $\mathcal{G}(\mathcal{A,B})$ realized by structured matrices $(\mathcal{A,B})$, these scenarios arise due to the presence of signed dilation or layer dilation in $\mathcal{G}_s$. 

\begin{enumerate}
    \item [(i)] In case of signed dilation:
 Since the $\mathcal{LUG^{H}}(\mathcal{G}_s)$ spans all the nodes, the unmatched nodes in the dilation set $\Delta_i$ are sign-matched outside $\Delta_i$. Let us consider all the nodes in $\Delta_i^N$ are sign-matched within $\Delta_i$, therefore, every $j \in  \Delta_i^P$ is being sign-matched outside the $\Delta_i$. This implies there is an alternate walk of length $p$ to node $j$ through which it is sign-matched. Then the entry $[\Psi_{p-1}]_j\neq0$ and therefore there exists a realization of $(\mathcal{A},\mathcal{B})$ such that $\mathcal{C}_{(j,p)} \neq -\alpha~\mathcal{C}_{(j,k)}$.

Hence, for all  $j \in \Delta^P_i$ and $k \in \Delta^N_i$, there exists a realization of $(\mathcal{A},\mathcal{B})$ such that $\mathcal{C}_j$ and $\mathcal{C}_k$ are not linearly dependent. Thus $\mathcal{C}_j \ne -\alpha~\mathcal{C}_k~; \alpha>0$

 \item [(ii)] In case of Layer dilation:
    
 The Layer dilation in $L_p$ causes sign changes in the entries of $[\Psi_{p+1}]$ as discussed in Section \ref{Sec:4}. Since $\mathcal{LUG^{H}}(\mathcal{G}_s)$ spans all the nodes of $\mathcal{G(A,B)}$, the unmatched nodes in $L_p$ will be sign-matched in other layers. Hence, those unmatched nodes have more than one nonzero entry in their corresponding rows other than $[\Psi_{p+1}]$. Therefore there exists a realization of $(\mathcal{A},\mathcal{B})$ for which these rows are linearly independent, Hence, $\mathcal{C}_j \ne -\alpha~\mathcal{C}_k~; \alpha>0$
\end{enumerate} 
 From the above cases, it follows that there is no $j,k \in \{2, \dots n\}$ such that $\mathcal{C}_j$ and $\mathcal{C}_k$ are nonnegatively linearly dependent. This implies there exists no $\mathbf{y} \ge \mathbf{0}$ such that $\mathcal{C}(\mathcal{A},\mathcal{B})^\top \mathbf{y} = \mathbf{0}$. By Proposition \ref{test for H}, the digraph is $\mathcal{SS}$ herdable.


\begin{remark}
    In contrast, suppose that no single $\mathcal{LUG}^{\mathcal{H}}(\mathcal{G}_s)$ spans all nodes within $n$ layers. This situation arises when an unmatched node lies in the acyclic component of the graph exhibiting layer dilation, or when an unmatched node in a signed dilation is not matched outside the corresponding signed dilation set. By Propositions~\ref{prop3:prop3} and Lemma~\ref{Prop:SD set}, it then follows that the system~(\ref{sys1:sys1eqn}) is not $\mathcal{SS}$ herdable.
\end{remark}


\textit{Necessity}: Let us consider that the system described by (\ref{sys1:sys1eqn}) is $\mathcal{SS}$ herdable. Then there exists a parametric realization of $(\mathcal{A,B})$ such that $\mathcal{G(A,B)}$ is herdable. This implies there exists a vector $\bm{\delta}$ such that \( \mathcal{C}(\mathcal{A},\mathcal{B}) \, \bm{\delta} = \bm{v} > 0.
\) This can be equivalently written as 
\vspace{-2mm}
\begin{equation}\label{lin-com}
\sum_{j=1}^{n} \delta_j \,[C_{(:,j)}] = v>0.
\end{equation}

where $C_{(:,j)}$ is the $j^{th}$ column of \( \mathcal{C}(\mathcal{A},\mathcal{B})\) and  \eqref{lin-com} represents a linear combination of the columns of $\mathcal{C}(\mathcal{A},\mathcal{B})$. For each $i \in \{1, 2, \dots, n\}$, the above expression arises from one of the following cases:

\begin{enumerate}
    \item [(i)]There exists an index \( x \in \{1, \dots, n\} \) such that
   \(\delta_x\, \mathcal{C}_{(i,x)} >
\sum_{\substack{j=1 \\ j \neq x}}^{n} 
\left| \delta_j\, \mathcal{C}_{(i,j)} \right|.\)
    In this situation, the effect of one component is greater than the combined influence of all others, indicating that the corresponding $v_i$ is positive.

\item  [(ii)] All terms \( \delta_j \, \mathcal{C}_{(i,j)} \) for \( j \in \{1, 2, \dots, n\} \) are positive but not necessarily of equal magnitude. Consequently, \( \sum_{j=1}^{n} \delta_j \, \mathcal{C}_{(i,j)} = v_i > 0, \) indicates a uniform or consistently positive contribution from all components.

\item  [(iii)] Consider all nonzero terms \( \mathcal{C}_{(i,j)} \) for \( j \in \{1, 2, \dots, n\} \).  Partition them into two distinct index sets such that 
\(
y = \{ j \mid \delta_j \, \mathcal{C}_{(i,j)} > 0 \}, \qquad
\bar{y} = \{ k \mid \delta_k \, \mathcal{C}_{(i,k)} < 0 \}.
\)
If both sets are nonempty, and 
\(
\sum_{j \in y} \delta_j \, \mathcal{C}_{(i,j)} > 
 \sum_{k \in \bar{y}} \left| \delta_k \, \mathcal{C}_{(i,k)} \right|,
\) implies that the total positive contribution dominates the total negative contribution. This results in  \(\sum_{j=1}^{n} \delta_j \,[C_{(i,j)}] = v_i>0\)
\end{enumerate}
In all of the above scenarios, for each \( i \in \{1, 2, \dots, n\} \), there exists at least one index \( j \in \{1, 2, \dots, n\} \) such that \( \delta_j \, \mathcal{C}_{(i,j)} > 0\). Therefore, assuming that \(\delta\) can be freely chosen, it is always possible to redesign \(\delta\) so that all cases can be transformed into case  (i).
With case (i) being true for all \( i \in \{1, 2, \dots, n\} \), a vector $\bm{\delta}$ can be chosen such that \eqref{delta satis} holds for some index $x$. 
 Using Lemma~\ref{lem1:lemma1}, the condition in \eqref{delta satis} can be equivalently expressed as follows:
\vspace{-4mm}
\begin{equation}\label{delta satis 2}
\hspace{-1em}
\bm{\delta}_x~ [\mathcal{A}^{x-1}\mathcal{B}]_i >~\sum_{j=1; j\neq x }^{n}| \bm{\delta}_j~ [\mathcal{A}^{j-1}\mathcal{B}]_i~|
~ \left|\raisebox{-0.4em}{\scriptsize$\begin{array}{c}
x \in \{1, \dots n\}
\end{array}$}
\right.
\end{equation}
Every $i^{th}$ entry of  $[\mathcal{A}^{x-1}\mathcal{B}]$ that satisfies (\ref{delta satis 2}) corresponds to a node $i$ in the layer $L_x$ of the signed layered graph $\mathcal{G}_s$ associated with the digraph $\mathcal{G(A,B)}$. 
Consider the set of nodes \( \{\, i \mid \bm{\delta}_x [\mathcal{A}^{x-1}\mathcal{B}]_i > 0 \,\},\) corresponding to paths of length $(x-1)$. This set forms a subgraph of $\mathcal{G}_s$, where $x \in \{1, 2, \dots, n\}$.
 Applying sign matching on the mapped nodes, the resulting subgraph is an $\mathcal{LUG}^{\mathcal{H}}(\mathcal{G}_s)$. 

For a node  \( i \in \{1, 2, \dots, n\} \), the products \( \delta_j \, \mathcal{C}_{(i,j)}, \quad j \in \{1, \dots, n\},\) may all have the same sign. This arises from the sign-matching of node \( i \) in the layer \( L_x \) of \( \mathcal{LUG}^{\mathcal{H}}(\mathcal{G}_s) \) within \( \mathcal{G}_s \). It can be observed that the occurrence of cases (i) and (ii) described above arises due to repeated sign matching of node \( i \) across different layers of \( \mathcal{LUG}^{\mathcal{H}}(\mathcal{G}_s) \).

 Hence, there exists an $\mathcal{LUG}^{\mathcal{H}}(\mathcal{G}_s)$ in the $\mathcal{G}_s$ associated with $\mathcal{G(A,B)}$ that is $\mathcal{SS}$ herdable. \hfill$\qed$

\end{proof}
\section{Proof of Theorem~\ref{Thm multi Driver}}\label{proof Thm multi Driver}
\text{Sufficiency:} 
Let $\mathcal{G}(\mathcal{A,B})$ be a digraph whose signed layered graph $\mathcal{G}_s$ contains a set of $\mathcal{LUG}_i^{\mathcal{H}}(\mathcal{G}_s)$, where $i$ is the leader node associated with corresponding $\mathcal{LUG}_i^{\mathcal{H}}(\mathcal{G}_s)$. 
Let us assume that all the nodes of $\mathcal{G}(\mathcal{A,B})$ are spanned by $~\bigcup_{i \in \mathbb{Z}^+}~ \mathcal{LUG}_i^{\mathcal{H}}(\mathcal{G}_s)$. This imples that for every node $j$ there exists atleast one $\mathcal{LUG}_i^{\mathcal{H}}(\mathcal{G}_s)$ such that $j \in L_p$ where $L_p$ is the $p^{th}$ layer of $\mathcal{LUG}_i^{\mathcal{H}}(\mathcal{G}_s)$. Then the following statements are true:
\begin{enumerate}
    \item[(i)] For all $i\in\{1,2..n\}$, the $i^{th}$ row of $\mathcal{C}(\mathcal{A},\mathcal{B})$ , i.e  $\mathcal{C}_{(i,:)}$ is nonzero.
    \item [(ii)] For all the sign-matched nodes in $L_{p+1}$ in the $\mathcal{LUG}_\ell^{\mathcal{H}}(\mathcal{G}_s)$ originating from leader node $\ell$ their corresponding entries in  \([\mathcal{A}^p \mathcal{B}_\ell]\) = \( [\Psi^\ell_p]\) are unisigned. 
    \item[(iii)] Other than the entries corresponding to sign-matched nodes, there can be other nonzero entries in $[\Psi^\ell_k]$. They are matched in the other columns of $ \mathcal{C}(\mathcal{A},\mathcal{B})$. 
    \item[(iv)] Applying Lemma \ref{lemma C herd} will ensure that all (+/-) entries in $ \mathcal{C}(\mathcal{A},\mathcal{B})$ are realized to the sign for which the corresponding node is sign-matched.
 \end{enumerate}

It follows from Proposition \ref{test for H}, that if for every realization of $(\mathcal{A,B})$ there exists a nonnegative vector $\mathbf{y}\in \mathbb{R}^{nm}$ such that $\mathcal{C}(\mathcal{A,B})^\top \mathbf{y}=0$, then \eqref{sys1:sys1eqn} is not $\mathcal{SS}$ herdable. Let $\mathcal{C}_1,\mathcal{C}_2 \ldots, \mathcal{C}_n$ denote the rows of $ \mathcal{C}(\mathcal{A},\mathcal{B})$, which are equivalently the column vectors of $ \mathcal{C}(\mathcal{A},\mathcal{B})^\top$, where $ \mathcal{C}(\mathcal{A},\mathcal{B})^\top \in \mathbb{R}^{nm \times n}$. The existence of $\mathbf{y}\geq0$ that satisfies the above expression implies that there exist some rows in $\mathcal{C}(\mathcal{A},\mathcal{B})$ that are nonnegatively linearly dependent, irrespective of any realization of $(\mathcal{A,B})$. 

Let us assume that for every realization of $(\mathcal{A},\mathcal{B})$ there exists a $\mathbf{y}\geq0$ such that $\mathcal{C}(\mathcal{A,B})^\top \mathbf{y}=0$. This implies there exists some $j,k~\in~\{1,2 \dots n\}$ such that
\begin{align}\label{c mul satisfy}
     \mathcal{C}_j = -\alpha ~ \mathcal{C}_k~;~ \alpha \in\mathbb{R}^+
     \end{align}
The expression \([\Psi^\ell_p]_j= - \alpha ~[\Psi^\ell_p]_k~~; \forall ~p ~\in \{1,2, \dots n-1\}\) implies that for any non-zero $[\Psi^\ell_p]_j$ and $[\Psi^\ell_p]_k$, nodes $j$ and $k$ are part of the same layer. Given that all nodes are spanned by union of $\mathcal{LUG_\ell^{H}}(\mathcal{G}_s), \ell\in \{1,2 \dots m\}$, every nodes of $\mathcal{G}(\mathcal{A,B})$ is sign-matched by an $\mathcal{LUG^{H}}(\mathcal{G}_s)$ originating from a leader $\ell$. This leads to the following possible scenarios.
\begin{enumerate}
    \item [($i$)] If node $j$ and $k$ are sign-matched at the same layer $L_p$, then $\text{sign}~[\Psi^\ell_p]_j= \text{sign}~[\Psi^\ell_p]_k$. This is in contradiction with the relation $[\Psi^\ell_p]_j= - \alpha ~[\Psi^\ell_p]_k$, given that $\alpha~\in \mathbb{R}^+$; consequently, condition~\eqref{c satisfy} fails to hold.
     \item [($ii$)] If the node $j$ is sign-matched at $L_p$ and node $k$ is sign-matched at a $L_{q \ne p}$, then there exists a walk $\mathcal{W}_{(\ell,k)}$ of length $q-1$ and the entry $[\Psi^\ell_q]_k \ne 0$  in $\mathcal{C}_k$. This leads to  $\mathcal{C}_j \ne \beta ~ \mathcal{C}_k~;~ \beta \in\mathbb{R}$ for some realization of $(\mathcal{A},\mathcal{B})$ given by Lemma \ref{lemma C herd}.
\end{enumerate}
From the previous cases, it can be inferred that there is a realization of $(\mathcal{A},\mathcal{B})$ where \eqref{c mul satisfy} does not hold. This results in the absence of a non-negative vector $\mathbf{y}$ such that $\mathcal{C}(\mathcal{A,B})^\top \mathbf{y}=0$ for the corresponding realization of $(\mathcal{A},\mathcal{B})$, thereby contradicting the initial assumption.
This implies there exists no $\mathbf{y} \ge \mathbf{0}$ such that $\mathcal{C}(\mathcal{A},\mathcal{B})^\top \mathbf{y} = \mathbf{0}$. By Proposition \ref{test for H}, the digraph is $\mathcal{SS}$ herdable.

\textit{Necessity}: Let us consider a system described by (\ref{sys1:sys1eqn}) having $m$ driver nodes, is $\mathcal{SS}$ herdable. Then there exists a parametric realization of $(\mathcal{A,B})$ such that $\mathcal{G(A,B)}$ is herdable. This implies there exists a vector $\bm{\delta \in \mathbb{R}^{nm}}$ such that \( \mathcal{C}(\mathcal{A},\mathcal{B}) \, \bm{\delta} = \bm{v} > 0.
\) Each component of vector $v_i$ can be equivalently written as 
\vspace{-2mm}
\begin{equation}\label{multi lin-com}
\sum_{j=1}^{nm} \delta_j \,[C_{(i,j)}] = v_i>0, \quad \forall i \in \{1, 2, \dots, n\}.
\end{equation}

where \eqref{multi lin-com} indicated that $v$ is a linear combination of the columns of $\mathcal{C}(\mathcal{A},\mathcal{B})$. For each $i \in \{1, 2, \dots, n\}$, the above expression arises from one of the following cases:

\begin{enumerate}
    \item [(i)]There exists an index \( x \in \{1, \dots, nm\} \) such that
    \(
    \delta_x \, \mathcal{C}_{(i,x)} >  
    \sum_{\substack{j=1 \\ j \neq x}}^{nm} \left|\delta_j \, \mathcal{C}_{(i,j)}
    \right|.
    \)
    In this situation, the effect of one component is greater than the combined influence of all others, indicating that the corresponding $v_i$ is positive.

\item  [(ii)] All terms \( \delta_j \, \mathcal{C}_{(i,j)} \) for \( j \in \{1, 2, \dots, nm\} \) are positive but not necessarily of equal magnitude. Consequently, \( \sum_{j=1}^{nm} \delta_j \, \mathcal{C}_{(i,j)} = v_i > 0, \) indicates a uniform or consistently positive contribution from all components.

\item  [(iii)] Consider all nonzero terms \( \mathcal{C}_{(i,j)} \) for \( j \in \{1, 2, \dots, nm\} \). Partition them into two distinct index sets such that 
\(
y = \{ j \mid \delta_j \, \mathcal{C}_{(i,j)} > 0 \}~\text{and}~
\bar{y} = \{ k \mid \delta_k \, \mathcal{C}_{(i,k)} < 0 \}.
\)
If
\(
\sum_{j \in y} \delta_j \, \mathcal{C}_{(i,j)} > 
 \sum_{k \in \bar{y}} \left|\delta_k \, \mathcal{C}_{(i,k)} \right|,
\) then it implies that the total positive contribution dominates the total negative contribution. This results in  \(\sum_{j=1}^{nm} \delta_j \,[C_{(i,j)}] = v_i>0\).
\end{enumerate}
In all of the above scenarios, for each \( i \in \{1, 2, \dots, n\} \), there exists at least one index \( j \in \{1, 2, \dots, mn\} \) such that \( \delta_j \, \mathcal{C}_{(i,j)} > 0\). Therefore, with the assumption that \eqref{sys1:sys1eqn} is herdable and \(\delta\) can be freely chosen, it is always possible to redesign \(\delta\) so that all cases can be transformed into case  (i).
With case (i) being true for all \( i \in \{1, 2, \dots, n\} \), a vector $\bm{\delta}$ can be chosen such that for some index $x$, the following expression holds:
\begin{equation}\label{multi delta satis}
    \delta_x~ \mathcal{C}_{(i,x)} >\sum_{j=1; j\neq x }^{nm} |\delta_j~ \mathcal{C}_{(i,j)}~|
~~ \left |\raisebox{-0.4em}{\scriptsize$\begin{array}{c}
x \in \{1, \dots nm\}
\end{array}$}
\right.
\end{equation}

 It follows from Remark \ref{Remark multi} and Lemma \ref{multi lemma} that \(\mathcal{C}_{(i,x)}\) is given by \([\mathcal{A}^k \mathcal{B}_\ell]_i\) = \( [\Psi^\ell_k]_i\), where \(k = \left\lfloor~\frac{x-1}{m}~\right\rfloor\), \( \ell=\big((x-1) \bmod m \big)+1\) and $m$ is the number of inputs.
 
 Therefore, the condition in \eqref{multi delta  satis} can be equivalently expressed as follows:
\vspace{-4mm}
\begin{equation}\label{dominance_condition}
 \delta_g [\mathcal{A}^k \mathcal{B}_\ell]_i 
>
\sum_{\substack{j = 0,\dots,n-1 \\ \ell' = 1,\dots,m \\ (j,\ell') \neq (k,\ell)}} 
\big| \delta_{\bar{g}} [\mathcal{A}^j \mathcal{B}_{\ell'}]_i ~~
\big|
\quad
\substack{
k = \{0, \dots, n-1\} \\
\ell = \{1, \dots, m\}
}
\end{equation}
where \( \delta_g \) and \([\mathcal{A}^k \mathcal{B}_\ell]\) has the same sign, $g=(k\cdot m+\ell)$ and $\bar{g}=(j\cdot m+\ell')$.

Every $i^{th}$ entry of  $[\mathcal{A}^k \mathcal{B}_\ell]$ that satisfies (\ref{dominance_condition}) corresponds to a node $i$ in the layer $L_{k+1}$ of the signed layered graph $\mathcal{G}_s$. For each $\delta_{{g}}[\mathcal{A}^k \mathcal{B}_\ell]_i >0 $ the node $i$ is mapped in the layer $L_{k+1}$ of $\mathcal{G}_s$ that has walk from leader $\ell$ with length $k$. Applying sign matching on the mapped nodes, the resulting subgraph is a set of $\mathcal{LUG}_\ell^{\mathcal{H}}(\mathcal{G}_s)$ having node $\ell$ as the leader and their  $~\bigcup_{i \in \mathbb{Z}^+}~ \mathcal{LUG}_i^{\mathcal{H}}(\mathcal{G}_s)$ spans all the nodes of $\mathcal{G(A,B)}$.
 \hfill$\qed$

\bibliographystyle{unsrt}        
\bibliography{autosam}           
\end{document}